\documentclass{lmcs}
\pdfoutput=1

\usepackage{lastpage}
\lmcsdoi{21}{1}{1}
\lmcsheading{}{\pageref{LastPage}}{}{}%
{Jan.~16,~2023}{Jan.~03,~2025}{}

\usepackage[utf8]{inputenc}
\usepackage{xr-hyper}

\usepackage{hyperref}

\usepackage{booktabs}   
\usepackage{subcaption} 

\usepackage{amsmath}
\usepackage{amsthm}
\usepackage{amsfonts}
\usepackage{amssymb}
\usepackage{cmll}

\usepackage{tikz}
\usetikzlibrary{matrix}
\usetikzlibrary{cd}
\usepackage{pgfplots}
\usepackage{stmaryrd}
\usepackage{ebproof}
\usepackage{url}

\newcommand\CMLLPAR{
\usepackage{cmll}
\newcommand\IPar{\mathord{\parr}}
}

\CMLLPAR
\newcommand\Etc{\textit{etc.}}
\newcommand\Ie{\textit{i.e.}}
\newcommand\Eg{\textit{e.g.}}

\makeatletter
\newcommand*{\inlineeq}[2][]{%
  \begingroup
    \refstepcounter{equation}%
    \ifx\\#1\\%
    \else
      \label{#1}%
    \fi
    \relpenalty=10000 %
    \binoppenalty=10000 %
    \ensuremath{%
      #2%
    }%
    ~\@eqnnum
  \endgroup
}
\makeatother

\newenvironment{Axicond}[1]
{\smallbreak\noindent{#1}\,}
{\smallbreak}

\newcommand\Proofcase{\smallbreak\noindent$\blacktriangleright$\ }

\newcommand{\Endproof}{
  \ifmmode 
  \else \leavevmode\unskip\penalty9999 \hbox{}\nobreak\hfill
  \fi
  \quad\hbox{$\Box$}
  \par\medskip}

\newcommand\Eqref[1]{(\ref{#1})}

\renewcommand{\phi}{\varphi}
\renewcommand\epsilon{\varepsilon}

\newcommand{\Implies}{\Rightarrow}

\newcommand\Equiv{\Leftrightarrow}
\newcommand{\St}{\mid}

\newcommand{\Sbot}{{\mathord{\perp}}}
\newcommand{\Top}{\top}

\newcommand\Seqempty{\Tuple{}}

\newcommand\cC{\mathcal{C}}

\newcommand\cF{\mathcal{F}}

\newcommand\cI{\mathcal{I}}

\newcommand\cK{\mathcal{K}}

\newcommand\cM{\mathcal{M}}
\newcommand\cN{\mathcal{N}}
\newcommand\cO{\mathcal{O}}
\newcommand\cP{\mathcal{P}}
\newcommand\cQ{\mathcal{Q}}
\newcommand\cR{\mathcal{R}}
\newcommand\cS{\mathcal{S}}
\newcommand\cT{\mathcal{T}}

\newcommand\cX{\mathcal{X}}
\newcommand\cY{\mathcal{Y}}
\newcommand\cZ{\mathcal{Z}}

\newcommand\Fini{{\mathrm{fin}}}

\newcommand\Part[1]{{\mathcal P}\left({#1}\right)}

\newcommand\Union{\bigcup}

\newcommand{\Linarrow}{\multimap}

\newcommand\Myleft{}
\newcommand\Myright{}

\newcommand\Web[1]{\Myleft|{#1}\Myright|}

\newcommand\Supp[1]{\operatorname{\mathsf{supp}}({#1})}

\newcommand\Emptytuple{(\,)}

\newcommand\ITens{\mathop\otimes}
\newcommand\Tens[2]{{#1}\ITens{#2}}
\newcommand\Tensp[2]{({#1}\ITens{#2})}

\newcommand\IWith{\mathrel{\&}}
\newcommand\With[2]{{#1}\IWith{#2}}

\newcommand\Withp[2]{\left({#1}\IWith{#2}\right)}
\newcommand\IPlus{\oplus}
\newcommand\Plus[2]{{#1}\IPlus{#2}}
\newcommand\Orth[2][]{#2^{\mathord\perp_{#1}}}

\newcommand\Bwith{\mathop{\&}}

\newcommand\Bplus{\mathop\oplus}

\newcommand\Mtinj[2]{\mathsf{in}_{#1}(#2)}

\newcommand\Biorth[1]{#1^{\bot\bot}}

\newcommand\Triorth[1]{{#1}^{\bot\bot\bot}}

\newcommand\One{1}

\newcommand\LL{\hbox{\textsf{LL}}}

\newcommand\Card[1]{\#{#1}}

\newcommand\Locun[1]{1^J}

\newcommand\Isom\simeq

\newcommand\Comp{\mathrel\circ}

\newcommand\Funinv[1]{{#1}^{-1}}

\newcommand\SET{\mathbf{Set}}

\newcommand\Limpl[2]{{#1}\Linarrow{#2}}
\newcommand\Limplp[2]{\left({#1}\Linarrow{#2}\right)}

\newcommand\Nat{{\mathbb{N}}}

\newcommand\Natnz{{\Nat^+}}

\newcommand\Biind[2]{\genfrac{}{}{0pt}{1}{#1}{#2}}

\newcommand\Snat{\mathsf N}

\newcommand\Zero{0}

\newcommand\List[3]{#1_{#2},\dots,#1_{#3}}

\newcommand\Kronecker[2]{\delta_{{#1},{#2}}}

\newcommand\Subst[3]{{#1}\left[{#2}/{#3}\right]}

\newcommand\Substbis[2]{{#1}[{#2}]}

\newcommand\Factor[1]{{#1}!}
\newcommand\Binom[2]{\genfrac{(}{)}{0pt}{}{#1}{#2}}

\newcommand\Real{\mathbb{R}}
\newcommand\Realp{\mathbb{R}_{\geq 0}}
\newcommand\Realpto[1]{(\Realp)^{#1}}

\newcommand\Realpc{\overline{\Realp}}
\newcommand\Realpcto[1]{\Realpc^{#1}}

\newcommand\Intercc[2]{[#1,#2]}

\newcommand\Interco[2]{[#1,#2)}

\newcommand\Mfin[1]{\mathcal M_\Fini({#1})}

\newcommand\Ev{\operatorname{\mathsf{Ev}}}

\newcommand\Norm[1]{\left\|{#1}\right\|}

\newcommand\Redst[1]{\mathop{\mathsf{Red}}}

\newcommand\Symgrp[1]{\mathfrak S_{#1}}

\newcommand\Tuple[1]{\langle{#1}\rangle}
\newcommand\Cotuple[1]{\left[{#1}\right]}

\newcommand\Msetofsubst[1]{\bar F}

\newcommand\Inv[1]{#1^{-1}}
\newcommand\Invp[1]{({#1})^{-1}}

\newcommand\Pcoh[1]{\mathsf P{#1}}
\newcommand\Pcohp[1]{\Pcoh{(#1)}}
\newcommand\Pcohcn{\mathsf{ic}}
\newcommand\Pcohc[1]{\Pcohcn(#1)}

\newcommand\Base[1]{\mathsf e(#1)}

\newcommand\PCOH{\mathbf{Pcoh}}

\newcommand\Leftu{\lambda}
\newcommand\Rightu{\rho}
\newcommand\Assoc{\alpha}
\newcommand\Sym{\gamma}

\newcommand\Retri\zeta
\newcommand\Retrp\rho

\newcommand\Impl[2]{{#1}\Rightarrow{#2}}
\newcommand\Implp[2]{({#1}\Rightarrow{#2})}

\newcommand\Tsem[1]{\llbracket{#1}\rrbracket}

\newcommand\Psem[2]{\llbracket{#1}\rrbracket_{#2}}

\newcommand\Tnat\iota

\newcommand\Num[1]{\underline{#1}}
\newcommand\Loop\Omega

\newcommand\Tseq[3]{{#1}\vdash{#2}:{#3}}

\newcommand\Timpl\Impl
\newcommand\Timplp\Implp

\newcommand\Simpls[2]{#1\Rightarrow_{\mathsf s}#2}
\newcommand\Simpla[2]{#1\Rightarrow_{\mathsf a}#2}

\newcommand\Der[1]{\operatorname{\mathsf{der}}_{#1}}

\newcommand\Digg[1]{\operatorname{\mathsf{dig}}_{#1}}

\newcommand\Lfunn{\operatorname{\mathsf{fun}}}
\newcommand\Lfun[1]{\operatorname{\Lfunn}(#1)}
\newcommand\Fun[1]{\widehat{#1}}

\newcommand\Id{\operatorname{\mathsf{Id}}}

\newcommand\Proj[1]{\mathsf{pr}_{#1}}
\newcommand\Inj[1]{\mathsf{in}_{#1}}

\newcommand\Excl[1]{\oc{#1}}

\newcommand\Exclp[1]{\oc({#1})}

\newcommand\Prom[1]{#1^!}
\newcommand\Proms[1]{#1^{\mathord\oc_{\mathsf s}}}
\newcommand\Proma[1]{#1^{\mathord\oc_{\mathsf a}}}
\newcommand\Promana[1]{#1^{\mathord\oc_{\mathsf a}}}
\newcommand\Promms[1]{#1^{\mathord\oc_{\mathsf s}\mathord\oc_{\mathsf s}}}
\newcommand\Promma[1]{#1^{\mathord\oc_{\mathsf a}\mathord\oc_{\mathsf a}}}

\newcommand\Kleisli[2]{{#1}_{#2}}

\newcommand\Relincl\eta
\newcommand\Relrestr\rho

\newcommand\Seely{\mathsf m}
\newcommand\Seelyz{\Seely^0}
\newcommand\Seelyt{\Seely^2}

\newcommand\Compl{\,}
\newcommand\Curlin{\operatorname{\mathsf{cur}}}

\newcommand\Op[1]{{#1}^{\mathsf{op}}}
\newcommand\Kl[1]{{#1}_\oc}
\newcommand\Em[1]{{#1}^\oc}

\newcommand\Eval[2]{\langle#1,#2\rangle}

\newcommand\Vect[1]{\overrightarrow{#1}}

\newcommand\Bnfeq{\mathrel{\mathord:\mathord=}}
\newcommand\Bnfor{\,\,\mathord|\,\,}

\newcommand\Cuball[1]{\mathcal B#1}
\newcommand\Cuballp[1]{\mathcal B(#1)}

\newcommand\Cantor{\cC}

\newcommand\Eset[1]{\{#1\}}

\newcommand\Matappa[2]{{#1}\cdot{#2}}

\newcommand\Sone{\One}
\newcommand\Sonelem{\ast}

\newcommand\Intcc[2]{[#1,#2]}

\newcommand\Adj{\mathrel{\dashv}}

\newcommand\Textsep{\hspace{4em}}

\newcommand\Stop{\Top}

\newcommand{\xref}{\nameref}
\makeatletter
\newcommand\xlabel[2][]{\phantomsection\def\@currentlabelname{#1}\label{#2}}
\makeatother

\newcommand\Pcsimpll{(\textbf{Cancel})\xlabel[(\textbf{Cancel})]{ax:c1}}
\newcommand\Pcsimplr{\xref{ax:c1}}

\newcommand\Pcposl{(\textbf{Pos})\xlabel[(\textbf{Pos})]{ax:c2}}
\newcommand\Pcposr{\xref{ax:c2}}

\newcommand\Cnormzl{(\textbf{Normz})\xlabel[(\textbf{Normz})]{ax:c3}}
\newcommand\Cnormzr{\xref{ax:c3}}

\newcommand\Cnormhl{(\textbf{Normh})\xlabel[(\textbf{Normh})]{ax:c4}}
\newcommand\Cnormhr{\xref{ax:c4}}

\newcommand\Cnormtl{(\textbf{Normt})\xlabel[(\textbf{Normt})]{ax:c5}}
\newcommand\Cnormtr{\xref{ax:c5}}

\newcommand\Cnormpl{(\textbf{Normp})\xlabel[(\textbf{Normp})]{ax:c6}}
\newcommand\Cnormpr{\xref{ax:c6}}

\newcommand\Cnormcl{(\textbf{Normc})\xlabel[(\textbf{Normc})]{ax:c7}}
\newcommand\Cnormcr{\xref{ax:c7}}

\newcommand\Msmesl{(\textbf{Msmeas})\xlabel[(\textbf{Msmeas})]{ax:c8}}
\newcommand\Msmesr{\xref{ax:c8}}

\newcommand\Mscompl{(\textbf{Mscomp})\xlabel[(\textbf{Mscomp})]{ax:c9}}
\newcommand\Mscompr{\xref{ax:c9}}

\newcommand\Mssepl{(\textbf{Mssep})\xlabel[(\textbf{Mssep})]{ax:c10}}
\newcommand\Mssepr{\xref{ax:c10}}

\newcommand\Msnorml{(\textbf{Msnorm})\xlabel[(\textbf{Msnorm})]{ax:c11}}
\newcommand\Msnormr{\xref{ax:c11}}

\newcommand\ARCAT{\mathbf{Ar}}

\newcommand\MEAS{\mathbf{Meas}}

\newcommand\Measterm{0}

\newcommand\Cdual[1]{{#1}'}

\newcommand\Absm[2]{\boldsymbol\lambda #1\cdot#2}

\newcommand\Mcca[1]{\underline{#1}}
\newcommand\Mcms[1]{\cM^{#1}}

\newcommand\CONES{\mathbf{Cones}}
\newcommand\MCONES{\mathbf{MCones}}
\newcommand\ICONES{\mathbf{ICones}}

\newcommand\Mtpath[2]{#1\mathrel\triangleright #2}
\newcommand\Mtlfun[2]{#1\mathrel\triangleright #2}
\newcommand\Mtfun[2]{#1\mathrel\triangleright #2}
\newcommand\Mtlfunm[2]{#1\mathrel\triangleright #2}

\newcommand\Cpath[2]{\mathsf{Path}(#1,#2)}
\newcommand\Cpathm[2]{\mathsf{Path}(#1,#2)}
\newcommand\Cpathf{\mathsf{Path}}

\newcommand\Pt[1]{\widehat{#1}}

\newcommand\Limplm{\Limpl}

\newcommand\Tensor{\tau}

\newcommand\Dirac[1]{\boldsymbol\delta^{#1}}

\newcommand\Charfun[1]{\chi_{#1}}

\newcommand\Cmeas{\mathsf{FMeas}}
\newcommand\Cmeast{\Cmeas^{\mathord\otimes}}
\newcommand\Sigalg[1]{\sigma_{#1}}
\newcommand\Emeas[1]{\widetilde{#1}}

\newcommand\Mcint[1]{\cI^{#1}}
\newcommand\Mcinti[1]{\cK^{#1}}
\newcommand\Pushf[1]{{#1}_\ast}

\newcommand\Measprod[2]{{#1}\times{#2}}
\newcommand\Flpath{\mathsf{fl}}

\newcommand\Mcofic[1]{{#1}}
\newcommand\Kernto{\leadsto}

\newcommand\Limpli[2]{#1\Linarrow#2}
\newcommand\Swlinpath{\mathsf{sw}}
\newcommand\Swlinlin{\mathsf{sw}}

\newcommand\Cloc[2]{{#1}_{#2}}

\newcommand\Cat[1]{\mathbf{#1}}

\newcommand\Npset[1]{\cP^-(#1)}
\newcommand\Ppset[1]{\cP^+(#1)}
\newcommand\Epset[2]{\cP^{#1}(#2)}

\newcommand\STAB{\mathbf{SCones}}
\newcommand\STA{\mathbf{SCones}}
\newcommand\SCONES{\STAB}
\newcommand\ANA{\mathbf{ACones}}
\newcommand\ACONES{\ANA}

\newcommand\Mlsym[3]{\mathbf{Sym}_{#1}(#2,#3)}
\newcommand\Monom[3]{\underline{\mathbf{Hpol}}_{#1}(#2,#3)}
\newcommand\Monomic[3]{\mathbf{Hpol}_{#1}(#2,#3)}

\newcommand\Linan[1]{\mathsf D^{(#1)}_0}
\newcommand\Linhp[1]{\mathsf L_{#1}}
\newcommand\Hpan[1]{\mathsf P_{#1}}
\newcommand\Mlmon[1]{\mathsf M_{#1}}
\newcommand\Intset[1]{[#1]}
\newcommand\Intfun[2]{\mathsf L(#1,#2)}

\newcommand\Derfun{\operatorname{\mathsf{Der}}}
\newcommand\Derfuns{\operatorname{\mathsf{Der}^{\mathsf s}}}
\newcommand\Derfuna{\operatorname{\mathsf{Der}^{\mathsf a}}}
\newcommand\Estab{\operatorname{\mathsf E}^{\mathsf s}}
\newcommand\Eana{\operatorname{\mathsf E}^{\mathsf a}}

\newcommand\Excls{\oc^{\mathsf s}}
\newcommand\Exclls{\oc^{\mathsf s}\oc^{\mathsf s}}
\newcommand\Excla{\oc^{\mathsf a}}
\newcommand\Exclana{\oc^{\mathsf a}}
\newcommand\Ders{\mathsf{der}^{\mathsf s}}
\newcommand\Diggs{\mathsf{dig}^{\mathsf s}}

\newcommand\Expadjs{\Theta^{\mathsf s}}
\newcommand\Expadja{\Theta^{\mathsf a}}
\newcommand\Unistab{\mathsf{st}}

\newcommand\Coalgs[1]{\mathsf h_{#1}}
\newcommand\Coalga[1]{\mathsf h^{\mathsf a}_{#1}}

\newcommand\ContinuousPart{\mathsf{cont}}

\newcommand\Rep[2]{\overline{#1}^{{#2}}}

\newcommand\Evreal[2]{{#1}\mid_{#2}}

\newcommand\Fdiffp[2]{\Delta^{+}#1(#2)}
\newcommand\Fdiffn[2]{\Delta^{-}#1(#2)}
\newcommand\Fdiffs[3]{\Delta^{#1}#2(#3)}
\newcommand\Fdiff[2]{\Delta#1(#2)}
\newcommand\Fdiffvar[3]{\Fdiff{#1}{#3}(#2)}

\newcommand\Csum[2]{\mathsf S^{#1}#2}

\newcommand\Inset[1]{\mathsf{inj}_{#1}}

\newcommand\SKERN{\mathbf{Skern}}

\newcommand\Funpushf{\Cmeas}

\newcommand\Tenspcs[2]{{#1}\mathrel{\overline{\mathord\otimes}}{#2}}

\newcommand\Cantbase[1]{\mathord\uparrow#1}
\newcommand\Cantopenbase[1]{\mathord\downarrow#1}

\newcommand\Opens[1]{\cO(#1)}
\newcommand\Cantormeas{\mathsf{meas}}
\newcommand\Cantorvect{\mathsf{rep}}

\newcommand\Sklin{\mathsf{Klin}}

\newcommand\Transfwd[1]{{#1}_\ast}

\newcommand\Pl{\mathord+}
\newcommand\Mn{\mathord-}

\newcommand\Ttreeo{\ast}
\newcommand\Ttreeb[2]{\langle#1,#2\rangle}
\newcommand\Ttrees[1]{\cT_{#1}}
\newcommand\Ttreet[1]{#1^{\otimes}}
\newenvironment{theorem}{\begin{thm}}{\end{thm}}
\newenvironment{lemma}{\begin{lem}}{\end{lem}}

\newenvironment{proposition}{\begin{prop}}{\end{prop}}

\newenvironment{definition}{\begin{defi}}{\end{defi}}

\newenvironment{remark}{\begin{rem}}{\end{rem}}
\newenvironment{example}{\begin{exa}}{\end{exa}}

\newcounter{examplectr}

\numberwithin{examplectr}{section}

\pgfplotsset{compat=1.18}

\begin{document}

\title{Integration in cones}
\author[T.~Ehrhard]{Thomas Ehrhard\lmcsorcid{0000-0001-5231-5504}}[a]
\author[G.~Geoffroy]{Guillaume Geoffroy\lmcsorcid{0009-0005-7102-3378}}[b]
\address{Université Paris Cité, CNRS, Inria, IRIF, F-75013, Paris, France}
\email{ehrhard@irif.fr}
\address{Université Paris Cité, CNRS, IRIF, F-75013, Paris, France}
\email{guillaume.geoffroy@irif.fr}

\begin{abstract}
 Measurable cones, with linear and measurable functions as morphisms,
are a model of intuitionistic linear logic and of call-by-name
probabilistic PCF which accommodates ``continuous data types'' such as
the real line. So far however, they lacked a major feature to make
them a model of more general probabilistic programming languages
(notably call-by-value and call-by-push-value languages): a theory of
integration for functions whose codomain is a cone, which is the key
ingredient for interpreting the sampling programming primitives. The
goal of this paper is to develop such a theory: our definition of
integrals is an adaptation to cones of Pettis integrals in topological
vector spaces. We prove that such integrable cones, with
integral-preserving linear maps as morphisms, form a model of Linear
Logic for which we develop two exponential comonads: the first based
on a notion of stable and measurable functions introduced in earlier
work and the second based on a new notion of integrable analytic
function on cones.


\end{abstract}

\maketitle

\section*{Introduction}

There are several approaches in the denotational semantics of
functional probabilistic programming languages that we can summarize
as follows:
\begin{itemize}
\item quasi-Borel spaces (QBSs)~\cite{KammarStatonVakar19} which are,
  roughly speaking, separated presheaves on the cartesian category of
  measurable spaces and measurable functions (or on a full cartesian
  sub-category thereof),
  and the considered category of QBSs must be given together with a
  well behaved probability monad (\emph{à la} Giry);
\item probabilistic games~\cite{DanosHarmer00} which are similar to
  deterministic games apart that now strategies are probability
  distributions on plays;
\item models based on categories of domains, possibly equipped with a
  probabilistic monad, and where morphisms are Scott continuous functions;
\item probabilistic coherence spaces~\cite{DanosEhrhard08} (PCSs) which
  are a refinement of the relational model of Linear Logic (\(\LL\)).
  In the PCS model, an object is a set equipped with a collection of
  ``valuations'', which are functions%
  \footnote{For objects corresponding to ground types, these
    valuations are the subprobability distributions.} %
  from this set to \(\Realp\), and a morphism is a linear functions on
  these valuations, or analytic functions in the CCC used for
  interpreting the programming languages.
  This approach can be understood as extending to higher types the
  basic idea of~\cite{Kozen81} which is to interpret programs as
  probability distribution transformers.
\end{itemize}

\paragraph{Main motivation.}
Modern probabilistic programming languages deal with probability
distributions on continuous data-types such as the real line, and PCSs
are not able to represent such types: PCSs are fundamentally of a
discrete nature.
On the other hand, QBS-based models accept continuous data-types by
construction, and give rise to cartesian closed categories for a very
general reason ---~they are essentially categories of
presheaves~---.
This also means that these models are not very informative about
morphisms: they are essentially only required to satisfy a hereditary
measurability condition and, accordingly, they have in general no
clear underlying linear structure (in the sense of the categorical
semantics of Linear Logic).
The benefit of such a linear structure is that it allows to take into
account in a modular way the various options in the design of a
programming language, and in particular the choice of operational
semantics (call-by-name or call-by-value, typically).
Also the linear structure provides tools ---~versions of the Taylor
expansion of analytic functions~--- allowing to analyze the resource
usage of programs.
In contrast to QBSs, PCSs are natively a model of \(\LL\) whose associated
cartesian closed category can be used as a model of probabilistic
functional languages.
In this CCC the morphisms are quite regular: they are analytic
functions described by generalized power series with nonnegative
coefficients.
This feature allowed the first author to prove, for instance, two full
abstraction results~\cite{EhrhardPaganiTasson18b,EhrhardTasson19}
wrt.~the PCS semantics.

The main purpose of the model presented in this paper is to extend to
the continuous probability setting these two main features of PCSs: the
model has a linear underlying structure and the programs are
interpreted as functions which are analytic in some generalized sense.
One essential feature of our semantics is that a functional program
\(M\) of type \(\rho\) (the type of real numbers) with only one variable
\(x\) of type \(\rho\) will be interpreted as a function \(f\) from
the set \(\cR\) of subprobability measures on \(\Real\) to \(\cR\).
With this intuition in mind, it is easier to understand what linearity
can mean for such a function (very roughly: commutation with existing
linear combinations of measures), and also what analyticity can mean:
the function \(g\) which maps a subprobability distribution \(\mu\) on
\(\Real\) to \(\mu\ast\mu\) (convolution product of measures) is
clearly not linear, but it is polynomial of degree \(2\).
More precisely, the addition program on real numbers will typically be
represented as a function \(a:\cR\times\cR\to\cR\) which will be
bilinear: it maps a pair \((\mu,\nu)\in\cR^2\) of subprobability
measures to \(\Pushf\alpha(\mu\times\nu)\) where
\(\alpha:\Real\times\Real\to\Real\) is the addition function,
\(\mu\times\nu\) is the usual product of \(\mu\) and \(\nu\), which is
a subprobability measure on \(\Real\times\Real\), and \(\Pushf\alpha\)
is the push-forward operation on measures associated with the
measurable function \(\alpha\).
The function \(g\) is polynomial of degree \(2\) because
\(g(\mu)=a(\mu,\mu)\).

\paragraph{Types as cones.}
In recent works~\cite{EhrhardPaganiTasson18,Ehrhard20} we have
developed such a continuous extension of the PCS semantics, using
quite a suitable notion of \emph{positive cone} introduced by Selinger
in~\cite{Selinger04} (we will often drop the adjective ``positive'').
Cones are similar to real Banach spaces, with the difference that, in
a cone, ``everything is positive''; for instance the coefficients are
taken in \(\Realp\) and not in \(\Real\) and \(x+y=0\) is possible
only if \(x=y=0\).
For that reason cones are naturally ordered and are required to
satisfy a completeness property expressed \emph{à la} Scott, in terms
of the norm and of this order relation.
This notion of completeness is very different from the standard
Cauchy-completeness of ordinary Banach spaces.
It has the benefit of making the interpretation of recursive programs
quite straightforward (no need for contractivity assumptions).

In this setting, the ground type \(\rho\) of real numbers of our
programming language is interpreted as the set \(\cR\) of finite
nonnegative measures on the real line equipped with its Borel
\(\sigma\)-algebra, this set \(\cR\) has indeed an obvious structure
of cone.
Cones are naturally equipped with a notion of linear morphisms, which
are also assumed to be Scott continuous, and with a notion of
non-linear morphism introduced in~\cite{EhrhardPaganiTasson18}, called
stable functions and characterized by a \emph{total monotonicity}
condition (plus Scott continuity) which allow to define a cartesian
closed category where fixpoint operators are available at all
types. With these morphisms, cones are a conservative extension of the
category of PCSs and analytic functions as shown in~\cite{Crubille18}.

\paragraph{Integration and sampling.}
The most essential feature of a probabilistic programming language
is the possibility of \emph{sampling} a value according to a given
probability distribution.
In our semantical setting and in the presence of continuous data-types
this requires some form of integration and therefore the morphisms
(here, the linear or the stable functions between cones) must satisfy
a suitable measurability condition.
Consider indeed a functional program \(M\) such that %
\(\Tseq{x:\rho}{M}{\sigma}\) for some type \(\sigma\), where we recall
that \(\rho\) is the type of real numbers, and a program %
\(N\) such that \(\Tseq{}N{\rho}\).
Then \(N\) will be interpreted as an element \(\mu\) of \(\cR\) (a
subprobability measure on \(\Real\) actually) and \(M\) as an analytic
function %
\(g:\cR\to P\) where \(P\) is the cone interpreting the type
\(\sigma\).
Then we typically would like to write a program
\(R=\mathtt{sample}(x,N,M)\) which should satisfy \(\Tseq{}R\sigma\).
The semantics of \(R\) should then be
\begin{align*}
  \int g(\Dirac{}(r))\mu(dr)
\end{align*}
because the Dirac probability measure at \(r\in\Real\),
\(\Dirac{}(r)\in\cR\), is the representation in our semantics of the
real number \(r\).
For instance if \(M=x+x\) (so that \(A=\rho\)) then the
semantics of \(\Subst NMx=M+M\) is \(g(\mu)=\mu\ast\mu\) and the
semantics \(\nu\in\cR\) of \(R=\mathtt{sample}(x,N,x+x)\) is
\begin{align*}
  \int\Dirac{}(2r)\mu(dr)=\Pushf\beta(\mu)
\end{align*}
where \(\beta:\Real\to\Real\) is defined by \(\beta(r)=2r\).
In Section~\ref{sec:stable-exp-meas-coalg}, we will understand that
this sampling operation is simply a \(\texttt{let}\) construct,
exactly as in the discrete PCS setting of~\cite{EhrhardTasson19}.
See Example~\ref{ex:sampling-programming} for a more developed
explanation.

In~\cite{EhrhardPaganiTasson18,Ehrhard20} the cones were accordingly
equipped with a \emph{measurability structure} defined in reference to
a collection of basic measurable spaces (such a collection can be
simply \(\Eset\Real\), what we assume in this introduction for
simplicity).
Given a cone \(P\) equipped with such a measurability structure
\(\cM\) it is then possible to define a class of bounded%
\footnote{With respect to the norm of \(P\).} %
functions \(\Real\to P\) that we call the \emph{measurable paths} of
\(P\).
And then a (linear or stable) function \(P\to Q\) is \emph{measurable}
from \((P,\cM)\) to \((Q,\cN)\) if its pre-composition with each
\(\cM\)-measurable path of \(P\) gives a \(\cN\)-measurable path of
\(Q\).
Equipped with their measurability paths, these measurable cones (more
precisely, their unit balls) can be considered as QBSs, and the
condition above of measurable path preservation is exactly the same as
the definition of a morphism of QBSs (however notions such as
linearity, stability or analyticity, which are crucial for us, do not
arise naturally in the framework of QBSs).

These measurable cones were sufficient in~\cite{EhrhardPaganiTasson18}
to allow sampling over the type \(\rho\) in a probabilistic extension
of PCF because all types in such a language can be written %
\(\Timpl{\sigma_1}{\Timpl\cdots{\Timpl{\sigma_n}{\rho}}}\) and hence
integrability for paths valued in such a type boils down to the
integrability of \(\cR\)-valued paths (with additional parameters in
\(\List \sigma 1n\)) which is possible by our measurability
assumptions.
But if we want to interpret a call-by-value (or even
call-by-push-value) language then we face the problem of integrating
functions valued in more general cones such as for instance
\(\Excl\cR\) (in the sense of \(\LL\), \(\cR\) being the cone of
finite measures on \(\Real\)).
So we must deal with cones where measurable paths can be
integrated. Fortunately it turns out that, thanks to the properties of
the measurability structure \(\cM\) of a cone \(P\), it is easy to
define the integral of a \(P\)-valued path \(\gamma:\Real\to P\)
wrt.~a finite measure \(\mu\) on \(\Real\): it is an \(x\in P\) such
that, for each measurability test \(m\) on \(P\), the real number
\(m(x)\) is equal to the standard Lebesgue integral
\(\int m(\gamma(r))\mu(dr)\) which is well defined and belongs to
\(\Realp\) since \(m\Comp\gamma\) is measurable and bounded, and
\(\mu\) is finite.
And when such an \(x\) exists it is unique by our assumptions that the
measurability tests associated with a cone separate it.
So we can define a cone to be integrable if such integrals always
exist, whatever be the choices of \(\gamma\) and \(\mu\).

In that way we are able to define a category of \emph{integrable
  cones} and \emph{linear and integrable maps}, that is, linear and
measurable maps of cones which moreover commute with all integrals, a
property which can be understood as a strong form of linearity.
Such linear maps will sometimes be called integrable.
It is rather easy to prove that this locally small category is
complete, has a cogenerator and is well-powered so that we know by the
special adjoint functor theorem that each continuous functor from this
category to any other locally small category has a left adjoint.
This allows first to equip our category with a tensor product: given
two integrable cones \(B,C\) (we keep the measurability structures
implicit), we can form the integrable cone \(\Limpl BC\) whose
elements are the linear integrable maps from \(B\) to \(C\), addition
is defined pointwise and the norm is defined by
\(\Norm f=\sup_{\Norm x\leq 1}\Norm{f(x)}\).
Then the functor \(\Limpl B\_\) is easily seen to preserve all limits
and hence has a left adjoint \(\Tens\_B\).
And we can prove that one defines in that way a tensor product
\(\Tens\_\_\) which makes our category symmetric monoidal closed%
\footnote{In~\cite{Ehrhard20} we used the fact that PCSs are dense in
  cones to prove this result but this is actually not necessary,
  thanks to a slightly stronger assumption on the measurability
  structure of cones.}.

There is a
faithful functor from the category of measurable spaces and
sub-probability kernels to the category of measurable cones which maps
a measurable space \(X\) to the cone \(\Cmeas(X)\) of finite
non-negative measures on \(X\).
As already explained in~\cite{Geoffroy22} (in a slightly different
context) the integral preservation property that we enforce on linear
morphisms on cones has the major benefit of making this functor not
only faithful but also full.

\paragraph{Nonlinear functions: stability and analyticity.}
In a second part of the paper we define two cartesian closed
categories of integrable cones and non-linear morphisms which are
Scott continuous and measurable.
We also develop the associated notions of exponential comonad (in the
sense of the semantics of \(\LL\), see for instance~\cite{Mellies09})
applying the special adjoint functor theorem to the continuous
inclusion functor from the category of integrable cones and integrable
linear functions to the non-linear category.
\begin{itemize}
\item In the first case the non-linear morphisms between integrable
  cones are the \emph{measurable and stable functions} that were
  introduced in~\cite{EhrhardPaganiTasson18}.
  These morphisms are Scott continuous functions satisfying a ``total
  monotonicity'' condition, which is an iterated form of monotonicity
  (plus preservation of measurable paths by post-composition of
  course).
  A peculiarity of this construction is that apparently no integral
  preservation condition is imposed on these morphisms%
  \footnote{Notice that it is not possible to expect that non-linear
    morphisms will preserve integrals but one could expect that they
    satisfy a weakened version of this condition.}. %
\item This fact can be considered as an issue for which we propose a
  solution by defining a notion of \emph{analytic morphism} as the
  bounded limits of polynomial functions which are themselves
  described as finite sums of functions of shape
  \(x\mapsto f(x,\dots,x)\) where \(f\) is an \(n\)-linear symmetric
  integrable and measurable function.
  These analytic functions are of course stable
  and measurable
  but not all stable
  and measurable
  functions are analytic because this latter notion is based on
  integrable linearity%
  \footnote{An \(n\)-ary integrable multilinear function
    is 
    a function with \(n\)-arguments which is linear and integrable in
    each parameter.}.
\end{itemize}
For each measurable space \(X\), we show that for both exponential
comonads \(\Excl\_\) described above, the integrable cone \(\Cmeas(X)\) has a
canonical structure of coalgebra, which means that this cone can be
considered as a \emph{data-type} in the sense of~\cite{Krivine90a} or
in the sense of the \emph{positive formulas} of Polarized Linear
Logic~\cite{Girard91a,LaurentRegnier03,Ehrhard16}.
It is very important to observe that this construction uses
integration in a crucial way:
as already explained above, the associated \(\mathtt{let}\) operator
can also be understood as a sampling construct, it is interpreted
using this coalgebra structure which is defined using integration in
the integrable cone \(\Excl{\Cmeas(X)}\).
Combined with the fact that the Kleisli categories of these comonads
are cartesian closed and \(\omega\)-cpo enriched, this means that
integrable cones provide a semantics for a large number of functional
programming languages with continuous data types and basic probability
features.

\paragraph{Convex QBSs.}
Besides measurable cones, one major source of inspiration of this work
is~\cite{Geoffroy22}, which introduces the notion of \emph{convex QBS},
which are a particular class of algebras on the Giry-Panangaden monad of
sub-probability measures in the category of QBSs.
In other words, a convex QBS is a QBS equipped with an abstract,
algebraic operation of ``integration'' from which all elementary
operations of a cone can be derived.
As in the present work, linear morphisms are required to commute with
integration, \Ie~to be morphisms of algebras on the sub-probability monad.
The main differences with respect to the present setting are, first,
that linear negation in convex QBSs is involutive (because they are
defined as dual pairs) whereas we strongly conjecture that this is not
true for integrable cones; and second, that measurability in convex
QBSs is axiomatized in the QBS manner, by equipping each object with a
collection of ``measurable paths'' from \(\Real\) to this object,
satisfying sheaf-like conditions\footnote{In fact, these two
  differences are closely linked: negation in convex QBSs can be
  involutive precisely because their measurability is axiomatized in
  the QBS manner, without restrictions on the QBS-structure.}.
In integrable cones, following~\cite{EhrhardPaganiTasson18},
measurability is axiomatized by means of a ``measurability
structure'', \Ie~a collection of ``test functions'' that map a real
number and an element of the cone to a non-negative real number,
measurably with respect to the first variable, and linearly and
continuously with respect to the second.
In turn, this measurability structure induces a class of measurable
maps from \(\Real\) to the cone, turning the latter into a QBS: a map
from \(\Real\) to the cone is measurable if and only if its
composition with each test function is a measurable map from
\(\Real\times\Real\) to \(\Realp\) (by composition, we mean that the
second argument of the test function is replaced by the map, and the
first argument is left alone).
A map between integrable cones is measurable when it is a morphism of
QBSs.
This means that, from the point of view of measurability alone (\Ie~if
we forget the algebraic structure), integrable cones can be seen as a
particular class of QBSs whose QBS structure can be defined as the
``dual'' of a set of test functions.
This restriction has the pleasant consequence of making the theory of
measurability and integration in cones quite easy, reducing it to
standard Lebesgue integration by means of post-composition with tests.

Similarly defined integrals of functions ranging in topological vector
spaces separated by their topological duals have been introduced by
Pettis a long time ago~\cite{Pettis38}, and are also known as
\emph{weak integrals} or \emph{Gelfand-Pettis integrals}.
The transposition of this definition in our positive cone setting
turns out to be quite suitable, thanks to its compatibility with
categorical limits.


\tableofcontents

\section{Preliminaries}
\subsection{Notations}
In the whole paper, we say that a set is countable if it is finite or
has the same cardinality as \(\Nat\).

We use notations borrowed from the lambda-calculus to denote
mathematical functions: if \(e\) is a mathematical expression for an
element of \(B\) depending on a parameter \(x\in A\), we use
\(\Absm{x\in A}{e}\) for the corresponding function \(A\to B\).

\subsubsection{Categorical notations borrowed from \(\LL\)}
We also borrow notations from intuitionistic \(\LL\) for denoting
objects of our categories and construction on these objects.
These notations are quite coherent although they somehow depart from the
categorical traditions.
In what follows, the word ``linear'' has to be understood in an
intuitive way: as explained in the Introduction, our constructions are
based on notions of linear morphisms which will be defined precisely
later.
\begin{itemize}
\item We use \(\Limpl{E}{F}\) to denote a space of linear morphisms
  from \(E\) to \(F\);
\item we use \(\Tens GE\) to denote the tensor product of \(G\) and
  \(E\), such that a linear morphism from \(\Tens GE\) to \(F\) is the
  same thing as linear morphism from \(G\) to \(\Limpl EF\);
\item we use \(\Sone\) for the unit of \(\ITens\) (instead of the more
  traditional \(I\));
\item we use \(\IWith\) (instead of the more traditional \(\times\)
  that we use for denoting the standard cartesian product of sets) for
  the categorical product (aka.~direct product) and \(\Stop\) (instead
  of the more traditional \(1\)) for the associated unit, which is the
  terminal object;
\item we use \(\IPlus\) for the coproduct (aka.~direct sum) and
  \(\Zero\) for the associated unit which is the initial object;
\item we use \(\Excl E\) for the linear logic exponential, which is
  not a symmetric tensor algebra but rather a symmetric tensor
  coalgebra.
\end{itemize}
Even if in our categories \(\Zero\) and \(\Stop\) are the same object
(just as in the category of vector spaces), we prefer to keep
distinct notations because we have in mind a refinement of our model
where these objects are distinct, and we use the two notations
depending on the context.
Similarly, in some context where \(\Sone\) is considered as a
dualizing object, we denote it as \(\Sbot\), again in accordance with
the tradition of \(\LL\).

\subsubsection{Measure theory and other notations}
We use \(\MEAS\) for the category of measurable spaces and measurable
functions.

If \(X\) and \(Y\) are measurable spaces, recall that a
\emph{kernel} from \(X\) to \(Y\) is a map
\(\kappa : X \times \Sigalg{Y} \to \Realpc\) (where
\(\Sigalg{Y}\) denotes the \(\sigma\)-algebra of \(Y\)) such that:
\begin{itemize}
\item for all \(x \in X\), the map \(\Absm {U} {\kappa(x, U)}\) is a
  measure on \(Y\),
\item for all \(U \in \Sigalg{Y}\), the map
  \(\Absm {x} {\kappa(x, U)}\) is measurable.
\end{itemize}
We write \(\kappa:X\Kernto Y\) for ``\(\kappa\) is a kernel from
\(X\) to \(Y\)''.
We say that \(\kappa\) is \emph{bounded} if the set
\(\{ \kappa(x,Y) \St x \in X \}\) has a finite upper bound.

If \(X\) is a measurable space, \(\mu\) a non-negative measure on
\(X\) and \(f:X\to\Realp\) a non-negative measurable function, we
use
\begin{align*}
  \int f(r)\mu(dr)
\end{align*}
for the integral, which belongs to \(\Realpc\), rather than the more
usual \(\int f(r)d\mu(r)\).
The reason of this choice is that it is much more convenient when the
measure arises as the image of a kernel \(\kappa:Y\Kernto X\) in
which case we can use the non ambiguous notation
\(\int f(r)\kappa(s,dr)\).
This notation is also intuitively compelling if we see \(dr\) as
representing metaphorically an ``infinitesimal'' measurable subset of
\(X\).

If \(a\) is an element and \(n\in\Nat\) we use \(\Rep an\) for the
\(n\)-tuple \((a,\dots,a)\).

We use \(\Natnz\) for \(\Nat\setminus\Eset 0\).

If \(n\in\Nat\) we set \(\Intset n=\Eset{1,\dots,n}\). 

If \(I\) is a set, we use \(\Mfin I\) for the set of all finite
multisets of elements of \(I\), which are the functions \(m:I\to\Nat\)
such that the set \(\Supp m=\{i\in I\St m(i)\not=0\}\) is finite.

\subsection{Categories}
The following is an easy consequence of the Yoneda lemma which gives a
simple tool for proving that two functors are naturally isomorphic by
checking that two associated indexed classes of homsets are in natural
bijective correspondence.
\begin{lemma}\label{lemma:functor-yoneda-iso}
  Let $\Cat C$ and $\Cat D$ be categories, $F,G:\Cat C\to\Cat D$ be
  functors and let $\psi_{C,D}:\Cat D(F(C),D)\to\Cat D(G(C),D)$ be a
  natural bijection.
  Then the family of morphisms
  $\eta_C=\psi_{C,F(C)}(\Id_{F(C)})\in\Cat D(G(C),F(C))$ is a natural
  isomorphism whose inverse is the family of morphisms
  $\theta_C=\Funinv{\psi_{C,G(C)}}(\Id_{G(C)})\in\Cat D(F(C),G(C))$.
\end{lemma}

\section{Cones}
\label{sec:algebraic-cones}

Cones are the basic objects of our model.
They are algebraic structures with numerical features (the non-negative real
half line acts on them) as well as domain theoretic features.
The algebraic and numerical aspects will be essential to account for
the probabilistic aspects of the model and the domain theoretic
aspects will be crucial to give to our model a suitable computational
expressive power, allowing to interpret arbitrary recursive definitions.

The purpose of the present section is to introduce this basic algebraic and
numerical infrastructure and give its basic properties.
Our definition of cones is borrowed without major modifications
from~\cite{Selinger04}.
As explained in that paper, they are close to the domain theoretic
treatment of positive cones developed in~\cite{Tix00}, with the
difference that Selinger's cones are equipped with a norm and that
their order-theoretic completeness is deeply related to this norm.

The notion of positive cone itself is pervasive in functional analysis
and it would be a very difficult task to describe its genealogy and
many avatars in the literature.
Our (and Selinger's) cones seem very similar to normal cones in Banach
spaces, and it seems actually possible, given one or our cones \(P\),
to define an enveloping Banach space of which \(P\) is a normal
positive cone.
However, the linear morphisms that we consider between our cones are
assumed to be continuous in a domain theoretic sense, and this seems
to be a stronger property than continuity wrt.~the topology induced by
the norm (when the linear morphism is extended to the associated
Banach space).
This difference in the definition of morphisms seems to be a major
drift wrt.~the standard uses of cones in analysis.

Both Selinger and Tix assume that their cones are continuous (in the
domain-theoretic sense) which makes it possible to prove a separation
property similar to a Hahn-Banach theorem.
This is an assumption that we cannot afford here because we will need
our category of cones and linear maps to be complete and continuity
does not seem to be preserved by equalizers in general.
We will see that dropping this assumption is essentially harmless in
the setting of this paper: our measurability structures of
Section~\ref{sec:measurable-cones} will provide us the required
separation properties.

Another difference between Selinger's cones and ours is that we
do not assume order-theoretic completeness wrt.~arbitrary norm-bounded
directed sets, as it is usual in domain theory, but only
wrt.~norm-bounded \(\omega\)-increasing sequences (or, equivalently,
to countable directed sets).
This assumption is sufficient for computing arbitrary fixpoints, see
Section~\ref{sec:ccc-fix}, and cannot be significantly strengthened
because of our constant use of the monotone convergence theorem.

Measurability notions for cones will be necessary as well to deal with
probabilities on arbitrary measurable spaces such as the real line;
this will be done in Section~\ref{sec:measurable-cones}.

\subsection{Basic definitions}

A \emph{precone} is a \(\Realp\)-semimodule \(P\) which satisfies %
\begin{Axicond}\Pcsimpll{}
  \(\forall x_1,x_2,x\in P\ x_1+x=x_2+x\Implies x_1=x_2\)
\end{Axicond}
\begin{Axicond}\Pcposl{}
  \(\forall x_1,x_2\in P\ x_1+x_2=0\Implies x_1=0\)
\end{Axicond}
Given \(x_1,x_2\in P\), we stipulate that \(x_1\leq x_2\) if %
\(\exists x\in P\ x_2=x_1+x\).
By \Pcsimplr{} and \Pcposr{} this defines a partial order relation on
\(P\): \emph{the cone order} of \(P\).
Moreover when \(x_1\leq x_2\) there is exactly one \(x\in P\) such
that \(x_2=x_1+x\), that we denote as \(x_2-x_1\).
Notice that this subtraction between elements of \(P\) is only
partially defined, and that it satisfies all the usual laws of
subtraction.

A \emph{cone} is a precone \(P\) equipped with a function %
\(\Norm\__P:P\to\Realp\) (or simply \(\Norm\_\)), called the
\emph{norm of \(P\)}, which satisfies the following properties.
\begin{Axicond}\Cnormhl{}
  \(\forall \lambda\in\Realp\forall x\in P\ \Norm{\lambda x}=\lambda\Norm x\)
\end{Axicond}
\begin{Axicond}\Cnormzl{}
  \(\forall x\in P\ \Norm x=0\Implies x=0\)
\end{Axicond}
\begin{Axicond}\Cnormtl{}
  \(\forall x_1,x_2\in P\ \Norm{x_1+x_2}\leq\Norm{x_1}+\Norm{x_2}\)
\end{Axicond}
\begin{Axicond}\Cnormpl{}
  \(\forall x_1,x_2\in P\ \Norm{x_1}\leq\Norm{x_1+x_2}\) or, equivalently %
  \(\forall x_1,x_2\in P\ x_1\leq x_2\Implies\Norm{x_1}\leq\Norm{x_2}\).
\end{Axicond}
Condition~\Cnormpr{} expresses the positiveness of \(P\) and implies
\Pcposr{}, but it is seems more sensible to require \Pcposr{} at the
beginning because of its purely algebraic nature, and because this
allows to define the useful notion of precone.
\begin{Axicond}\Cnormcl{}
  Each sequence \((x_n)_{n\in\Nat}\) of elements of \(P\) which is
  increasing%
  \footnote{A reader acquainted with domain-theory might expect here a
    stronger completeness requirement using arbitrary directed sets
    instead of \(\omega\)-chains (or, equivalently, countable directed
    sets).
    It is absolutely crucial to use this restricted definition because
    we will often have to use the monotone convergence theorem to
    prove this property, and this theorem is valid only for countable
    families.} %
  (for the cone order relation of \(P\)) and satisfies
  \(\forall n\in\Nat\ \Norm{x_n}\leq 1\) has a lub %
  \(x=\sup_{n\in\Nat}x_n\) in \(P\) which satisfies \(\Norm x\leq 1\).
\end{Axicond}
\noindent 
A subset \(A\) of \(P\) is
\begin{itemize}
\item \emph{bounded} if
  \(\exists\lambda\in\Realp\forall x\in A\ \Norm x\leq\lambda\).
  We set %
  \(\Cuball P=\{x\in P\St\Norm x\leq 1\}\) and call this set the
  \emph{unit ball} of \(P\) (\emph{unit tip} might be more appropriate
  but seems less standard).
  With this notation, \(A\) is bounded iff %
  \(\exists\lambda\in\Realp\ A\subseteq\lambda\Cuball P\).
\item \emph{\(\leq\)-bounded} if there is \(y\in P\) such that %
  \(\forall x\in A\ x\leq y\).
  This implies that \(A\) is bounded (but the converse is not true).
\item \emph{\(\omega\)-closed}
  if %
  \(\forall x_1,x_2\in P\ (x_1\leq x_2\text{ and }x_2\in A)\Implies
  x_1\in A\) and %
  for each bounded increasing sequence \((x_n)_{n\in\Nat}\) of elements
  of \(A\) one has \(\sup_{n\in\Nat}x_n\in A\).
\end{itemize}
\noindent 
Notice that \(P\) and \(\Cuball P\) are \(\omega\)-closed subsets of \(P\).

\begin{definition}
  Let \(S\) be a set and \(P\) be a cone.
  A function \(f:S\to P\) is \emph{bounded} if \(f(S)\) is bounded in
  \(P\).
\end{definition}

\begin{definition}
  \label{def:monotone-scott}
Let \(P\) and \(Q\) be cones, let \(A\subseteq P\) be \(\omega\)-closed and
let \(f:A\to Q\) be a function.
\begin{itemize}
\item \(f\) is \emph{increasing} if %
  \(\forall x_1,x_2\in A\ x_1\leq x_2\Implies f(x_1)\leq
  f(x_2)\).
  Notice that if \(f\) is increasing and \((x_n)_{n\in\Nat}\) is a
  bounded and increasing sequence in \(A\) then %
  the sequence \((f(x_n))_{n\in\Nat}\) is bounded by %
  \(\Norm{f(\sup_{n\in\Nat}x_n)}\) in \(Q\), by \Cnormpr{} and
  monotonicity of \(f\).
\item \(f\) is \emph{\(\omega\)-continuous}, or simply continuous (no
  other notion of continuity will be considered in this paper), if
  \(f\) is monotonic and for each bounded increasing sequence %
  \((x_n)_{n\in\Nat}\) of elements of \(A\), one has %
  \(f(\sup_{n\in\Nat}x_n)=\sup_{n\in\Nat}f(x_n)\), that is %
  \(f(\sup_{n\in\Nat}x_n)\leq\sup_{n\in\Nat}f(x_n)\) since the
  converse holds by monotonicity of \(f\).
\item \(f\) is \emph{linear} if \(A=P\), \(f(\lambda x)=\lambda f(x)\) and %
  \(f(x_1+x_2)=f(x_1)+f(x_2)\), for all \(\lambda\in\Realp\) and %
  \(x,x_1,x_2\in P\).
  Notice that if \(f\) is linear then \(f\) is increasing because, given
  \(x_1,x_2\in P\), if \(x_1\leq x_2\) then
  \(f(x_2-x_1)+f(x_1)=f(x_2)\), and moreover we have
  \(f(x_2-x_1)=f(x_2)-f(x_1)\).
  One says that \(f\) is linear and continuous if it is linear and
  \(\omega\)-continuous.
\item If \(f:P\to Q\) is linear, one says that \(f\) is \emph{bounded}
  if its restriction to \(\Cuball P\) is a bounded function.
\end{itemize}  
\end{definition}
\noindent 
One major interest of this kind of continuity is the fact that
separate continuity implies continuity (see
Lemma~\ref{lemma:seprate-cont-implies-cont}), a property that usual
topological continuity does not satisfy.

There are plenty of examples of cones:
\begin{example}
  \label{ex:cones-measurable-functions}
  Let \(X\) be a measurable space. The space of all bounded measurable maps
  from \(X\) to \(\Realp\) forms a cone: the operations are defined pointwise,
  and the norm is given by the supremum.
\end{example}
  
\begin{example}
  \label{ex:cones-analytic-functions}
  Section~\ref{sec:basic-cone-finite-meas} describes the cone of
  finite measures on a measurable space which provides one of the main
  motivations for this work.
  All the objects of the probabilistic coherence space model of
  \(\LL\) can be seen as cones; the interested reader can have a look
  at the beginning of Section~\ref{sec:pcs-integrable} to see more
  about them.
  Here are some instances of this particular class.
  \begin{itemize}
  \item The cone \(\Snat\) whose elements are the \(u\in\Realpto\Nat\)
    such that \(\sum_{n\in\Nat}u_n<\infty\), with algebraic
    operations defined pointwise, and \(\Norm u=\sum_{n\in\Nat}u_n\).
    This is also a special case of the cones of
    Section~\ref{sec:basic-cone-finite-meas} where the measurable
    space is \(\Nat\) with the discrete \(\sigma\)-algebra.
  \item The dual of \(\Snat\) (in the sense of
    Definition~\ref{def:cone-dual}) which can be described as the cone
    \(\Orth N\) of bounded families \(u\in\Realpto\Nat\) with norm
    defined by \(\Norm u=\sup_{n\in\Nat}u_n\).
  \item The cone \(\Timpl P\Sone\) where
    \(P\in\{\Snat,\Orth\Snat\}\), whose elements are the families
    \(t\in\Realpto{\Mfin\Snat}\) such that there is
    \(\lambda\in\Realp\) such that
    \begin{align*}
      \forall u\in P\quad
      \Norm u\leq 1\Implies\sum_{m\in\Mfin\Nat}t_mu^m\leq\lambda
    \end{align*}
    where \(u^m=\prod_{n\in\Nat}u_n^{m(n)}\), with algebraic
    operations defined componentwise and norm defined by %
    \(\Norm t=\sup_{u\in\Cuball P}\Fun t(u)\) where
    \(\Fun t(u)=\sum_{m\in\Mfin\Nat}t_mu^m\).
    In both cases \(P=\Snat\) and \(P=\Orth\Snat\), \(\Fun t\) can
    be seen as a bounded function %
    \(\Cuball P\to\Realp\).
    The set of these families \(t\) equipped with that norm is easily
    seen to be a cone.
    An element of \(\Timpl{\Orth\Snat}\Sone\) can be seen as a power series
    with infinitely many parameters, defining a function
    \(\Cuball{\Orth\Snat}=\Intercc01^\Nat\to\Realp\).
    An element of \(\Timpl\Snat\Sone\) is an analytic function on the
    subprobability distributions on the natural numbers, we give an
    example of such a function.
    Given two such distributions \(u\) and \(v\), we can define their
    convolution product \(u\ast v\in\Snat\) by
    \((u\ast v)_n=\sum_{i=0}^nu_iv_{n-i}\) which is again a
    subprobability distribution (the push-forward of addition).
    If \(u\) (resp.~\(v\)) is the probability distribution of a
    \(\Nat\)-valued random variable \(X\) (resp.~\(Y\)) and \(X\) and
    \(Y\) are independent, then \(u\ast v\) is the probability
    distribution associated with \(X+Y\).
    Given a family \((a_n\in\Realp)_{n\in\Nat}\) such that
    \(\sum_{n\in\Nat}a_n=1\), a non trivial example of element of
    \(\Timpl\Snat\Sone\) is \(t\) given by
    \begin{align*}
      \Fun t(u)=\sum_{n=0}^\infty a_n
      \overbrace{u\ast\cdots\ast u}^n\,.
    \end{align*}
    whose coefficients are
    \begin{align*}
      t_m=\frac{\Factor{\Card m}}{\Factor m}\,a_{\Card m+1}
    \end{align*}
    where \(\Card m\) is the number of elements of \(m\) (taking
    multiplicities into account) and
    \(\Factor m=\prod_{i\in\Nat}\Factor{m(i)}\).
    Such power series are typical examples of the analytic functions
    that we will meet in Section~\ref{sec:analytic-functions-exp} in
    the general setting of integrable cones.
  \end{itemize}
\end{example}

\begin{remark}
  \label{rk:analytic-dedf-on-balls}
  We can already observe that our analytic functions will be defined
  in general only on the unit ball of their domain.
  The reason is that the analytic functions which interpret programs
  will in general be characterized by recursive equations.
  For instance, it is quite easy to define a probabilistic program (of
  type \(\texttt{unit}\to\texttt{unit}\)) whose interpretation is a
  function \(f:\Intercc01\to\Intercc01\) such that
  \(f(u)=\frac12u+\frac12 f(u)^2\), so that \(f(u)=1-\sqrt{1-u}\):
  this is the only solution of the quadratic equation which gives a
  power series with nonnegative coefficients \((t_n)_{n\in\Nat}\) such
  that \(f(u)=\sum_{n\in\Nat}t_nu^n\) for all \(u\in\Intercc01\), but
  the series diverges for \(u>1\).
\end{remark}

\begin{remark}
  As we have seen in the basic definitions and in the first examples,
  all the real numbers we consider are non-negative.
  A natural question is whether this restriction could be dropped
  and we argue that this issue is more tricky than it might seem at
  first sight.
  Consider an analytic function (in the sense described above)
  \(f:\Intercc01^2\to\Intercc01\), given by a family \(t\in\Realp^2\),
  so that \(f(u,v)=\sum_{n,p\in\Nat}t_{n,p}u^nv^p\).
  Then for each \(u\in\Intercc01\), the function
  \(f_u:\Intercc01\to\Intercc01\) such that \(f_u(v)=f(u,v)\) has a
  least fixpoint \(g(u)\), and we shall see in Section
  \ref{sec:ccc-fix} that the function \(g\) is itself analytic, that
  is \(g(u)=\sum_{n\in\Nat}s_nu^n\) for some
  \((s_n\in\Realp)_{n\in\Nat}\) such that
  \(\sum_{n\in\Nat}s_n\leq 1\).
  If we relax this positivity requirement, then our function \(f\)
  could be \(f(u,v)=1-(1-u)(1-v)=u+v-uv\) (mentioned in particular
  in~\cite{EscardoHofmannStericher04}).
  It is still true that \(f_u\) has a least fixpoint \(g(u)\), but one
  checks easily that \(g(0)=0\) and \(g(u)=1\) if \(u>0\), which shows
  that \(g\) cannot be analytic, even with possibly negative
  coefficients.
\end{remark}

\subsection{An archetypal example: the cone of finite measures}
\label{sec:basic-cone-finite-meas}

Let \(X\) be a measurable space.
The set \(\Mcca {\Cmeas(X)}\) of all \emph{finite} (non-negative,
real-valued) measures on \(X\) is naturally equipped with the
structure of a cone:
\begin{itemize}
\item the algebraic operations of \(\Mcca {\Cmeas(X)}\) are defined
  pointwise (\Eg~\((\mu_1+\mu_2)(U)=\mu_1(U)+\mu_2(U)\) for all
  \(U\in\Sigalg X\));
\item the norm is given by \(\Norm\mu=\mu(X)\) (this is the
  total variation norm of \(\mu\) since \(\mu\) is non-negative);
\item observing that \(\mu_1\leq\mu_2\) means %
  \(\forall U\in\Sigalg X\ \mu_1(U)\leq\mu_2(U)\), it is clear that
  each increasing sequence \((\mu_n)_{n\in\Nat}\) in \(\Cuball P\) has
  a least upper bound \(\mu\in\Cuball P\) which is computed pointwise:
  \(\mu(U)=\sup_{n\in\Nat}\mu_n(U)\).
\end{itemize}
The set \(\Mcca {\Cmeas(X)}\) itself can be equipped with a
\(\sigma\)-algebra, in the spirit of the Giry monad.

Let \(\kappa : X \Kernto Y\) be a bounded kernel.
Then the map
\(\widehat \kappa : \Mcca {\Cmeas(X)} \Kernto \Mcca {\Cmeas(Y)}\)
defined by
\(\widehat \kappa(\mu)(V) = \int_{x \in X} \kappa(x,V) \, \mu(dx) \)
is linear, continuous and measurable.

In fact, this map \(\widehat \kappa\) has a stronger property: it
preserves (S-)finite integrals.
Namely, whenever \(\mu\) is a finite measure on \(\Cmeas(X)\) (or
more generally an S-finite%
\footnote{Recall that a measure is S-finite
  if it is a countable sum of finite measures.
  This is a weaker property than \(\sigma\)-finiteness.} %
measure such that
\(\int_{\nu \in \Cmeas(X)} \nu(X) ~ \mu(d\nu) < \infty\)), we have
\[
  \widehat \kappa\Big(\int_{\nu \in \Cmeas(X)} \nu\,\mu(d\nu)\Big) =
  \int_{\nu \in \Cmeas(X)} \widehat \kappa(\nu)\,\mu(d\nu)\,,
\]
where the integrals are defined pointwise (recall that a measure on
\(X\) is in particular a map from \(\Sigalg X\) to \(\Realpc\)).
Conversely, if a map
\( f : \Mcca {\Cmeas(X)} \Kernto \Mcca {\Cmeas(Y)} \) preserves
S-finite integrals, then there exists a unique bounded kernel
\(\kappa : X \Kernto Y\) such that \( f = \widehat \kappa \).
It is given by \(\kappa(x) = f(\Dirac {X}(x))\), where
\(\Dirac {X}(x)\) denotes the Dirac measure at \(x\) on \(X\).
If we think of S-finite measures as a generalization of formal linear
combinations, then commutation with S-finite integrals is simply a
generalization of linearity.

Preservation of S-finite integrals implies continuity and linearity,
but the converse does not hold in general as illustrated in
Remark~\ref{rk:continuous-part-measure}.

\begin{remark}
  \label{rk:continuous-part-measure}
  Consider the map
  \(\ContinuousPart:\Mcca{\Cmeas(\Real)}\to\Mcca{\Cmeas(\Real)}\)
  defined by
  \begin{align*}
  	\ContinuousPart(\mu) = \frac{d\mu}{d\lambda}\,\lambda
  \end{align*}
  where \(\lambda\) is the Lebesgue measure on \(\Real\) and
  \(\frac{d\mu}{d\lambda}\) is the Radon-Nikodym derivative of \(\mu\)
  with respect to \(\lambda\).
  In other words, this map extracts the continuous part of the measures
  on \(\Real\).
  
  This map is linear, \(\omega\)-continuous and measurable (because the map
  \(\mu \mapsto \frac{d\mu}{d\lambda}\) is measurable~\cite[Theorem
    1.28]{kallenberg17}).
  On the other hand,
  \[
    \ContinuousPart\Big(\int \Dirac {\Real}(r)
      \lambda_{[0,1]}(dr)\Big) =
    \ContinuousPart\left(\lambda_{[0,1]}\right) = \lambda_{[0,1]}\,,
  \]
  (where \(\lambda_{[0,1]}\) is the Lebesgue measure on \([0,1]\)),
  while
  \[
    \int \ContinuousPart(\Dirac {\Real}(r)) \lambda_{[0,1]}(dr)
    = \int 0 \lambda_{[0,1]}(dr) = 0\,.
  \]
\noindent 
  Therefore, there exists no kernel \(\kappa : \Real \Kernto \Real\)
  such that \(\ContinuousPart = \widehat \kappa\).
  As we shall see, avoiding this kind of situation is one of our main
  motivations for introducing integrability.
\end{remark}

\subsection{Basic properties}

The following means that the notion of continuity we consider for
linear maps behaves in an essentially algebraic way.
\begin{lemma}
  \label{lemma:linear-inverse}
  Let \(P\) and \(Q\) be cones and let \(f:P\to Q\) be linear and
  continuous.
  If \(f\) is bijective then \(\Inv f\) is linear and continuous.
\end{lemma}
\begin{proof}
  Linearity follows from the injectivity of \(f\): let
  \(y_1,y_2\in Q\), \(x_1=\Inv f(y_1+y_2)\) and
  \(x_2=\Inv f(y_1)+\Inv f(y_2)\), we have \(f(x_1)=y_1+y_2\) and
  \(f(x_2)=y_1+y_2\) by linearity of \(f\), hence \(x_1=x_2\).
  Scalar multiplication is dealt with similarly. Since \(\Inv f\) is
  linear, it is increasing.

  Let \((y_n\in\Cuball Q)_{n=1}^\infty\) be an increasing sequence and
  let \(y\in\Cuball Q\) be its lub.
  The sequence %
  \((\Inv f(y_n)\in P)_{n=1}^\infty\) is increasing and upper bounded
  by %
  \(\Inv f(y)\) and hence bounded in norm by \(\Norm{\Inv f(y)}_P\),
  so it has a lub \(x\in P\) such that \(x\leq \Inv f(y)\).
  By continuity of \(f\) we have %
  \(f(x)=f(\sup_{n=1}^\infty\Inv f(y_n))=\sup_{n=1}^\infty y_n=y\) and
  hence \(x=\Inv f(y)\) which shows that \(\Inv f\) is continuous.
\end{proof}
\noindent 
Using the notations of \(\LL\) for the multiplicative constants, there
is a cone \(\Sone=\Sbot\) whose set of elements is \(\Realp\) and
\(\Norm x=x\): the \(1\)-dimensional cone.
And using the notations of \(\LL\) for the additive constants, there
is also a cone \(\Zero=\Top\) whose only element is \(0\): the
\(0\)-dimensional cone.
\begin{lemma}
  \label{lemma:cones-cone-add-scal}
  Let \(P\) be a cone.
  Addition \(P\times P\to P\) and scalar multiplication
  \(\Sone\times P\to P\) are increasing and \(\omega\)-continuous.
\end{lemma}
\begin{proof}
  See~\cite{Selinger04}.
\end{proof}

\noindent 
The following lemma will be quite useful to prove that the difference
between two linear and continuous functions is also linear and
continuous, when it exists.
\begin{lemma} %
  \label{lemma:fun-diff-Scott}
  Let \(P\) and \(Q\) be cones, let \(A\subseteq P\) be \(\omega\)-closed
  and let \(f,g:A\to Q\) be functions such that %
  \(f\) is increasing, \(g\) is \(\omega\)-continuous,
  \(\forall x\in P\ f(x)\leq g(x)\) and the function
  \(g-f=\Absm{x\in P}{(g(x)-f(x))}\) is increasing.
  Then \(g-f\) is \(\omega\)-continuous.
\end{lemma}
\begin{proof}
  Let \((x_n)_{n\in\Nat}\) be a bounded increasing sequence in \(A\) and
  let %
  \(x=\sup_{n\in\Nat}x_n\).
  For all \(n\in\Nat\) we have \(f(x_n)\leq f(x)\) and hence %
  \(g(x_n)\leq f(x)+g(x_n)-f(x_n)\).
  The sequence \((f(x)+g(x_n)-f(x_n))_{n\in\Nat}\) is increasing by our
  assumption that \(g-f\) is increasing, and it is \(\leq\)-bounded by
  \(f(x)+g(x)\).
  We have 
  \begin{align*}
    g(x)&=g(\sup_{n\in\Nat}x_n)\\
        &=\sup_{n\in\Nat}g(x_n)
          \text{\quad since }g\text{ is \(\omega\)-continuous}\\
        &\leq\sup_{n\in\Nat}(f(x)+g(x_n)-f(x_n))\\
        &=f(x)+\sup_{n\in\Nat}(g(x_n)-f(x_n))
    \text{\quad by Lemma~\ref{lemma:cones-cone-add-scal}}
  \end{align*}
  and hence \(g(x)-f(x)\leq\sup_{n\in\Nat}(g(x_n)-f(x_n))\).
  Since \(g-f\) is increasing, if follows that \(g-f\) is
  \(\omega\)-continuous.
\end{proof}

\begin{lemma} %
  \label{lemma:line-cont-bounded}
  If \(f:P\to Q\) is linear then
  \(f(\Cuball P)\) is bounded.
  We set
  \[
    \Norm f=\sup_{x\in\Cuball P}\Norm{f(x)}\in\Realp\,.
  \]
\end{lemma}
\begin{proof}
  See~\cite{Selinger04}, we give the proof because it is short and
  interesting.
  If the lemma does not hold there is a sequence \((x_n)_{n\in\Nat}\)
  such that \(\forall n\in\Nat\ \Norm{x_n}\leq 1\) and
  \(\forall n\in\Nat\ \Norm{f(x_n)}\geq 4^n\). Then let %
  \(y_n=\sum_{k=1}^n\frac1{2^k}x_k\in P\), we have %
  \(\Norm{y_n}\leq \sum_{k=1}^n\frac1{2^k}\Norm{x_n}\leq 1\) and %
  \((y_n)_{n\in\Nat}\) is an increasing sequence which therefore has a
  lub \(y\in\Cuball P\), and we have %
  \(\Norm{f(y)}\geq \Norm{f(y_n)} \geq\frac1{2^n}\Norm{f(x_n)}\geq
  2^n\) %
  by \Cnormpr{} and linearity of \(f\).
  Since this holds for all \(n\in\Nat\) we have a contradiction.
\end{proof}
\begin{remark}
  We have obtained this property without even requiring \(f\) to be
  \(\omega\)-continuous.
\end{remark}

\subsection{The category of cones and linear and continuous maps}
Given cones \(P\) and \(Q\) the set \(\Limpl PQ\) of all linear and
continuous maps from \(P\) to \(Q\), equipped with obvious pointwise
defined algebraic operations, is a precone. Notice that
\(f_1\leq f_2\) simply means that %
\(\forall x\in P\ f_1(x)\leq f_2(x)\) %
and then the difference is given by %
\((f_2-f_1)(x)=f_2(x)-f_1(x)\): it suffices to check that this latter
map is linear which is obvious, and that it is continuous which
results from Lemma~\ref{lemma:fun-diff-Scott}.
\begin{lemma} %
  \label{lemma:limpl-cone}
  Equipped with the norm defined in
  Lemma~\ref{lemma:line-cont-bounded}, the precone \(\Limpl PQ\) is a
  cone.
\end{lemma}
\noindent 
The proof is easy.
By definition of \(\Norm f\) and by \Cnormhr{} we have %
\begin{align*}
  \forall x\in P\quad \Norm{f(x)}\leq\Norm f\Norm x\,.
\end{align*}

\begin{definition}
  \label{def:cone-dual}
  We set \(\Cdual P=(\Limpl P\Sbot)\). %
  If \(x\in P\) and \(x'\in\Cdual P\) we write
  \(\Eval x{x'}=x'(x)\in\Realp\).
\end{definition}
\noindent 
Notice that, with these notations,
\(\Norm{x'}=\sup_{x\in\Cuball P}\Eval{x}{x'}\).

\begin{remark}
  The cone \(P'\) should be understood as an analog of the topological
  dual of a normed vector space; for instance one can define a linear
  and continuous morphism \(\eta:P\to P''\) which corresponds to the
  usual embedding of a vector space into its bidual.
  However, in the case of cones, this morphism \(\eta\) is not
  necessarily injective: see the counter-example in
  Remark~\ref{rk:cone-non-separe}.
  As already mentioned we know from~\cite{Selinger04} that an
  additional requirement on morphisms ---~namely, continuity, in the
  domain-theoretic sense~--- could guarantee the injectivity of this
  morphism; however continuity is too strong a requirement for what we
  are trying to accomplish.
  For this reason, the structures we define below (measurable, and
  later, integrable cones) will contain the axiom \(\Mssepr\) which
  states precisely that this morphism is injective.
\end{remark}

\begin{remark}[A cone whose dual is zero]
  \label{rk:cone-non-separe}
  The following construction provides a non-trivial cone \(P\) whose
  dual \(P'\) contains only \(0\). This construction was suggested to
  us by one of the reviewers: many thanks to her/him. Consider the
  quotient set \(P=P_0/\mathord\sim\), where \(P_0\) is the set of all
  bounded measurable maps from \([0,1]\) to \(\Realp\) (Example
  \ref{ex:cones-measurable-functions}), and \(\sim\) is the following
  equivalence relation: \(f\sim g\) if and only if \(f\) and \(g\)
  differ only on a meager\footnote{Recall that a subset of a
    topological space is meager if it is contained in a countable
    union of closed sets whose interiors are empty. Conversely, a
    subset is co-meager if it contains a countable intersection of
    dense open subsets.} subset of \([0,1]\).

  The set \(P_0\) inherits the structure of a cone from \(P\), with
  \(\Norm{[f]}\) defined as \(\inf\{\Norm g \St g \sim f\}\) (where
  \([f]\) denotes the equivalence class of \(f\) and
  \(\Norm g=\sup_{x\in[0,1]}g(x)\)). This is indeed a cone: the least
  obvious part is the fact that \(\Norm{[f]} = 0\) implies
  \([f] = 0\). Indeed, assume \(\Norm{[f]} = 0\). Then for all
  \(n > 0\), there exists \(f_n \sim f\) such that
  \(\Norm{f_n} \leq \frac{1}{n}\). Thus, there exists a co-meager set
  \(X_n\) such that \(f(x) \leq \frac{1}{n}\) for all \(x \in
  X_n\). The intersection of all the \(X_n\) is itself co-meager, and
  therefore \(f \sim 0\).

  We shall prove that \(P' = \{0\}\). Let \(\alpha \in P'\). The map
  \(f \mapsto \alpha([f])\) (from \(P_0\) to \(\Realp\)) must commute
  with countable sums, therefore there must exist a finite measure
  \(\mu\) on \([0,1]\) such that
  \(\alpha([f]) = \int_{x\in[0,1]} f(x) \mu(dx)\) for all
  \(f\). Moreover, we have that \(\mu(Y) = 0\) for all meager sets
  \(Y\). On the other hand, there exists a co-meager set \(X\) such
  that \(\mu(X) = 0\). Indeed, let \(D\) be a dense countable subset
  of \([0,1]\) such that \(\mu(D) = 0\) (such a set must exist because
  \(\mu\) has at most countably many atoms). For all \(n > 0\), there
  exists an open set \(X_n\) that contains \(D\) (and is therefore
  dense) and such that \(\mu(X_n) \leq \frac{1}{n}\). The intersection
  \(X\) of all the \(X_n\) is co-meager, and its measure is
  \(0\). Therefore,
  \(\mu([0,1]) = \mu(X) + \mu([0,1] \setminus X) = 0\), and thus
  \(\alpha = 0\).
\end{remark}

\begin{definition}
  The category \(\CONES\) has the cones as objects, and
  \(\CONES(P,Q)\) is the set of all linear and continuous \(f:P\to Q\)
  such that \(\Norm f\leq 1\).
\end{definition}

\begin{theorem}
  \label{th:cone-product}
  The category \(\CONES\) has all small products. Given a family \((P_i)_{i\in I}\) of cones (with
  no cardinality restrictions on \(I\)), their categorical product
  \((\Bwith_{i\in I}P_i,(\Proj i)_{i\in I})\) is defined as follows:
  \begin{itemize}
  \item \(\Bwith_{i\in I}P_i\) is the set of all
    \(\Vect x=(x_i)_{i\in I}\in\prod_{i\in I}P_i\) such that the
    family \((\Norm{x_i})_{i\in I}\) is bounded, equipped with the
    obvious algebraic operations defined componentwise
  \item and \(\Norm{\Vect x}=\sup_{i\in I}\Norm{x_i}\).
  \end{itemize}
  In particular, the terminal object (corresponding to the case where
  \(I=\emptyset\)) is \(\Stop\).
  
  The projections are the standard projections of the cartesian
  product in \(\SET\).
  Given a family of morphisms \((f_i\in\CONES(Q,P_i))_{i\in I}\), the
  associated morphism %
  \(f=\Tuple{f_i}_{i\in I}\in\CONES(Q,\Bwith_{i\in I}P_i)\) is
  characterized by \(f(y)=(f_i(y))_{i\in I}\).

  The cone order of \(\Bwith_{i\in I}P_i\) is the product of the cone
  orders of the \(P_i\)'s and the lubs of bounded sequences of
  elements of \(\Bwith_{i\in I}P_i\) are computed componentwise.
\end{theorem}
See~\cite{Selinger04}.
There is a clear similarity with the \(\ell^\infty\) construct of
Banach spaces.

As announced, we can check now that separate continuity implies
continuity.
\begin{lemma}
  \label{lemma:seprate-cont-implies-cont}
  Let \(P\), \(Q\) and \(R\) be cones, let \(A\subseteq P\) and
  \(B\subseteq Q\) be \(\omega\)-closed, so that \(A\times B\) is an
  \(\omega\)-closed subset of the product cone $\With PQ$, and let
  \(f:A\times B\to R\) be separately \(\omega\)-continuous (that is,
  for all \(y\in B\) the function \(\Absm{x\in A}f(x,y)\) is
  \(\omega\)-continuous and for all \(x\in A\) the function
  \(\Absm{y\in B}f(x,y)\) is \(\omega\)-continuous).
  Then \(f\) is \(\omega\)-continuous.
\end{lemma}
\begin{proof}
  Let \(((x_n,y_n)\in A\times B)_{n\in\Nat}\) be an increasing bounded
  sequence in \(A\times B\) so that \((x_n)_{n\in\Nat}\) is increasing
  and bounded in \(A\), and \((y_n)_{n\in\Nat}\) is increasing and
  bounded in \(B\), and %
  \(\sup_{n\in\Nat}(x_n,y_n)=(\sup_{n\in\Nat}x_n,\sup_{n\in\Nat}y_n)\).
  We have
  \begin{align*}
    f(\sup_{n\in\Nat}(x_n,y_n))
    &=\sup_{n\in\Nat}\sup_{k\in\Nat}f(x_n,y_k)
      \text{\quad by separate continuity}\\
    &=\sup_{n\in\Nat}f(x_n,y_n)
  \end{align*}
  because \(f\) is increasing.
\end{proof}

\begin{theorem}
  \label{th:alg-cones-equalizers}
  The category \(\CONES\) has all binary equalizers and therefore is
  complete.
  Moreover, if \((E,e\in\CONES(E,P))\) is the equalizer of
  \(f,g\in\CONES(P,Q)\) then \(e\) reflects the order relation: if
  \(x,y\in E\) satisfy \(e(x)\leq_P e(y)\) then \(x\leq_E y\).
\end{theorem}
\begin{proof}
  Let \(f,g\in\CONES(P,Q)\).
  Let \(E=\{x\in P\St f(x)=g(x)\}\).
  We equip \(E\) with the algebraic operations of \(P\) which makes
  sense since \(f\) and \(g\) are linear: if \(x_1,x_2\in E\) then
  \(f(x_1+x_2)=f(x_1)+f(x_2)=g(x_1)+g(x_2)=g(x_1+x_2)\) so that
  \(x_1+x_2\in E\) and similarly for scalar multiplication.
  Next, for \(x\in E\), we set \(\Norm x_E=\Norm x_P\) which easily
  satisfies \Cnormhr{}, \Cnormzr{}, \Cnormtr{} and \Cnormpr{}.
  Let \(x_1,x_2\in E\).
  It is obvious that \(x_1\leq_E x_2\Implies x_1\leq_P x_2\).
  Conversely assume that \(x_1\leq_Px_2\) so that \(x_2-x_1\) exists
  in \(P\), we have %
  \(f(x_2-x_1)=f(x_2)-f(x_2)\) by linearity of \(f\) and similarly
  \(g(x_2-x_1)=g(x_2)-g(x_2)\) and hence \(x_2-x_1\in E\) so that
  \(x_1\leq_E x_2\).
  We have proven that \(x_1\leq_Ex_2\Equiv x_1\leq_Px_2\).

  We prove that the norm of \(E\) satisfies \Cnormcr{}, so let
  \((x_n)_{n\in\Nat}\) be a sequence which is increasing in \(E\) and
  such that \(\forall n\in\Nat\ \Norm{x_n}_E\leq 1\).
  Then this sequence is increasing in \(P\) and satisfies
  \(\forall n\in\Nat\ \Norm{x_n}_P\leq 1\) and hence it has a supremum
  \(x\in P\) such that \(\Norm x_P\leq 1\).
  Moreover by continuity of \(f\) and \(g\) we have \(x\in E\).
  We have \(\forall n\in\Nat\ x_n\leq_P x\) and hence
  \(\forall n\in\Nat\ x_n\leq_E x\).
  Let \(y\in E\) be such that \(\forall n\in\Nat\ x_n\leq_E y\), we
  have \(\forall n\in\Nat\ x_n\leq_P y\) and hence \(x\leq_P y\) which
  implies \(x\leq_E y\).
  This shows that \(x\) is the supremum of the \(x_n\)'s in \(E\), and
  since \(\Norm x_E=\Norm x_P\leq 1\), we have proven \Cnormcr{} and
  hence \(E\) is a cone.
  Let \(e:E\to P\) be the inclusion of \(E\) into \(P\), it is clear
  that \(e\in\CONES(E,P)\).

  Finally, if \(h\in\CONES(Q,P)\) satisfies \(f\Compl h=g\Compl h\),
  we have that \(h(u)\in E\) for all \(u\in Q\) by definition of
  \(E\).
  So \(h=e\Compl h'\) where \(h':Q\to E\) is defined exactly as \(h\)
  (the only difference is the codomain), and we have
  \(h'\in\CONES(Q,E)\) since the operations in \(E\) are defined as in
  \(P\).
  The uniqueness of \(h'\) results from the injectivity of \(e\).
  So \((E,e)\) is the equalizer of \(f\) and \(g\).

  It follows that \(\CONES\) is complete since it has also all small
  products.
\end{proof}

\begin{lemma}
  \label{lemma:norm-inv-iso-cones}
  Let \(P\) and \(Q\) be cones and assume that \(P\not=\Zero\).
  Let \(f:P\to Q\) be linear and continuous, and bijective.
  Then \(\Norm f\not=0\) and \(\Norm{\Inv f}\geq\Inv{\Norm f}\).
\end{lemma}
\begin{proof}
  By assumption there is \(x\in P\) such that \(x\not=0\) and since
  \(f\) is bijective and \(f(0)=0\), we have that \(Q\not=0\).
  If follows that \(\Norm f\not=0\) and similarly \(\Norm{\Inv f}\not=0\).
  Let \(\epsilon>0\), we can find \(x\in P\) such that
  \(\Norm{f}\geq\frac{\Norm{f(x)}}{\Norm x}-\epsilon\) hence %
  \(\frac{\Norm x}{\Norm{f(x)}}\geq\frac 1{\Norm f+\epsilon}\).
  Setting \(y=f(x)\), we have %
  \(\frac{\Norm{\Inv f(y)}}{\Norm y}\geq\frac 1{\Norm f+\epsilon}\) and
  hence %
  \(\Norm{\Inv f}\geq\frac 1{\Norm f+\epsilon}\) and since this holds
  for all \(\epsilon>0\) we get \(\Norm{\Inv f}\geq\Inv{\Norm f}\).
\end{proof}

\begin{proposition}
  If \(f\in\CONES(P,Q)\) is an iso and \(P\not=\Zero\), then
  \(\Norm f=1\).
\end{proposition}
\begin{proof}
  By Lemma~\ref{lemma:norm-inv-iso-cones} we have
  \(\Norm{\Inv f}\geq\Inv{\Norm f}\) and since \(f\) is a morphism
  \(\Norm f\leq 1\) so \(\Norm{\Inv f}\geq 1\).
  Hence \(\Norm{\Inv f}=1\) since \(\Inv f\) is a morphism.
  Applying this to \(\Inv f\) we get \(\Norm f=1\).
\end{proof}
\noindent 
The following technical lemma will be useful for proving that our
category of integrable cones is well-powered in
Theorem~\ref{th:icones-conditions-saft}, a crucial property for being
able to apply the special adjoint functor theorem.
\begin{lemma}
  \label{lemma:alg-cone-transport}
  Let \(P\) be a cone, \(S\) be a set and \(f:P\to S\) be a bijective
  function.
  There is exactly one cone structure on \(S\) for which \(f\) becomes
  an iso in \(\CONES\).
\end{lemma}
\begin{proof}
  Given \(s_1,s_2\in S\), we set
  \(s_1+s_2=f(\Inv f(s_1)+\Inv f(s_2))\) and similarly
  \(\lambda s=f(\lambda\Inv f(s))\) for \(s\in S\) and
  \(\lambda\in\Realp\).
  And we set \(\Norm s_S=\Norm{\Inv f(s)}_P\).
  It is straightforward that one defines a cone in that way, and that
  \(f\) is an iso.
  It is also obvious that this structure of cone one \(S\) is the
  only one such that \(f\) is an iso.
\end{proof}

\section{Measurable cones}
\label{sec:measurable-cones}
Let \(\ARCAT\) be a \emph{small}%
\footnote{This assumption is crucial for making the use of the special
  adjoint functor theorem possible.} %
full subcategory of \(\MEAS\) (the category of measurable spaces and
measurable functions) which is closed under cartesian products and
contains the terminal object \(\Measterm\) which is the one point
measurable space (we use this notation because the one point
measurable space is geometrically \(0\)-dimensional).
We also assume all the objects of \(\ARCAT\) to be non-empty
measurable spaces.
We also use \(\ARCAT\) for the set of all objects of the category
\(\ARCAT\).

\begin{remark}\label{rmk:ar-main-example}
  In most situations we could assume that \(\ARCAT\) has all finite
  products \(\Real^n\) as objects, and measurable maps as morphisms.
  We could even assume that \(\ARCAT\) has \(\Real\) as single object,
  or more precisely, two objects: \(\Real\) and \(\Measterm\) since
  all the \(\Real^n\) are isomorphic%
  \footnote{Such isomorphisms involve however non canonical encoding
    methods so we prefer to avoid using this property explicitly.} %
  to \(\Real\) in \(\MEAS\) for \(n>0\). %
\end{remark}

\begin{definition}
  \label{def:meas-structure}
  A \emph{measurability structure} on a cone \(P\) is a family %
  \(\cM=(\cM_X)_{X\in\ARCAT}\) with %
  \(\cM_X\subseteq(\Cdual P)^{X}\) (where we recall that
  \(\Cdual P=\Limplp P\Sbot\) is the dual of the cone \(P\)) %
  satisfying the four next conditions %
  \Msmesr{}, \Mscompr{}, \Mssepr{} and \Msnormr{}.
  When \(X=\Measterm\) we consider \(m\in\cM_X\) as an element of %
  \(\Cdual P\).
  \begin{Axicond}\Msmesl{}
    For each \(m\in\cM_X\) and \(x\in \Cuball P\), one has
    \(\Absm{r\in X}{m(r,x)}\in\MEAS(X,\Intcc01)\) where
    \(\Intcc01\subseteq\Real\) is equipped with its standard Borel
    \(\sigma\)-algebra.
    This implies in particular that if \(r\in X\), then %
    \(\Absm{x\in P}{m(r,x)}\in\CONES(P,\Sone)\).
  \end{Axicond}
  \begin{Axicond}\Mscompl{}
    For each \(m\in\cM_X\) and \(\phi\in\ARCAT(Y,X)\) one has %
    \( \Absm{(s,x)\in(Y\times P)}{m(\phi(s),x)}
    =m\Comp(\phi\times P) \in\cM_Y\).
  \end{Axicond}
  \noindent
  In particular, since \(\Measterm\) is the terminal object of
  \(\ARCAT\), each element \(m\in\cM_\Measterm\) induces an element %
  \(\Absm{(r,x)\in(X\times P)}{m(x)}\in\cM_X\) and in this way we
  consider \(\cM_\Measterm\) as a subset of \(\cM_X\) for each
  \(X\in\ARCAT\).
  \begin{Axicond}\Mssepl{}
    If \(x_1,x_2\in P\) satisfy %
    \(\forall m\in\cM_0\ m(x_1)= m(x_2)\) then \(x_1= x_2\).
  \end{Axicond}
\begin{Axicond}\Msnorml{}
  For all \(x\in P\), one has %
  \[
    \Norm x=\sup\Big\{\frac{m(x)}{\Norm m}\St m\in\cM_0\text{ and }m\not=0\Big\}
  \]
  or, equivalently,
  \(\Norm x\leq\sup\{m(x)/\Norm m\St m\in\cM_0\text{ and
  }m\not=0\}\). %
\end{Axicond}
\end{definition}
Indeed, for each \(x'\in\Cdual P\setminus\Eset 0\) and \(x\in P\) one
has %
\(\Norm x\geq\frac{\Eval x{x'}}{\Norm{x'}}\).
\begin{remark}
  The condition \Msnormr{} can also be formulated as follows: for
  each \(x\in P\setminus\Eset 0\) and for each \(\epsilon>0\) there
  exists %
  \(m\in\cM_0\setminus\Eset 0\) such that
  \(\Norm x\leq\frac{m(x)}{\Norm m}+\epsilon\).
  The condition that \(x\not=0\) is required because we could possibly
  have \(\cM_0=\Eset 0\), but in that situation, by \Mssepr{} we must
  have \(P=\Eset 0\).
  The condition \Msnormr{} was absent in~\cite{Ehrhard20} which made it
  much more difficult to prove that the category of cones and linear
  measurable cones is symmetric monoidal: we had to use a property of
  density of the category of PCSs.
  \Msnormr{} makes the whole theory much better behaved.
\end{remark}

\begin{remark}
  \label{rk:small-test-set}
  We do not require \(\cM_0\) to be the whole unit ball of the dual
  \(P'\), but only a subset of it, sufficiently large for satisfying
  our requirements \Mscompr{}, \Mssepr{} and \Msnormr{}.
  As we will see in various constructions, \(\cM_0\) will often be a
  very small part of this unit ball.
\end{remark}

\begin{remark}
  Instead of~\Mssepr{} we could also consider the following stronger
  separation condition: %
  if for all \(m\in\cM_0\) one has \(m(x)\leq m(y)\) then \(x\leq y\).
  However this would complicate the definition of the measurability
  structures of the spaces of stable and measurable functions in
  Section~\ref{sec:icone-stable-fun} and of analytic functions in
  Section~\ref{sec:analytic-functions-exp}.
  This stronger separability does not seem to be necessary (at least
  for the purpose of what we do in this paper) but one should keep in
  mind that all our constructions could be performed within this
  restricted class.
\end{remark}

\begin{definition}[Measurable cone]
  A \emph{measurable cone} is a pair %
  \(C=(\Mcca C,\Mcms C)\) where \(\Mcca C\) is a cone and %
  \(\Mcms C\) is a measurability structure on \(\Mcca C\).
\end{definition}

\noindent 
The main purpose of the measurability structure of a measurable cone
is to equip the underlying cone with a structure of QBS by defining a
collection of paths ranging in the cone.

\begin{definition}[Measurable path]
  Let \(X\in\ARCAT\) and let \(C\) be a measurable cone.
  A \emph{(measurable) path} of arity \(X\) is a function
  \(\gamma:X\to\Mcca C\) which is bounded and such that, for each %
  \(Y\in\ARCAT\) and \(m\in\Mcms C_{Y}\), the function %
  \(\Absm {(s,r)\in{Y\times X}}{m(s,\gamma(r))}: {Y\times
    X}\to\Realp\) is measurable.
  We use \(\Mcca{\Cpath XC}\) for the set of measurable paths of arity
  \(X\) of the measurable cone \(C\).
\end{definition}

\begin{remark}
  Measurable paths should be thought of as a generalization of finite
  kernels, and in the case where \(C\) is the measurable cone of
  finite measures on a measurable space \(Y\), each bounded kernel from
  \(X\) to \(Y\) is a measurable path from \(X\) to \(C\).
  One of the purposes of introducing integrals is to make the converse
  true.
\end{remark}

\begin{lemma} %
  \label{lemma:cst-path}
  Let \(x\in\Mcca C\) and \(\gamma=\Absm{r\in X}{x}:X\to\Mcca C\) be
  the constant function.
  Then \(\gamma\) is a measurable path.
\end{lemma}
\noindent 
This immediately results from the definitions.

\begin{lemma} %
  \label{lemma:precomp-path}
  Let \(\gamma:X\to\Mcca C\) be a measurable path and let %
  \(\phi\in\ARCAT(Y,X)\) for some \(Y\in\ARCAT\). Then %
  \(\gamma\Comp\phi:Y\to\Mcca C\) is also a measurable path.
\end{lemma}
\begin{proof}
  Let \(Y'\in\ARCAT\) and \(m\in\Mcms C_{Y'}\), we have %
  \(\Absm{(s',s)\in{Y'\times Y}}{m(s',\gamma(\phi(s)))}
  =(\Absm{(s',r)\in{Y'\times X}}{m(s',\gamma(r))})
  \Comp({Y'}\times\phi)\)
  which is measurable as the composition of two measurable maps.
\end{proof}

\noindent 
We turn the cone \(\Sone=\Sbot\) into a measurable cone by defining %
\(\Mcms\Sone_X\) as the set of all functions mapping each element
\(r\in X\) to \(\Id:\Realp\to\Realp\) for all \(X\in\ARCAT\). 

\begin{proposition} %
  \label{th:norm-dual}
  Let \(B\) be a measurable cone and let \(x\in\Mcca B\).
  Then
  \begin{align*}
    \Norm x=\sup_{x'\in\Cuball{\Cdual{\Mcca B}}}\Eval x{x'}\,.
  \end{align*}
\end{proposition}
\begin{proof}
  By definition of the norm in \(\Cdual{{\Mcca B}}\) we have %
  \(\Norm x\geq\Eval x{x'}\) for all \(x'\in\Cuball{\Cdual{\Mcca B}}\).
  We can assume that \(x\not=0\) since otherwise the announced
  equation trivially holds.
  Let \(\epsilon>0\) and \(m\in\Mcms B_\Measterm\setminus\Eset 0\) %
  be such that %
  \(\Norm x\leq\frac{m(x)}{\Norm m}+\epsilon\).
  Let \(x'=m/\Norm m\), we have \(x'\in\Cuball{\Cdual{\Mcca B}}\).
  Since \(\Norm x\leq\Eval{x}{x'}+\epsilon\) our contention is proven.
\end{proof}

\begin{remark}
  One main purpose of the condition \Msnormr{} is to get the above
  highly desirable property.
  We could have expected to get \Mssepr{} and \Msnormr{} for free by
  means of a Hahn Banach theorem for cones as in~\cite{Selinger04}.
  However, the counter-example of Remark~\ref{rk:cone-non-separe},
  suggested to us by one of the reviewers of this paper, shows that such a
  separation property does not hold in our setting.
  The very nice Hahn Banach theorem proven in~\cite{Selinger04} relies
  on the assumption that cones are continuous domains, an assumption
  that we cannot afford here because we need our cones to define a
  complete category in order to apply the special adjoint functor
  theorem which is our main tool for equipping \(\ICONES\) with a
  tensor product and with an exponential.
  Fortunately, we can take this Hahn Banach separation property as one
  of our axioms on the measurability tests, and proving that this
  property is preserved by all limits does not induce noticeable
  technical difficulties.
\end{remark}

\subsection{The category of measurable cones and linear, continuous and measurable maps}
We can now define our first main category of interest.
\begin{definition}
  The category \(\MCONES\) has measurable cones as objects and an
  element of \(\MCONES(B,C)\) is an \(f\in\CONES(\Mcca B,\Mcca C)\)
  such that for each \(X\in\ARCAT\) and each measurable path %
  \(\beta:X\to\Mcca B\) the function \(f\Comp\beta\) is a measurable
  path.
  Equivalently %
  \begin{align*}
    \forall Y\in\ARCAT\ \forall m\in\Mcms C_Y
    \quad\Absm{(s,r)\in{X\times Y}}{m(s,f(\beta(r)))}
    \text{ is measurable.}
  \end{align*}
\end{definition}

\begin{remark} %
  \label{rk:mcones_iso}
  An isomorphism \(f\in\MCONES(B,C)\) is a bijection %
  \(f:\Mcca B\to\Mcca C\) which is linear and continuous, satisfies
  \(\forall x\in\Mcca B\ \Norm{f(x)}_C=\Norm x_B\) and, for each
  \(X\in\ARCAT\) and each function \(\beta:X\to\Mcca B\), one
  has
  \(\beta\in\Mcca{\Cpath XB}\Equiv f\Comp\beta\in\Mcca{\Cpath XC}\).
  This means that \(B\) and \(C\) can be isomorphic even if the
  measurability structures \(\Mcms C\) and \(\Mcms D\) are quite
  different: it suffices that they induce the same measurable paths.
\end{remark}

\begin{definition}%
  \label{def:mes-cone-homothetie}  
  Let \(B\) be a measurable cone and \(\alpha\in\Real\) with
  \(\alpha>0\).
  Then \(\alpha B\) is the measurable cone which is defined exactly as
  \(B\) apart for the norm which is given by %
  \(\Norm x_{\Mcca{\alpha B}}=\Inv\alpha\Norm x_{\Mcca B}\).
\end{definition}
\noindent 
Notice that
\(\Cuball{(\Mcca{\alpha B})}=\alpha\Cuball{\Mcca B}=\Eset{x\in\Mcca
  B\St\Norm x_{\Mcca B}\leq\alpha}\).

\subsection{Examples: the measurable cones of measures and paths}
\label{sec:examples-of-cones}
We introduce two important examples of measurable cones.
The measurable cone of finite measures on an object of \(\ARCAT\) will
allow us to understand \(\ARCAT\) as our category of basic data-types.

As noticed by one of the reviewers, it would not be strictly necessary
to introduce the measurable cone of measurable paths since we will see
that, in the setting of integrable cones, the cone of measurable paths
can be described as an internal linear hom, see
Theorem~\ref{th:meas-path-equiv}.
Our motivations for presenting this construction are:
\begin{itemize}
\item it illustrates for the first time the reason why our tests have
  parameters in objects of \(\ARCAT\);
\item it is a simple and natural construction, quite different from
  the cone of finite measures (and somehow dual to it), and completely
  independent from our integrability assumptions.
\end{itemize}

\subsubsection{The measurable cone of finite measures} %
\label{sec:cone-finite-meas}

Let \(X\) be a measurable space (not necessarily in \(\ARCAT\)).
Recall that in Section~\ref{sec:basic-cone-finite-meas} we defined the
cone \(\Mcca {\Cmeas(X)}\) of all finite measures on \(X\).

For all \(Y \in \ARCAT\) and all \(U \in \Sigalg{X}\) we define
\(\Emeas U:Y \times \Mcca {\Cmeas(X)}\to\Realp\) by
\(\Emeas U(s, \mu)=\mu(U)\).
Then we define \(\cM_Y=\Eset{\Emeas U\St U \in \Sigalg{X}}\), and
\(\Cmeas(Y)=(\Mcca {\Cmeas(X)},(\cM_Y)_{Y\in\ARCAT})\) is clearly a
measurable cone.

\begin{remark}
  \label{rk:measure-cone-two-ms}
  We could have taken another measurability structure as follows.
  For all \(Y \in \ARCAT\) and all \(W \in \Sigalg{Y\times X}\) we
  define \(\Emeas W:Y \times \Mcca {\Cmeas(X)}\to\Realp\) by
  \(\Emeas W(s, \mu)=\mu(\{r \St (s, r) \in W\})\).
  Then we define
  \(\cM_Y=\Eset{\Emeas W\St W \in \Sigalg{Y\times X}}\).
  Then \(\Cmeas(X)=(\Mcca {\Cmeas(X)},(\cM_Y)_{Y\in\ARCAT})\) is a
  measurable cone.
  As easily checked, these two measurability structures define exactly
  the same measurable paths on \(\Cmeas(X)\).
  This example shows that a given cone (namely \(\Mcca {\Cmeas(X)}\))
  can be given two distinct measurability structures.
  This is also an example of the situation mentioned in
  Remark~\ref{rk:mcones_iso}: the measurability cones defined by these
  two measurability structures are isomorphic in the category \(\MCONES\).
\end{remark}

Notice that if \(Y\in\ARCAT\) then a measurable path
\(\gamma:Y\to\Mcca{\Cmeas(X)}\) is the same thing as a bounded
kernel from \(Y\) to \(X\).

Let \(\phi\in\MEAS(X,Y)\) (remember that this means that \(\phi\) is
a measurable function \(X\to Y\)), then given %
\(\mu\in\Mcca{\Cmeas(X)}\) we can define %
\(\nu={\Pushf\phi(\mu)}\in\Mcca{\Cmeas(Y)}\) by %
\(\nu(V)=\mu(\Inv\phi(V))\) for each \(V\in\Sigalg Y\) (the push-forward
of \(\mu\) along \(\phi\)).
\begin{lemma}
  \label{lemma:pushf-measurable}
  We have \(\Pushf\phi\in\MCONES(\Cmeas(X),\Cmeas(Y))\).
  The operation \(\Cmeas\) on measurable cones extends to a functor %
  \(\Funpushf:\ARCAT\to\MCONES\), acting on morphisms by measure
  push-forward: \(\Cmeas(\phi)=\Pushf\phi\).
\end{lemma}
\begin{proof}
  Linearity and continuity being obvious, as well as the fact that %
  \(\Norm{\Cmeas(f)}\leq 1\), we only have to check measurability.
  Let \(\kappa : Y'\to\Mcca{\Cmeas(X)}\) be a measurable
  path.
  We must prove that \(\kappa'=\Cmeas(\phi)\Comp\kappa\) is a
  measurable path.
  Let \(p\in\Mcms{\Cmeas(Y)}_{Y''}\) for some \(Y''\in\ARCAT\), that
  is \(p=\Emeas V\) for some \(V\in\Sigalg Y\).
  For \((s'',s')\in Y''\times Y'\) we have
  \(p(s'',\kappa'(s'))=\kappa'(s')(V)=\kappa(s')(\Inv\phi(V))\) and
  hence \(\Absm{(s'',s')\in Y''\times Y'}{p(s'',\kappa'(s'))}\) is
  measurable because \(\kappa\) is a kernel and \(\phi\) is
  measurable.
  Functoriality of \(\Cmeas\) is obvious.
\end{proof}

\subsubsection{The measurable cone of paths}
\label{sec:meas-cone-of-paths}
Let \(C\) be an object of \(\MCONES\) and \(X\in\ARCAT\). %
Let \(P\) be the set of all measurable paths %
\(\gamma:X\to\Mcca C\).
We turn \(P\) into a precone by defining the algebraic laws in the
obvious pointwise manner.
For instance let \(\gamma_1,\gamma_2\in P\), we define
\(\gamma=\gamma_1+\gamma_2\) by \(\gamma(r)=\gamma_1(r)+\gamma_2(r)\)
which is bounded by \Cnormtr.
To check measurability, take \(m\in\Mcms C_Y\), we have %
\( \Absm {(s,r)\in Y\times X}{m(s,\gamma(r))} =\Absm {(s,r)\in Y\times
  X}{m(s,{\gamma_1}(r))} +\Absm {(s,r)\in Y\times
  X}{m(s,{\gamma_2}(r))} \) (pointwise addition) by linearity of \(m\)
in its second parameter, which is measurable in \(r\) by continuity of
addition on \(\Realp\).

Then we have \(\gamma_1\leq\gamma_2\) iff %
\(\forall r\in X\ \gamma_1(r)\leq\gamma_2(r)\): it suffices to check
that, when this latter condition holds, the map
\(\Absm{r\in X}{(\gamma_2(r)-\gamma_1(r))}\) is a path which results
from the continuity (and hence measurability) of subtraction of real
numbers.

Given \(\gamma\in P\) we set
\begin{align*}
  \Norm\gamma=\sup_{r\in X}\Norm{\gamma(r)}\in\Realp
\end{align*}
which is well defined by our assumption that \(\gamma\) is
bounded.
This satisfies all the required conditions for turning \(P\) into a
cone, the only non obvious one being \(\Cnormcr\).
So let \((\gamma_n)_{n\in\Nat}\) be an increasing sequence of elements of
\(P\) such that %
\(\forall n\in\Nat\,\forall r\in X\ \Norm{\gamma_n(r)}\leq
1\).
We define \(\gamma:X\to P\) by %
\(\gamma(r)=\sup_{n\in\Nat}\gamma_n(r)\in\Cuball{\Mcca C}\) which is
well defined since for each \(r\in X\) the sequence
\((\gamma_n(r))_{n\in\Nat}\) is increasing in \(\Cuball{\Mcca C}\).
It suffices to check that \(\gamma\) satisfies the measurability
condition, so let \(Y\in\ARCAT\) and \(m\in\Mcms C_Y\), we have by
\(\omega\)-continuity of \(m\) in its second argument
\begin{align*}
  \Absm{(s,r)\in Y\times X}{m(s,\gamma(r))}=
  \Absm{(s,r)\in Y\times X}{\sup_{n\in\Nat}m(s,\gamma_n(r))}
\end{align*}
which is measurable by the monotone convergence theorem of measure
theory (observing that
\((\Absm{(s,r)\in Y\times X}{m(s,\gamma_n(r))})_{n\in\Nat}\) is an increasing
sequence of measurable functions \(Y\times X\to\Intcc 01\)).

\begin{remark}
  Remember that it is precisely for being able to prove this kind of
  properties that we assume the unit balls of cones to be complete
  only for increasing chains and not for arbitrary directed sets.
\end{remark}
\noindent 
So far we have equipped \(P\) (the set of measurable paths from \(X\)
to \(C\)) with a structure of cone in the algebraic sense of
Section~\ref{sec:algebraic-cones}.
We equip now this cone with a measurability structure.
This definition will illustrate, for the first time in this paper, the
usefulness of the ``additional'' parameter of tests, spanning
measurable spaces taken in \(\ARCAT\).

Let \(Y\in\ARCAT\), \(\phi\in\ARCAT(Y,X)\) and
\(m\in\Mcms C_Y\), we define
\begin{align*}
  \Mtpath\phi m: Y\times P&\to\Realp\\
  (s,\gamma)&\mapsto m(s,\gamma(\phi(s)))\,.
\end{align*}
Observe first that for each \(s\in Y\), the map %
\(\Absm{\gamma\in P}{(\Mtpath\phi m)}(s,\gamma)\) is
linear and continuous by linearity and continuity of \(m\) in its
second argument.
We check that the family %
\((\cM_Y\subseteq {P'}^{Y})_{Y\in\ARCAT}\) %
defined by %
\( \cM_Y =\{\Mtpath\phi m\St \phi\in\ARCAT(Y,X)\text{ and
}m\in\Mcms{C}_Y\} \) is a measurability structure on \(P\).

\Proofcase\Msmesr{}.
Let \(p\in\cM_Y\) and \(\gamma\in P\), so that \(p=\Mtpath\phi m\) for
some \(\phi\in\ARCAT(Y,X)\) and \(m\in\Mcms C_Y\), then let %
\( \theta=\Absm{s\in Y}{p(s,\gamma)} =\Absm{s\in
  Y}{m(s,\gamma(\phi(s)))}\). %
We know that %
\(\psi=\Absm{(s,r)\in Y\times X}{m(s,\gamma(r))}\) is measurable %
\(Y\times X\to\Intcc 01\) and hence %
\(\theta=\psi\Comp\Tuple{Y,\phi}\) is measurable %
\(Y\to\Intcc01\) since \(\phi\) is measurable.

\Proofcase\Mscompr{}.
Let \(p\in\cM_Y\) and \(\psi\in\ARCAT(Y',Y)\).
We have %
\(p=\Mtpath\phi m\) for some \(\phi\in\ARCAT(Y,X)\) and %
\(m\in\Mcms C_Y\).
Then we have %
\(
p\Comp(\psi\times P)
=\Mtpath{(\phi\Comp\psi)}{(m\Comp(\psi\times\Mcca C))}
\in\cM_{Y'}
\).

\Proofcase\Mssepr{}.
Let \(\gamma_1,\gamma_2\in P\) and assume that %
\(\forall p\in\cM_0\ p(\gamma_1)=p(\gamma_2)\). %
Let \(r\in X\) that we consider as an element of %
\(\ARCAT(\Measterm,X)\).
Let \(m\in\Mcms C_\Measterm\), by our assumption we have %
\((\Mtpath rm)(\gamma_1)=(\Mtpath rm)(\gamma_2)\), that is %
\(m(\gamma_1(r))=m(\gamma_2(r))\) and since this holds for all %
\(m\in\Mcms C_\Measterm\) we have \(\gamma_1(r)=\gamma_2(r)\) by \Mssepr{}
in \(C\).

\Proofcase\Msnormr{}.
Let \(\gamma\in P\setminus\Eset 0\) and \(\epsilon>0\).
We can find \(r\in X\) such that \(\gamma(r)\not=0\) and
\(\Norm\gamma\leq\Norm{\gamma(r)}+\frac\epsilon 2\). %
By \Msnormr{} holding in \(C\) we can find
\(m\in\Mcms C_\Measterm\setminus\Eset 0\) such that %
\(\Norm{\gamma(r)}\leq\frac{m(\gamma(r))}{\Norm m}+\frac\epsilon
2\).
Remember that \(\Mtpath rm\in\cM_\Measterm\)
and notice that %
\(\Norm{\Mtpath rm}=\sup\{m(\delta(r))\St\delta\in\Cuball P\}=\Norm
m\) %
by Lemma~\ref{lemma:cst-path}.
So we have %
\begin{align*}
  \Norm\gamma
  \leq\Norm{\gamma(r)}+\frac\epsilon 2
  \leq\frac{(\Mtpath rm)(\gamma)}{\Norm{\Mtpath rm}}+\epsilon
\end{align*}
and hence %
\(\Norm\gamma=\sup\{\frac{p(\gamma)}{\Norm p}\St p\in\cM_0\text{ and
}p\not=0\}\) as required since
\(\cM_\Measterm
=\{\Mtpath rm\St r\in X\text{ and }m\in\Mcms C_\Measterm\}\).

We use \(\Cpathm XC\) for the measurable cone \((P,\cM)\) defined above.
We end this section with the following lemma which will be useful when
dealing with the tensor product of measurable cones, and in particular
for proving Theorem~\ref{th:fmeas-prod-tensor}.
\begin{lemma} %
  \label{lemma:meas-path-flat}
  Let \(B\) be a cone and \(X,Y\in\ARCAT\).
  There is an iso %
  \[
    \Flpath_{X,Y}\in\MCONES(\Cpathm X{\Cpathm Y B},\Cpathm{X\times Y}B)
  \]
  which ``flattens'' %
  \(\eta\in\Mcca{\Cpathm X{\Cpathm Y B}}\) into %
  \(\Flpath_{X,Y}(\eta)=\Absm{(r,s)\in{X\times Y}}{\eta(r)(s)}\). %
  As a consequence %
  \[
    \Inv{\Flpath_{Y,X}}\Compl\Flpath_{X,Y}\in
    \MCONES(\Cpathm X{\Cpathm Y B},\Cpathm Y{\Cpathm X B})\,,
  \]
  the function which swaps the parameters of a path of paths, is an
  iso in \(\MCONES\).
\end{lemma}
\begin{proof}
  Let \(\eta\in\Mcca{\Cpathm X{\Cpathm Y B}}\), we need first to prove
  that %
  \(\eta'=\Flpath(\eta)\in\Mcca{\Cpathm{X\times Y}B}\) so let %
  \(Y'\in\ARCAT\) and let \(m\in\Mcms B_{Y'}\), we must prove that
  \begin{align*}
    \phi=\Absm{(s',r,s)\in{Y'\times X\times Y}}{m(s',\eta'(r,s))}
    =\Absm{(s',r,s)\in{Y'\times X\times Y}}{m(s',\eta(r)(s))}
  \end{align*}
  is measurable.
  Let %
  \(m'=m\Comp(\Proj 1\times\Mcca B)\in\Mcms B_{Y'\times Y}\) (that is
  \(m'(s',s,x)=m(s',x)\)) so that %
  \(\Mtpath{\Proj 2}{m'}\in\Mcms{\Cpathm YB}_{Y'\times Y}\), we know that %
  \(\Absm{(s',s,r)\in{Y'\times Y\times X}}{(\Mtpath{\Proj
      2}{m'})(s',s,\eta(r))}
  =\Absm{(s',s,r)\in{Y'\times Y\times X}}{m(s',\eta(r)(s))}\) %
  is measurable from which it follows that \(\phi\) is measurable. %
  Moreover it is clear that %
  \(\eta'({X\times Y})\subseteq\Norm\eta\Cuball{\Mcca B}\) is bounded
  in %
  \(\Mcca B\) and hence \(\eta'\in\Mcca{\Cpathm{X\times Y}B}\) as announced.

  The linearity and \(\omega\)-continuity of \(\Flpath\) are clear so we
  check its measurability.
  Let \(Y'\in\ARCAT\) and let %
  \(\eta\in\Mcca{\Cpathm{Y'}{\Cpathm{X}{\Cpathm YB}}}\), we must prove
  that %
  \[
    \Flpath\Comp\eta\in\Mcca{\Cpathm{Y'}{\Cpathm{X\times Y}B}}\,.
  \]
  So let %
  \(Y''\in\ARCAT\) and let %
  \(p\in\Mcms{\Cpathm{X\times Y}B}_{Y''}\).
  Let \(\phi'=\Tuple{\phi,\psi}\in\ARCAT(Y'',X\times Y)\) and %
  \(m\in\Mcms{B}_{Y''}\) be such that %
  \(p=\Mtpath{\phi'}m\), we have that %
  \begin{align*}
    \phi''&=\Absm{(s'',s')\in{Y''\times Y'}}{p(s'',\Flpath(\eta(s')))}\\
    &=\Absm{(s'',s')\in{Y''\times Y'}}{m(s'',\Flpath(\eta(s'))
      (\phi(s''),\psi(s'')))}\\
    &=\Absm{(s'',s')\in{Y''\times Y'}}{m(s'',\eta(s')
      (\phi(s''))(\psi(s'')))}
  \end{align*}
  is measurable because %
  \begin{align*}
    \phi''
    =\Absm{(s'',s')\in{Y''\times Y'}}
    {(\Mtpath{\phi}{(\Mtpath{\psi}{m})})
    (s'',\eta(s'))}
  \end{align*}
  and by our assumption about \(\eta\).
  Last notice that %
  \(\Norm{\Flpath(\eta)}=\Norm{\eta}\) which shows that %
  \(\Flpath\in\MCONES(\Cpathm X{\Cpathm Y B},\Cpathm{X\times Y}B)\).

  As to the converse direction, given %
  \(\eta\in\Mcca{\Cpathm{X\times Y}B}\) let %
  \(\Flpath'(\eta)=\Absm{r\in X}{\Absm{s\in Y}{\eta(r,s)}}\), %
  we must first prove that %
  \(\Flpath'(\eta)\in\Mcca{\Cpathm X{\Cpathm Y B}}\), we just check
  measurability, boundedness being obvious. Let %
  \(p\in\Mcms{\Cpathm YB}_{Y'}\) for some \(Y'\in\ARCAT\).
  Let \(\phi\in\ARCAT(Y',Y)\) and \(m\in\Mcms B_{Y'}\) be such that
  \(p=\Mtpath{\phi}{m}\), we must prove that %
  \(\psi=\Absm{(s',r)\in{Y'\times X}}{p(s',\Flpath'(\eta)(r))}
  =\Absm{(s',r)\in{Y'\times X}}{m(s',\eta(r,\phi(s')))}\) is
  measurable.
  This follows from the fact that \(\phi\) and %
  \(\Absm{(s',r,s)\in{Y'\times X\times Y}}{m(s',\eta(r,s))}\) are
  measurable, the latter by our assumption about \(\eta\).

  Checking that \(\Flpath'\) is a morphism in \(\MCONES\) follows
  exactly the same pattern as for \(\Flpath\), using the obvious
  bijection between %
  \(\Mcms{\Cpathm{X}{\Cpathm YB}}_{Y'}\) and %
  \(\Mcms{\Cpathm{X\times Y}{B}}_{Y'}\) induced by the fact that %
  \(\ARCAT\) is cartesian.
  Finally the observation that %
  \(\Flpath'=\Inv\Flpath\) shows that \(\Flpath\) is an iso in
  \(\MCONES\).
\end{proof}

\begin{remark}
  \label{rk:test-are-paramerized}
  So a test on the space of \(C\)-valued and \(X\)-parameterized paths
  is provided by a test \(m\in\Mcms C_Y\) ---~itself parameterized by
  a space \(Y\in\ARCAT\)~--- and a ``variable argument'' which is a
  measurable function \(\phi\) from \(Y\) to the space \(X\).
  When \(\phi\) is not a constant function, the value of
  \((\Mtpath\phi m)(s,\gamma)=m(s,\gamma(\phi(s)))\) depends in
  general on \(s\) when \(\gamma\) is not a constant path, \emph{even
    if the function \(m:Y\times\Mcca C\to\Realp\) does not depend on
    its first argument}.
  This definition of tests in the cones of paths plays a crucial role
  in the proof of %
  Lemma~\ref{lemma:meas-path-flat}.

  Imagine that we want to use a simpler notion of tests, with
  \(\Mcms C\subseteq\Cdual{\Mcca C}\), that is, assume that our tests
  do not have the further parameter taken in a \(Y\in\ARCAT\).
  The first difficulty we face consists in finding a suitable
  definition for \(\Mcms{\Cpathm XC}\).
  The simplest option consists in taking all the %
  \(\Mtpath rm\) where \(r\in X\) and \(m\in\Mcms C\), defined by %
  \((\Mtpath rm)(\beta)=m(\beta(r))\).
  This choice fulfills all the expected separation properties.
  With this definition, an element of \(\Mcca{\Cpathm Y{\Cpathm XC}}\)
  is the same thing as a bounded function
  \(\gamma:Y\times X\to\Mcca B\) such that, for each \(m\in\Mcms C\),
  the function %
  \(\Absm{s\in Y}{m(\beta(s,r))}:Y\to\Realp\) is measurable for all
  \(r\in X\) and the function %
  \(\Absm{r\in X}{m(\beta(s,r))}:X\to\Realp\) is measurable for all
  \(s\in Y\).
  Let us assume that \(C=\Cmeas(Z)\) for some \(Z\in\ARCAT\), let
  \(\phi:Y\times X\to Z\) be a function, and let
  \(\gamma:Y\times X\to\Mcca C\) be given by
  \(\gamma(s,r)=\Dirac Z(\phi(s,r))\).
  Saying that \(\gamma\in\Mcca{\Cpathm Y{\Cpathm XC}}\) means that the
  function \(\phi\) is separately measurable in both arguments, a
  condition which is strictly weaker than measurability on
  \(Y\times X\) in general.
  On the other hand, with our definition of measurability tests for
  \(\Cpathm XC\), Lemma~\ref{lemma:meas-path-flat} tells us that
  \(\gamma\in\Mcca{\Cpathm Y{\Cpathm XC}}\) iff \(\phi\) is measurable
  \(Y\times X\to Z\) for the very simple reason that
  \(Y\times X\in\ARCAT\).
  Notice that we have equipped \(P\), the cone of measurable paths
  from \(X\) to \(C\), with two different measurability structures:
  the original one made of all tests \(\Mtpath\phi m\) where
  \(\phi\in\ARCAT(Y,X)\) and \(m\in\Mcms C_Y\), and the simplified
  one, made of tests \(\Mtpath\phi m\) where \(\phi\in\ARCAT(Y,X)\) is
  a \emph{constant} function and \(m\in\Mcms C_Y\).
  The two measurable cones obtained in that way are not isomorphic%
  \footnote{More precisely, the identity function between these two
    cones is not an isomorphism.} %
  in \(\MCONES\) since the associated measurable paths are distinct as
  we have seen.
  This complements Remarks~\ref{rk:measure-cone-two-ms}
  and~\ref{rk:mcones_iso}.

  We will meet a completely similar definition of tests for the space
  \(\Limpl BC\) of linear, continuous and integrable functions from
  \(B\) to \(C\) in Section~\ref{sec:lin-hom-tensor}.
\end{remark}

\section{Integrable cones} %
\label{sec:int-cones}
We now introduce the main novelties of this paper, which are the
definition of the integral of a measurable path wrt.~a finite measure,
the notion of integrable cone, and the notion of linear, continuous,
measurable and integral preserving functions between integrable cones.

The following definition is quite similar to Definition~2.1
in~\cite{Pettis38} of the integral of a function valued in a
topological vector space.
Our integrals are valued in cones instead of vector spaces.
\begin{definition}
  Let \(B\) be a measurable cone, \(X\in\ARCAT\),
  \(\beta\in\Mcca{\Cpath XB}\) and \(\mu\in\Mcca{\Cmeas(X)}\).
  An \emph{integral of \(\beta\) over \(\mu\)} is an element \(x\) of
  \(\Mcca B\) such that, for all \(m\in\Mcms B_\Measterm\), one has
  \begin{align*}
    m(x)=\int m(\beta(r))\mu(dr)\,.
  \end{align*}
\end{definition}
\noindent 
Notice indeed that \(m\Comp\beta:X\to\Realp\) is a bounded
measurable function so that the integral above is well defined and
belongs to \(\Realp\) (remember that the measure \(\mu\) is finite).
Notice also that by \Mssepr{} if such an integral \(x\) exists, it is
unique, so we can introduce a notation for it, we write
\begin{align*}
  x=\int\beta(r)\mu(dr)\,.
\end{align*}
When we want to stress the cone \(B\) where this integral is computed
we denote it as \(\int^B\beta(r)\mu(dr)\) and when we want to insist
on the measurable space on which the integral is computed we write %
\(\int_{X}\beta(r)\mu(dr)\) or \(\int_{r\in X}\beta(r)\mu(dr)\).

\begin{lemma} %
  \label{lemma:integral-bounded}
  If \(\beta\in\Mcca{\Cpathm XB}\) is integrable over
  \(\mu\in\Cmeas(X)\) then
  \begin{align*}
    \Norm{\int_X\beta(r)\mu(dr)}_B
    \leq\Norm\beta_{\Cpath XB}\Norm\mu_{\Cmeas(X)}\,.
  \end{align*}
\end{lemma}
\begin{proof}
  Let \(x=\int\beta(r)\mu(dr)\). If \(x=0\) there is nothing to prove
  so assume that \(x\not=0\). %
  Let \(\epsilon>0\) and let %
  \(m\in\Mcms B_0\setminus\Eset 0\) be such that %
  \(\Norm x\leq\epsilon+\frac{m(x)}{\Norm m}\), that is
  \begin{align*}
    \Norm x&\leq \epsilon+\frac1{\Norm m}\int m(\beta(r))\mu(dr)\,.
  \end{align*}
  For each \(r\in X\) we have %
  \(m(\beta(r))\leq\Norm m\Norm{\beta(r)}\leq\Norm m\Norm \beta\). Our
  contention follows from %
  \(\Norm\mu=\mu(X)=\int\mu(dr)\).
\end{proof}

\begin{definition}
  \label{def:integral-in-cone}
  A measurable cone is \emph{integrable} if, for all \(X\in\ARCAT\),
  each \(\beta\in\Mcca{\Cpath XB}\) has an integral in \(\Mcca B\)
  over each measure \(\mu\in\Mcca{\Cmeas(X)}\).
  When this is the case we use \(\Mcint B_X\) for the uniquely defined
  function %
  \(\Mcca{\Cpath XB}\times\Mcca{\Cmeas(X)}\to\Mcca B\) such that %
  \(\Mcint B_X(\beta,\mu)=\int\beta(r)\mu(dr)\).
\end{definition}

\begin{remark}
  \label{rk:exist-cones-non-integ}
  A very natural question is whether there are measurable cones which
  are not integrable.
  We strongly conjecture that such cones do exist but we have not yet
  tried to exhibit some.
\end{remark}
\noindent 
The fundamental example of an integrable cone is the measurable cone
of finite measures described in Section~\ref{sec:cone-finite-meas}.
\begin{theorem}
  For each measurable space \(X\), %
  the measurable cone \(\Cmeas(X)\) is integrable.
\end{theorem}
\noindent 
This is just a reformulation of the standard integration of a kernel.
\begin{proof}
  Let \(Y\in\ARCAT\), \(\kappa\in\Mcca{\Cpathm Y{\Cmeas(X)}}\), which
  means that %
  \(\kappa\) is a bounded kernel \(Y\Kernto X\), and let %
  \(\nu\in\Mcca{\Cmeas(Y)}\), which means that \(\nu\) is a finite
  measure.
  We define \(\mu:\Sigalg X\to\Realp\)  by
  \begin{align*}
    \forall U\in\Sigalg X\quad
    \mu(U)=\int\kappa(s,U)\nu(ds)\in\Realp\,.
  \end{align*}
  The fact that \(\mu\) defined in that way is a finite measure is
  completely standard in measure theory and \(\mu\) is the integral of
  \(\kappa\) by the very definition of \(\Mcms{\Cmeas(X)}_\Measterm\).
\end{proof}

\noindent 
In the sequel we assume that \(B\) is an integrable cone. We state and
prove some basic expected properties of integration.

\begin{lemma} %
  \label{lemma:integral-measurable}
  Let \(\phi:Y\times X\to\Realp\) be measurable and bounded and
  let %
  \(\kappa:Y\to\Mcca{\Cmeas(X)}\) be a bounded kernel.
  Then the function %
  \(\Absm{s\in Y}{\int\phi(s,r)\kappa(s,dr)}\) is measurable.
\end{lemma}
\begin{proof}
  The property is obvious when \(\phi\) is simple%
  \footnote{A \(\Real\)-valued measurable function is simple iff it
    ranges in a finite subset of \(\Real\).}, %
  and the result follows from the monotone convergence theorem by the
  fact that each \(\Realp\)-measurable function is the lub of a
  increasing sequence of simple functions.
\end{proof}

\begin{lemma} %
  \label{lemma:int-mesurable}
  For each \(X\in\ARCAT\), the map \(\Mcint B_X\) is bilinear,
  continuous and measurable.
  This means that
  \(\Mcint B_X:\With{\Mcca{\Cpath XB}}{\Mcca{\Cmeas(X)}}\to\Mcca B\)
  is continuous, separately linear in each of its two arguments and
  that for each \(Y\in\ARCAT\), %
  \(\eta\in\Mcca{\Cpathm{Y}{\Cpathm XB}}\) and %
  \(\kappa\in\Mcca{\Cpathm Y{\Cmeas(X)}}\), %
  the function %
  \(\beta=\Mcint B_X\Comp\Tuple{\eta,\kappa}:Y\to\Mcca B\) is a
  measurable path.
\end{lemma}
\begin{proof}
  Separate linearity in both argumets results from the linearity of
  integration and from \Mssepr{} satisfied by \(B\), let us prove
  separate continuity (which implies continuity by
  Lemma~\ref{lemma:seprate-cont-implies-cont}).
  Let \((\beta_n)_{n\in\Nat}\) be an increasing sequence in
  \(\Cuball{\Mcca{\Cpathm XB}}\) and let \(\mu\in\Mcca{\Cmeas(X)}\).
  The sequence %
  \((\Mcint B_X(\beta_n,\mu)\in\Mcca B)_{n\in\Nat}\) is increasing by
  linearity of %
  \(\Mcint B_X\) and for all \(n\in\Nat\) we have %
  \(\Norm{\Mcint
    B_X(\beta_n,\mu)}\leq\Norm{\beta_n}\Norm\mu\leq\Norm\mu\) %
  so that \(\sup_{n\in\Nat}\Mcint B_X(\beta_n,\mu)\in\Mcca B\) exists.
  Let \(\beta=\sup_{n\in\Nat}\beta_n\in\Cuball{\Mcca{\Cpathm XB}}\),
  that is %
  \(\forall r\in X\ \beta(r)=\sup_{n\in\Nat}\beta_n(r)\).
  Let \(m\in\Mcms B_\Measterm\), since \((m\Comp\beta_n)_{n\in\Nat}\)
  is an increasing sequence of measurable functions by linearity of
  \(m\) and since %
  \(m\Comp\beta=\sup_{n\in\Nat}m\Comp\beta_n\) (pointwise) by
  continuity of \(m\), we have
  \begin{align*}
    \int m(\beta(r))\mu(dr)=\sup_{n\in\Nat}\int m(\beta_n(r))\mu(dr)
  \end{align*}
  by the monotone convergence theorem.
  That is %
  \(m(\Mcint B_X(\beta,\mu)) =\sup_{n\in\Nat}m(\Mcint
  B_X(\beta_n,\mu)) =m(\sup_{n\in\Nat}\Mcint B_X(\beta_n,\mu))\) by
  continuity of \(m\). %
  By \Mssepr{} we get %
  \(\Mcint B_X(\beta,\mu)=\sup_{n\in\Nat}\Mcint B_X(\beta_n,\mu)\) as
  required.

  Let \(\beta\in\Mcca{\Cpathm XB}\) and let %
  \((\mu_n\in\Cuball{\Mcca{\Cmeas(X)}})_{n\in\Nat}\) be an increasing
  sequence with lub \(\mu\).
  %
  %
  It is a standard fact that for each measurable and bounded
  \(\phi:X\to\Realp\) the sequence %
  \((\int\phi(r)\mu_n(dr))_{n\in\Nat}\) is increasing and has %
  \(\int\phi(r)\mu(dr)\) as lub: this is due to the fact that %
  \(\int\phi(r)\mu(dr)=\sup_{k\in\Nat}\int\phi_k(r)\mu(dr)\) %
  where \((\phi_k\leq\phi)_{k\in\Nat}\) is an increasing family of
  simple functions whose pointwise lub is \(\phi\), and to the fact
  that \(\int\psi(r)\mu(dr)=\sup_{n\in\Nat}\int\psi(r)\mu_n(dr)\)
  holds trivially when \(\psi\) is simple.
  As above the sequence \((\Mcint B_X(\beta,\mu_n))_{n\in\Nat}\) %
  is increasing with %
  \(\forall n\in\Nat\ \Norm{\Mcint B_X(\beta,\mu_n)}
  \leq\Norm\beta\Norm\mu\) and therefore has a lub %
  \(\sup_{n\in\Nat}\Mcint B_X(\beta,\mu_n)\in\Mcca B\).
  Let \(m\in\Mcms B_\Measterm\), we have %
  \begin{align*}
    m(\sup_{n\in\Nat}\Mcint B_X(\beta,\mu_n))
    &=\sup_{n\in\Nat}m(\Mcint B_X(\beta,\mu_n))\\
    &=\sup_{n\in\Nat}\int m(\beta(r))\mu_n(dr)\\
    &=\int m(\beta(r))\mu(dr)\\
    &=m(\Mcint B_X(\beta,\mu))
  \end{align*}
  and the announced continuity follows by \Mssepr{} in \(B\).

  Now we prove measurability, so let \(Y\in\ARCAT\), %
  \(\eta\in\Mcca{\Cpathm{Y}{\Cpathm XB}}\) and let %
  \(\kappa\in\Mcca{\Cpathm Y{\Cmeas(X)}}\), we prove that the
  function %
  \(\beta=\Mcint B_X\Comp\Tuple{\eta,\kappa}:Y\to\Mcca B\) %
  belongs to \(\Mcca{\Cpathm YB}\).
  The fact that \(\beta(X)\) is bounded results from
  Lemma~\ref{lemma:integral-bounded}.
  Let \(Y'\in\ARCAT\) and \(m\in\Mcms B_{Y'}\), we have %
  \begin{align*}
    \Absm{(s',s)\in{Y'\times Y}}{m(s',\beta(s))}
    &=\Absm{(s',s)\in{Y'\times Y}}{m(s',\Mcint B_X(\eta(s),\kappa(s)))}\\
    &=\Absm{(s',s)\in{Y'\times Y}}{\int m(s',\eta(s,r))\kappa(s,dr)}
  \end{align*}
  and this function is measurable by
  Lemma~\ref{lemma:integral-measurable} and by our assumption about
  \(\eta\).
\end{proof}

\begin{lemma}[Change of variable]
  Let \(X,Y\in\ARCAT\), \(\beta\in\Mcca{\Cpathm XB}\),
  \(\nu\in\Mcca{\Cmeas(Y)}\) and %
  \(\phi\in\ARCAT(Y,X)\).
  We have
  \begin{align*}
    \int_{s\in Y}\beta(\phi(s))\nu(ds)
    =\int_{r\in X}\beta(r)\Pushf\phi(\nu)(dr)\,.
  \end{align*}
  In other words \(\Mcint B_X\) is extranatural in \(X\).
\end{lemma}
\begin{proof}
  By the usual change of variable formula, through the use of
  measurability tests \(m\in\Mcms B_\Measterm\) and \Mssepr{} for
  \(B\).
\end{proof}

\begin{lemma}
  If \(B\) is an integrable cone and \(\alpha\in\Real\) is such that
  \(\alpha>0\) then the measurable cone \(\alpha B\) is integrable,
  and has the same integrals as \(B\).
\end{lemma}
\noindent 
We can define now the category which is at the core of the present study.
\begin{definition}
  \label{def:icones-category}
  The category \(\ICONES\) has integrable cones as objects and an
  element of \(\ICONES(B,C)\) is an
  \(f\in\MCONES(\Mcofic B,\Mcofic C)\) such that, for all %
  \(X\in\ARCAT\) and all %
  \(\beta\in\Mcca{\Cpathm X{\Mcofic B}}\) and \(\mu\in\Mcca{\Cmeas(X)}\) %
  one has %
  \begin{align*}
    f\Big(\int\beta(r)\mu(dr)\Big)=\int f(\beta(r))\mu(dr)\,.
  \end{align*}
  This property of \(f\) will be called \emph{integral preservation}
  and when it holds we often simply say that \(f\) is
  \emph{integrable}.
\end{definition}
Notice that the right hand term of the above equation is well defined
because %
\(f\Comp\beta\in\Mcca{\Cpathm X{\Mcofic C}}\) by our assumption on
\(f\). It is obvious that we define a category in that way.

\begin{lemma}
  \label{lemma:pushf-functor-icones}
  The functor \(\Funpushf:\ARCAT\to\MCONES\) introduced in
  Lemma~\ref{lemma:pushf-measurable} is a functor %
  \(\ARCAT\to\ICONES\).
\end{lemma}
\begin{proof}
  Let \(\phi\in\ARCAT(X,Y)\) and %
  \(\kappa\in\Mcca{\Cpathm{Y'}{\Cmeas(X)}}\) be a bounded kernel.
  Given \(\mu'\in\Mcca{\Cmeas(Y')}\) and \(V\in\Sigalg Y\) we have %
  \begin{align*}
    \Pushf\phi\Big(\int\kappa(s')&\mu'(ds')\Big)(V)
    =\Big(\int\kappa(s')\mu'(ds')\Big)(\Inv\phi(V))\\
    &=\int\kappa(s',\Inv\phi(V))\mu'(ds')
      \text{\quad by def.~of integration in }\Cmeas(X)\\
    &=\int\Pushf\phi(\kappa(s'))(V)\mu'(ds')\\
    &=\Big(\int\Pushf\phi(\kappa(s'))\mu'(ds')\Big)(V)
      \text{\quad by def.~of integration in }\Cmeas(Y)
  \end{align*}
  so that \(\Pushf\phi\) preserves integrals.
\end{proof}

\subsection{Integrable cones as quasi-Borel spaces with additional structure}
\label{sec:icones-qbs}

In this section, we assume, as in Remark \ref{rmk:ar-main-example}, that
\(\ARCAT\) has only two objects \(\Real\) and \(\Measterm\).

Then every integrable cone \(C\) can be seen as a QBS
\cite{KammarStatonVakar19} by letting \(M_C\) (which is by definition
the set of all QBS-morphisms from \(\Real\) to \(C\)) be the set
of all maps \(\alpha : \Real \to \Mcca C\) such that for all
\(m \in \Mcms C_\Real\), the map \(\Absm {(r,s)} {m(\alpha(r),s)}\)
from \(\Real\times\Real\) to \(\Real\) is measurable.

There is a well-defined notion of S-finite measure (respectively:
probability measure, sub-probability measure) on QBSs.
The operation that maps a QBS \(Q\) to the set of all S-finite
(respectively: probability, sub-probability) measures on \(Q\) defines
a commutative strong monad on the category of QBSs \cite[end of \S
2]{heunen2018} (this is similar to the Giry monad on the category of
measurable spaces and measurable maps).
For each S-finite measure \(\mu\) on an integrable cone \(C\), there
exists at most one element \(y \in\Mcca C\) such that for all
\(m \in \Mcms C_\Measterm\),
\(m(y) = \int_{x\in\Mcca C} m(x)\,\mu(dx)\).
If this element exists (which is always the case if \(\mu\) is
finite and has a bounded support), we will denote it by
\(\int_{x \in\Mcca C} x\,\mu(dx)\).
One can check that for each integrable cone \(C\), this construction makes the unit ball
\(\Cuball{C}\) into an algebra over the monad of sub-probability measures on QBSs.

In this situation a map \(f : \Mcca C \to \Mcca B\) between two
integrable cones is a morphism in \(\ICONES\) if and only if:
\begin{itemize}
\item it is a morphism of QBSs,
\item it preserves S-finite integrals: for each S-finite measure
  \(\mu\) on \(C\), if \(\int_{x \in \Mcca C} x\,\mu(dx)\) exists (as
  an element of \(\Mcca C\)), then
  \[\int_{x \in \Mcca C} f(x)\,\mu(dx) = f\Big(\int_{x \in \Mcca C}
    x\,\mu(dx)\Big),\]
\item and it is non-expansive: for all \(x \in \Mcca C\),
  \(\Norm{f(x)} \leq \Norm{x}\).
\end{itemize}
\noindent 
In particular, for each morphism
\(f \in \ICONES(B,C)\), the restriction of \(f\) to
\(\Cuball{\Mcca B}\) is a morphism of algebras:
this was one of our main guidelines in the design of integrable cones.

However, it is not clear whether or not each morphism of algebras from
\(\Cuball{\Mcca B}\) to \(\Cuball{\Mcca C}\) is the restriction of
some morphism in \(\ICONES(B,C)\) (which would make \(\ICONES\) a full
subcategory of the category of algebras over the monad of
sub-probability measures on QBSs).
It is also not clear whether the construction of Section
\ref{sec:ccc-fix} (which will define a fixpoint operator on integrable
cones) can be replicated in the category of algebras over the monad of
sub-probability measures on QBSs.
Indeed, we conjecture that in both cases the answer is no. On the
other hand, integrable cones are not quite algebras over the monad of
S-finite measures, because the would-be monad multiplication (namely,
S-finite integration) is only partially defined. We do not know
whether there exists a monad on (a full subcategory of) QBSs such that
\(\ICONES\) is equivalent to the category of algebras over this monad.

\subsection{The integrable cone of paths and a Fubini theorem for cones} %
\label{sec:int-cpath}
\begin{theorem}
  For each \(X\in\ARCAT\) and each integrable cone \(B\), the measurable
  cone \(\Cpathm X{\Mcofic B}\) is integrable.
\end{theorem}
\begin{proof}
  Let \(Y\in\ARCAT\), \(\eta\in\Mcca{\Cpathm Y{\Cpathm XB}}\) and %
  \(\nu\in\Mcca{\Cmeas(Y)}\), we define %
  \(\beta:X\to\Mcca B\) by %
  \(\beta(r) %
  =\int\eta(s)(r)\nu(ds)\), in other words the integral of a path of
  paths is defined pointwise.
  For each \(r\in X\) we have
  \begin{align*}
    \Norm{\beta(r)}
    &=\Norm{\int\eta(s)(r)\nu(ds)}\\
    &\leq\Norm{\Absm{s\in Y}{\eta(s)(r)}}\Norm{\nu}
      \text{\quad by Lemma~\ref{lemma:integral-bounded}}\\
    &\leq\Norm\eta\Norm\nu
  \end{align*}
  so \(\beta\) is a bounded function.
  This function is a measurable path by
  Lemma~\ref{lemma:int-mesurable} so %
  \(\beta\) belongs to \(\Mcca{\Cpathm X{\Mcofic B}}\). %
  Let \(p\in\Mcms{\Cpathm X{\Mcofic B}}_\Measterm\) so that %
  \(p=\Mtpath rm\) for some %
  \(r\in X\) and \(m\in\Mcms B_\Measterm\), we have %
  \begin{align*}
    p(\beta)
    &=m(\beta(r))\\
    &=m\Big(\int\eta(s)(r)\nu(ds)\Big)\\
    &=\int m(\eta(s)(r))\nu(ds)
    \\
    &=\int p(\eta(s))\nu(ds)\text{\quad by definition of }p
  \end{align*}
  and hence \(\beta=\int\eta(s)\nu(ds)\).
\end{proof}

\begin{theorem}
  The operation \(\Cpathf\), extended to morphisms by %
  pre- and post-composition, is a functor
  \(\Op\ARCAT\times\ICONES\to\ICONES\).  In other words, given
  \(f\in\ICONES(B,C)\) and \(\phi\in\ARCAT(Y,X)\), we have %
  \[
    \Cpath\phi f
    =\Absm{\beta\in\Mcca{\Cpath XB}}{(f\Comp\beta\Comp\phi)}
    \in\ICONES({\Cpath XB},{\Cpath YC})\,.
  \]
\end{theorem}
\begin{proof}
  Functoriality is obvious.
  We check first measurability of %
  \(\Cpath\phi f\) so let \(Y'\in\ARCAT\) and let %
  \(\eta\in\Mcca{\Cpath{Y'}{\Cpath XB}}\), we must check that %
  \(\Cpath\phi f\Comp\eta\in\Mcca{\Cpath{Y'}{\Cpath YC}}\).
  Let \(Y''\in\ARCAT\) and \(p\in\Mcms{\Cpath YC}_{Y''}\), we check
  that %
  \(\psi=\Absm{(s'',s')\in{Y''\times Y'}} {p(s'',\Cpath\phi
    f(\eta(s')))}\) is measurable.
  So let %
  \(\rho\in\ARCAT(Y'',Y)\) and \(m\in\Mcms{C}_{Y''}\) be such that %
  \(p=\Mtpath\rho m\).
  Give \(s''\in Y''\) and \(s'\in Y'\), we set
  \begin{align*}
    \psi(s'',s')
    &=
    {m(s'',\Cpath\phi f(\eta(s'))(\rho(s'')))}\\
    &=
      {m(s'',f(\eta(s')(\phi(\rho(s'')))))}\\
    &=
      {m(s'',f(\Flpath(\eta)(s',\phi(\rho(s'')))))}
  \end{align*}
  and the map \(\psi\) is measurable by
  Lemma~\ref{lemma:meas-path-flat} because %
  \(f\Comp\Flpath(\eta)\in\Mcca{\Cpath{Y'\times X}{C}}\) and %
  \(\phi\Comp\rho\) is measurable.
  We prove that %
  \(\Cpath\phi f\) preserves integrals.
  Let \(Y'\in\ARCAT\), %
  \(\eta\in\Mcca{\Cpath{Y'}{\Cpath XB}}\), %
  \(\nu'\in\Mcca{\Cmeas(Y')}\) and let \(s\in Y\).
  We have
  \begin{align*}
    \Cpath\phi f\Big(&\int^{\Cpath XB}_{s'\in Y'}\eta(s')
    \nu'(ds')\Big)(s)
      =f\Big(\Big(\int^{\Cpath XB}_{s'\in Y'}\eta(s')\nu'(ds')\Big)
      (\phi(s))\Big)\\
    &=f\Big(\int^B_{s'\in Y'}\eta(s')(\phi(s))\nu'(ds')\Big)
      \text{\quad by def.~of integration in }\Cpath XB\\
    &=\int^C_{s'\in Y'} f(\eta(s')(\phi(s)))\nu'(ds')
      \text{\quad since }f\text{ preserves integrals}\\
    &=\Big(\int^{\Cpath YC}_{s'\in Y'}\Cpath\phi f(\eta(s'))\nu'(ds')\Big)(s)
  \end{align*}
  which proves our contention.
\end{proof}

\begin{lemma} %
  \label{lemma:lemma:int-path-flat}
  The bijection \(\Flpath_{X,Y}\) defined in
  Lemma~\ref{lemma:meas-path-flat}, as well as its inverse, preserve
  integrals and hence
  \begin{align*}
    \Flpath_{X,Y}\in\ICONES(\Cpath X{\Cpath Y B},\Cpath{X\times Y}B)
  \end{align*}
  is an iso in \(\ICONES\).
\end{lemma}
\begin{proof}
  Results straightforwardly from the ``pointwise'' definition of
  integration in the cones of paths.
\end{proof}

\begin{theorem}[Fubini]
  \label{th:paths-Fubini}
  Let \(X,Y\in\ARCAT\), \(\eta\in\Mcca{\Cpath{X}{\Cpath YB}}\), %
  \(\mu\in\Mcca{\Cmeas(X)}\) and \(\nu\in\Mcca{\Cmeas(Y)}\).
  We have %
  \begin{align*}
    \int_Y\Big(\int_X\eta(r)\mu(dr)\Big)(s)\nu(ds)
    =\int_{X\times Y}\Flpath(\eta)(t)(\Measprod\mu\nu)(dt)
  \end{align*}
\end{theorem}
\begin{proof}
  Denoting by \(x_1\) and \(x_2\) these two elements of \(\Mcca B\) it
  suffices to prove that for each \(m\in\Mcms B_\Measterm\) one has
  \(m(x_1)=m(x_2)\).
  Setting \(\eta'=\Flpath(\eta)\) we have
  \begin{align*}
    x_1=\int_Y\Big(\int_X \eta'(r,s)\mu(dr)\Big)\nu(ds)
    \Textsep
    x_2=\int_{X\times Y}\eta'(t)(\Measprod\mu\nu)(dt)
  \end{align*}
  and the equation follows by application of the usual Fubini theorem
  to the bounded non-negative measurable function \(m\Comp\eta'\) and
  to the finite measures \(\mu\) and \(\nu\).
  Notice that in the expression of \(x_2\) the variable \(t\) ranges
  over pairs.
\end{proof}

\subsection{The category of integrable cones}

We start with proving that the category \(\ICONES\) of
Definition~\ref{def:icones-category} has all (projective) limits.
This is not only a very pleasant property of the probabilistic model
of \LL{} that we are defining%
\footnote{Which is not so common among models of \LL{}; there is
  however a price to pay, it is very likely that the category
  \(\ICONES\) has no \(\ast\)-autonomous structure.}, %
it will play a crucial role in our definition of the tensor product
and of the exponentials.

\begin{theorem} %
  \label{th:mcones-complete}
  The category \(\ICONES\) is complete.
\end{theorem}
There is a faithful forgetful functor %
\(\ICONES\to\SET\) which maps each integrable cone \(C\) to
\(\Mcca C\), considered as a set, and each morphism to itself; we will
see that this functor actually creates all the small limits in
\(\ICONES\).
\begin{proof}
  We prove first that \(\ICONES\) has all small products. We use
  implicitly Theorem~\ref{th:cone-product} at several places.
  Let \((C_i)_{i\in I}\) be a collection of integrable cones and let %
  \(P=\Bwith_{i\in I}\Mcca{C_i}\) which is the product of the %
  \(\Mcca{C_i}\)'s in \(\CONES\).
  Given \(X\in\ARCAT\), \(i\in I\)
  and \(m\in\Mcms{C_i}_X\) we define %
  \(\Mtinj im:X\times P\to\Realp\) by %
  \(\Mtinj im(r,\Vect x)=m(r,x_i)\).
  We set \(\cM=(\cM_X)_{X\in\ARCAT}\) where %
  \(\cM_X=\{\Mtinj im\St i\in I\text{ and }m\in\Mcms{C_i}_X\}\). %
  With the notations above, given \(\Vect x\in P\), the function %
  \(\Absm{r\in X}{\Mtinj im(r,\Vect x)}
  =\Absm{r\in X}{m(r,x_i)}\) is measurable
  since \(\Mcms{C_i}\) satisfies \Msmesr.

  Let \(\phi\in\ARCAT(Y,X)\), we have %
  \(\Mtinj im\Comp(\phi\times P) =\Mtinj
  i{m\Comp(\phi\times\Mcca{C_i})}\in\cM_Y\) since %
  \(m\Comp(\phi\times\Mcca{C_i})\in\Mcms{C_i}_Y\) by \Mscompr{} in
  \(C_i\).

  Let \(\Vect{x(1)},\Vect{x(2)}\in P\) be such that %
  \(\forall p\in\cM_\Measterm\ p(\Vect{x(1)})= p(\Vect{x(2)})\).
  Then for each \(i\in I\) we have \(x(1)_i= x(2)_i\) by %
  \Mssepr{} holding in \(C_i\) and hence %
  \(\Vect{x(1)}=\Vect{x(2)}\).

  Let \(\Vect x\in P\setminus\Eset 0\) and \(\epsilon>0\). %
  Since \(\Norm{\Vect x}=\sup_{i\in I}\Norm{x_i}\) there is %
  \(i\in I\) such that %
  \(\Norm{\Vect x}\leq\Norm{x_i}+\epsilon/2\) and \(x_i\not=0\).
  We can find %
  \(m\in\Mcms{C_i}_0\setminus\Eset \Measterm\) such that %
  \(\Norm{x_i}\leq m(x_i)/\Norm m+\epsilon/2\).
  Let \(p=\Mtinj im\in\cM_0\), notice that %
  \(\Norm p=\Norm{m_i}\) since for each %
  \(x\in\Cuball{\Mcca{C_i}}\) the family \(\Vect y\) defined by %
  \(y_i=x\) and \(y_j=0\) if \(j\not=i\) satisfies %
  \(\Vect y\in\Cuball P\).
  So we have %
  \(\Norm{\Vect x}\leq p(\Vect x)/\Norm p+\epsilon\) which shows that
  \(\cM\) satisfies \Msnormr.

  So the pair \((P,\cM)\) is a measurable cone
  \(C=\Bwith_{i\in I}C_i\), we prove that it is integrable.
  An element of \(\Mcca{\Cpathm X{C}}\) is a family %
  \((\gamma_i\in\Mcca{\Cpathm X{C_i}})_{i\in I}\) such that %
  \((\Norm{\gamma_i})_{i\in I}\) is bounded and, given
  \(\mu\in\Mcca{\Cmeas(X)}\), the family
  \begin{align*}
    \Vect x=\Big(\int\gamma_i(r)\mu(dr)\Big)_{i\in I}
  \end{align*}
  is in \(P\) by Lemma~\ref{lemma:integral-bounded} and is the
  integral of \(\gamma\) over \(\mu\) by definition of \(\Mcms C\).

  With the same notations as above, for each \(i\in I\), the map %
  \(\Proj i\Comp\gamma\) is a measurable path since, given %
  \(Y\in\ARCAT\) and \(m\in\Mcms{C_i}_Y\) one has
  \(\Absm{(s,r)\in{Y\times X}}{\Mtinj im(s,\gamma(r))}
  =\Absm{(s,r)\in{Y\times X}}{m(s,\Proj i(\gamma(r))}\).
  The fact that \(\Proj i\in\ICONES(C,C_i)\) results from the
  definition of integration in \(C\). %
  
  Let %
  \((f_i\in\ICONES(D,C_i))_{i\in I}\), then we know that %
  \(f=\Tuple{f_i}_{i\in I}\in\CONES(\Mcca D,\Mcca C)\). %
  Let \(\delta\in\Mcca{\Cpath XD}\), we prove that
  \(f\Comp\delta\in\Mcca{\Cpath XC}\) so let \(i\in I\) and
  \(m\in\Mcms{C_i}_Y\).
  We have %
  \(\Absm{(s,r)\in{Y\times X}}{\Mtinj im(s,f(\delta(r)))}
  =\Absm{(s,r)\in{Y\times X}}{m(s,f_i(\delta(r)))}\) and this latter
  map is measurable for each \(i\in I\) thus proving that
  \(f\Comp\delta\) is measurable.
  Using the same notations, let furthermore
  \(\mu\in\Mcca{\Cmeas(X)}\), we have
  \begin{align*}
    f\Big(\int\delta(r)\mu(dr)\Big)
    &=\Big(f_i\Big(\int\delta(r)\mu(dr)\Big)\Big)_{i\in I}\\
    &=\Big(\int f_i(\delta(r))\mu(dr))\Big)_{i\in I}
      \text{\quad since each }f_i\text{ preserves integrals}\\
    &=\int f(\delta(r))\mu(dr)
  \end{align*}
  which shows that \(f\in\ICONES(D,C)\) as required.
  This proves that \(\ICONES\) has all small products.

  We prove now that \(\ICONES\) has equalizers, so let %
  \(f,g\in\ICONES(C,D)\).
  Let \((P,e\in\CONES(P,\Mcca C))\) be the equalizer of \(f\) and \(g\)
  in \(\CONES\), see Theorem~\ref{th:alg-cones-equalizers}.
  Remember that if \(x,y\in P\) satisfy \(x\leq_{\Mcca C}y\) then
  \(y-x\in P\).

  We define \(\cM_X\) as the set of
  all 
  \(p:X\times P\to\Realp\) such that there is %
  \(m\in\Mcms C_X\) satisfying %
  \(\forall x\in P\,\forall r\in X\ p(r,x)=m(r,x)\).
  Then it is clear that \(p\in(\Cdual P)^{X}\) and we actually
  identify \(\cM_X\) with \(\Mcms C_X\) although several elements of
  the latter can induce the same element of the former.
  We prove that \((\cM_X)_{X\in\ARCAT}\) defines a measurability
  structure on \(P\), the only non trivial property being \Msnormr{}.
  Let \(x\in P\setminus\Eset 0\) and \(\epsilon>0\).
  Let \(\epsilon'>0\) be such that %
  \(\epsilon'\leq\epsilon\) and \(\epsilon'<\Norm x\) (remember that
  we have assumed that \(x\not=0\) and hence \(\Norm x>0\)).
  Applying \Msnormr{} in \(C\) we can find
  \(m\in\Mcms C_\Measterm\setminus\Eset 0\) such that %
  \(\Norm x=\Norm x_C\) satisfies
  \(\Norm x\leq m(x)/\Norm m^{\Mcca C}+\epsilon'\) where we have added
  the superscript to \(\Norm m\) to insist on the fact that it is
  computed in \(\Mcca C\), that is
  \(\Norm m^{\Mcca C}=\sup_{y\in\Cuball{\Mcca C}}m(y)\).
  By our assumption that \(\epsilon'<\Norm x\) we must have
  \(m(x)\not=0\).
  By definition of \(\Norm\__P\) we have
  \(\Cuball P=P\cap\Cuball{\Mcca C}\subseteq\Cuball{\Mcca C}\) and
  hence %
  \(\Norm m^P= %
  \sup_{y\in\Cuball P}m(y) %
  \leq\sup_{y\in\Cuball{\Mcca C}}m(y) %
  =\Norm m^{\Mcca C}\) %
  (and \(\Norm m^P\not=0\) since \(m(x)\not=0\) and \(x\in P\)) and
  hence
  \begin{align*}
    \Norm x_P=\Norm x_C
    \leq\frac{m(x)}{\Norm m^{\Mcca C}}+\epsilon'
    \leq\frac{m(x)}{\Norm m^{P}}+\epsilon
  \end{align*}
  since \(\Norm m^P\leq\Norm m^{\Mcca C}\) and
  \(\epsilon'\leq\epsilon\), and since this holds for all
  \(\epsilon>0\), it follows that \(P\) satisfies \Msnormr.

  So we have defined a measurable cone \(E=(P,\cM)\), we check that it
  is integrable.
  Let \(X\in\ARCAT\), %
  \(\beta\in\Mcca{\Cpathm X{E}}\) and \(\mu\in\Mcca{\Cmeas(X)}\), we
  have
  \begin{align*}
    f\Big(\int\beta(r)\mu(dr)\Big)
    =\int f(\beta(r))\mu(dr)=\int g(\beta(r))\mu(dr)
    =g\Big(\int\beta(r)\mu(dr)\Big)
  \end{align*}
  since \(\beta\) ranges in \(\Mcca E=P\) and \(f\) and \(g\) preserve
  integrals.
  Hence %
  \(\int\beta(r)\mu(dr)\in\Mcca E\) and this element of \(\Mcca E\) is
  the integral of \(\beta\) over \(\mu\) by definition of \(\Mcms E\).

  We check now that \((E,e)\) is the equalizer of \(f,g\) in
  \(\ICONES\).
  The inclusion \(e\in\CONES(\Mcca E,\Mcca C)\) is measurable
  \(E\to C\) by definition of the measurability structure of \(E\)
  which is essentially the same as that of \(C\) and preserves
  integrals because the integral in \(E\) is defined as in \(C\).

  We already know that \(f\Compl e=g\Compl e\).
  Let \(H\) be 
  an integrable cone and \(h\in\ICONES(H,C)\) be such that
  \(f\Compl h=g\Compl h\).
  Let \(h_0\) be the unique element of \(\CONES(\Mcca H,\Mcca E)\)
  such that \(h=e\Compl h_0\).
  Let \(X\in\ARCAT\) and %
  \(\gamma\in\Mcca{\Cpath XH}\) be a measurable path of \(H\).
  Let %
  \(Y\in\ARCAT\) and \(m\in\Mcms E_Y\) so that actually %
  \(m\in\Mcms C_Y\).
  We have %
  \(\Absm{(s,r)\in{Y\times X}}{m(s,h_0(\gamma(r)))}
  =\Absm{(s,r)\in{Y\times X}}{m(s,h(\gamma(r)))}\) which is
  measurable since \(h\) is.
  With the same notation, taking also %
  \(\mu\) in \(\Mcca{\Cmeas(X)}\), we have
  \begin{align*}
    h_0\Big(\int^H\gamma(r)\mu(dr)\Big)
    &=h\Big(\int^H\gamma(r)\mu(dr)\Big)
      \text{\quad by definition of }h_0\\
    &=\int^C h(\gamma(r))\mu(dr)
      \text{\quad since }h\text{ preserves integrals}\\
    &=\int^E h_0(\gamma(r))\mu(dr)
  \end{align*}
  and hence \(h_0\in\ICONES(H,E)\).
  Since \(h=e\Compl h_0\) and is unique with this property in
  \(\CONES\), it has the same properties in \(\ICONES\).
  
  This shows that \(\ICONES\) has all small limits.
\end{proof}

\begin{lemma}
  \label{lemma:int-cone-transp-iso}
  Let \(C\) be an integrable cone, \(S\) be a set, and let
  \(f:\Mcca C\to S\) be a bijection.
  There is an integrable cone structure on \(S\) such that \(f\)
  is an iso in \(\ICONES\).
\end{lemma}
\noindent 
This structure is not unique \emph{a priori} (other choices for the
measurability structure are possible in general), but this is not an
issue for the use that we will make of this lemma.
\begin{proof}
  By Lemma~\ref{lemma:alg-cone-transport} we can equip \(S\) with a
  cone structure such that \(f\in\CONES(\Mcca C,S)\) (we use \(S\) for
  the cone obtained by equipping \(S\) with this structure), and hence
  \(f\) is an iso in \(\CONES\) since \(f\) is a bijection.
  Let \(X\in\ARCAT\) and let \(m\in\Mcms C_X\).
  We define \(\Transfwd f(m):X\times S\to\Realp\) by %
  \(\Transfwd f(m)(r,z)=m(r,\Inv f(z))\).
  We set \(\cM_X=\{\Transfwd f(m)\St m\in\Mcms C_X\}\).
  In view of the definition of the algebraic structure and of the norm
  of \(S\), it is clear that \((\cM_X)_{X\in\ARCAT}\) is a
  measurability structure on \(S\), we still use \(S\) for denoting
  this measurable cone and we observe that \(f\) is an iso from \(C\)
  to \(S\) in \(\MCONES\).
  It is also easy to check that
  \(\Cpath XS=\{f\Comp\gamma\St \gamma\in\Cpath XC\}\).
  This cone \(S\) is integrable: %
  given \(\mu\in\Cmeas(X)\) and \(\gamma\in\Cpath XS\), we have %
  \(\Inv f\Comp\gamma\in\Cpath XC\) and hence the integral %
  \(x=\int_C\Inv f(\gamma(r))\mu(dr)\in\Mcca C\) exists.
  Then for each \(p=\Transfwd f(m)\in\Mcms S_\Measterm\) where
  \(m\in\Mcms C_\Measterm\), we have %
  \(p(f(x))
  =m(x)
  =\int_{\Realp} m(\Inv f(\gamma(r)))\mu(dr)
  =\int_{\Realp} p(\gamma(r))\mu(dr)\) %
  which shows that \(\int_S\gamma(r)\mu(dr)\) exists and is \(f(x)\).
  It follows also trivially that \(f\) preserves integrals.
\end{proof}

\begin{theorem}
  \label{th:icones-conditions-saft}
  In the category \(\ICONES\) the object \(\Sone\) is a coseparator
  and a separator%
  \footnote{In the literature one also finds the words
    \emph{generator} and \emph{cogenerator} for such objects.} %
  and \(\ICONES\) is well-powered.
\end{theorem}
\begin{proof}
  Let \(f\not=g\in\ICONES(C,D)\) and let \(x\in\Mcca C\) be such
  that %
  \(f(x)\not=g(x)\). By \Mssepr{} there is \(m\in\Mcms C_\Measterm\)
  such that %
  \(m(f(x))\not=m(g(x))\) and since \(m\in\ICONES(C,\Sone)\) (using
  the definition of integrals) this shows that \(\Sone\) is
  a coseparator.

  Given \(x\in\Mcca C\) we check that the function
  \(\Pt x:\Realp\to\Mcca C\) defined by %
  \(\Pt x(\lambda)=\lambda x\) belongs to \(\ICONES(\Sone,C)\).
  It is clearly linear and continuous.
  Let %
  \(\beta\in\Mcca{\Cpath X\Sone}\) for some \(X\in\ARCAT\).
  This simply means that \(\beta\) is a measurable and bounded
  function \(X\to\Realp\), we must check that %
  \(\Pt x\Comp\beta\in\Mcca{\Cpath XC}\) so let %
  \(Y\in\ARCAT\) and \(m\in\Mcms C_Y\), we have %
  \begin{align*}
    \Absm{(s,r)\in Y\times X}{m(s,\Pt x(\beta(r)))}
    &=\Absm{(s,r)\in Y\times X}{m(s,\beta(r)x)}\\
    &=\Absm{(s,r)\in Y\times X}{\beta(r)m(s,x)}
  \end{align*}
  which is measurable by measurability of multiplication.
  With the same notations and using moreover some
  \(\mu\in\Mcca{\Cmeas(X)}\) we must prove that
  \begin{align*}
    \Pt x\Big(\int^\Sone\beta(r)\mu(dr)\Big)
    =\int^C\Pt x(\beta(r))\mu(dr)\,.
  \end{align*}
  Remember that the second member of this equation is well defined
  since we have shown that \(\Pt x\Comp\beta\in\Mcca{\Cpath XC}\). To
  check the equation, let %
  \(m\in\Mcms C_\Measterm\), we have
  \begin{align*}
    m\Big(\Pt x\Big(\int^\Sone\beta(r)\mu(dr)\Big)\Big)
    &=m\Big(\Big(\int^\Sone\beta(r)\mu(dr)\Big)x\Big)\\
    &=\Big(\int^\Sone\beta(r)\mu(dr)\Big)m(x)\\
    &=\int^\Sone\beta(r)m(x)\mu(dr)\\
    &=\int^\Sone m(\Pt x(\beta(r)))\mu(dr)\\
    &=m\Big(\int^C\Pt x(\beta(r))\mu(dr)\Big)\,.
  \end{align*}
  So \(\Pt x\in\ICONES(\Sone,C)\) as contended. Since \(\Pt x(1)=x\)
  this shows that \(\Sone\) is a separator.

  Let \(D\) be a subobject of \(C\); more precisely let %
  \(h\in\ICONES(D,C)\) be a mono.
  This implies that \(h\) is injective because \(\Sone\) is
  a separator.
  Let \(S=h(\Mcca D)\subseteq\Mcca C\).
  By Lemma~\ref{lemma:int-cone-transp-iso}, there is an integrable
  cone \(D'\) such that \(\Mcca{D'}=S\) (as sets) and \(f\), the
  corestriction of \(h\) to \(S\), is an iso in \(\ICONES\) from \(D\)
  to \(D'\).
  Moreover the inclusion \(e\) of \(S\) into the set \(\Mcca C\)
  satisfies \(e=h\Comp\Inv f\) and hence \(e\in\ICONES(D',C)\).

  We have proven that, in the slice category \(\ICONES/C\), each subobject
  \((D,h)\) of \(C\) is isomorphic to a subobject \((D',h')\) of \(C\)
  such that \(h'\) is an inclusion (that is
  \(\forall y\in\Mcca{D'}\ h'(y)=y\)).
  Notice finally that that the class of subobjects \((D',h')\) of \(C\) such
  that \(h'\) is an inclusion is a set because \(\ARCAT\) is a set%
  \footnote{It is only here that we use this assumption but it is
    essential.}.
  Consider indeed a subset \(S\) of \(\Mcca C\).
  The class of all structures of measurable cones \(D'\) whose underlying set is \(S\) is contained in
  \begin{align*}
    &S^{S\times S}
      \text{\quad (contains all possible additions)}\\
    \times
    &S^{\Realp\times S}
      \text{\quad (contains all possible scalar multiplications)}\\
    \times
    & \Realp^S
      \text{\quad (contains all possible norms)}\\
    \times
    &\prod_{X\in\ARCAT}\Part{\Realp^{X\times S}}
      \text{\quad (contains all possible measurability structures)}
  \end{align*}
  which is a set \(\cF(S)\) \emph{because
    \(\ARCAT\) is small}.
  Now the class of all subobjects \((D',h')\) of \(C\) such that
  \(h'\) is an inclusion is
  contained in %
  \(\{(S,F)\St S\subseteq\Mcca C\text{ and }F\in\cF(S)\}\), which is a
  set.
  This shows that the class of subobjects of \(C\) is essentially
  small, that is \(\ICONES\) is well-powered.
\end{proof}

\begin{theorem} %
  \label{th:Icones-adjoint-functor}
  If \(\cC\) is a locally small category and %
  \(R:\ICONES\to\cC\) is a functor which preserves all limits, then
  \(R\) has a left adjoint.
\end{theorem}
\begin{proof}
  Apply the special adjoint functor theorem.
\end{proof}

\begin{remark}
  This implies in particular that the forgetful functor
  \(\ICONES\to\MCONES\) (which obviously preserves all limits) has a
  left adjoint, meaning that each measurable cone can be ``completed
  with integrals''.
\end{remark}

\subsection{Colimits and coproducts}
\label{sec:icones-coproducts}
\begin{theorem} %
  \label{th:Icones-all-colimits}
  The category \(\ICONES\) has all small colimits.
\end{theorem}
\begin{proof}
  Let \(I\) be a small category, we use \(\ICONES^I\) for the category
  whose objects are the functors \(I\to\ICONES\) and the morphisms are
  the natural transformations, which is locally small since \(I\) is
  small.
  Then we have a ``diagonal'' functor \(\Delta:\ICONES\to\ICONES^I\)
  which maps each object of \(\ICONES\) to the corresponding constant
  functor and each morphism to the identity natural transformation.
  It is easily checked that \(\Delta\) preserves all limits and hence
  it has a left adjoint by Theorem~\ref{th:Icones-adjoint-functor}.
  By definition of an adjunction, this functor maps each functor
  \(I\to\ICONES\) to its colimit which shows that \(\ICONES\) is
  cocomplete.
\end{proof}
\noindent 
This theorem does not give any insight on the structure of these
colimits%
\footnote{In particular it would be interesting to have a more
  explicit description of coequalizers.}, %
so it is reasonable to have at least a closer look at coproducts.

\paragraph{Coproducts of cones}
Let \(I\) be a set, without any restrictions on its cardinality for
the time being.
Let \((P_i)_{i\in I}\) be a family of cones. Let \(P\) be the set of
all families %
\(\Vect x=(x_i)_{i\in I}\in\prod_{i\in I}P_i\) such that %
\(\sum_{i\in I}\Norm{x_i}<\infty\).
Notice that for such a family \(\Vect x\), the set
\(\{i\in I\St x_i\not=0\}\) is countable.
We turn \(P\) into a cone by defining the operations componentwise and
by setting \(\Norm{\Vect x}=\sum_{i\in I}\Norm{x_i}\).
The induced cone order relation on \(P\) is the pointwise order
and \(\omega\)-completeness is easily proven (by commutations of lubs with
sums in \(\Realpc\)).
In \(\CONES\), this cone \(P\) is the coproduct of
the \(P_i\)'s with obvious injections \(\Inj i\in\CONES(P_i,P)\)
mapping \(x\) to the family \(\Vect x\) such that \(x_i=x\) and
\(x_j=0\) for \(j\not=i\).
Given a family \((f_i\in\CONES(P_i,Q))_{i\in I}\) the unique map %
\(\Cotuple{f_i}_{i\in I}\in\CONES(P,Q)\) such that %
\(\forall j\in I\ \Cotuple{f_i}_{i\in I}\Compl\Inj j=f_j\) %
is given by
\begin{align*}
  \Cotuple{f_i}_{i\in I}(\Vect x)
  =\sum_{i\in I}f_i(x_i)\,.
\end{align*}
This sum converges because for each finite \(J\subseteq I\) one has
\begin{align*}
  \Norm{\sum_{i\in J}f_i(x_i)}
  \leq\sum_{i\in J}\Norm{f_i(x_i)}
  \leq\sum_{i\in J}\Norm{x_i}=\Norm x
\end{align*}
and this map \(\Cotuple{f_i}_{i\in I}\) is easily seen to be linear
and continuous.
We use \(\Bplus_{i\in I}P_i\) for the cone \(P\) defined in that way.

\begin{lemma}
  For each cone \(Q\) the cones %
  \(\Limpl{(\Bplus_{i\in I}P_i)}Q\) and %
  \(\Bwith_{i\in I}\Limplp{P_i}Q\) are isomorphic in \(\CONES\).
\end{lemma}
\begin{proof}
  The fact that \(\Bplus_{i\in I}P_i\) is the coproduct of the
  \(P_i\)'s means that the function %
  \[\Cuballp{\Limpl{(\Bplus_{i\in I}P_i)}Q}
    \to\Cuball{(\Bwith_{i\in I}\Limplp{P_i}Q)}\] which maps %
  \(f\) to \((f\Compl\Inj i)_{i\in I}\) is a bijection. It is linear
  and continuous by linearity and continuity of composition of
  morphisms. So this bijection is an isomorphism.
\end{proof}
\noindent 
In particular \((\Bplus_{i\in I}P_i)'\Isom\Bwith_{i\in
  I}{P_i'}\). Given %
\(\Vect{x'}\in\Bwith_{i\in I}P_i'\) the associated linear and
continuous form \(\Lfun{\Vect{x'}}\) on \(\Bplus_{i\in I}P_i\) is
given by
\begin{align*}
  \Lfun{\Vect{x'}}(\Vect x)
  =\Eval{\Vect x}{\Vect{x'}}=\sum_{i\in I}\Eval{x_i}{x'_i}
  \leq\Norm{\Vect x}\Norm{\Vect{x'}}\,.
\end{align*}
We use these observations in the sequel.

\paragraph{Coproduct of measurable cones}
Let \((C_i)_{i\in I}\) be a family of measurable cones. Let
\(P=\Bplus_{i\in I}\Mcca{C_i}\).
Let \(\cM=(\cM_X)_{X\in\ARCAT}\) where \(\cM_X\) is the set of all
\(p\in(P')^{X}\) such that there is a family of coefficients
\((\lambda_i\in\Realp)_{i\in I}\) with %
\(\Absm{r\in X}{\lambda_ip(r)_i}\in\Mcms{C_i}_X\), identifying
\(P'\) %
with \(\Bwith_{i\in I}\Mcca{C_i}'\) as explained above.
In other words %
\(p\in\cM_X\) means that there are families %
\(\Vect m=(m_i\in\Mcms{C_i}_X)_{i\in I}\) and %
\(\Vect\lambda=(\lambda_i\in\Realp)_{i\in I}\) such that, for all
\(r\in X\), the family %
\((\lambda_i\Norm{m_i(r)})_{i\in I}\) is bounded by \(1\), and we have
\(p(r)=\Lfun{\Vect\lambda\Vect m(r)}\) (where
\(\Vect\lambda\Vect x=(\lambda_ix_i)_{i\in I}\)).
Remember indeed from \Msmesr{} that, for each measurable cone \(C\), it
is assumed that each \(m\in\Mcms C_X\) satisfies that
\(m(r,x)\in\Intercc 01\) for all \(r\in X\) and
\(x\in\Cuball{\Mcca C}\).

\begin{remark}
  These coefficients \(\lambda_i\in\Realp\) are necessary because in the
  definition of measurable cones, we make very weak assumptions about
  the sets of measurability tests, in particular we do not assume that
  they are closed under multiplication by nonnegative coefficients
  \(\leq1\).
  Such assumptions ---~and stronger ones, for instance, as suggested
  by one of the reviewers, we could require these sets of tests to be
  cones with operations defined pointwise~--- would be quite
  meaningful, but would require to check additional conditions in the
  proofs, for artificial reasons.
  The sets of measurability tests of a cone \(C\) should be understood
  as a kind of ``predual'' of the cone of paths \(\Cmeas(X,C)\), in
  the sense that the criterion for a bounded map \(X\to\Mcca C\) to
  belong to this cone is the measurability of the (suitably defined)
  composition of this map with all measurability tests.
\end{remark}

We prove that \(\cM\) is a measurability structure on \(P\).
Given \(p\in\cM_X\) and \(\Vect x\in P\) the map %
\(\Absm{r\in X}{p(r)(\Vect x)}\) is measurable by the monotone
convergence theorem because the set \(\Eset{i\in I\St x_i\not=0}\) is
countable so the condition \Msmesr{} holds.
The conditions \Mscompr{} and \Mssepr{} obviously hold, let us check
\Msnormr{}. Let \(\Vect x\in P\setminus\Eset 0\) and let
\(\epsilon>0\).
Let \(J=\Eset{i\in I\St x_i\not=0}\) which is countable and let
\((i(n))_{n\in\Nat}\) be an enumeration of this set (assuming that it
is infinite; the case where it is finite is simpler).
For each \(n\in\Nat\) let \(m_n\in\Mcms{C_{i(n)}}_X\) be such that
\(m_n\not=0\) and
\begin{align*}
  \Norm{x_{i(n)}}_{C_{i(n)}}
  \leq \frac{m_n(x_{i(n)})}{\Norm{m_n}}+\frac\epsilon{2^{n+1}}\,.
\end{align*}
Let \(p\in\cM_\Measterm\) be given by
\(p(\Vect y)=\sum_{n\in\Nat}\frac{m_n(y_{i(n)})}{\Norm{m_n}}\).
We have \(p(y)\leq\sum_{n\in\Nat}\Norm{y_{i(n)}}\leq 1\) and hence
\(0<\Norm p\leq1\).
So we have
\begin{align*}
  \Norm{\Vect x}
  =\sum_{n\in\Nat}\Norm{x_{i(n)}}
  \leq\sum_{n\in\Nat}\frac{m_n(x_{i(n)})}{\Norm{m_n}}
  +\sum_{n\in\Nat}\frac\epsilon{2^{n+1}}
  =p(\Vect x)+\epsilon
  \leq\frac{p(\Vect x)}{\Norm p}+\epsilon
\end{align*}
proving our contention. We have shown that \(C=(P,\cM)\) is a
measurable cone that we denote as \(\Bplus_{i\in I}C_i\).

\begin{theorem}
  For each \(j\in I\) one has
  \((\Inj j\in\MCONES(C_j,\Bplus_{i\in I}C_i))_{i\in I}\). %
  
  If \(I\) is countable then
  \((\Bplus_{i\in I}C_i,(\Inj i)_{i\in I})\) is the coproduct of the
  \(C_i\)'s in \(\MCONES\).
  If moreover the \(C_i\)'s are integrable then so is
  \(\Bplus_{i\in I}C_i\), the \(\Inj i\)'s preserve integrals and
  \((\Bplus_{i\in I}C_i,(\Inj i)_{i\in I}\) is the coproduct of the
  \(C_i\)'s in \(\ICONES\).
\end{theorem}
\begin{proof}
  The measurability of the \(\Inj i\)'s is easy to prove.

  We assume that \(I\) is countable.
  Let %
  \((f_i\in\MCONES(C_i,D))_{i\in I}\), we have already defined %
  \(f=\Cotuple{f_i}_{i\in I}\in\CONES(\Bplus_{i\in
    I}{\Mcca{C_i}},\Mcca D)\) and we must prove that this function is
  measurable.
  Let \(X\in\ARCAT\) and %
  \(\gamma\in\Mcca{\Cpath X{\Bplus_{i\in I}C_i}}\), we prove that %
  \(f\Comp\gamma\in\Mcca{\Cpath XD}\) so let \(Y\in\ARCAT\) and %
  \(q\in\Mcms D_Y\), we have
  \begin{align*}
    \Absm{(s,r)\in{Y\times X}}{q(s,f(\gamma(r)))}
    &=\Absm{(s,r)\in{Y\times X}}{q\big(s,\sum_{i\in I}f_i(\gamma(r)_i)\big)}\\
    &=\Absm{(s,r)\in{Y\times X}}
      {\sum_{i\in I}q(s,f_i(\gamma(r)_i))}
  \end{align*}
  which is measurable by the monotone convergence theorem since \(I\)
  is countable.

  Assume moreover that the \(C_i\)'s are integrable and let
  \(\mu\in\Mcca{\Cmeas(X)}\).
  For each \(i\in I\) we have %
  \(\Absm{r\in X}{\gamma(r)_i}\in\Mcca{\Cpath X{C_i}}\)
  because, for each \(Y\in\ARCAT\) and \(m\in\Mcms{C_i}_Y\) we know
  that %
  \(\Absm{(s,r)\in{Y\times X}}{p(s,\gamma(r))}\) %
  is measurable, where \(p=\Lfun{\Vect m}\) with \(m_j=m\) if \(i=j\)
  and \(m_j=0\) otherwise%
  \footnote{It is harmless to assume that \(0\in\Mcms B_X\) for each
    measurable cone \(B\) and \(X\in\ARCAT\): if \(B\) is a measurable
    cone and \(C=(\Mcca B,\cM)\) where \(\cM_X=\Mcms{B}_X\cup\Eset 0\)
    then \(C\) is a measurable cone and \(B\) and \(C\) are isomorphic
    in \(\MCONES\), see Remark~\ref{rk:mcones_iso}.
    And similarly in the category of integrable cones.}. %
  Therefore we can define %
  \(\Vect x\in\prod_{i\in I}\Mcca{C_i}\) by %
  \(x_i=\int\gamma(r)_i\mu(dr)\).

  Given \(p=\Lfun{\Vect\lambda\Vect m}\in\Mcms C_0\), the map
  \(p\Comp\gamma: X\to\Realp\) is bounded and measurable, and
  we have
  \begin{align*}
    \int p(\gamma(r))\mu(dr)
    &=\int\big(\sum_{i\in I}\lambda_im_i(\gamma(r)_i)\big)\mu(dr)\\
    &=\sum_{i\in I}\int \lambda_im_i(\gamma(r)_i)\mu(dr)
      \text{\quad by the monotone convergence theorem,}\\
    &\Textsep\text{since }I\text{ is countable}\\
    &=\sum_{i\in I}\lambda_im_i\Big(\int\gamma(r)_i\mu(dr)\Big)
      \text{\quad by definition of integrals}\\
    &=p(\Vect x)\,.
  \end{align*}
  By \Msnormr{} holding in \(C\) as shown above, this proves that
  \(\Norm{\Vect x}<\infty\) and hence \(\Vect x\in\Mcca C\) and the
  computation above shows also that \(\Vect x=\int\gamma(r)\mu(dr)\)
  and hence the cone \(C\) is integrable. The proof that it is the
  coproduct of the \(C_i\)'s in \(\ICONES\) is routine.
\end{proof}
\noindent 
Even if \(I\) is not countable we know that \((C_i)_{i\in I}\) has a
coproduct in \(\ICONES\) by Theorem~\ref{th:Icones-all-colimits}, but
we don't know yet how to describe it concretely.

\section{Internal linear hom and the tensor product}
\label{sec:lin-hom-tensor}
The main goal of this section is to define a tensor product of two
integrable cones, to prove that this operation is functorial and that
\(\ICONES\) can be equipped with a structure of symmetric monoidal
category (SMC) which is closed (SMCC).

\begin{remark}
  \label{rk:tensor-concrete-pres}
  Of course we first tried to define the tensor product concretely as
  one usually does in algebra, using some quotient.
  However the complicated interaction between the algebraic and the
  order theoretic properties of cones made the resulting description
  ineffective for proving basic properties expected from a tensor
  product.
  Now that we know that the tensor product exists for abstract
  reasons, and has the required structures and properties, the quest
  for a reasonably simple concrete description can be undertaken with
  a more relaxed mind.
\end{remark}

We define first the integrable cone \(\Limpl CD\) of
linear, continuous, measurable and integrable morphisms \(C\to D\).
There are two good reasons for doing so.
\begin{itemize}
\item The definition of this object is easy and natural.
\item Building this object will be necessary for proving that the SMC
  we define is closed.
\end{itemize}
Moreover, it is easy to prove that the associated internal hom functor
\(\Limpl C\_:\ICONES\to\ICONES\) preserves all limits.
Then, thanks to Theorem~\ref{th:Icones-adjoint-functor} this functor
has a left adjoint \(\Tens \_C:\ICONES\to\ICONES\), and this operation
is also functorial wrt.~\(C\) because \(\Limpl\_\_\) is a functor
\(\Op\ICONES\times\ICONES\to\ICONES\) (thanks to a standard result in
category theory), and we have a natural bijection of sets %
\(\ICONES(\Tens BC,D)\to\ICONES(B,\Limpl CD)\).

Last we prove that this natural bijection is actually an isomorphism
\((\Limpl{\Tens BC}D)\to(\Limpl B{(\Limpl CD)})\) in \(\ICONES\) and
we show how to derive the SMC structure of \(\ICONES\) from this
property.

We are convinced that this method is exactly the one described
axiomatically in~\cite{EilenberKelly65}.
However we do not prove explicitly that \(\ICONES\) is closed in the
sense of that paper and do not apply explicitly its results, first for
the sake of self-containedness and also because, due to the concrete
features of our category (basically: our morphisms are functions) the
direct approach remains tractable.

Another approach, equivalent but conceptually more
elegant, would have been to describe first
\(\ICONES\) as a multicategory, introducing from the beginning a
notion of multilinear morphism on \(\ICONES\) in a completely standard
way.
Then the tensor product would have been defined by a familiar
universal property wrt.~bilinear morphisms.

\subsection{The cone of linear morphisms} %
\label{sec:cone-linear-fun}
Let \(C\) and \(D\) be objects of \(\ICONES\) and let \(P\) be the set
of all \(f:\Mcca C\to\Mcca D\) such that, for some \(\epsilon>0\), one
has \(\epsilon f\in\ICONES(C,D)\), equipped with the same algebraic
structure as %
\(\Limpl{\Mcca C}{\Mcca D}\) (see Lemma~\ref{lemma:limpl-cone}).
This makes sense since the algebraic laws of the cone
\(\Limpl{\Mcca C}{\Mcca D}\) preserve measurability and since
integration is linear.
Moreover given an increasing sequence \((f_n)_{n\in\Nat}\) of measurable
and integral preserving elements of \(\Limpl{\Mcca C}{\Mcca D}\) such
that \(\Norm{f_n}\leq 1\) (remember that
\(\Norm f=\sup_{x\in\Cuball{\Mcca C}}\Norm{f(x)}\)), the linear and
continuous map \(f=\sup_{n\in\Nat}f_n\) is measurable and preserves
integrals by the monotone convergence theorem, as we show now.

Let \(\gamma\in\Mcca{\Cpath XC}\) be a measurable path and let
\(m\in\Mcms D_{Y}\) for some \(Y\in\ARCAT\).
The function %
\(\phi %
=\Absm{(s,r)\in{Y\times X}}{m(s,f(\gamma(r))} %
:{Y\times X}\to\Intercc 01\) satisfies %
\(\phi(s,r)=\sup_{n\in\Nat}\phi_n\) where %
\((\phi_n=\Absm{(s,r)\in{Y\times X}}{m(s,f_n(\gamma(r))})_{n\in\Nat}\) %
is an increasing sequence of measurable functions by measurability of
the %
\(f_n\)'s and linearity and continuity of %
\(m\) in its second parameter, so \(\phi\) is measurable which shows
that \(f\) is measurable.
Next, with the same notations and taking moreover some
\(\mu\in\Mcca{\Cmeas(X)}\) we have, for each \(m\in\Mcms D_\Measterm\),
\begin{align*}
  m\Big(f\Big(\int^C\gamma(r)\mu(dr))\Big)\Big)
  &=m\Big(\sup_{n\in\Nat}f_n\Big(\int^C\gamma(r)\mu(dr)\Big)\Big)\\
  &=m\Big(\sup_{n\in\Nat}\int^D f_n(\gamma(r))\mu(dr)\Big)\\
  &=\sup_{n\in\Nat}m\Big(\int^D f_n(\gamma(r))\mu(dr)\Big)
    \text{\quad by cont.~of }m\\
  &=\sup_{n\in\Nat}\int^\Sone m(f_n(\gamma(r)))\mu(dr)
    \text{\quad by def.~of integration in }D\\
  &=\int^\Sone \sup_{n\in\Nat} m(f_n(\gamma(r)))\mu(dr)
    \text{\quad by the monotone conv.~th.}\\
  &=\int^\Sone m(f(\gamma(r)))\mu(dr)
    \text{\quad by continuity of }m\\
  &=m\Big(\int^D f(\gamma(r))\mu(dr)\Big)
    \text{\quad by def.~of integration in }D
\end{align*}
and hence \(f\) preserves integrals by \Mssepr.

Given \(\gamma\in\Mcca{\Cpath XC}\) and %
\(m\in\Mcms D_X\) we define %
\(\Mtlfun\gamma m =\Absm{(r,f)\in X\times P}{m(r,f(\gamma(r)))}
:X\times P\to\Realp\).
For each \(r\in X\) the function %
\(l=(\Mtlfun\gamma m)(r,\_):\Mcca{\Limpl CD}\to\Realp\) is linear and
continuous by linearity and continuity of \(m\) in its second
argument.
We define %
\(\cM_X =\{\Mtlfun\gamma m\St\gamma\in\Mcca{\Cpath XC}\text{ and
}m\in\Mcms D_X\}\).
\begin{remark}
  We use the same notations for measurability tests on the cone of
  linear, continuous and integrable morphisms as for the cone of
  measurable paths in Section~\ref{sec:meas-cone-of-paths} for two reasons.
  The first one is that these tests are defined in a very similar way,
  the second one is Theorem~\ref{th:meas-path-equiv} which allows to
  see measurable paths as linear, continuous and integrable maps.
\end{remark}

\noindent 
We check that the family \((\cM_X)_{}\) is a measurability structure
on the cone \(P=\Limplp{\Mcca C}{\Mcca D}\).

\Proofcase\Msmesr{}
Let \(f\in\Cuball P\), %
\(\gamma\in\Mcca{\Cpath XC}\) and \(m\in\Mcms D_X\), then the map %
\(\phi=\Absm{r\in X}{m(r,f(\gamma(r)))}\) is measurable because
\(f\Comp\gamma\in\Mcca{\Cpath XD}\) by measurability
of %
\(f\) and hence
\(\Absm{(s,r)\in{X\times X}}{m(s,f(\gamma(r)))}: {X\times
  X}\to\Intcc01\) is measurable from which follows the measurability
of \(\phi\).
The fact that \(\phi\) ranges in \(\Intercc 01\)
results from the assumption that \(\Norm f\leq 1\).

\Proofcase\Mscompr{}
Let \(\gamma\in\Mcca{\Cpath XC}\) and \(m\in\Mcms D_X\), and let %
\(\phi\in\ARCAT(Y,X)\) for some \(Y\in\ARCAT\).
Then we have %
\((\Mtlfun\gamma m)\Comp(\phi\times P) =\Absm{(s,f)\in Y\times
  P}{m(\phi(s),f(\gamma(\phi(s)))))}
=\Mtlfun{(\gamma\Comp\phi)}{(m\Comp(\phi\times\Mcca D))}\) and since
\(\gamma\Comp\phi\in\Mcca{\Cpath YC}\) by
Lemma~\ref{lemma:precomp-path} and %
\(m\Comp(\phi\times\Mcca D)\in\Mcms D_Y\) by property \Mscompr{}
satisfied in \(D\), we have %
\((\Mtlfun\gamma m)\Comp(\phi\times P)\in\cM_Y\).

\Proofcase\Mssepr{}
Let \(f_1,f_2\in P\) and assume that for all %
\(x\in\Mcca C\) and \(m\in\Mcms D_\Measterm\) one has %
\((\Mtlfun xm)(f_1)=(\Mtlfun xm)(f_2)\), that is %
\(m(f_1(x))=m(f_2(x))\).
By \Mssepr{} in \(D\) we have %
\(f_1(x)=f_2(x)\), and since this holds for all \(x\in\Mcca C\) %
we have \(f_1=f_2\).

\Proofcase\Msnormr{}
Let \(f\in P\setminus\Eset 0\).
Let \(\epsilon>0\), we can assume without loss of generality that
\(\epsilon<2\Norm f\).
By definition of \(\Norm f\) there is \(x\in\Cuball{\Mcca C}\) such
that \(\Norm{f}\leq \Norm{f(x)}+\epsilon/3\) and hence
\(\Norm{f}<\Norm{f(x)}+\epsilon/2\).
This implies in particular that \(f(x)\not=0\) by our assumption that
\(\epsilon<2\Norm f\).
By \Msnormr{} in \(D\) there is
\(m\in\Mcms D_\Measterm\setminus\Eset 0\) such that %
\[
  \Norm{f(x)}\leq m(f(x))/\Norm m+\min(\epsilon/2,\Norm
  f-\epsilon/2)\,.
\] %
If \(m(f(x))=0\) we have \(\Norm{f(x)}\leq\Norm f-\epsilon/2\) which
is not possible since \(\Norm{f}<\Norm{f(x)}+\epsilon/2\), so
\((\Mtlfun xm)(f)=m(f(x))\not=0\).
This implies in particular that \(\Norm{\Mtlfun xm}\not=0\).
We have %
\(\Norm{\Mtlfun xm}=\sup_{g\in\Cuball P} m(g(x)) \leq\Norm
m\sup_{g\in\Cuball P}\Norm{g(x)}\leq\Norm m\) since
\(x\in\Cuball{\Mcca C}\), and hence %
\begin{align*}
  \Norm f\leq \Norm{f(x)}+\epsilon/2
  &\leq \frac{m(f(x))}{\Norm m}+\epsilon
  =\frac{(\Mtlfun xm)(f)}{\Norm m}+\epsilon\\
  &\leq\frac{(\Mtlfun xm)(f)}{\Norm{\Mtlfun xm}}+\epsilon
    \text{\quad since }0<\Norm{\Mtlfun xm}\leq \Norm m\,.
\end{align*}
So we have defined a measurable cone %
that we denote as %
\(\Limpl CD\).

We will need a technical lemma whose intuitive meaning is interesting
\emph{per se}: a linear morphisms valued in a cone of paths is the
same thing as a path valued in the measurable cone of linear morphisms
just defined.
This lemma will be essential for proving that the measurable cone
\(\Limpl CD\) is integrable.

\begin{lemma} %
  \label{lemma:swap-lin-path}
  There is an argument swapping isomorphism %
  \[
    \Swlinpath\in\MCONES(\Limpli C{\Cpath XD},\Cpathm X{\Limpli CD}
  \] %
  which maps \(f\) to %
  \(\Absm{r\in X}{\Absm{x\in\Mcca C}{f(x)(r)}}\).
\end{lemma}
This iso preserves integrals, but this property is not required for
what follows.
\begin{proof}
  Let \(f\in\Mcca{\Limpli C{\Cpath XD}}\).
  If \(r\in X\), the map %
  \(g=\Absm{x\in\Mcca C}{f(x)(r)}:\Mcca C\to\Mcca D\) is linear and
  continuous because \(f\) is, and the algebraic operations and the
  lubs are computed pointwise in \(\Mcca{\Cpath XD}\), we prove that
  \(g\) is measurable.
  Let \(Y,Y'\in\ARCAT\),  \(\gamma\in\Mcca{\Cpath YC}\) and %
  \(m\in\Mcms D_{Y'}\), we set %
  \(\phi=\Absm{(s',s)\in{Y'\times Y}}{m(s',g(\gamma(s)))}
  =\Absm{(s',s)\in{Y'\times Y}}{m(s',f(\gamma(s))(r))}\).
  Notice that, identifying \(r\) with the constant \(r\)-valued
  measurable function \(Y'\to X\), we have
  \(\Mtpath rm\in\Mcms{\Cpath XD}_{Y'}\) and %
  \(\phi=\Absm{(s',s)\in Y'\times Y}{(\Mtpath rm)(s',f(\gamma(s)))}\)
  which is measurable because %
  \(f\Comp\gamma\in\Mcca{\Cpath{Y}{\Cpath XD}}\) since %
  \(f\in\Mcca{\Limpli C{\Cpath XD}}\).
  This shows that \(g\) is measurable, we prove that \(g\) preserves
  integrals so let moreover \(\nu\in\Mcca{\Cmeas(Y)}\), we have %
  \begin{align*}
    g\Big(\int_{s\in Y}^{C}\gamma(s)\nu(ds)\Big)
    &=f\Big(\int_{s\in Y}^C\gamma(s)\nu(ds)\Big)(r)\\
    &=\Big(\int_{s\in Y}^{\Cpath XD}f(\gamma(s))\nu(ds)\Big)(r)
      \text{\quad since }f\text{ preserves integrals.}\\
    &=\int_{s\in Y}^D f(\gamma(s))(r)\nu(ds)\\
    &=\int_{s\in Y}^D g(\gamma(s))\nu(ds)\,.
  \end{align*}
  since integrals in \(\Cpath XD\) are computed pointwise.
  This shows that \(g=\Swlinpath(f)(r)\in\Mcca{\Limpl CD}\) for all
  \(r\in X\).

  We prove next that \(\eta=\Swlinpath(f)\) belongs to %
  \(\Mcca{\Cpathm X{\Limpli CD}}\) so let \(Y\in\ARCAT\) and %
  \(p\in\Mcms{\Limpli CD}_Y\).
  Let \(\gamma\in\Mcca{\Cpath YC}\) and %
  \(m\in\Mcms D_Y\) be such that \(p=\Mtlfun\gamma m\).
  The function %
  \(\phi=\Absm{(s,r)\in Y\times X}{p(s,\eta(r))}\) satisfies
  \begin{align*}
    \phi
    &=\Absm{(s,r)\in Y\times X}{m(s,\eta(r)(\gamma(s)))}\\
    &=\Absm{(s,r)\in Y\times X}{m(s,f(\gamma(s))(r))}\,.
  \end{align*}
  We know that %
  \(\delta=f\Comp\gamma\Comp\Proj1\in\Mcca{\Cpath{Y\times X}{\Cpath XD}}\)
  because \(\gamma\Comp\Proj 1\in\Mcca{\Cpath{Y\times X}{C}}\)
  and \(f\in\MCONES(C,\Cpath XD)\).
  Let \(m'\in\Mcms D_{Y\times X}\) be defined by %
  \(m'(s,r,y)=m(s,y)\), we have %
  \(\Mtpath{\Proj2}{m'}\in\Mcms{\Cpath XD}_{Y\times X}\) and hence %
  \(\phi'
  =\Absm{(s,r)\in Y\times X}{(\Mtpath{\Proj2}{m'})(s,r,\delta(s,r))}\) %
  is measurable.
  But
  \begin{align*}
    \phi'(s,r)
    ={m'(s,r,\delta(s,r)(\Proj2(s,r)))}
    ={m(s,f(\gamma(s))(r))}=\phi(s,r)
  \end{align*}
  so that \(\phi\) is measurable, this shows that %
  \(\Swlinpath(f)\in\Mcca{\Cpathm X{\Limpli CD}}\).
  The linearity and continuity of \(\Swlinpath\) are obvious (the
  algebraic operations and lubs are defined pointwise) as well as the
  fact that \(\Norm{\Swlinpath}\leq 1\).
  Its measurability relies on the obvious bijection between %
  \(\Mcms{\Limpli C{\Cpath XD}}_{Y'}\) and %
  \(\Mcms{\Cpathm X{\Limpli CD}}_{Y'}\) which maps %
  \(\Mtlfun{\gamma}{(\Mtpath\phi m)}\) to %
  \(\Mtpath\phi{(\Mtlfun\gamma m)}\) for all \(Y'\in\ARCAT\) (with
  \(\gamma\in\Mcca{\Cpath{Y'}{C}}\), \(\phi\in\ARCAT(Y',X)\) %
  and \(m\in\Mcms D_{Y'}\)).
  We have proven that \(\Swlinpath\) is a morphism in \(\MCONES\).

  Conversely given \(\eta\in\Mcca{\Cpathm X{\Limpli CD}}\) we define %
  \(f=\Swlinpath'(\eta) =\Absm{x\in\Mcca C}{\Absm{r\in
      X}{\eta(r)(x)}}\) %
  and prove first that %
  \(f\in\Mcca{\Limpli C{\Cpath XD}}\).
  Let \(x\in\Mcca C\) and \(\delta=f(x):X\to\Mcca D\).
  If \(r\in X\) we have \(\eta(r)\leq\Norm\eta\) and hence %
  \(\Norm{\delta(r)}=\Norm{\eta(r)(x)}\leq\Norm\eta\Norm x\) which
  shows that the function \(\delta\) is bounded.
  Let \(Y\in\ARCAT\) and \(m\in\Mcms D_Y\), we set %
  \(\phi=\Absm{(s,r)\in Y\times X}{m(s,\delta(r))}\). 
  For \(s\in Y\) and \(r\in X\), we have
  \begin{align*}
    \phi(s,r)
    ={m(s,\delta(r))}
    ={m(s,\eta(r)(x))}
    ={(\Mtlfun xm)}(s,\eta(r))
  \end{align*}
  where we identify \(x\) with the path \(\gamma\in\Mcca{\Cpath YC}\)
  such that \(\gamma(s)=x\), so that %
  \(\Mtlfun xm\in\Mcms{\Limpl CD}_Y\).
  It follows that \(\phi\) is measurable and hence %
  \(\delta\in\Mcca{\Cpath XD}\).

  Linearity of \(f\) is obvious and continuity results from the fact that lubs
  in \(\Mcca{\Cpath XD}\) are computed pointwise.
  Let \(\gamma\in\Mcca{\Cpath YC}\) for some \(Y\in\ARCAT\), we must
  prove next that %
  \(f\Comp\gamma\in\Mcca{\Cpath{Y}{\Cpath XD}}\).
  Let \(Y'\in\ARCAT\) and \(p\in\Mcms{\Cpath XD}_{Y'}\), we must
  prove that %
  \(\psi=\Absm{(s',s)\in Y'\times Y}{p(s',f(\gamma(s)))}\) %
  is measurable.
  Let \(\phi\in\ARCAT(Y',X)\) and \(m\in\Mcms D_{Y'}\) be such that %
  \(p=\Mtpath{\phi}{m}\).
  For \(s'\in Y'\) and \(s\in Y\) we have
  \begin{align*}
    \psi(s',s)
    &={m(s',f(\gamma(s))(\phi(s'))}\\
    &={m(s',\eta(\phi(s'))(\gamma(s)))}\\
    &={(\Mtpath{(\gamma\Comp\Proj2)}{(m\Comp(\Proj1\times\Mcca D))})
      (s',s,\eta\Comp\phi)}
  \end{align*}
  and hence \(\psi\) is measurable since \(\eta\Comp\phi\) is a
  measurable path.
  This shows that \(f\) is measurable, we prove last that \(f\)
  preserves integrals.
  So let \(\gamma\in\Mcca{\Cpath YC}\) and \(\mu\in\Mcca{\Cmeas(Y)}\).
  Given \(r\in X\) we have
  \begin{align*}
    f\Big(\int^C_{s\in Y}\gamma(s)\mu(ds)\Big)(r)
    &=\eta(r)\Big(\int^C_{s\in Y}\gamma(s)\mu(ds)\Big)\\
    &=\int^D_{s\in Y}\eta(r)(\gamma(s))\mu(ds)
      \text{\quad since }\eta(r)\in\Mcca{\Limpl CD}\\
    &=\int^D_{s\in Y}f(\gamma(s))(r)\mu(ds)
      \text{\quad by definition of }f\\
    &=\Big(\int^{\Cpath XD}_{s\in Y} f(\gamma(s))\mu(ds)\Big)(r)
  \end{align*}
since integrals in \(\Cpath XD\) are computed pointwise.

The proof that \(\Swlinpath'\) is a morphism in %
\(\MCONES\) follows the same pattern as for \(\Swlinpath\).
\end{proof}

\begin{lemma}
  The measurable cone \(\Limpl CD\) is integrable.
\end{lemma}
\begin{proof}
  Let \(X\in\ARCAT\), %
  \(\eta\in\Cpath X{\Limpli CD}\) and %
  \(\mu\in\Mcca{\Cmeas(X)}\).
  Let
\begin{align*}
  f=\Absm{x\in\Mcca C}{\int^D\eta(r)(x)\mu(dr)}
  =\Absm{x\in\Mcca C}{\int^D\Swlinpath(\eta)(x)(r)\mu(dr)}\,.
\end{align*}
This function is well defined since for each \(x\in\Mcca C\) one has %
\(\Swlinpath(\eta)(x)\in\Mcca{\Cpath XD}\) by
Lemma~\ref{lemma:swap-lin-path} so that the integral %
\(\int\Swlinpath(\eta)(x)(r)\mu(dr)\in\Mcca D\) is well defined.
The fact that \(f:\Mcca C\to\Mcca D\) is linear and continuous results
from the linearity of integration and from the monotone convergence
theorem.
Let us check that \(f\) is measurable so let \(Y\in\ARCAT\)
and let \(\gamma\in\Mcca{\Cpath YC}\), we must prove that
\begin{align*}
  \Absm{s\in Y}{\int^D\Swlinpath(\eta)(\gamma(s))(r)\mu(dr)}
  \in\Mcca{\Cpath YD}
\end{align*}
so let \(Y'\in\ARCAT\) and \(m\in\Mcms D_{Y'}\), we must check that
the function
\begin{align*}
  \psi
  &=\Absm{(s',s)\in Y'\times Y}
  {m\Big(s',\int^D\Swlinpath(\eta)(\gamma(s))(r)\mu(dr)\Big)}\\
  &=\Absm{(s',s)\in Y'\times Y}
  {\int^{\Realp} m(s',\Swlinpath(\eta)(\gamma(s))(r))\mu(dr)}
\end{align*}
is measurable.
We know that the function %
\(\Absm{(s',s,r)\in Y'\times Y\times X}
{m(s',\Swlinpath(\eta)(\gamma(s))(r))}\) is measurable and bounded
because %
\(\Swlinpath(\eta)\Comp\gamma\in\Mcca{\Cpath Y{\Cpath XD}}\) %
by Lemma~\ref{lemma:swap-lin-path} and we get the announced
measurability by Lemma~\ref{lemma:int-mesurable} (in the special case
where \(\kappa\) is the kernel constantly equal to \(\mu\)). %
Next we prove that \(f\) preserves integrals, so let moreover
\(\nu\in\Mcca{\Cmeas(Y)}\), we have
\begin{align*}
  f\Big(\int^C_{s\in Y}\gamma(s)\nu(ds)\Big)
  &=\int^D_{r\in X} \eta(r)\Big(\int^C_{s\in Y}\gamma(s)\nu(ds)\Big)\mu(dr)\\
  &=\int^D_{r\in X}\Big(\int^D_{s\in Y}\eta(r)(\gamma(s))\nu(ds)\Big)\mu(dr)
  \text{\quad since }\eta(r)\in\Mcca{\Limpli CD}\\
  &=\int^D_{r\in X}\Big(\int^D_{s\in Y}
    \Swlinpath(\eta)(\gamma(s))(r)\nu(ds)\Big)\mu(dr)\\
  &=\int^D_{s\in Y}\Big(\int^D_{r\in X}
    \Swlinpath(\eta)(\gamma(s))(r)\mu(dr)\Big)\nu(ds)\\
  &\hspace{3em}\text{by Th.~\ref{th:paths-Fubini} (Fubini), since }
    \Swlinpath(\eta)\Comp\gamma\in\Mcca{\Cpath Y{\Cpath XD}}\\
  &=\int^D_{s\in Y} f(\gamma(s))\nu(ds)\,.
\end{align*}
This completes the proof that %
\(f\in\Mcca{\Limpli CD}\) as contended.

Let \(p\in\Mcms{\Limpli CD}_\Measterm\).
Let %
\(x\in\Mcca C\) and \(m\in\Mcms D_\Measterm\) be such that %
\(p=\Mtlfun xm\), we have %
\begin{align*}
  p(f)
  &=m\Big(\int^D\eta(r)(x)\mu(dr)\Big)\\
  &=\int m(\eta(r)(x))\mu(dr)\\
  &=\int p(\eta(r))\mu(dr)\,,
\end{align*}
so \(\eta\) is integrable over \(\mu\), and
\(\int^{\Limpl CD}\eta(r)\mu(dr)=f\).
\end{proof}

\noindent 
This is the right place to insert a lemma very similar to
Lemma~\ref{lemma:swap-lin-path} which will be useful for exhibiting
the symmetry of our tensor product.
\begin{lemma} %
  \label{lemma:swap-lin-lin}
  There is an argument swapping natural isomorphism %
  \[
    \Swlinlin\in
    \ICONES(\Limpli{B_1}{(\Limpli{B_2}C)},\Limpli{B_2}{(\Limpli{B_1}C)})
  \] %
  which maps \(f\) to %
  \(\Absm{x_1\in\Mcca{B_1}}{\Absm{_2\in\Mcca{B_2}}{f(x_1)(x_2)}}\).
\end{lemma}
\begin{proof}
  One checks that
  \[
    \Swlinlin\in
    \MCONES(\Limpli{B_1}{(\Limpli{B_2}C)},\Limpli{B_2}{(\Limpli{B_1}C)})
  \] %
  as in the proof of~\ref{lemma:swap-lin-path}, and this morphism
  preserves integrals because integrals are computed pointwise in
  \(\Limpl DE\) for all integrable cones \(D\) and \(E\) and by the
  Fubini theorem because all the considered measures are finite.
\end{proof}

\subsection{Bilinear maps} %
\label{sec:bilin-measurable}
After the linear function space \(\Limpl CD\) that we have just
defined, the next concept deeply related to the tensor product is of
course the concept of bilinear map that we introduce now.

\begin{definition}
  \label{def:bilinear}
  Let \(C_1,C_2,D\) be integrable cones, we define formally
  \begin{align*}
    \Limplm{C_1,C_2}{D}=\Limpl{C_1}{\Limplp{C_2}{D}}
  \end{align*}
  and call this integrable cone \emph{the cone of integrable bilinear and
  continuous maps %
  \(C_1,C_2\to D\)}. %
\end{definition}
Indeed, thanks to Lemma~\ref{lemma:seprate-cont-implies-cont}, an
element of \(\Mcca{\Limplm{C_1,C_2}{D}}\) can be seen as a function %
\(f:\Mcca{C_1}\IWith\Mcca{C_2}\to\Mcca D\) which is separately linear
and \(\omega\)-continuous.  
Measurability of \(f\) is expressed equivalently by saying that given
\((X_i\in\ARCAT)_{i=1,2}\) and %
\((\gamma_i\in\Mcca{\Cpath{X_i}{C_i}})_{i=1,2}\) the map %
\(\Absm{(r_1,r_2)\in X_1\times X_2}{f(\gamma_1(r_1),\gamma_2(r_2))}:
X_1\times X_2\to\Mcca D\) is a measurable path
or that, given \(X\in\ARCAT\) and %
\((\gamma_i\in\Mcca{\Cpath X{C_i}})_{i=1,2}\), the map %
\(\Absm{r\in X}{f(\gamma_1(r),\gamma_2(r))}\) is a measurable path.
Preservation of integrals means that, given moreover %
\((\mu_i\in\Mcca{\Cmeas(X_i)})_{i=1,2}\), we have %
\begin{align*}
  f\Big(\int^{C_1}\gamma_1(r_1)\mu_1(dr_1),
  \int^{C_2}\gamma_2(r_2)\mu_2(dr_2)\Big)
  =\iint^D f(\gamma_1(r_1),\gamma_2(r_2))\mu_1(dr_1)\mu_2(dr_2)
\end{align*}
where we can use the double integral symbol by
Theorem~\ref{th:paths-Fubini}.

Continuing to spell out the definition above of the integrable cone %
\(\Limplm{C_1,C_2}{D}\), we see that, given \(X\in\ARCAT\), %
an element of \(\Mcms{\Limplm{C_1,C_2}{D}}_X\) is a %
\begin{align*}
  \Mtlfunm{\gamma_1,\gamma_2}{m}
  =\Absm{(r,f)\in X\times\Mcca{(\Limplm{C_1,C_2}{D})}}
  {m(r,f(\gamma_1(r),\gamma_2(r))}
\end{align*}
where %
\((\gamma_i\in\Mcca{\Cpath X{C_i}})_{i=1,2}\) and %
\(m\in\Mcms D_X\).
Last the integral of a measurable path %
\(\eta\in\Mcca{\Cpath X{(\Limplm{C_1,C_2}{D})}}\) %
over \(\mu\in\Mcca{\Cmeas(X)}\) is characterized by %
\begin{align*}
  \Big(\int\eta(r)\mu(dr)\Big)(x_1,x_2)
  =\int\eta(r)(x_1,x_2)\mu(dr)\,.
\end{align*}

\subsection{The linear hom functor}
\label{sec:lin-hom-functor}
In order to define the tensor product as a left adjoint, we need to
consider the operation \(\Limpl{}{}\) as an operation on morphisms of
\(\ICONES\), not only on objects.
We define this operation and prove that it preserves all limits in its
second argument.

\begin{definition}
  Let \(g\in\ICONES(D_1,D_2)\) and \(h\in\ICONES(C_2,C_1)\).
  The function
  \(\Limpl hg:\Mcca{\Limpl{C_1}{D_1}}\to \Mcca{\Limpl{C_2}{D_2}}\) is
  defined by \((\Limpl hg)(f)=g\Compl f\Compl h\).
\end{definition}

\begin{proposition}
  If \(g\in\ICONES(D_1,D_2)\) and \(h\in\ICONES(C_2,C_1)\) then %
  \(\Limpl hg\in\ICONES(\Limpl{C_1}{D_1},\Limpl{C_2}{D_2})\).
\end{proposition}
\begin{proof}
  The linearity and continuity of \(\Limpl hg\) result from the same
  properties satisfied by \(g\) and \(h\).
  The fact that %
  \(\Norm{\Limpl hg}\leq 1\) results from the fact that %
  \(\Norm g,\Norm h\leq 1\), so let us check that \(\Limpl hg\) is
  measurable. %
  Let \(\eta_1\in\Mcca{\Cpath X{\Limpl{C_1}{D_1}}}\) for some
  \(X\in\ARCAT\).
  We must prove that %
  \((\Limpl hg)\Comp\eta_1\in\Mcca{\Cpath X{\Limpl{C_2}{D_2}}}\) %
  so let %
  \(p\in\Mcms{\Limpl{C_2}{D_2}}_Y\) for some \(Y\in\ARCAT\), we must
  prove that
  \begin{align*}
    \phi=\Absm{(s,r)\in Y\times X}{p(s,(\Limpl hg)(\eta_1(r)))}
  \end{align*}
  is measurable.
  Let \(\gamma\in\Mcca{\Cpath Y{C_2}}\) and \(m\in\Mcms{D_2}_Y\) %
  be such that %
  \(p=\Mtlfun\gamma m\).
  For \(s\in Y\) and \(r\in X\) we have
  \begin{align*}
    \phi(s,r)
    ={m(s,g(\eta_1(r)(h(\gamma(s))))}
    ={m(s,g(\delta_1(s)(r)))}
    ={m(s,g(\Flpath(\delta_1)(s,r)))}
  \end{align*}
  where %
  \(\delta_1=\Inv\Swlinpath(\eta_1)\Comp h\Comp\gamma \in\Mcca{\Cpath
    Y{\Cpath X{D_1}}}\) by Lemma~\ref{lemma:swap-lin-path} and hence %
  \(g\Comp\Flpath(\delta_1)\in\Mcca{\Cpath{Y\times X}{D_2}}\) by
  Lemma~\ref{lemma:meas-path-flat} so that \(\phi\) is measurable.
  We need last to prove that %
  \(\Limpl hg\) preserves integrals so let moreover %
  \(\mu\in\Mcca{\Cmeas(X)}\), we have %
  \begin{align*}
    (\Limpl hg)\Big(\int^{\Limpl{C_1}{D_1}}\eta_1(r)\mu(dr)\Big)
    &=\Absm{x\in\Mcca{C_2}}
      {g\Big(\Big(\int^{\Limpl{C_1}{D_1}}\eta_1(r)\mu(dr)\Big)(h(x))\Big)}\\
    &=\Absm{x\in\Mcca{C_2}}{g\Big(\int^{D_1}\eta_1(r)(h(x))\mu(dr)\Big)}\\
    &=\Absm{x\in\Mcca{C_2}}{\int^{D_2} g\Big(\eta_1(r)(h(x))\Big)\mu(dr)}\\
    &=\Absm{x\in\Mcca{C_2}}{
      \int^{D_2}(\Limpl hg)(\eta_1(r))(x)\mu(dr)}\\
    &=\int^{\Limpl{C_2}{D_2}}(\Limpl hg)(\eta_1(r))\mu(dr)\,.
      \qedhere
  \end{align*}
\end{proof}

\noindent 
So we have defined a functor %
\(\Limpl\_\_:\Op\ICONES\times\ICONES\to\ICONES\). We identify
\(\Limpl \Sone\_\) with the identity functor: we make no distinction %
between \(x\in\Mcca C\) and the function %
\(\Pt x\in\Mcca{\Limpl\Sone C}\) (this notation is introduced in the
proof of Theorem~\ref{th:icones-conditions-saft}).

\begin{theorem} %
  \label{th:limpl-has-left-adj}
  For each integrable cone \(C\), the functor \(\Limpl C\_\) has a
  left adjoint.
\end{theorem}
\begin{proof}
  By Theorem~\ref{th:Icones-adjoint-functor} it suffices to prove that %
  \(\Limpl C\_\) preserves all limits.

  \Proofcase{} Products.
  Let \((D_i)_{i\in I}\) be a family of measurable cones and let %
  \(D=\Bwith_{i\in I}D_i\) as described in the proof of
  Theorem~\ref{th:mcones-complete}.
  We have a morphism
  \begin{align*}
    k=\Tuple{\Limpl C{\Proj i}}_{i\in I}
    \in\ICONES(\Limpl CD,\Bwith_{i\in I}\Limplp C{D_i})
  \end{align*}
  and we must prove that \(k\) is an iso.
  It is clearly injective, to prove surjectivity, let %
  \(\Vect f=(f_i\in\Mcca{\Limpl C{D_i}})_{i\in I}
  \in\Mcca{\Bwith_{i\in I}{\Limplp C{D_i}}}\) so that %
  \((\Norm{f_i})_{i\in I}\) is bounded in \(\Realp\) and hence for
  each \(x\in\Mcca C\) the family \((\Norm{f_i(x)})_{i\in I}\) is
  bounded.
  So we can define a function %
  \(f:\Mcca C\to\Mcca D\) by \(f(x)=(f_i(x))_{i\in I}\).
  This function is clearly linear and continuous.
  To prove measurability, take %
  \(\gamma\in\Mcca{\Cpath XC}\) for some \(X\in\ARCAT\) and %
  \(p\in\Mcms D_Y\) for some \(Y\in\ARCAT\).
  This means that %
  \(p=\Mtinj im\) for some \(i\in I\) and \(m\in\Mcms{D_i}_Y\).
  Then %
  \(\Absm{(s,r)\in Y\times X}{p(s,f(\gamma(r)))}
  =\Absm{(s,r)\in Y\times X}{m(s,f_i(\gamma(r)))}\) %
  is measurable because \(f_i\) is.
  Last let moreover %
  \(\mu\in\Mcca{\Cmeas(X)}\), we have
  \(f(\int^C\gamma(r)\mu(dr))=\int^D f(\gamma(r))\mu(dr)\) by
  definition of \(f\) and of integration in \(D\).
  This shows that %
  \(f\in\Mcca{\Limpl CD}\) and hence that \(k\) is a bijection since
  \(f_i=\Proj i\Compl f\) for each \(i\in I\) and hence
  \(\Vect f=k(f)\).

  We prove that %
  \(\Inv k\in\ICONES(\Bwith_{i\in I}\Limplp C{D_i},\Limpl
  CD)\).
  Linearity and continuity follow from the fact that all
  operations are defined componentwise in %
  \(\Mcca{\Bwith_{i\in I}}\Limplp C{D_i}\).
  Next, given
  \(\Vect f\in\Cuball{(\Mcca{\Bwith_{i\in I}}\Limplp C{D_i})}\), we
  have %
  \begin{align*}
    \Norm{\Inv k(\Vect f)}
    &=\sup_{x\in\Cuball{\Mcca C}}\Norm{\Inv k(\Vect f)(x)}\\
    &=\sup_{x\in\Cuball{\Mcca C}}\sup_{i\in I}\Norm{f_i(x)}\\
    &=\sup_{i\in I}\sup_{x\in\Cuball{\Mcca C}}\Norm{f_i(x)}\\
    &=\sup_{i\in I}\Norm{f_i}\leq 1\,.
  \end{align*}
  Next we prove that \(\Inv k\) is measurable so let %
  \(\eta\in\Mcca{\Cpath{X}{\Bwith_{i\in I}\Limplp C{D_i}}}\) %
  for some \(X\in\ARCAT\), we must prove that %
  \(\eta'=\Inv k\Comp\eta\in\Mcca{\Cpath{X}{\Limpl CD}}\). %
  Notice that for all \(r\in X\) we can write %
  \(\eta(r)=(\eta_i(r))_{i\in I}\) where %
  \(\eta_i=\Proj i\Comp\eta\in\Mcca{\Cpath X{\Limpl C{D_i}}}\) for
  each \(i\in I\).
  Let \(Y\in\ARCAT\), \(\gamma\in\Mcca{\Cpath YC}\) and %
  \(p\in\Mcms D_Y\), so that %
  \(p=\Mtinj im\) for some \(i\in I\) and \(m\in\Mcms{D_i}_Y\).
  We have %
  \( %
  \Absm{(s,r)\in Y\times X}{(\Mtlfun{\gamma}{p})(s,\eta'(r))}
  =\Absm{(s,r)\in Y\times X}{m_i(s,\eta_i(r))} %
  \) %
  which is measurable since each %
  \(\eta_i\) is a measurable path.
  Last let moreover \(\mu\in\Mcca{\Cmeas(X)}\), we must prove that %
  \(g_1=\Inv k(\int^{\Bwith_{i\in I}(\Limpl C{D_i})}\eta(r)\mu(dr))\) %
  and %
  \(g_2=\int^{\Limpl CD}\Inv k(\eta(r))\mu(dr)\) are the same function. %
  Let \(x\in\Mcca{C}\), we have %
  \(g_1(x)=(\int^{D_i} \eta_i(r)(x)\mu(dr))_{i\in I}=g_2(x)\).
  This ends the proof that \(k\) is an iso in \(\ICONES\) and hence
  that \(\Limpl C\_\){} preserves all products.
  
  \Proofcase{} Equalizers. Let \(f,g\in\ICONES(D_1,D_2)\) and let %
  \((E,e)\) be the corresponding equalizer in \(\ICONES\), as
  described in the proof of Theorem~\ref{th:mcones-complete}.
  Then we have %
  \(\Limplp Cf\Compl\Limplp Ce=\Limplp Cg\Compl\Limplp Ce\) by
  functoriality of \(\Limpl C\_\) and it will be sufficient to prove
  that %
  \((\Limpl CE,\Limpl Ce)\) has the universal property of an
  equalizer. Let \(H\) be an integrable cone and %
  \(h\in\ICONES(H,\Limpl C{D_1})\) be such that %
  \(\Limplp Cf\Compl h=\Limplp Cg\Compl h\).
  Identifying \(h\) with its ``uncurried'' version %
  \(h'\in\Mcca{\Limplm{H,C}{D_1}}\), the integrable bilinear and
  continuous map (see Section~\ref{sec:bilin-measurable}) given by
  \(h'(z,x)=h(z)(x)\), we have %
  \(f\Comp h'=g\Comp h'\).
  In other words \(h'\) ranges in %
  \(\Mcca E\subseteq\Mcca{D_1}\), allowing to define %
  \(h'_0\in\Mcca{\Limplm{H,C}{E}}\) which is the same function as
  \(h'\) and is bilinear continuous and integrable by definition of
  \(E\) (which inherits the norm, the measurability and integrability
  structure of \(C\)).
  We use \(h_0\) for the corresponding element of %
  \(\ICONES(H,\Limpl CE)\), so that %
  \(h=(\Limpl Ce)h_0\).
  The fact that \(h_0\) is unique with this property results from the
  fact that \(\Limpl Ce\) is a mono (it is actually the inclusion of %
  \(\Mcca{\Limpl CE}\) into \(\Mcca{\Limpl C{D_1}}\) resulting from
  the inclusion \(e\) of \(\Mcca E\) into \(\Mcca{D_1}\)). %
  This shows that \((\Limpl CE,\Limpl Ce)\) is the equalizer of %
  \(\Limpl Cf\) and \(\Limpl Cg\) and ends the proof that %
  \(\Limpl C\_\) preserves all limits.
\end{proof}

\begin{lemma} %
  \label{lemma:lfun-path-swap}
  Let \(X\in\ARCAT\) and let \(B,C,D\) be measurable cones. Let %
  \(f\) be an element of \(\ICONES(B,\Cpath X{\Limpl CD})\).
  Then %
  \(f'=\Absm{(y,r,x)\in\Mcca C\times X\times\Mcca B}
  {f(x,r,y)}\) belongs to \(\ICONES(C,\Cpath X{\Limpl BD})\).
\end{lemma}
\begin{proof}
  This results from the following sequence of isos in \(\ICONES\):
  \begin{align*}
    &\Limpl B{\Cpath X{\Limpl CD}}\\
    &\Textsep\Isom \Limpl B{\Limplp C{\Cpath XD}}
    \quad\text{by Lemma~\ref{lemma:swap-lin-path}
      and functoriality~of }\Limpl C\_\\
    &\Textsep=\Limplm{B,C}{\Cpath XD}\\
    &\Textsep\Isom\Limplm{C,B}{\Cpath XD}
      \text{\quad by Section~\ref{sec:bilin-measurable}.}\qedhere
  \end{align*}
\end{proof}

\subsection{The tensor product of integrable cones} %
\label{sec:tensor}
Let \(C\) be an integrable cone. We denote by %
\(\Tens\_ C\) the left adjoint of the functor \(\Limpl C{\_}\), %
see Theorem~\ref{th:limpl-has-left-adj}.
Because \(\Limpl\_\_\) is a functor %
\(\Op\ICONES\times\ICONES\to\ICONES\) (see
Section~\ref{sec:lin-hom-functor}), we know by the adjunction with a
parameter theorem (\cite{Maclane71}, Chapter~IV, Section~7,
Theorem~3), that the so defined operation%
\footnote{According to one of the reviewers of this paper, this tensor
  product can be understood as the adaptation to the setting of
  integrable cones of the standard projective tensor product of
  locally convex spaces.} %
\(\ITens\) can uniquely be extended in a bifunctor %
\(\ITens:\ICONES^2\to\ICONES\) in such a way that the bijection
\begin{align*}
  \Phi_{B,C,D}:\ICONES(\Tens BC,D)\to\ICONES(B,\Limpl CD)
\end{align*}
given by the adjunction for each \(C\) is natural in \(B,C,D\).
We define %
\begin{align*}
  \Tensor_{B,C}=\Phi_{B,C,\Tens BC}(\Id_{\Tens BC})
  \in\ICONES(B,\Limpl C{\Tens BC})
  =\Cuball{(\Mcca{\Limplm{B,C}{\Tens BC}})}
\end{align*}
and, for \(x\in\Mcca B\) and \(y\in\Mcca C\) we use the notation %
\(\Tens xy=\Tensor_{B,C}(x,y)\).
By naturality of \(\Phi\) we have that, for each
\(f\in\Mcca{\Limpl{\Tens BC}{D}}\),
\begin{align}
  \label{eq:tens-adj-bij-comp}
  \Phi_{B,C,D}(f)=f\Comp\Tensor_{B,C}.
\end{align}

The next lemma is the key observation for proving that the above
bijection is a cone isomorphism.
\begin{lemma} %
  \label{lemma:path-tens-to-one}
  Let \(X\in\ARCAT\) and \(B,C\) be integrable cones. %
  Let \(\eta:X\to\Mcca{\Limpl{\Tens BC}\Sone}\) be a function.
  One has \(\eta\in\Mcca{\Cpath{X}{\Limpl{\Tens BC}\Sone}}\) as soon as 
  \begin{itemize}
  \item \(\eta(X)\subseteq\Mcca{\Limpl{\Tens BC}\Sone}\) %
    is bounded
  \item and for all \(Y\in\ARCAT\), %
    \(\beta\in\Mcca{\Cpath YB}\) and \(\gamma\in\Mcca{\Cpath YC}\),
    the function %
    \(\Absm{(s,r)\in Y\times X}{\eta(r)(\Tens{\beta(s)}{\gamma(s)})}:
    Y\times X\to\Realp\) %
    is measurable.
  \end{itemize}
\end{lemma}
\begin{proof}
  Let %
  \(\eta':\Mcca B\times\Mcca C\times X\to\Realp\) be defined
  by %
  \(\eta'(x,y,r)=\eta(r)(\Tens xy)\).
  We have %
  \(\eta'\in\ICONES(B,\Limplm C{\Cpath X\Sone})\) by our
  assumptions.
  Let us check that \(\eta'\) indeed preserves integrals.
  Let \(Y,Z\in\ARCAT\), \(\mu \in \Cmeas(Y)\),
  \(\nu \in \Cmeas(Z)\), \(\beta \in \Cpath{Y}{B}\),
  \(\gamma \in \Cpath{Z}{C}\), and
  \(m = \Mcms{\Cpath X\Sone}_\Measterm\) (\Ie~there is
  \(r \in X\) such that \(m(\xi) = \xi(r)\) for all
  \(\xi \in \Cpath X\Sone\)).
  Then
  \begin{align*}
    m\Big(\eta'\Big(\int \beta(s)\mu(ds),& \int \gamma(t)\nu(dt)\Big)\Big)
    =\eta'\Big(\int \beta(s)\mu(ds), \int \gamma(t)\nu(dt)\Big)(r)\\
    &=\eta(r)\Big(\Tens{\Big(\int \beta(s)\mu(ds)\Big)}
      {\Big(\int \gamma(t)\nu(dt)\Big)}\Big)\\
          &=\iint \eta(r)(\Tens{\beta(s)}{\gamma(t)})\mu(ds)\nu(dt)
    \text{\quad since }\Tensor\text{ preserves integrals}\\
    &=\iint m(\eta'(\beta(s),\gamma(t)))\mu(ds)\nu(dt)\\
    &=m\Big(\iint \eta'(\beta(s),\gamma(t))\mu(ds)\nu(dt)\Big)\,.
  \end{align*}
  Let %
  \[\eta''
    =\Inv{\Phi_{B,C,\Cpath X\Sone}}(\eta')\in\ICONES(\Tens BC,\Cpath X\Sone)
    =\ICONES(\Tens BC,\Cpath X{\Limpl\Sone\Sone})
  \] %
  up to a trivial \(\ICONES\) iso %
  and so by Lemma~\ref{lemma:lfun-path-swap} there is a %
  \(h\in\ICONES(\Sone,\Cpath X{\Limpl{\Tens BC}\Sone})\) %
  such that \(h(1)(r)(z)=\eta''(z)(r)\) for all %
  \(z\in\Mcca{\Tens BC}\) and \(r\in X\).
  So we have \(h(1)(r)(\Tens xy)=\eta(r)(\Tens xy)\) and hence %
  \(\eta(r)=h(1)(r)\) since both are elements of %
  \(\Mcca{\Limpl{\Tens BC}\Sone}\).
  Since this holds for all \(r\in X\) we have proven that %
  \(\eta=h(1)\) and hence \(\eta\in\Cpath{X}{\Limpl{\Tens BC}\Sone}\)
  as contended.
\end{proof}

\noindent 
Now we can prove the main property of our tensor product which will
allow us to prove that it has a structure of monoidal product on
\(\ICONES\).
\begin{theorem} %
  \label{th:icones-tens-limpl-isom}
  For each integrable cones \(B,C,D\), the function %
  \(\Phi_{B,C,D}\) is an isomorphism of integrable cones from %
  \(\Limpl{\Tens BC}{D}\) to %
  \(\Limpl B{\Limplp CD}=(\Limplm{B,C}D)\).
\end{theorem}
\begin{proof}
  By linearity and \(\omega\)-continuity of composition on the left, the
  function \(\Phi_{B,C,D}\) ---~characterized
  by~\Eqref{eq:tens-adj-bij-comp}~--- is a linear and continuous map %
  \(\Phi_{B,C,D}:\Mcca{\Limpl{\Tens BC}{D}}\to\Mcca{\Limplm{B,C}D}\) %
  which satisfies \(\Norm{\Phi_{B,C,D}(f)}\leq\Norm f\) for all %
  \(f\in\Mcca{\Limplm{\Tens BC}{D}}\).
  This latter property is due to the fact that if
  \(f\in\Mcca{\Limpl{\Tens BC}{D}}\) satisfies \(\Norm f\leq 1\)
  then %
  \(f\in\ICONES(\Tens BC,D)\) and hence %
  \(\Phi_{B,C,D}(f)\in\ICONES(B,\Limpl CD)\) so that %
  \(\Norm{\Phi_{B,C,D}(f)}\leq 1\) and hence for an arbitrary %
  \(f\in\Mcca{\Limpl{\Tens BC}{D}}\) such that \(f\not=0\) we have %
  \(\Norm{(1/\Norm f)f}\leq 1\) and hence %
  \(\Norm{\Phi_{B,C,D}((1/\Norm f)f)}\leq 1\) which is exactly our
  contention, which trivially also holds when \(f=0\).

Let us prove that \(\Phi_{B,C,D}\)
is measurable, so let \(X\in\ARCAT\) and %
\(\eta\in\Mcca{\Cpath X{\Limpl{\Tens BC}D}}\), we must prove that %
\(\Phi_{B,C,D}\Comp\eta\in\Mcca{\Cpath{X}{(\Limplm{B,C}D)}}\).
So let \(Y\in\ARCAT\) %
and \(p\in\Mcms{\Limpl{B,C}D}_Y\), which means that %
\(p=\Mtlfunm{\beta,\gamma}{m}\) for some \(\beta\in\Mcca{\Cpath YB}\),
\(\gamma\in\Mcca{\Cpath YC}\) and %
\(m\in\Mcms D_Y\).
We have
\begin{align*}
  \Absm{(s,r)\in Y\times X}{p(s,\Phi_{B,C,D}(\eta(r)))}
  &=\Absm{(s,r)\in Y\times X}{p(s,\eta(r)\Comp\Tensor_{B,C})}\\
  &=\Absm{(s,r)\in Y\times X}{m(s,\eta(r)(\Tens{\beta(s)}{\gamma(s)}))}
\end{align*}
which is measurable because %
\(\Tens\beta\gamma\in\Mcca{\Cpath Y{\Tens BC}}\) (defining
\(\Tens\beta\gamma\) by %
\(\Tensp\beta\gamma(s)=\Tens{\beta(s)}{\gamma(s)}\)) %
by measurability of \(\Tensor\), and by our assumption that \(\eta\)
is a measurable path.
Altogether we have proven that %
\begin{align*}
  \Phi_{B,C,D}
  \in\MCONES(\Mcofic{(\Limpl{\Tens BC}D)},\Mcofic{(\Limplm{B,C}D)})
\end{align*}
and we prove now that this morphism preserves integrals, so let
moreover \(\mu\in\Mcca{\Cmeas(X)}\), we have %
\begin{align*}
  \Phi_{B,C,D}\Big(\int\eta(r)\mu(dr)\Big)
  &=\Absm{(x,y)\in\Mcca B\times\Mcca C}
    {\Big(\int\eta(r)\mu(dr)\Big)(\Tens xy)}\\
  &=\Absm{(x,y)\in\Mcca B\times\Mcca C}
    {\Big(\int\eta(r)(\Tens xy)\mu(dr)\Big)}\\
  &=\Absm{(x,y)\in\Mcca B\times\Mcca C}
    {\Big(\int\Phi_{B,C,D}(\eta(r))(x,y)\mu(dr)\Big)}\\
  &=\int\Phi_{B,C,D}(\eta(r))\mu(dr)\,.
\end{align*}
This shows that %
\(\Phi_{B,C,D}
\in\ICONES({(\Limpl{\Tens BC}D)},{(\Limplm{B,C}D)})\)
and we show now that this morphism is an iso.

We know that this function is bijective, let us use %
\(\Psi_{B,C,D}\) for its inverse, which is linear and continuous by
Lemma~\ref{lemma:linear-inverse}.
Since %
\(\Psi_{B,C,D}:\ICONES(B,\Limpl CD)\to\ICONES(\Tens BC,D)\), we have %
\(\Norm{\Psi_{B,C,D}(g)}\leq\Norm g\) for all %
\(g\in\Mcca{B,\Limpl CD}\), using also the linearity of
\(\Psi_{B,C,D}\).
We prove that \(\Psi_{B,C,D}\) is measurable.
Let \(X\in\ARCAT\) and %
\(\eta\in\Mcca{\Cpath{X}{(\Limplm{B,C}{D})}}\), we must prove that %
\(\Psi_{B,C,D}\Comp\eta\in\Mcca{\Cpath{X}{\Limplp{\Tens BC}{D}}}\). %
Without loss of generality we assume that \(\Norm\eta\leq 1\). %
Let \(Y\in\ARCAT\) and \(p\in\Mcms{\Limpl{\Tens BC}D}_Y\), %
we must check that %
\(\Absm{(s,r)\in Y\times X}{p(s,\Psi_{B,C,D}(\eta(r)))}\) %
is measurable.
There is \(\theta\in\Cpath{Y}{\Tens BC}\) %
and \(m\in\Mcms D_Y\) such that %
\(p=\Mtlfun{\theta}{m}\), and we must check that %
\(\Absm{(s,r)\in Y\times X}{m(s,\Psi_{B,C,D}(\eta(r))(\theta(s)))}\) %
is measurable. For this, since \(\ARCAT\) is cartesian, it suffices
to prove that %
\begin{align*}
  \Absm{(s,s',r)\in Y\times Y\times X}{m(s,\Psi_{B,C,D}(\eta(r))(\theta(s')))}
\end{align*}
is measurable.
Since \(\theta\in\Mcca{\Cpath{Y}{\Tens BC}}\) it suffices to prove that %
\begin{align*}
  \eta'=\Absm{(s,r,z)\in Y\times X\times\Mcca{\Tens BC}}
  {m(s,\Psi_{B,C,D}(\eta(r))(z))}
  \in\Mcca{\Cpath{Y\times X}{\Limpl{\Tens BC}{\Sone}}}
\end{align*}
and to this end we apply Lemma~\ref{lemma:path-tens-to-one}.
The boundedness assumption is satisfied because %
\(\Norm\eta\leq 1\) %
and hence %
\(\Norm{\Psi_{B,C,D}(\eta(r))}\leq 1\) for each %
\(r\in X\).
So let %
\(Y'\in\ARCAT\), \(\beta\in\Cpath{Y'}{B}\) %
and \(\gamma\in\Cpath{Y'}{C}\).
We have
\begin{align*}
  &\Absm{(s',s,r)\in Y'\times Y\times X}
  {\eta'(s,r)(\Tens{\beta(s')}{\gamma(s')})}\\
  &\Textsep=\Absm{(s',s,r)\in Y'\times Y\times X}
    {m(s,\Psi_{B,C,D}(\eta(r))(\beta(s')\ITens\gamma(s')))}\\
  &\Textsep=\Absm{(s',s,r)\in Y'\times Y\times X}
    {m(s,\eta(r)(\beta(s'),\gamma(s')))}
\end{align*}
which is measurable by our assumption about \(\eta\).
Last we must prove that \(\Psi_{B,C,D}\) preserves integrals.
Using the same path \(\eta\) let furthermore %
\(\mu\in\Mcca{\Cmeas(X)}\), we must prove that
\begin{align*}
  g_1&=\Psi_{B,C,D}\Big(\int \eta(r)\mu(dr)\Big)
       \in\Mcca{\Limpl{\Tens BC}D}\\
  \text{and }
  g_2&=\int\Psi_{B,C,D}(\eta(r))\mu(dr)\in\Mcca{\Limpl{\Tens BC}D}
\end{align*}
are equal.
Since \(\Phi_{B,C,D}\) preserves integrals we have %
\(\Phi_{B,C,D}(g_1)=\Phi_{B,C,D}(g_2)\) %
and the required property follows from the injectivity of %
\(\Phi_{B,C,D}\).
\end{proof}

\begin{theorem}
  For each \(x\in\Mcca B\) and \(y\in\Mcca C\) we have %
  \(\Norm{\Tens xy}=\Norm x\Norm y\).
\end{theorem}
\begin{proof}
  Since \(\Tensor_{B,C}\in\ICONES(B,\Limpl C{\Tens BC})\) we have %
  \(\Norm{\Tens xy}\leq\Norm x\Norm y\), we just have to prove the
  converse.
  If \(x=0\) or \(y=0\) our contention trivially holds so we can
  assume without loss of generality that \(\Norm x=\Norm y=1\) and let
  \(\epsilon>0\).
  By Proposition~\ref{th:norm-dual} there is
  \(x'\in\Cuball{\Cdual{\Mcca B}}\) and %
  \(y'\in\Cuball{\Cdual{\Mcca C}}\) such
  that %
  \(\Eval{x}{x'}\geq 1-\epsilon/2\) and
  \(\Eval{y}{y'}\geq 1-\epsilon/2\).
  Let \(g:\Mcca B\times\Mcca C\to\Realp\) be defined by %
  \(g(x_0,y_0)=\Eval{x_0}{x'}\Eval{y_0}{y'}\).
  Then
  \(g\in\Mcca{\Limplm{B,C}\Sone}\) and moreover \(\Norm g\leq 1\).
  Let %
  \(z'=\Inv{\Phi_{B,C,\Sone}}(g)\in\Cuball{\Cdual{(\Mcca{\Tens BC})}}\), %
  we have %
  \begin{align*}
    \Norm{\Tens xy}
    &\geq\Eval{\Tens xy}{z'}\\
    &=\Eval{x}{x'}\Eval{y}{y'}\\
    &\geq\big(1-\frac\epsilon 2\big)^2> 1-\epsilon
  \end{align*}
  so that \(\Norm{\Tens xy}\geq 1\).
\end{proof}

\subsection{The symmetric monoidal structure of \(\ICONES\)}
\label{sec:icones-is-smcc}
We want now to exploit Theorem~\ref{th:icones-tens-limpl-isom} to show
that the category \(\ICONES\) can be endowed with a symmetric monoidal
structure whose monoidal functor is our tensor product \(\ITens\).

One very convenient tool for proving the associated coherence diagrams
will be Proposition~\ref{prop:fun-ttree-charact} which uses binary
trees given by the following syntax: \(\Ttreeo\) is a tree (a leaf)
and if \(t_1\) and \(t_2\) are trees then \(\Ttreeb{t_1}{t_2}\) is a
tree.
We use \(\Ttrees n\) for the set of trees which have \(n\) leaves (for
\(n\in\Natnz\)).

These trees are used to specify arbitrary ``tensor expressions'' as
follows.
If \(n\in\Natnz\), \(\Vect B=(B_i)_{i=1}^n\) is a sequence of objects
of \(\ICONES\) and \(t\in\Ttrees n\), we define an object
\(\Ttreet t(\Vect B)\) of \(\ICONES\) by a straightforward induction,
for instance
\(\Ttreet{\Ttreeb{\Ttreeo}{\Ttreeb{\Ttreeo}{\Ttreeo}}}(B_1,B_2,B_3)
=\Tens{B_1}{\Tensp{B_1}{B_2}}\).
In the same way, given \(\Vect x=(x_i\in\Mcca{B_i})_{i=1}^n\) one
defines \(\Ttreet t(\Vect x)\in\Mcca{\Ttreet t(\Vect B)}\), for
instance
\(\Ttreet{\Ttreeb{\Ttreeo}{\Ttreeb{\Ttreeo}{\Ttreeo}}}(x_1,x_2,x_3)
=\Tens{x_1}{\Tensp{x_1}{x_2}}\).
\begin{proposition}
  \label{prop:fun-ttree-charact}
  Let \(n\in\Natnz\), \(\List B1n,C\) be integrable cones and
  \(t\in\Ttrees n\).
  Let \(f,g\in\ICONES(\Ttreet t(\Vect B),C)\).
  If, for all \((x_i\in\Mcca{B_i})_{i=1}^n\) one has %
  \(f(\Ttreet t(\Vect x))=g(\Ttreet t(\Vect x))\), then \(f=g\).
\end{proposition}
\begin{proof}
  By induction on the structure of \(t\).
  The base case \(t=\Ttreeo\) being trivial, assume that
  \(t=\Ttreeb{t_1}{t_2}\) with \((t_i\in\Ttrees{n_i})_{i=1,2}\) and
  \(n_1+n_2=n\) (notice that \(n_1,n_2<n\)) so that we can write %
  \(\Vect
  B=(B_1^1,\dots,B_{n_1}^1,B_1^2,\dots,B_{n_2}^2)=(\Vect{B^1},\Vect{B^2})\)
  and we have \(\Ttreet t{(\Vect{B})} =\Tens{D_1}{D_2}\) where %
  \((D_i=\Ttreet{t_i}{(\Vect{B^i})})_{i=1,2}\).
  With these notations, we have %
  \(f,g\in\ICONES(\Tens{D_1}{D_2},C)\) and so it suffices to prove that %
  \(\Phi_{D_1,D_2,C}(f)=\Phi_{D_1,D_2,C}(g)\in\ICONES(D_1,\Limpl{D_2}C)\).
  By inductive hypothesis, it suffices to prove that for all %
  \(\Vect{x^1}=(x^1_i\in\Mcca{B^1_i})_{i=1}^{n_1}\), one has %
  \(\Phi_{D_1,D_2,C}(f)(\Ttreet{t_1}(\Vect{x^1}))
  =\Phi_{D_1,D_2,C}(g)(\Ttreet{t_1}(\Vect{x^1}))\in\ICONES(B_2,C)\)
  (the fact that both morphisms have norm \(\leq 1\) is true but not
  essential), and so, by inductive hypothesis again, it suffices to
  prove that for all %
  \(\Vect{x^2}=(x^2_i\in\Mcca{B^2_i})_{i=1}^{n_2}\), one has %
  \(\Phi_{D_1,D_2,C}(f)(\Ttreet{t_1}(\Vect{x^1}))(\Ttreet{t_2}(\Vect{x^2}))
  =\Phi_{D_1,D_2,C}(g)(\Ttreet{t_1}(\Vect{x^1}))(\Ttreet{t_2}(\Vect{x^2}))
  \in\Mcca C\) which results from %
  \(\Phi_{D_1,D_2,C}(f)(\Ttreet{t_1}(\Vect{x^1}))(\Ttreet{t_2}(\Vect{x^2}))
  =f(\Tens{\Ttreet{t_1}(\Vect{x^1})}{\Ttreet{t_2}(\Vect{x^2})})
  =f(\Ttreet t(\Vect x))\) and similarly for \(g\), and from our
  assumption about \(f\) and \(g\).
\end{proof}

\begin{theorem} %
  \label{th:icones-smcc}
  The category \(\ICONES\), equipped with the bifunctor \(\ITens\) and
  unit \(\Sone\) has a structure of symmetric monoidal category, and
  this SMC is closed.
\end{theorem}
\begin{proof}
  Let us deal first with the associator.
  We have two natural bijections
  \begin{center}
    \begin{tikzcd}
      \ICONES(\Tens{\Tensp{B_1}{B_2}}{B_3},C)
      \ar[d,"\Phi_{\Tens{B_1}{B_2},B_3,C}"]\\[-0.4em]
      \ICONES(\Tens{B_1}{B_2},\Limpl{B_3}C)
      \ar[d,"\Phi_{B_1,B_2,\Limpl{B_3}C}"]\\[-0.4em]
      \ICONES(B_1,\Limpl{B_2}{\Limplp{B_3}{C}})
    \end{tikzcd}
  \end{center}
  and ---~notice that here we use Theorem~\ref{th:icones-tens-limpl-isom}
  in a crucial way~---
  \begin{center}
    \begin{tikzcd}
      \ICONES(\Tens{B_1}{\Tensp{B_2}{B_3}},C)
      \ar[d,"\Phi_{B_1,\Tens{B_2}{B_3},C}"]\\[-0.4em]
      \ICONES(B_1,\Limpl{\Tens{B_2}{B_3}}{C})
      \ar[d,"\ICONES(B_1{,}\Phi_{B_2,B_3,C})"]\\[-0.4em]
      \ICONES(B_1,\Limpl{B_2}{\Limplp{B_3}{C}})
    \end{tikzcd}
  \end{center}
  that we call respectively \(\Psi_{B_1,B_2,B_3,C}\) and
  \(\Psi'_{B_1,B_2,B_3,C}\) so that %
  \(\Invp{\Psi_{B_1,B_2,B_3,C}}\Comp\Psi'_{B_1,B_2,B_3,C}\) is a
  natural bijection %
  \(\ICONES(\Tens{B_1}{\Tensp{B_2}{B_3}},C)
  \to\ICONES(\Tens{\Tensp{B_1}{B_2}}{B_3},C)\) and hence, setting
  \(C=\Tens{B_1}{\Tensp{B_2}{B_3}}\), we know by
  Lemma~\ref{lemma:functor-yoneda-iso} that
  \begin{align*}
    \Assoc_{B_1,B_2,B_3}
    =\Invp{\Psi_{B_1,B_2,B_3,C}}(\Psi'_{B_1,B_2,B_3,C}(\Id_C))
    \in\ICONES(\Tens{\Tensp{B_1}{B_2}}{B_3},\Tens{B_1}{\Tensp{B_2}{B_3}})
  \end{align*}
  is a natural iso.
  Moreover the definition of the natural iso \(\Phi\) implies that for
  all \(x_1\in\Mcca{B_1}\), \(x_2\in\Mcca{B_2}\) and
  \(x_3\in\Mcca{B_3}\), one has
  \begin{align}
    \label{eq:associso-tens-args}
    \Assoc_{B_1,B_2,B_3}(\Tens{\Tensp{x_1}{x_2}}{x_3})
    =\Tens{x_1}{\Tensp{x_2}{x_3}}
  \end{align}
  Indeed \(\Psi'_{B_1,B_2,B_3,\Tens{B_1}{\Tensp{B_2}{B_3}}}(\Id_C)\)
  is
  \begin{align*}
    f=\Absm{x_1\in\Mcca{B_1}}{\Absm{x_2\in\Mcca{B_2}}
    {\Absm{x_3\in\Mcca{B_3}}{\Tens{x_1}{\Tensp{x_2}{x_3}}}}}
  \end{align*}
  and
  \(\Assoc_{B_1,B_2,B_3}\)
  must satisfy \(\Psi_{B_1,B_2,B_3,C}(\Assoc_{B_1,B_2,B_3})=f\) which
  is exactly Equation~\Eqref{eq:associso-tens-args}.
  
  Similarly one defines natural isos
  \(\Leftu_B\in\ICONES(\Tens\Sone B,B)\) (using the obvious natural
  bijection \(\ICONES(\Sone,\Limpl BC)\to\ICONES(B,C)\)), %
  \(\Rightu_B\in\ICONES(\Tens B\Sone,C)\) (using the obvious natural
  iso \((\Limpl\Sone C)\to C\) in \(\ICONES\)) and %
  \(\Sym_{B_1,B_2}\in\ICONES(\Tens{B_1}{B_2},\Tens{B_2}{B_1}))\)
  (using the natural iso of Lemma~\ref{lemma:swap-lin-lin}).
  These isos satisfy the following equations
  \begin{align}
    \forall x\in\Mcca B,\,\forall u\in\Realp
    \quad &\Leftu_B(\Tens ux)=ux=\Rightu_B(\Tens xu)
    \label{eq:unitiso-tens-args}\\
    \forall x_1\in\Mcca{B_1},\,\forall x_2\in\Mcca{B_2}
    \quad &\Sym_{B_1,B_2}(\Tens{x_1}{x_2})=\Tens{x_2}{x_1}\,.
            \label{eq:symiso-tens-args}
  \end{align}
  The required coherence diagrams are easily proven using
  Equations~\Eqref{eq:associso-tens-args},
  \Eqref{eq:unitiso-tens-args} and~\Eqref{eq:symiso-tens-args}
  combined with Proposition~\ref{prop:fun-ttree-charact}.
  In that way, we have endowed \(\ICONES\) with an SMC structure whose
  monoidal product is our tensor product \(\ITens\).
  The natural isomorphism \(\Phi\) tells us moreover that this SMC is
  closed.
\end{proof}

\section{Categorical properties of integration} %
\label{sec:cat-prop-integ}
From now on all the cones we consider are integrable cones, unless
otherwise specified.
We use letters \(B\), \(C\), \(D\) and \(E\), possibly with
subscripts, to denote such cones.

In Lemma~\ref{lemma:pushf-functor-icones} we have defined the
functor %
\(\Funpushf:\ARCAT\to\ICONES\) which maps each \(X\in\ARCAT\) to
the integrable cone \(\Funpushf(X)\) of finite non-negative measures
on \(X\) and acts on measurable functions by the standard
push-forward operation, \(\Funpushf(\phi)=\Pushf\phi\). %

Notice that for each \(X\in\ARCAT\) we have a specific element %
\(\Dirac X\in\Mcca{\Cpath X{\Cmeas(X)}}\) such that %
\(\Dirac X(r)\) is the Dirac mass at \(r\in X\), the measure
defined by
\begin{align*}
  \Dirac X(r)(U)=
  \begin{cases}
    1&\text{if }r\in U\\
    0&\text{otherwise.}
  \end{cases}
\end{align*}
The boundedness of \(\Dirac X\) is obvious and its measurability
results from the observation that if \(m=\Emeas U\) (for some
\(U\in\Sigalg X\)) we have \(m\Comp\Dirac X=\Charfun U\) (the
characteristic function of \(U\)) which is measurable.

\begin{theorem} %
  \label{th:meas-path-equiv}
  For each \(X\in\ARCAT\) and integrable cone \(B\), one has %
  \begin{align*}
    \Mcint B_X\in\ICONES(\Cpath XB,
    \Limpl{\Cmeas(X)}{B})
  \end{align*}
  and \(\Mcint B_X\) (this notation is introduced in
  Definition~\ref{def:integral-in-cone}) is an isomorphism which is
  natural in %
  \(X\) and in \(B\) (between functors %
  \(\Op\ARCAT\times\ICONES\to\ICONES\)).
\end{theorem}
This means that \(\Mcint B_X\) is bilinear continuous and
measurable, and preserves integrals on both sides, and that,
considered as a linear morphism acting on \(\Cpath XB\), it is an iso
in \(\ICONES\).
\begin{proof}
  For the first statement we just have to prove preservation of
  integrals in both arguments since bilinearity, continuity and
  measurability have already been proven in
  Lemma~\ref{lemma:int-mesurable}.
  So let %
  \(Y\in\ARCAT\), \(\nu\in\Mcca{\Cmeas(Y)}\), %
  \(\eta\in\Mcca{\Cpath Y{\Cpath XB}}\) and %
  \(\mu\in\Mcca{\Cmeas(X)}\), we have
  \begin{align*}
    \Mcint B_{X}\Big({\int^{\Cpath XB}_{s\in Y}\eta(s)\nu(ds)},{\mu}\Big)
    &=\int^B_{r\in X}\Big(\int^{\Cpath XB}_{s\in Y}\eta(s)\nu(ds)\Big)(r)\mu(dr)\\
    &=\int^B_{r\in X}\Big(\int^B_{s\in Y}\eta(s)(r)\nu(ds)\Big)\mu(dr)\\
    &=\int^B_{s\in Y}\Big(\int^B_{r\in X}\eta(s)(r)\mu(dr)\Big)\nu(ds)
      \text{\quad by Theorem~\ref{th:paths-Fubini}}\\
    &=\int\Mcint B_{X}(\eta(s),\mu)\nu(ds)\,.
  \end{align*}
  Next let \(\beta\in\Mcca{\Cpath XB}\) and %
  \(\kappa\in\Mcca{\Cpath Y{\Cmeas(X)}}\), we have
  \begin{align*}
    \Mcint B_X\Big(\beta,\int^{\Cmeas(X)}_{s\in Y}\kappa(s)\nu(ds)\Big)
    &=\int^B_{r\in X} \beta(r)\Big(\int^{\Cmeas(X)}_{s\in Y}\kappa(s)\nu(ds)\Big)(dr)
  \end{align*}
  where one should remember that the value of the integral %
  \(\int\kappa(s)\nu(ds)\) is the finite measure on \(X\) %
  which maps \(U\in\Sigalg X\) to \(\int\kappa(s,U)\nu(ds)\in\Realp\). %
  We claim that \(x_1=x_2\) where
  \begin{align*}
    x_1=\int^B_{r\in X} \beta(r)\Big(\int^{\Cmeas(X)}_{s\in Y}\kappa(s)\nu(ds)\Big)(dr)
    \Textsep
    x_2=\int^B_{s\in Y}\Big(\int^B_{r\in X}\beta(r)\kappa(s,dr)\Big)\nu(ds)\,.
  \end{align*}
  Upon applying to both members an element of \(\Mcms B_\Measterm\)
  and using \Mssepr{} for \(B\) we can assume that \(B=\Sone\).
  By the monotone convergence theorem and the fact that each measurable
  function is the lub of a monotone sequence of simple measurable
  functions, we can assume that \(\beta\) is simple, and by linearity
  of integrals we can assume that \(\beta=\Charfun U\) for some
  \(U\in\Sigalg X\).
  Then we have %
  \(x_1=\int\kappa(s,U)\nu(ds)=x_2\).

  Now we define a function %
  \(\Mcinti B_X:\Mcca{\Limpl{\Cmeas(X)}{B}}\to\Mcca{\Cpath XB}\) by
  setting %
  \begin{align*}
    \Mcinti B_X(f)=f\Comp\Dirac X\,,
  \end{align*}
  which belongs indeed to \(\Mcca{\Cpath XB}\) because \(\Dirac X\) is
  a bounded measurable path.
  Linearity and continuity of \(\Mcinti B_X\) result from linearity
  and continuity of composition.
  We prove measurability so let \(Y\in\ARCAT\) and %
  \(\eta\in\Mcca{\Cpath Y{\Limpl{\Cmeas(X)}{B}}}\), we contend that %
  \(\Mcinti B_X\Comp\eta\in\Mcca{\Cpath Y{\Cpath XB}}\).
  So let \(Y'\in\ARCAT\) and let \(p\in\Mcms{\Cpath XB}_{Y'}\), we
  must prove that
  \begin{align*}
    \psi=\Absm{(s',s)\in Y'\times Y}{p(s',\Mcinti B_X(\eta(s)))}
  \end{align*}
  is measurable.
  Let \(\phi\in\ARCAT(Y',X)\) and \(m\in\Mcms B_{Y'}\) be such that %
  \(p=\Mtpath\phi m\), we have
  \begin{align*}
    \psi
    &=\Absm{(s',s)\in Y'\times Y}
    {m(s',\Mcinti B_X(\eta(s))(\phi(s')))}\\
    &=\Absm{(s',s)\in Y'\times Y}
    {m(s',\eta(s)(\Dirac X(\phi(s'))))}\\
    &=\Absm{(s',s)\in Y'\times Y}
    {((\Mtpath{(\Dirac X\Comp\phi)}m)(s',\eta(s))}
  \end{align*}
  which is measurable since %
  \(\Mtpath{(\Dirac X\Comp\phi)}m\in\Mcms{\Limpl{\Cmeas(X)}B}_{Y'}\) and %
  \(\eta\in\Mcca{\Cpath Y{\Limpl{\Cmeas(X)}{B}}}\).

  We prove that \(\Mcinti B_X\) preserves integrals so let moreover %
  \(\nu\in\Mcca{\Cmeas(Y)}\), we have
  \begin{align*}
    \Mcinti B_X\Big(\int^{\Limpl{\Cmeas(X)}{B}}_{s\in Y}\eta(s)\nu(ds)\Big)
    &=\Absm{r\in X}{\int^B_{s\in Y}\eta(s)(\Dirac X(r))\nu(ds)}\\
    &=\Absm{r\in X}{\int^B_{s\in Y}\Mcinti B_X(\eta(s))(r)\nu(ds)}\\
    &=\int^{\Cpath XB}_{s\in Y}\Mcinti B_X(\eta(s))\nu(ds)
  \end{align*}
  so that
  \(\Mcinti B_X\in\ICONES(\Limpl{\Cmeas(X)}{B},\Cpath XB)\).

  Let \(f\in\Mcca{\Limpl{\Cmeas(X)}{B}}\), we have
  \begin{align*}
    \Mcint B_{X}(\Mcinti B_X(f))
    &=\Absm{\mu\in\Mcca{\Cmeas(X)}}{\int^B_{r\in X} f(\Dirac X(r))\mu(dr)}\\
    &=\Absm{\mu\in\Mcca{\Cmeas(X)}}
      {f\Big(\int^{\Cmeas(X)}_{r\in X} \Dirac X(r)\mu(dr)\Big)}
      \text{\quad since }f\text{ preserves integrals}\\
    &=\Absm{\mu\in\Mcca{\Cmeas(X)}}{f(\mu)}=f
  \end{align*}
  and let \(\beta\in\Mcca{\Cpath XB}\), we have
  \[
    \Mcinti B_X(\Mcint B_{r\in X}(\beta))
    =\Absm{r\in X}{\Big(\int^{B}_{r'\in X} \beta(r')\Dirac X(r,dr')\Big)}
    =\beta\,.
  \]
  Checking naturality is routine.
\end{proof}

\begin{theorem} %
  \label{th:dirac-dense} %
  Let \(X\in\ARCAT\), \(B\) be an object of \(\ICONES\) and
  \(f_1,f_2\in\ICONES(\Cmeas(X),B)\).
  If, for all \(r\in X\), one has
  \(f_1(\Dirac X(r))=f_2(\Dirac X(r))\) then \(f_1=f_2\).
\end{theorem}
\begin{proof}
  This results from Theorem~\ref{th:meas-path-equiv}.
\end{proof}

\begin{remark}
  In other words the Dirac measures \(\Dirac X(r)\) are ``dense'' in
  the integrable cone \(\Cmeas(X)\), in the sense that two \(\ICONES\)
  morphisms which take the same values on Dirac measures are equal.
  This property is one of the main benefits of integrability of cones
  and it does not hold in \(\MCONES\) as shown
  in Remark~\ref{rk:continuous-part-measure}.
\end{remark}

It is easy to check that for each \(X\in\ARCAT\) the functor %
\(\Cpath X\_:\ICONES\to\ICONES\) preserves all limits.
It follows by Theorem~\ref{th:Icones-adjoint-functor} that it has a
left adjoint.
We provide an explicit description of this adjoint.
We define the functor %
\(\Cmeast:\ARCAT\times\ICONES\to\ICONES\) by %
\(\Cmeast(X,B)=\Tens{\Cmeas(X)}{B}\) and similarly for morphisms.

\begin{theorem}
  \label{th:path-equiv-lin}
  For each \(X\in\ARCAT\) we have \(\Cmeast(X,\_)\Adj\Cpath X\_\)
\end{theorem}
\begin{proof}
  We have the following sequence of natural bijections:
  \begin{align*}
    \ICONES(
    &\Tens{\Cmeas(X)}{B},C)\\
    &\Isom\ICONES(\Tens B{\Cmeas(X)},C)
      \text{\quad by symmetry of }\ITens\text{, Th.~\ref{th:icones-smcc}}\\
    &\Isom\ICONES(B,\Limpl{\Cmeas(X)}C)
      \text{\quad since }\ICONES
      \text{ is an SMCC, Th.~\ref{th:icones-smcc}}\\
    &\Isom\ICONES(B,\Cpath XC)
      \text{\quad by Theorem~\ref{th:meas-path-equiv}.}
      \qedhere
  \end{align*}
\end{proof}

\subsection{The category of substochastic kernels as a full
  subcategory of \(\ICONES\)}

If \(X,Y\in\ARCAT\), a substochastic kernel from \(X\) to \(Y\) is
an element of \(\SKERN(X,Y)=\Cuball{\Mcca{\Cpath X{\Cmeas(Y)}}}\): this
is an equivalent characterization of this standard measure theory and
probability notion.
Then \(\SKERN\) is the category whose objects are those or \(\ARCAT\)
and:
\begin{itemize}
\item the identity at \(X\) is \(\Dirac X\in\SKERN(X,X)\)
\item and given \(\kappa_1\in\SKERN(X_1,X_2)\) and
  \(\kappa_2\in\SKERN(X_2,X_3)\), their composite
  \(\kappa=\kappa_2\Compl\kappa_1\) is given by
  \begin{align*}
    \kappa(r_1)(U_3)=\int_{r_2\in X_2}^\Sone
    \kappa_2(r_2,U_3)\kappa_1(r_1,dr_2)
  \end{align*}
  for \(U_3\in\Sigalg{X_3}\), that is %
  \(\kappa(r_1)=\Mcint{\Cmeas(X_3)}_{X_2}(\kappa_2)(\kappa_1(r_1))\):
  this formula is a continuous generalization of the product of
  substochastic matrices.
\end{itemize}
As is well known the category of measurable spaces and substochastic
kernels can be presented as the Kleisli category of the Giry monad
(or more precisely, of the Panangaden monad since we consider
substochastic kernels instead of stochastic kernels),
but this point of view does not apply to our case because the set of
objects of our small category \(\ARCAT\) has no reason to be
stable under the action of the Panangaden monad.%

If \(\kappa\in\SKERN(X,Y)\), we set %
\(\Sklin(\kappa)
=\Mcint{\Cmeas(Y)}_X(\kappa)\in\ICONES(\Cmeas(X),\Cmeas(Y))\)
defining a functor \(\Sklin:\SKERN\to\ICONES\) which maps
\(X\in\ARCAT\) to \(\Cmeas(X)\).
Remember that we use \(\Funpushf\) for the functor
\(\ARCAT\to\ICONES\) defined on morphisms by
\(\Funpushf(\phi)=\Pushf\phi=\Sklin(\Dirac
Y\Comp\phi)\in\ICONES(\Cmeas(X),\Cmeas(Y))\) for
\(\phi\in\ARCAT(X,Y)\).

\begin{theorem}
  The functor \(\Sklin:\SKERN\to\ICONES\) is full and faithful.
\end{theorem}
\begin{proof}
  By Theorem~\ref{th:meas-path-equiv}.
\end{proof}

\begin{remark}
  So we can consider the category of measurable spaces (at least those
  sorted out by \(\ARCAT\)) and \emph{substochastic kernels} as a full
  subcategory of \(\ICONES\) and again, this is a major consequence of
  the assumption that linear morphisms must preserve integrals.
  This has to be compared with QBSs which form a cartesian closed
  category which contain the category of measurable spaces and
  \emph{measurable functions} (or a full subcategory thereof such as
  our \(\ARCAT\)) as a full subcategory through the Yoneda embedding.
  %
  %

  In Section~\ref{sec:arcat-full-subcat-EM} we will see that, under
  very reasonable assumptions about its objects, \(\ARCAT\) arises as
  a full subcategory of \(\Em\ICONES\), the category of
  \(\oc\)-coalgebras for an exponential comonad \(\oc\) based on
  stable and measurable functions, or on analytical morphisms.
\end{remark}

\begin{theorem}
  \label{th:fmeas-prod-tensor}
  There is an iso in \(\ICONES(\Cmeas(\Measterm),\Sone)\) and, given
  \(X,Y\in\ARCAT\), there is an iso in
  \[
    \ICONES(\Cmeas(X\times Y),\Tens{\Cmeas(X)}{\Cmeas(Y)})
  \]
  which is natural in \(X\) and \(Y\) on the category \(\ARCAT\).
  These isos turn \(\Cmeas\) into a strong monoidal functor
  \((\ARCAT,\mathord\times)\to(\ICONES,\mathord\otimes)\).
\end{theorem}
\begin{proof}
  Remember first that \(\Cmeas\) is a functor \(\ARCAT\to\ICONES\)
  which acts on morphisms by push-forward (see
  Lemma~\ref{lemma:pushf-functor-icones}).
  
  The first statement is obvious since a finite measure on a one
  element space is the same thing as an element of \(\Realp\) whose
  norm is its value, and in that case our definition of measurability
  and integrals coincide with the usual ones.
  
  Given an object \(B\) of \(\ICONES\) we have the following sequence
  of natural bijections between functors %
  \(\Op\ARCAT\times\Op\ARCAT\times\ICONES\to\SET\)
  \begin{align*}
    \ICONES(\Cmeas(X\times Y),B)
    &\Isom\Cuball\Mcca{\Limpl{\Cmeas(X\times Y)}{B}}\\
    &\Isom\Cuball\Mcca{\Cpath{X\times Y}{B}}
    \text{\quad by Theorem~\ref{th:meas-path-equiv}}\\
    &\Isom\Cuball\Mcca{\Cpath X{(\Cpath Y B)}}
    \text{\quad by Lemma~\ref{lemma:meas-path-flat}}\\
    &\Isom\Cuball\Mcca{\Limpl{\Cmeas(X)}{\Limplp{\Cmeas(Y)}{B}}}
    \text{\quad by Theorem~\ref{th:meas-path-equiv}}\\
    &\Isom\ICONES(\Cmeas(X),{\Limplp{\Cmeas(Y)}{B}})\\
    &\Isom\ICONES(\Tens{\Cmeas(X)}{\Cmeas(Y)},B)
  \end{align*}
  because \(\ICONES\) is an SMCC and so by
  Lemma~\ref{lemma:functor-yoneda-iso} we have a natural transformation
  \begin{align*}
    \psi_{X,Y}\in\ICONES(\Tens{\Cmeas(X)}{\Cmeas(Y)},\Cmeas(X\times Y))
  \end{align*}
  between functors \(\ARCAT\times\ARCAT\to\ICONES\) and this natural
  transformation is completely characterized by
  \(\psi_{X,Y}(\Tens\mu\nu)=\mu\times\nu\) by the definition of the
  iso \(\Flpath\) used in Lemma~\ref{lemma:meas-path-flat}.
  Using this characterization as well as
  Proposition~\ref{prop:fun-ttree-charact}, it is easy to prove that
  this natural isomorphism (together with its \(0\)-ary version) turns
  \(\Cmeas\) into a monoidal functor from the cartesian category
  \(\ARCAT\) to the monoidal category \(\ICONES\).
  The proof uses also the fact that
  \(\Pushf{(\phi_1\times\phi_2)}{(\mu_1\times\mu_2)}
  =\Pushf{(\phi_1)}{(\mu_1)}\times\Pushf{(\psi_2)}({\mu_2})\) for
  \(\mu_i\in\Mcca{\Cmeas{X_i}}\) and \(\phi_i\in\ARCAT(X_i,Y_i)\) for
  \(i=1,2\).
\end{proof}

\begin{remark}
  This is yet another highly desirable property of the tensor product
  which results from the preservation of integrals by linear
  morphisms.
  It means that an element \(\pi\) of \(\Tens{\Cmeas(X)}{\Cmeas(Y)}\)
  whose norm is \(1\) can be considered as the joint probability
  distribution of two (not necessarily independent) random variables
  valued in \(X\) and \(Y\) respectively.
  The case where \(\pi=\Tens \mu\nu\) corresponds to the situation
  where the random variables are independent, of associated
  distributions \(\mu\) and \(\nu\).
\end{remark}

\section{Stable and measurable functions} %
\label{sec:stbale-func}

We start studying the non-linear maps between integrable cones and we
will consider actually two kinds of non-linear morphisms:
\begin{itemize}
\item the stable and measurable morphisms in the present section
\item and the analytic ones in Section~\ref{sec:analytic-functions-exp}.
\end{itemize}
The first ones were introduced in~\cite{EhrhardPaganiTasson18} in a
weaker setting (no general notion of integration was considered in
that paper).
The second ones are naturally derived from the monoidal structure of
\(\ICONES\) and from the \(\omega\)-completeness of cones.
The two notions are deeply connected: analytic morphisms are in
particular stable, and some properties proven in
Section~\ref{sec:stbale-func} will be useful in
Section~\ref{sec:analytic-functions-exp}.
We will also see in Remark~\ref{rem:discrete-meas-comp} that there are
stable functions which are not analytic, the fundamental reason for
that being that the definition of stability does not refer to
integrals.
In the \emph{discrete probability} setting of probabilistic coherence
spaces (see Section~\ref{sec:pcs-integrable}), it has been proved that
the two notions are equivalent, see~\cite{Crubille18}.


Stable morphisms satisfy a strong form of monotonicity, which, in
ordinary real analysis, can be expressed in terms of derivatives:
\(\forall n\in\Nat\ f^{(n)}(x)\geq0\) (absolute monotonicity).
Because we are working in the setting of cones, we use iterated
differences instead of derivatives, following an idea first introduced
by Bernstein in~1914.

The definition of these iterated differences uses a notion of
\emph{local cone} introduced in~\cite{EhrhardPaganiTasson18} and that
we recall now.
This construction could be rephrased in terms of summability
structure~\cite{Ehrhard23}, but this is not necessary for our purpose
here.

\subsection{The local cone} %
\label{sec:local-cone}
One major feature of stable functions%
\footnote{This will be also the case of analytic morphisms in
  Section~\ref{sec:analytic-functions-exp}.} %
from an integrable cone \(B\) to an integrable cone \(C\) is that,
contrarily to linear morphisms, they will be defined, in general, only
on the closed ``unit ball'' \(\Cuball{\Mcca B}\) of \(B\), see
Remark~\ref{rk:analytic-dedf-on-balls}.
A typical example is the already mentioned function
\(f:\Cuball{\Mcca{\Sone}}=\Intercc 01\to\Mcca\Sone\) given by
\(f(x)=1-\sqrt{1-x}\), see Remark~\ref{rk:analytic-dedf-on-balls}.
To deal with such a function \(\Cuball{\Mcca B}\to\Mcca C\) and its
strong monotonicity properties at a given \(x\in\Cuball{\Mcca B}\), it
will be important to be able to consider the set \(U\) of all
\(u\in\Mcca B\) such that \(x+u\in\Cuball{\Mcca B}\), or rather of all
elements of \(\Cuball{\Mcca B}\) of shape \(\lambda u\) for such an
\(u\) and for \(\lambda\in\Realp\): we will see that \(U\) can itself
be considered as a cone, the \emph{local cone} of \(B\) at \(x\).

Let \(B\) be an integrable cone and \(x\in\Cuball{\Mcca B}\).
Let
\begin{align*}
  P=\{u\in\Mcca B\St\exists\epsilon>0\ x+\epsilon u\in\Cuball{\Mcca B}\}\,.
\end{align*}
Then \(P\) is a precone whose \(0\) element is the \(0\) of
\(\Mcca B\).
Indeed if \(u_1,u_2\in P\) there is \(\epsilon>0\) such that
\((x+\epsilon u_i\in\Cuball{\Mcca B})_{i=1,2}\) and hence
\begin{align*}
  x+\frac\epsilon 2(u_1+u_2)
  =\frac 12(x+\epsilon u_1)+\frac 12(x+\epsilon u_2)\in\Cuball{\Mcca B}
\end{align*}
so that \(u_1+u_2\in P\).
The fact that
\(u\in P\Implies\forall\lambda\in\Realp\ \lambda u\in P\) is
clear. The condition \Pcsimplr{} and \Pcposr{} result easily from the
fact that they hold in \(\Mcca B\).

We can equip \(P\) with a map (sometimes called a gauge) %
\(N:P\to\Realp\) defined by
\begin{align*}
  N(u)=\Inv{(\sup\{\lambda>0\St x+\lambda u\in\Cuball{\Mcca B}\})}
  =\inf\{\Inv\lambda\St \lambda>0\text{ and }x+\lambda u\in\Cuball{\Mcca B}\}
\end{align*}
\begin{lemma}
  \label{lemma:local-cone-is-a-subcone}
  The function \(N\) is a norm on \(P\) and, equipped with this norm,
  \(P\) is a cone whose \(0\) element and algebraic operations are
  those of \(\Mcca B\).
\end{lemma}
\begin{proof}
  Assume that \(N(u)=0\), this means that %
  \(\forall\lambda\in\Realp\ \Norm{x+\lambda u}\leq 1\) and hence %
  \(\forall\lambda\in\Realp\ \lambda\Norm u\leq1\) so that \(u=0\). Let %
  \(u_1,u_2\in P\) and let \(\epsilon>0\).
  We can find \(\lambda_1,\lambda_2>0\) such that
  \(x+\lambda_i u_i\in\Cuball{\Mcca B}\) and %
  \(\Inv{\lambda_i}\leq N(u_i)+\epsilon/2\) for \(i=1,2\).
  We have %
  \begin{align*}
    \frac{\lambda_1}{\lambda_1+\lambda_2}(x+\lambda_2u_2)
    +\frac{\lambda_2}{\lambda_1+\lambda_2}(x+\lambda_1u_1)
    \in\Cuball{\Mcca B}
  \end{align*}
  so that
  \begin{align*}
    x+\frac{\lambda_1\lambda_2}{\lambda_1+\lambda_2}(u_1+u_2)
    \in\Cuball{\Mcca B}
  \end{align*}
  so that %
  \(N(u_1+u_2)
  \leq\Inv{(\frac{\lambda_1\lambda_2}{\lambda_1+\lambda_2})}
  =\Inv{\lambda_1}+\Inv{\lambda_2}\leq N(u_1)+N(u_2)+\epsilon\).
  Since this holds for all \(\epsilon>0\) we have
  \(N(u_1+u_2)\leq N(u_1)+N(u_2)\).
  The property \Cnormpr{} is easy, let us prove \Cnormcr.
  Observe first that, for \(u,v\in P\), we have %
  \(u\leq_P v\) iff \(u\leq_B v\).
  Let \((u_n)_{n\in\Nat}\) be an increasing sequence in \(P\) such that %
  \(\forall n\in\Nat\ N(u_n)\leq 1\).
  Then we have %
  \(\forall n\in\Nat\ x+u_n\in\Cuball{\Mcca B}\) and hence the
  sequence %
  \((x+u_n)_{n\in\Nat}\) is an increasing sequence in
  \(\Cuball{\Mcca B}\) and so it has a lub in \(\Cuball{\Mcca B}\) %
  which coincides with \(x+u\) where \(u\) is the lub of
  \((u_n)_{n\in\Nat}\) in \(\Mcca B\).
  Thus we have \(u\in P\) and \(N(u)\leq 1\).
  Last observe that \(u\) is \emph{a fortiori} the lub of the
  \(u_n\)'s in \(P\).
\end{proof}
We use now the standard notation \(\Norm\_\) for that norm. Notice that
\begin{align*}
  \Cuball P=\{u\in\Mcca B\St x+u\in\Cuball{\Mcca B}\}
\end{align*}
and also that \(\Norm u_B\leq\Norm u\), for all \(u\in P\).

For each \(X\in\ARCAT\) we define \(\cM_X\) as the set of all
test functions %
\(\Absm{r\in X}{\Absm{u\in P}{m(r,u)}}\) for the elements
\(m\) of \(\Mcms B_X\) (such an element of \(\cM_X\) will still be
denoted by \(m\) even if two different elements of \(\Mcms B_X\)
possibly induce the same test).
The fact that %
\(\forall r\in X,\,u\in\Cuball{P}\ m(r,u)\leq 1\) whenever
\(m\in\cM_X\) results from \(\Cuball P\subseteq\Cuball{\Mcca B}\).

Then it is straightforward to check that
\((P,(\cM_X)_{X\in\ARCAT})\) is a measurable cone, that we denote as
\(\Cloc Bx\) and call
\emph{the local cone of \(B\) at \(x\)} and it is also clear that this
measurable cone is integrable, the integral of a path in \(\Cloc Bx\)
being defined exactly as in \(B\).

It is important to keep in mind the meaning of the norm in this local
cone, which is most usefully described as follows.
\begin{lemma} %
  \label{lemma:cloc-norm-charact}
  Let \(x\in\Cuball{\Mcca B}\) and
  \(u\in\Mcca{\Cloc Bx}\setminus\Eset 0\). Then we have
  \(x+\Inv{\Norm u}_{\Cloc Bx}u\in\Cuball{\Mcca B}\) and %
  \(x+\lambda u\notin\Cuball{\Mcca B}\) for all %
  \(\lambda>\Inv{\Norm u}_{\Mcca{\Cloc Bx}}\).
\end{lemma}
\begin{proof}
  By definition of the norm in \(\Cloc Bx\) and by the \(\omega\)-closedness
  of \(\Cuball{\Mcca B}\).
\end{proof}

\begin{example}
  Let \(B=\Sone\) so that \(\Mcca\Sone=\Realp\) and
  \(\Cuball{\Mcca\Sone}=\Intercc01\). If \(x\in\Intercc 01\) we have
  two cases: if \(x=1\) then \(\Cloc Bx=\Zero\) and if \(x<1\) we
  have %
  \(\Mcca{\Cloc Bx}=\Realp\). If \(u\in\Mcca B_x\setminus\Eset 0\)
  (which implies \(x<1\)) then the largest \(\lambda>0\) such that
  \(x+\lambda u\in\Intercc01\) is %
  \(\frac{1-x}u\) and hence %
  \(\Norm u_{\Cloc Bx}=\frac u{1-x}=\frac1{1-x}\Norm u_B\) so the
  local cone \(\Cloc Bx\) can be seen as an homothetic image of \(B\)
  by a factor \(\frac 1{1-x}\) which goes to \(\infty\) when \(x\)
  gets closer to the border \(1\) of the unit ball
  \(\Intercc01\). This very simple example gives an intuition of what
  happens in general, with the difference that \(B_x\) has no reason
  to be always homothetic to \(B\), think for instance of the case
  where \(B=\With\Sone\Sone\) and \(x=(1,0)\): then \(\Cloc Bx\) is
  isomorphic to \(\Sone\).
\end{example}

\subsection{The integrable cone of stable and measurable functions }%
\label{sec:icone-stable-fun}
Given \(n\in\Nat\), we use \(\Npset n\) (resp.~\(\Ppset n\)) for the
set of all subsets \(I\) of \(\Eset{1,\dots,n}\) such that
\(n-\Card I\) is odd (resp.~even).
In particular \(\{1,\dots,n\}\in\Ppset n\).

\begin{lemma} %
  \label{lemma:inset-corresp}
  Let \(n\in\Nat\), \(j\in\Eset{1,\dots,n+1}\) %
  and \(\epsilon\in\Eset{-,+}\). Given %
  \(I\in\Epset\epsilon n\), the set
  \begin{align*}
    \Inset j(I)=\Eset{i\in I\St i<j}\cup\{j\}
    \cup\Eset{i+1\St i\in I\text{ and }i\geq j}
  \end{align*}
  belongs to \(\Epset\epsilon{n+1}\) and %
  \(\Inset j\) defines a bijection between \(\Epset\epsilon n\) and
  the set of all \(J\subseteq\Eset{1,\dots,n+1}\) such that
  \(J\in\Epset\epsilon{n+1}\) and \(j\in J\).
\end{lemma}
This is obvious.

\begin{definition}
  \label{def:totally-monotonic}
  Let \(P\) and \(Q\) be cones, a function \(f:\Cuball P\to Q\) is %
  \emph{totally monotonic} if for each \(n\in\Nat\) and each %
  \(x,\List u1n\in P\) such that %
  \(x+\sum_{i=1}^nu_i\in\Cuball P\) one has
  \begin{align}
    \label{eq:totally-monotone-def}
    \sum_{I\in\Npset n}f(x+\sum_{i\in I}u_i)
    \leq\sum_{I\in\Ppset n}f(x+\sum_{i\in I}u_i)\,.
  \end{align}
\end{definition}

For \(n=1\) this condition means that \(f\) is increasing.
For \(n=2\) we have \(\Npset2=\{\{1\},\{2\}\}\) and
\(\Ppset2=\{\{1,2\},\emptyset\}\), so
Condition~\Eqref{eq:totally-monotone-def} means
\begin{align*}
  f(x+u_1)+f(x+u_2)\leq f(x+u_1+u_2)+f(x)
\end{align*}
that is, assuming that \(f\) is increasing,
\begin{align*}
  f(x+u_1)-f(x)\leq f(x+u_1+u_2)-f(x+u_2)
\end{align*}
in other words, the map \(x\mapsto f(x+u_1)-f(x)\) is increasing (where
it is defined).
For \(n=3\) we have \(\Npset3=\{\{1,2\},\{2,3\},\{1,3\},\emptyset\}\)
and \(\Ppset3=\{\{1,2,3\},\{1\},\{2\},\{3\}\}\), so
Condition~\Eqref{eq:totally-monotone-def} means:
\begin{multline*}
  f(x+u_1+u_2)+f(x+u_2+u_3)+f(x+u_1+u_3)+f(x)\\
  \leq f(x+u_1+u_2+u_3)+f(x+u_1)+f(x+u_2)+f(x+u_3)\,.
\end{multline*}

\begin{remark}
  \label{rk:total-monotonicity}
  This kind of definition appears in many places in the literature in
  real analysis, differential equations, Laplace transforms \Etc{}
  The corresponding conditions, first considered by Hausdorff, are
  then usually expressed in terms of derivatives: for instance a
  function \(f\) from some open interval \(I\) of \(\Real\) to
  \(\Real\) is \emph{absolutely monotonic} (resp.~\emph{completely
    monotonic}) if it is \(C^\infty\) and satisfies
  \(f^{(n)}(x)\geq 0\) (resp.~\((-1)^nf^{(n)}(x)\geq 0\)) for all
  \(x\in I\) and \(n\in\Nat\).
  Bernstein introduces in~\cite{Bernstein1914}, in the case of real
  functions of one real parameter, iterated differences allowing to
  characterize absolutely monotonic functions ---~which in turn can be
  shown to be analytic~--- by a condition which is equivalent
  to~\Eqref{eq:totally-monotone-def}.
  We use the expression ``totally monotonic'' for this extension to
  cones of Bernstein's definition to avoid confusion with ``completely
  monotonic'' and ``absolutely monotonic''.

  This definition arose in denotational semantics during the work of
  the first author reported in~\cite{EhrhardPaganiTasson18}, when the
  authors of that paper tried to build a \emph{cartesian closed
    category} whose objects are Selinger's cones, ordered by the
  cone order (\(x\leq y\) if there is a \(z\) such that
  \(x+z=y\)).
  A careful analysis of these constraints (which are actually quite
  strong) leads unavoidably to the conclusion that the morphisms of
  the sought CCC should be totally monotonic.
  We will see in Section~\ref{sec:stable-CCC} how total monotonicity
  leads indeed to cartesian closedness.
\end{remark}

\begin{definition}
  Let \(P\) and \(Q\) be cones.
  A function \(f:\Cuball P\to Q\) is \emph{stable} if %
  \(f\) is totally monotonic, bounded and \(\omega\)-continuous (see
  Definition~\ref{def:monotone-scott}).
\end{definition}

\begin{remark}
  This terminology is motivated by the fact that stable functions (in
  the sense of Berry~\cite{Berry78}) on Girard's coherence
  spaces~\cite{Girard86} can be characterized by a property completely
  similar to total monotonicity.
  We have the intuition that our stable morphisms on cones are a
  quantitative analog of the notion of stable function introduced in
  ``qualitative'' denotational semantics.
\end{remark}

\begin{example}
  \label{ex:tot-mon-on-bool}
  Let \(B=\Cmeas{(\{0,1\})}\) so that %
  \(\Mcca B=\Realp^2\) with the norm given by \(\Norm x=x_0+x_1\).
  Then for each \(\Vect a\in\Realp^{\Nat\times\Nat}\) which satisfies
  \(
    \forall t\in\Intercc 01\quad \sum_{p,q\in\Nat}a_{p,q}t^p(1-t)^q\leq 1
  \),
  the function \(f:\Cuball{\Mcca B}\to\Realp\) defined by %
  \(f(x)=\sum_{p,q\in\Nat}a_{p,q}x_0^px_1^q\) is totally monotonic.
  An example of such a function is %
  \(f(x)=\sum_{n=1}^\infty 2^nx_0^nx_1^n\) since for
  \(x\in\Cuball{\Mcca B}\) one has \(x_0x_1\leq\frac 14\).
  One might think that \(f\) is convex since all its iterated partial
  derivatives are \(\geq0\) everywhere, but this is not true since for
  instance \(f(0,1)=f(1,0)=0\) whereas \(f(\frac 12,\frac 12)=1\).
  In general, total monotonicity of functions defined by power series
  with one or more parameters as in this example correspond to the
  fact that all the partial derivatives are everywhere $\geq 0$, which
  is related simply to convexity only in the one parameter case.
\end{example}

\begin{definition}
  Let \(C,D\) be measurable cones. A stable function %
  \(f:\Cuball{\Mcca C}\to\Mcca D\) is measurable if for each
  \(X\in\ARCAT\) and %
  \(\gamma\in\Cuball{\Mcca{\Cpath XC}}\) one has %
  \(f\Comp\gamma\in\Mcca{\Cpath XD}\).
\end{definition}

\begin{lemma}
  The set of stable and measurable functions \(C\to D\), equipped with
  algebraic operations defined pointwise, is a precone.
\end{lemma}
\noindent 
This is straightforward, we use \(P\) for this precone.
We need first to understand the order on stable functions induced by
the addition operation of \(P\).
As usual, this order relation is simply denoted as \(\leq\) of \(\leq_P\).
\begin{lemma} %
  \label{lemma:stable-order-charact}
  Let \(f,g\in P\). One has \(f\leq g\) iff for each \(n\in\Nat\) and
  each %
  \(x,\List u1n\in\Cuball{\Mcca C}\) such that %
  \(x+\sum_{i=1}^nu_i\in\Cuball{\Mcca C}\) one has
  \begin{multline*}
    \sum_{I\in\Npset n}g(x+\sum_{i\in I}u_i)
    +\sum_{I\in\Ppset n}f(x+\sum_{i\in I}u_i)\\
    \leq
    \sum_{I\in\Ppset n}g(x+\sum_{i\in I}u_i)
    +\sum_{I\in\Npset n}f(x+\sum_{i\in I}u_i)\,.
  \end{multline*}
\end{lemma}
\begin{proof}
  Assume first that \(f\leq g\) and let \(h=g-f\).
  Notice that since addition is defined pointwise in \(P\) we have
  \(h(x)=g(x)-f(x)\) for all \(x\in\Cuball{\Mcca C}\), and by
  definition of the order relation of \(P\), this function \(h\) is
  stable.
  Given \(n\in\Nat\) and \(x,\List u1n\in\Cuball{\Mcca C}\) such that
  \(x+\sum_{i=1}^nu_i\in\Cuball{\Mcca C}\) we have %
  \begin{align*}
    \sum_{I\in\Npset n}h(x+\sum_{i\in I}u_i)
    \leq\sum_{I\in\Ppset n}h(x+\sum_{i\in I}u_i)
  \end{align*}
  and the announced inequality follows. Conversely if the property
  expressed in the lemma holds we have in particular
  \(\forall x\in\Cuball P\ f(x)\leq g(x)\) (take \(n=0\)) and so we
  can define a function \(h:\Cuball{\Mcca C}\to\Mcca D\) by
  \(h(x)=g(x)-f(x)\) and this function is totally monotonic.
  This function is \(\omega\)-continuous by
  Lemma~\ref{lemma:fun-diff-Scott} and is measurable because
  subtraction is measurable on \(\Real^2\).
\end{proof}

\begin{remark}
  Notice that if \(f\leq g\) (still for the cone order relation of
  \(P\), characterized by Lemma~\ref{lemma:stable-order-charact}) then
  \(f(x)\leq g(x)\) for all \(x\in\Cuball{\Mcca B}\), but the converse
  is not true.
  As an example take \(f(x)=x\) and \(g(x)=1\), defining two stable
  functions (for \(B=C=\Sone\)) which do not satisfy \(f\leq_P g\) but
  are such that \(f(x)\leq g(x)\) holds for all \(x\in\Intercc 01\).
  It is natural to call this order relation on stable functions the
  \emph{stable order} in reference to~\cite{Berry78,Girard86} where
  the stable order behaves in a very similar way, and admits a similar
  characterization, in terms of differences.
\end{remark}

We equip this precone \(P\) with the norm
\(\Norm f=\sup_{x\in\Cuball{\Mcca C}}\Norm{f(x)}\) which is well
defined by our assumptions that stable functions are bounded. %
\begin{lemma} %
  \label{lemma:sup-stable-fns}
  Let \((f_n\in\Cuball P)_{n\in\Nat}\) be an increasing sequence (for the
  stable order).
  Then \(f:\Cuball{\Mcca C}\to\Mcca D\) defined by
  \(f(x)=\sup_{n\in\Nat}f_n(x)\) is bounded, totally monotonic,
  \(\omega\)-continuous and measurable, that is, \(f\in P\).
  This map \(f\) is the lub of the \(f_n\)'s in
  \(P\).
\end{lemma}
\begin{proof}[Proof sketch]
  Total monotonicity follows from \(\omega\)-continuity of addition,
  \(\omega\)-continuity is straightforward and measurability results from
  the monotone convergence theorem.
  The fact that \(f\) is the lub of the \(f_n\)'s results from the
  fact that it is defined as a pointwise lub.
\end{proof}

So \(P\) is a cone that we equip with a measurability structure
\(\cM\) defined as in \(\Limpl CD\): a \(p\in\cM_X\) is a function
\(p=\Mtfun\gamma m\) where %
\(\gamma\in\Mcca{\Cpath XC}\) and \(m\in\Mcms D_Y\), given by
\begin{align*}
  \Mtfun\gamma m=\Absm{(r,f)\in X\times P}
  {m(r,f(\gamma(r)))}\,.
\end{align*}
Then we check that \(\cM\) satisfies the required conditions exactly
as we did for \(\Limpl CD\) in Section~\ref{sec:cone-linear-fun}.
We have defined a measurable cone \(\Simpls CD\) that we prove now to
be integrable.

Let \(X\in\ARCAT\) and \(\eta\in\Mcca{\Cpath X{\Simpls CD}}\), and
let \(\mu\in\Mcca{\Cmeas(X)}\).
We define a function \(f:\Cuball{\Mcca C}\to\Mcca D\) by
\begin{align*}
  f(x)=\int_{r\in X}\eta(r)(x)\mu(dr)\,.
\end{align*}
This integral is well defined because, for each
\(x\in\Cuball{\Mcca B}\), the function \(\Absm{r\in X}{\eta(r)(x)}\)
is measurable and bounded since \(\eta\) is a measurable path.
The function \(f\) is totally monotonic by bilinearity of integration,
\(\omega\)-continuous by the monotone convergence theorem, we check
that it is measurable.
Let \(Y\in\ARCAT\) and \(\gamma\in\Mcca{\Cpath Y C}\), we have
\begin{align*}
  f\Comp\gamma
  &=\Absm{s\in Y}{\int_{r\in X}\eta(r)(\gamma(s))\mu(dr)}
\end{align*}
and we must prove that \(f\Comp\gamma\in\Mcca{\Cpath Y{D}}\),
so let %
\(Y'\in\ARCAT\) and \(m\in\Mcms{D}_{Y'}\), we have
\begin{align*}
  \Absm{(s',s)\in Y'\times Y}{m(s',(f\Comp\gamma)(s))}
  &=\Absm{(s',s)\in Y'\times Y}
    {m\Big(s',\int_{r\in X}\eta(r)(\gamma(s))\mu(dr)\Big)}\\
  &=\Absm{(s',s)\in Y'\times Y}{\int_{r\in X} m(s',\eta(r)(\gamma(s)))\mu(dr)}
\end{align*}
which is measurable because %
\(\Absm{(s',s,r)\in Y'\times Y\times X}{m(s',\eta(r)(\gamma(s)))}\) %
is measurable by our assumption about \(\eta\) and by
Lemma~\ref{lemma:integral-measurable}.
This shows that %
\(f\in\Mcca{\Simpls CD}\).
Let \(p\in\Mcms{\Impl CD}_\Measterm\) so that %
\(p=\Mtfun xm\) for some \(x\in\Mcca C\) and
\(m\in\Mcms D_\Measterm\), we have
\begin{align*}
  p(f)
  &=m\Big(\int_{r\in X}\eta(r)(x)\mu(dr)\Big)\\
  &=\int_{r\in X} m(\eta(r)(x))\mu(dr)\\
  &=\int_{r\in X} p(\eta(r))\mu(dr)
\end{align*}
so that \(f\) is the integral of \(\eta\) over \(\mu\), this shows that
\(\Simpls CD\) is an integrable cone.

\subsection{Finite differences} %
\label{sec:finite-diffs}
We introduce a natural difference operator on totally monotonic
functions which provides an inductive characterization of total
monotonicity that we will use to prove two basic properties of
totally monotonic functions, Lemmas~\ref{lemma:fdiff-comp}
and~\ref{lemma:lin-tot-mon-is-tot-mon}, which will show useful to
establish a property which is not completely obvious: the composition
of two stable functions is stable.
In the setting of complete and absolute monotonicity in real analysis,
the corresponding property can be obtained by means of the
Faà~di~Bruno formula (higher derivative of composite of functions) as
mentioned in~\cite{LorchNewman83}, Section~7.
Our notion of total monotonicity being defined in terms of iterated
differences, we need a specific reasoning.

Given \(\Vect u\in{\Mcca B}^n\) such that %
\(\sum_{i=1}^n u_i\in\Cuball{\Mcca B}\) we use \(\Cloc B{\Vect u}\)
for the local cone \(\Cloc B{\sum_{i=1}^nu_i}\) (see
Section~\ref{sec:local-cone}).

Let \(B,C\) be cones, \(f:\Cuball{\Mcca }B\to\Mcca C\) be a function,
\(n\in\Nat\) and \(\List u1n\in\Mcca B\) such that
\(\sum_{i=1}^nu_i\in\Cuball{\Mcca B}\) we define
\begin{align*}
  \Fdiffs\epsilon f{\Vect u}:\Cuball{\Mcca{\Cloc B{\Vect u}}}
  &\to\Mcca C\\
  x&\mapsto\sum_{I\in\Epset\epsilon n}f(x+\sum_{i\in I}u_i)
\end{align*}
for \(\epsilon\in\Eset{-,+}\)
and if \(f\) is totally monotonic we set
\begin{align*}
  \Fdiff f{\Vect u}
  =\Fdiffp f{\Vect u}-\Fdiffn f{\Vect u}:
  \Cuball{\Mcca{\Cloc B{\Vect u}}}\to\Mcca C\,,
\end{align*}
the difference being computed pointwise.
Notice that for \(n=0\) %
(so that \(\Vect u=\Emptytuple\)) we have \(\Fdiff f\Emptytuple=f\)
since \(\Ppset 0=\Eset\emptyset\) and \(\Npset 0=\emptyset\).

\begin{definition}
  Let \(f\in\Cuball{\Mcca B}\to\Mcca C\) be a function and let
  \(n\in\Nat\). We say that \(f\) is %
  \emph{\(n\)-increasing from \(B\) to \(C\)} if
  \begin{itemize}
  \item \(n=0\) and \(f\) is increasing
  \item or \(n>0\), \(f\) is increasing and, for all
    \(u\in\Cuball{\Mcca B}\) the function
    \(\Fdiff fu:\Cuball{\Mcca{\Cloc Bu}}\to\Mcca C\) (which maps \(x\)
    to \(f(x+u)-f(x)\)) is \((n-1)\)-increasing from \(\Cloc Bu\) to
    \(C\).
  \end{itemize}
\end{definition}

\begin{lemma}
  \label{lemma:fdiff-inf-increasing}
  Let \(f\in\Cuball{\Mcca B}\to\Mcca C\) be a function which is
  \(n\)-increasing for all \(n\in\Nat\).
  Then for all \(u\in\Cuball{\Mcca B}\), the function %
  \(\Fdiff fu:\Cuball{\Mcca{\Cloc Bu}}\to\Mcca C\) is \(n\)-increasing
  for all \(n\in\Nat\).
\end{lemma}
\begin{proof}
  Immediate consequence of the definition of \(n\)-increasing functions.
\end{proof}

\begin{lemma}
  \label{lemma:fdiffs-expr-extended}
  Let \(f:\Mcca B\to\Mcca C\) be totally monotonic.
  For \(u,u_1,\dots,u_n\in\Mcca B\) and
  \(x\in\Cuball{\Mcca{\Cloc B{u,\Vect u}}}\), one has
  \(\Fdiffs\epsilon f{u,\Vect u}(x)=\Fdiffs\epsilon f{\Vect
    u}(x+u)+\Fdiffs{-\epsilon}f{\Vect u}(x)\) for %
  \(\epsilon\in\{\Pl,\Mn\}\).
  Moreover
  \begin{align*}
    \Fdiff f{\Vect u}(x)\leq\Fdiff f{\Vect u}(x+u)
  \text{\quad and\quad}
    \Fdiff f{u,\Vect u}=\Fdiff{(\Fdiff f{\Vect u})}u\,.
  \end{align*}
\end{lemma}
\begin{proof}
  Let \(\Vect v=(u,\Vect u)\), of length \(n+1\).
  We have
  \begin{align*}
    \Fdiffs\epsilon f{\Vect v}(x)
    =\sum_{I\in\Epset\epsilon{n+1}}f(x+\sum_{i\in I}v_i)
    &=\sum_{\Biind{I\in\Epset\epsilon{n+1}}{1\in I}}f(x+\sum_{i\in I}v_i)
      +\sum_{\Biind{I\in\Epset\epsilon{n+1}}{1\notin I}}
      f(x+\sum_{i\in I}v_i)\,.
  \end{align*}
  Now observe that
  \begin{align*}
    \sum_{\Biind{I\in\Epset{\epsilon}{n+1}}{1\notin I}}
      f(x+\sum_{i\in I}v_i)
    =\sum_{I\in\Epset{-\epsilon}{n}}f(x+\sum_{i\in I}u_i)
    =\Fdiffs{-\epsilon}f{\Vect u}(x)
  \end{align*}
  by definition of \(\Vect v\) and, using
  Lemma~\ref{lemma:inset-corresp}, observe also that
  \begin{align*}
    \sum_{\Biind{I\in\Epset\epsilon{n+1}}{1\in I}}f(x+\sum_{i\in I}v_i)
    =\sum_{I\in\Epset\epsilon{n}}f(x+u+\sum_{i\in I}u_i)
    =\Fdiffs\epsilon{f}{\Vect u}(x+u)\,.
  \end{align*}
  So we have
  \(\Fdiffs\epsilon f{u,\Vect u}(x)=\Fdiffs\epsilon f{\Vect
    u}(x+u)+\Fdiffs{-\epsilon}f{\Vect u}(x)\).
  Since \(f\) is totally monotonic, we have
  \begin{align*}
    \Fdiffn f{u,\Vect u}(x)\leq\Fdiffp f{u,\Vect u}(x)
  \end{align*}
  and hence
  \begin{align*}
    \Fdiffn f{\Vect u}(x+u)-\Fdiffn f{\Vect u}(x)
    \leq
    \Fdiffp f{\Vect u}(x+u)-\Fdiffp f{\Vect u}(x)
  \end{align*}
  both subtractions being defined since \(f\) is increasing.
  Therefore \(\Fdiff f{\Vect u}(x)\leq\Fdiff f{\Vect u}(x+u)\).
  Moreover
  \begin{align*}
    \Fdiff f{u,\Vect u}(x)
    &=\Fdiffp f{u,\Vect u}(x)-\Fdiffn f{u,\Vect u}(x)\\
    &=(\Fdiffp f{\Vect u}(x+u)+\Fdiffn f{\Vect u}(x))
    -(\Fdiffn f{\Vect u}(x+u)+\Fdiffp f{\Vect u}(x))\\
    &=(\Fdiffp{f}{\Vect u}(x+u)-\Fdiffn{f}{\Vect u}(x+u))
      -(\Fdiffp{f}{\Vect u}(x)-\Fdiffn{f}{\Vect u}(x))\\
    &=\Fdiff f{\Vect u}(x+u)-\Fdiff f{\Vect u}(x)\\
    &=\Fdiff{(\Fdiff f{\Vect u})}u(x)\,.
      \qedhere
  \end{align*}
\end{proof}

\begin{lemma}
  \label{lemma:fdiff-tot-mono}
  If a function \(f\in\Cuball{\Mcca B}\to\Mcca C\) is totally
  monotonic, then for each \(u\in\Cuball{\Mcca B}\), the function
  \(\Fdiff fu:\Cuball{\Mcca{\Cloc Bu}}\to\Mcca C\) is totally
  monotonic.
\end{lemma}
\begin{proof}
  Let \(\Vect u\in\Cuball{\Mcca{\Cloc{B}u}}\), notice that %
  \(\Cloc{(\Cloc Bu)}{\Vect u}=\Cloc B{u,\Vect u}\).
  Let \(x\in\Cuball{\Mcca{\Cloc B{u,\Vect u}}}\).
  We have
  \[
    \Fdiffs\epsilon{(\Fdiff f{u})}{\Vect u}(x)
    =\Fdiffs\epsilon{f}{\Vect u}(x+u)-\Fdiffs\epsilon{f}{\Vect u}(x)
  \]
  where the subtraction makes sense because \(f\) is increasing.
  By our assumption on \(f\) we have
  \begin{align*}
    \Fdiffn{f}{u,\Vect u}(x)\leq\Fdiffp{f}{u,\Vect u}(x)
  \end{align*}
  that is
  \begin{align*}
    \Fdiffn f{\Vect u}(x+u)+\Fdiffp f{\Vect u}(x)
    \leq\Fdiffp f{\Vect u}(x+u)+\Fdiffn f{\Vect u}(x)
  \end{align*}
  by Lemma~\ref{lemma:fdiffs-expr-extended}, and hence
  \begin{align*}
    \Fdiffn{(\Fdiff f{u})}{\Vect u}(x)&\leq\Fdiffp{(\Fdiff f{u})}{\Vect u}(x)
    \qedhere
  \end{align*}
\end{proof}

\begin{theorem}
  \label{th:induct-total-nonotone}
  A function \(f\in\Cuball{\Mcca B}\to\Mcca C\) is totally monotonic
  iff it is \(n\)-increasing for all \(n\in\Nat\).
\end{theorem}
\begin{proof}
  Remember that \(f\) is totally monotonic iff for all \(n\in\Nat\),
  all \(\Vect u\in\Mcca B^n\) and \(x\in\Cloc{\Mcca{B}}{\Vect u}\) we
  have %
  \(\Fdiffn f{\Vect u}(x)\leq\Fdiffp f{\Vect u}(x)\).

  We prove first by induction on \(k\in\Nat\) that for all
  \(f\in\Cuball{\Mcca B}\to\Mcca C\), if \(f\) is totally monotonic
  then \(f\) is \(k\)-increasing.

  For \(k=0\), we have to prove that \(f\) is increasing, which results
  from total monotonicity applied with \(n=1\).

  For \(k>0\) we have to prove that \(f\) is increasing (which results
  from total monotonicity applied with \(n=1\)) and that for all
  \(u\in\Cuball{\Mcca B}\) the function
  \(\Fdiff f u:\Cuball{\Mcca{\Cloc Bu}}\to\Mcca C\) is
  \((k-1)\)-increasing, for which, by inductive hypothesis, it
  suffices to prove that \(\Fdiff fu\) is totally monotonic, and this
  property results from Lemma~\ref{lemma:fdiff-tot-mono}.

  Conversely, we prove by induction on \(n\in\Nat\) that, for each
  \(f\in\Cuball{\Mcca B}\to\Mcca C\), if \(f\) is \(k\)-increasing for
  all \(k\in\Nat\) then for all \(\Vect u\in\Cuball{\Mcca B}^n\) and
  \(x\in\Cuball{\Mcca{\Cloc B{\Vect u}}}\), one has
  \(\Fdiffn f{\Vect u}(x)\leq\Fdiffp f{\Vect u}(x)\).
  For \(n=0\) there is nothing to prove so assume that \(n>0\).
  Let \((u,\Vect u)\in\Cuball{\Mcca B}^n\) and let %
  \(x\in\Cuball{\Mcca{\Cloc B{u,\Vect u}}}\).
  Notice that \(\Vect u\in\Cuball{\Mcca B}^{n-1}\) and
  \(x\in\Cuball{\Mcca{\Cloc B{\Vect u}}}\).
  Since \(\Fdiff fu\) is \(k\)-increasing for all \(k\in\Nat\) by
  Lemma~\ref{lemma:fdiff-inf-increasing}, we know by applying the
  inductive hypothesis to \(\Fdiff fu\) that %
  \begin{align*}
    \Fdiffn{(\Fdiff fu)}{\Vect u}(x)\leq\Fdiffp{(\Fdiff fu)}{\Vect u}(x)
  \end{align*}
  that is
  \begin{align*}
    \Fdiffn f{\Vect u}(x+u)-\Fdiffn f{\Vect u}(x)
    \leq\Fdiffp f{\Vect u}(x+u)-\Fdiffp f{\Vect u}(x)
  \end{align*}
  which implies
  \begin{align*}
    \Fdiffn f{\Vect u}(x+u)+\Fdiffp f{\Vect u}(x)
    \leq\Fdiffp f{\Vect u}(x+u)+\Fdiffn f{\Vect u}(x)
  \end{align*}
  that is \(\Fdiffn f{u,\Vect u}(x)\leq\Fdiffp f{u,\Vect u}(x)\)
  by Lemma~\ref{lemma:fdiffs-expr-extended}, as expected.
\end{proof}

\begin{lemma} %
  \label{lemma:fdiff-tot-mon}
  Let \(f:\Cuball{\Mcca B}\to\Mcca C\) be totally monotonic and
  \(\Vect u\in{\Mcca B}^n\) be such that %
  \(\sum_{i=1}^n u_i\in\Cuball{\Mcca B}\). Then the functions %
  \(
  \Fdiffp f{\Vect u},\Fdiffn f{\Vect u},\Fdiff f{\Vect u}:
  \Cuball{\Mcca{\Cloc B{\Vect u}}}\to\Mcca C
  \)
  are totally monotonic.
\end{lemma}
\begin{proof}
  The total monotonicity of \(\Fdiffs\epsilon f{\Vect u}\) results
  from the easy observation that if
  \(g:\Cuball{\Mcca B}\to\Cuball{\Mcca C}\) and
  \(u\in\Cuball{\Mcca B}\) then the map %
  \(g_u:\Cuball{\Mcca{\Cloc Bu}}\to\Mcca C\) defined by %
  \(g_u(x)=g(x+u)\) is also totally monotonic, and each sum of totally
  monotonic functions is totally monotonic.

  The total monotonicity of \(\Fdiff f{\Vect u}\) results from
  Theorem~\ref{th:induct-total-nonotone}.
\end{proof}

\begin{lemma}
  \label{lemma:fdiff-upper}
  Let \(f:\Cuball{\Mcca B}\to\Mcca C\) be totally monotonic.
  Then for each \(\Vect u\in\Mcca B^n\) such that %
  \(\sum_{i=1}^nu_i\in\Cuball{\Mcca B}\) and \(x\in\Mcca{\Cloc B{\Vect u}}\)
  we have
  \(
  \Fdiff f{\Vect u}(x)\leq f(x+\sum_{i=1}^nu_i)
  \).
\end{lemma}
\begin{proof}
  By induction on \(n\).
  The base case \(n=0\) is trivial since then
  \(\Fdiff f{\Vect u}(x)=f(x)\).
  For the inductive case, let \((u,\Vect u)\in\Mcca B^{n+1}\)
  with \(u+\sum_{i=1}^nu_i\in\Cuball B\) and
  \(x\in\Mcca{\Cloc B{u,\Vect u}}\), that is
  \(x+u\in\Mcca{\Cloc B{\Vect u}}\).
  We have
\(
    \Fdiff f{u,\Vect u}(x)
    =\Fdiff f{\Vect u}(x+u)-\Fdiff f{\Vect u}(x)
    \leq \Fdiff f{\Vect u}(x+u)
    \leq f(x+u+\sum_{i=1}^nu_i)
\)
  by inductive hypothesis.
\end{proof}

\begin{lemma}\label{lemma:fdiff-sommes1}
  Let $f:\Cuball{\Mcca B}\to\Mcca C$ be a totally monotonic function.
  Let $n\in\Nat$, $u,v\in\Cuball{\Mcca B}$ and
  $\Vect u\in\Cuball{\Mcca B}^n$, and assume that
  $u+v+\sum_{i=1}^nu_i\in\Cuball{\Mcca B}$.
  Then for each \(x\in\Cuball{\Mcca{\Cloc B{\Vect u}}}\) we have
\begin{align*}
  \Fdiff f{\Vect u}(x+u)&=\Fdiff f{\Vect u}(x)+\Fdiff f{u,\Vect u}(x)\\
  \Fdiff f{u+v,\Vect u}(x)&=\Fdiff f{u,\Vect u}(x)+\Fdiff f{v,\Vect u}(x+u)\,.
\end{align*}
\end{lemma}
\begin{proof}
  The first equation results from
  \(\Fdiff f{u,\Vect u}=\Fdiff{(\Fdiff f{\Vect u})}u\), see
  Lemma~\ref{lemma:fdiff-tot-mono}.
  For the second equation take
  \(x\in\Cuball{\Mcca{\Cloc B{u+v,\Vect u}}}\).
  Let \(n\) be the length of \(\Vect u\).
  Setting \(\Vect v=(u,\Vect u)\) and \(\Vect w=(v,\Vect u)\) (both of
  length \(n+1\)), we have
  \begin{align*}
    &\Fdiff f{\Vect v}(x)+\Fdiff f{\Vect w}(x+u)
      =\sum_{I\in\Epset\Pl{n+1}}f(x+\sum_{i\in I}v_i)
      -\sum_{I\in\Epset\Mn{n+1}}f(x+\sum_{i\in I}v_i)\\
    &\Textsep+\sum_{I\in\Epset\Pl{n+1}}f(x+u+\sum_{i\in I}w_i)
      -\sum_{I\in\Epset\Mn{n+1}}f(x+u+\sum_{i\in I}w_i)\\
    &\quad=\sum_{I\in\Epset\Mn n}f(x+\sum_{i\in I}u_i)
      +\sum_{I\in\Epset\Pl n}f(x+u+\sum_{i\in I}u_i)\\
    &\Textsep-(\sum_{I\in\Epset\Pl n}f(x+\sum_{i\in I}u_i)
      +\sum_{I\in\Epset\Mn n}f(x+u+\sum_{i\in I}u_i))\\
    &\Textsep+\sum_{I\in\Epset\Mn n}f(x+u+\sum_{i\in I}u_i)
      +\sum_{I\in\Epset\Pl n}f(x+u+v+\sum_{i\in I}u_i)\\
    &\Textsep-(\sum_{I\in\Epset\Pl n}f(x+u+\sum_{i\in I}u_i)
      +\sum_{I\in\Epset\Mn n}f(x+u+v+\sum_{i\in I}u_i))\\
    &\quad=\sum_{I\in\Epset\Pl n}f(x+u+v+\sum_{i\in I}u_i)
      +\sum_{I\in\Epset\Mn n}f(x+\sum_{i\in I}u_i)\\
    &\Textsep-(\sum_{I\in\Epset\Pl n}f(x+\sum_{i\in I}u_i)
      +\sum_{I\in\Epset\Mn n}f(x+u+v+\sum_{i\in I}u_i))\,.
      \qedhere
  \end{align*}
\end{proof}

\begin{lemma} %
  \label{lemma:fdiff-sommes}
  Let $f:\Cuball{\Mcca B}\to\Mcca C$ be totally monotonic. Let
  $n\in\Nat$, $u\in\Mcca B$ and $\Vect u,\Vect v\in\Mcca B^n$,
  and assume that $u+\sum_{i=1}^n(u_i+v_i)\in\Cuball{\Mcca B}$.
  Then for each \(x\in\Cuball{\Mcca{\Cloc B{u,\Vect u,\Vect v}}}\) we
  have
  \begin{align*}
    \Fdiff f{\Vect u+\Vect v}(x+u)
    &= \Fdiff f{\Vect u}(x)+\Fdiff f{u,\Vect u+\Vect v}(x)\\
    &\quad+\Fdiff f{v_1,u_2+v_2,\dots,u_n+v_n}(x+u_1)\\
    &\quad+\Fdiff f{u_1,v_2,u_3+v_3,\dots,u_n+v_n}(x+u_2)
      +\cdots\\
    &\quad+\Fdiff f{u_1,\dots,u_{n-1},v_n}(x+u_n)\,.
  \end{align*}
\end{lemma}
\begin{proof}
  Simple computations using Lemma~\ref{lemma:fdiff-sommes1}.
\end{proof}
\noindent 
Let \(\Csum nB\) be the cone defined by %
\(\Mcca{\Csum nB}=\Mcca B^{n+1}\) with operations defined pointwise and
norm defined by %
\(\Norm{(x,\Vect u)}_{\Csum nB}=\Norm{x+\sum_{i=1}^nu_i}_{B}\).
It is easy to check that one actually defines a cone in that way.

\begin{remark}
  This cone is not the \((n+1)\)-fold coproduct of \(\Mcca B\) with itself.
  Take indeed \(B=\With\Sone\Sone\) and \(n=1\), then %
  \(\Norm{((1,0),(0,1)}_{\Csum 1B}=\Norm{(1,1)}_{\Mcca B}=1\) whereas
  \(\Norm{((1,0),(0,1)}_{\Plus{\Mcca B}{\Mcca B}}=2\).
  Neither is it the \((n+1)\)-fold product of \(\Mcca B\); it is
  actually (isomorphic to)
  \(\Limpl{\overbrace{\Sone\IWith\cdots\IWith\Sone}^{n+1}}{B}\).
  This construct is at the origin of coherent
  differentiation~\cite{Ehrhard23}.
\end{remark}

\begin{lemma} %
  \label{lemma:fdiff-glob-increasing}
  If $f:\Cuball{\Mcca B}\to\Mcca C$ is totally monotonic, the map %
  \((x,\Vect u)\to\Fdiff f{\Vect u}(x)\) is increasing %
  \(\Cuball{\Csum nB}\to\Mcca C\).
\end{lemma}
\begin{proof}
  Follows easily from Theorem~\ref{th:induct-total-nonotone}.
\end{proof}
\noindent 
Now we can state and prove the main lemma which allows to prove that
totally monotonic functions are closed under composition.
Remember that even in the setting of completely and absolutely
monotonic functions (where derivatives instead of differences are
used), the corresponding result is not completely trivial as it
requires the use of the Faà~di~Bruno formula.
\begin{lemma} %
  \label{lemma:fdiff-comp}  
  Let $n\in\Nat$, $f,\List h1n:\Cuball{\Mcca B}\to\Mcca C$
  and $g:\Cuball{\Mcca C}\to\Mcca D$ be totally monotonic functions
  such that
  $\forall x\in\Cuball{\Mcca B}\
  f(x)+\sum_{i=1}^nh_i(x)\in\Cuball{\Mcca C}$.
  Then the function $k:\Cuball{\Mcca B}\to\Mcca D$ defined by
  \(k(x)=\Fdiff g{h_1(x),\dots,h_n(x)}(f(x))\)
  is totally monotonic.
\end{lemma}
\begin{proof}
  With the notations and conventions of the statement, we prove by
  induction on $p$ that, for all $p\in\Nat$, for all $n\in\Nat$, for
  all $f,\List h1n,g$ which are totally monotonic and satisfy
  $\forall x\in\Cuball{\Mcca B}\
  f(x)+\sum_{i=1}^nh_i(x)\in\Cuball{\Mcca C}$, the function $k$ is
  $p$-increasing.

  For $p=0$, the property results from
  Lemma~\ref{lemma:fdiff-glob-increasing}.

  We assume the property for $p$ and prove it for $p+1$. Let
  $u\in\Cuball{\Mcca B}$ we have to prove that the function
  $\Fdiff ku$ is $p$-increasing from $\Cloc Bu$ to $D$.
  Let $x\in\Cuball{\Mcca{\Cloc Bu}}$, we have %
  \begin{align*}
    \Fdiffvar kxu
    &=\Fdiffvar g{f(x+u)}{h_1(x+u),\dots,h_n(x+u)}-
      \Fdiffvar g{f(x)}{h_1(x),\dots,h_n(x)}\\
    &=\Fdiffvar g{f(x)+\Fdiffvar fxu}{h_1(x)+\Fdiffvar{h_1}{x}{u},
      \dots,h_n(x)+\Fdiffvar{h_n}{x}{u}}
    \\
    &\quad-\Fdiffvar g{f(x)}{h_1(x),\dots,h_n(x)}
    \text{\quad by definition of }\Fdiff{h_i}u\\
    &=\Fdiffvar g{f(x)}{\Fdiffvar fxu,h_1(x)+\Fdiffvar{h_1}{x}{u},
      \dots,h_n(x)+\Fdiffvar{h_n}{x}{u}}\\
    &\quad+\Fdiffvar g{f(x)+h_1(x)}{\Fdiffvar{h_1}{x}{u},
      h_2(x)+\Fdiffvar{h_2}xu,\dots,h_n(x)+\Fdiffvar{h_n}xu}\\
    &\quad+\Fdiffvar g{f(x)+h_2(x)}{h_1(x),\Fdiffvar{h_2}xu,
      h_3(x)+\Fdiffvar{h_3}xu,\dots,\\
    &\hspace{18em}h_n(x)+\Fdiffvar{h_n}xu}\\
    &\quad+\cdots+\Fdiffvar g{f(x)+h_n(x)}{h_1(x),
      \dots,h_{n-1}(x),\Fdiffvar{h_n}{x}{u}}
  \end{align*}
  by Lemma~\ref{lemma:fdiff-sommes}, observing that the first term of
  the sum which appears in Lemma~\ref{lemma:fdiff-sommes} is
  annihilated by the subtraction above.

  We can apply the inductive hypothesis to each of the terms of this
  sum.
  Let us consider for instance the first of these expressions:
  \begin{align*}
    \Fdiffvar g{f(x)}{\Fdiffvar fxu,h_1(x)+\Fdiffvar{h_1}{x}{u},
      \dots,h_n(x)+\Fdiffvar{h_n}{x}{u}}\,.
  \end{align*}
  We know that the functions \(h'_1,\dots,h'_{n+1}\) defined by
  \(h'_1(x)=\Fdiffvar fxu\),
  \(h'_2(x)=h_1(x)+\Fdiffvar{h_1}{x}{u}=h_1(x+u)\),\dots,
  \(h'_{n+1}(x)=h_n(x)+\Fdiffvar{h_n}{x}{u}=h_n(x+u)\) are totally
  monotonic from $\Mcca{\Cloc Bu}$ to $\Mcca C$: %
  this results from Lemma~\ref{lemma:fdiff-tot-mon}.
  Moreover we have
  \[
    \forall x\in\Cuball{\Mcca B}\
    f(x)+\sum_{i=1}^{n+1}h'_i(x)=f(x+u)+\sum_{i=1}^nh_i(x+u)\in\Cuball{\Mcca
      C}\,.
  \]
  Therefore the inductive hypothesis applies and we know that the
  function
  \[
    x\mapsto\Fdiffvar g{f(x)}{\Fdiffvar
      fxu,h_1(x)+\Fdiffvar{h_1}{x}{u},
      \dots,h_p(x)+\Fdiffvar{h_p}{x}{u}}
  \] %
  is $p$-increasing.
  The same reasoning applies to all terms and hence the function
  $\Fdiff ku$ is $p$-increasing from $\Cuball{\Mcca{\Cloc Bu}}$ to
  $\Mcca C$, as contended.
\end{proof}

\begin{lemma} %
  \label{lemma:lin-tot-mon-is-tot-mon}
  Let \(f:\Mcca B\times\Cuball{\Mcca C}\to\Mcca D\) be linear in its
  first argument and totally monotonic in its second argument.
  Then, when restricted to %
  \(\Cuball{\Mcca B}\times\Cuball{\Mcca C}\), the function \(f\) is
  totally monotonic.
\end{lemma}
\begin{proof}
  Let \(n\in\Nat\),
  \((x,y),(u_1,v_1),\dots,(u_n,v_n)\in\Mcca B\times\Mcca C\) %
  be such that %
  \((x,y)+\sum_{i=1}^n(u_i,v_i)\in\Cuball{\Mcca B}\times\Cuball{\Mcca
    C}\). For \(\epsilon\in\Eset{\Pl,\Mn}\), we have
  \begin{align*}
    \Fdiffs\epsilon f{(u_1,v_1),\dots,(u_n,v_n)}(x,y)
    &=\sum_{I\in\Epset\epsilon n}f(x+\sum_{i\in I}u_i,y+\sum_{i\in I}v_i)\\
    &=\sum_{I\in\Epset\epsilon n}f(x,y+\sum_{i\in I}v_i)+
      \sum_{I\in\Epset\epsilon n}\sum_{j\in I}f(u_j,y+\sum_{i\in I}v_i)
  \end{align*}
  by linearity of \(f\) in its first argument.
  By total monotonicity of \(f\) in its second argument we have
  \begin{align*}
    \sum_{I\in\Ppset n}f(x,y+\sum_{i\in I}v_i)
    \geq
    \sum_{I\in\Npset n}f(x,y+\sum_{i\in I}v_i)\,.
  \end{align*}
  Next, assuming that \(n>0\), we have
  \begin{align*}
    \sum_{I\in\Epset\epsilon n}\sum_{j\in I}f(u_j,y+\sum_{i\in I}v_i)
    &=\sum_{j=1}^n\sum_{\Biind{I\in\Epset\epsilon n}{j\in I}}
      f(u_j,y+\sum_{i\in I}v_i)\\
    &=\sum_{j=1}^n\sum_{I\in\Epset\epsilon{n-1}}
      f(u_j,y+\sum_{i\in\Inset j(I)}v_i)
      \text{ by Lemma~\ref{lemma:inset-corresp}}\\
    &=\sum_{j=1}^n\sum_{I\in\Epset\epsilon{n-1}}
      f(u_j,y+v_j+\sum_{i\in I}v(j)_i)
  \end{align*}
  where \((v(j)_i)_{i=1}^{n-1}\) is defined by
  \begin{align*}
    v(j)_i=
    \begin{cases}
      v_i & \text{if }i<j\\
      v_{i+1} &\text{if }i\geq j\,.
    \end{cases}
  \end{align*}
  By our assumption that \(f\) is totally monotonic in its second
  argument we have, for each \(j=1,\dots,n\),
  \begin{align*}
    \sum_{I\in\Ppset{n-1}}
    f(u_j,y+v_j+\sum_{i\in I}v(j)_i)
    \geq
    \sum_{I\in\Npset{n-1}}
    f(u_j,y+v_j+\sum_{i\in I}v(j)_i)
  \end{align*}
  from which it follows that
  \begin{align*}
    \sum_{I\in\Ppset n}\sum_{j\in I}f(u_j,y+\sum_{i\in I}v_i)
    \geq
    \sum_{I\in\Npset n}\sum_{j\in I}f(u_j,y+\sum_{i\in I}v_i)
  \end{align*}
  and hence
  \begin{align*}
    \Fdiffp f{(u_1,v_1),\dots,(u_n,v_n)}(x,y)
    \geq
    \Fdiffn f{(u_1,v_1),\dots,(u_n,v_n)}(x,y)
  \end{align*}
  for \(n>0\).
  This inequation also holds trivially for \(n=0\).
\end{proof}

\noindent 
An immediate consequence of this lemma is the following observation
which will be useful in Section~\ref{sec:analytic-functions-exp}.
\begin{lemma}
  \label{lemma:multilin-tot-mon}
  Let \(f:\prod_{i=1}^n\Mcca{B_i}\to\Mcca C\) be linear
  in each of its \(n\) arguments.
  Then the restriction of \(f\) to %
  \(\prod_{i=1}^n\Cuball{\Mcca{B_i}}\) is totally monotonic %
  \(\Cuball{\Mcca{\Bwith_{i=1}^nB_i}}\to\Mcca C\).
\end{lemma}
\begin{proof}
  Simple induction on \(n\) using Lemma~\ref{lemma:lin-tot-mon-is-tot-mon}.
\end{proof}

\begin{remark}
  In this statement, and other similar ones, we restrict \(f\) to the
  unit ball not for deep reason but only because the notion of totally
  monotonic function has been defined on unit balls, see
  Definition~\ref{def:totally-monotonic}.
\end{remark}

\subsection{The cartesian closed category of integrable cones and
  stable and measurable functions}
\label{sec:stable-CCC}
A quite remarkable property of total monotonicity is that it gives
rise to cartesian \emph{closed} categories, as we will see in this
section.
We do not know if this phenomenon has been observed
before~\cite{EhrhardPaganiTasson18}.

Let \(\STAB(B,C)\) be the set of all stable and measurable functions
from \(B\) to \(C\) whose norm is \(\leq 1\).
\begin{theorem}
  If \(f\in\STAB(B,C)\) and \(g\in\STAB(C,D)\) then %
  \(g\Comp f\in\STAB(B,D)\).
\end{theorem}
\begin{proof}
  The only non-obvious fact is that \(g\Comp f\) is totally monotonic,
  which is obtained by Lemma~\ref{lemma:fdiff-comp} (applied with
  \(n=0\)).
\end{proof}
\noindent 
So we have defined a category \(\STAB\) whose objects are the integrable
cones, and the morphisms are the stable and measurable functions.

\begin{lemma}
  \(\ICONES(B,C)\subseteq\STAB(B,C)\).
\end{lemma}
\begin{proof}
  Indeed linearity clearly implies total monotonicity.
\end{proof}
\noindent 
So we have a functor \(\Derfuns:\ICONES\to\STAB\) which acts as the
identity on objects and morphisms.
We can consider this functor as a forgetful functor since it forgets
linearity, whence its name: in \(\LL\) the purpose of the
\emph{dereliction} rules allows to forget the linearity of morphisms.
The functor \(\Derfuns\) is obviously faithful but of course not full:
see Examples~\ref{ex:cones-analytic-functions}
and~\ref{ex:tot-mon-on-bool} which provide nonlinear totally
monotonic functions.

\begin{theorem}
  The category \(\STAB\) has all products and is cartesian closed.
\end{theorem}
\begin{proof}
  If \((B_i)_{i\in I}\) is a family of integrable cones, we have
  already defined \(B=\Bwith_{i\in I}B_i\) which is the categorical
  product of the \(B_i\)'s in \(\ICONES\) (when equipped with the
  projections \(\Proj i\in\ICONES(B,B_i)\)).
  So \(\Derfuns(\Proj i)\in\STAB(B,B_i)\) for each \(i\in I\). Let
  \((f_i\in\STAB(C,B_i))_{i\in I}\), we define %
  \(f:\Cuball{\Mcca C}\to\Cuball{\Mcca B}\) by %
  \(f(x)=(f_i(x))_{i\in I}\) which is well defined by our assumption
  that %
  \(\forall i\in I\ \Norm{f_i}\leq 1\).
  Then \(f\) is easily seen to be stable because all the operations of
  \(B\), as well as its cone order relation, are defined
  componentwise.
  Measurability of \(f\) is proven as in the proof of
  Theorem~\ref{th:mcones-complete}.
  This shows that \(B\) is the categorical product of the \(B_i\)'s in
  \(\STAB\).

  Let \(B\) and \(C\) be integrable cones.
  We have defined in Section~\ref{sec:icone-stable-fun} the integrable
  cone \(\Simpls BC\) of stable and measurable functions \(B\to C\),
  we show that it is the internal hom of \(B\) and \(C\) in \(\STAB\).
  We define %
  \(\Ev:\Cuball{\Mcca{\Withp{(\Simpls BC)}{B}}}\to\Mcca C\) by %
  \(\Ev(f,x)=f(x)\).
  The total monotonicity of \(\Ev\) results from
  Lemma~\ref{lemma:lin-tot-mon-is-tot-mon}.
  We have
  \begin{align*}
    \Cuball{\Mcca{\Withp{(\Simpls BC)}{B}}}
    =\Cuball{\Mcca{(\Simpls BC)}}\times\Cuball{\Mcca B}
  \end{align*}
  by definition of the norm in the categorical product.
  It follows that \(\Norm\Ev\leq 1\).
  We prove that \(\Ev\) is measurable so let %
  \(X\in\ARCAT\) and %
  \(\delta\in\Cuball{\Mcca{\Cpath X{\Withp{(\Simpls BC)}{B}}}}\) which
  means that \(\delta=\Tuple{\eta,\beta}\) with %
  \(\eta\in\Cuball{\Mcca{\Cpath X{\Simpls BC}}}\) and %
  \(\beta\in\Cuball{\Mcca{\Cpath XB}}\), we must prove that %
  \(\Ev\Comp\delta\in\Mcca{\Cpath XC}\) so let %
  \(Y\in\ARCAT\) and \(m\in\Mcms C_Y\).
  We must prove that %
  \begin{align*}
    \phi=\Absm{(s,r)\in Y\times X}{m(s,\Ev(\delta(r)))}
    &=\Absm{(s,r)\in Y\times X}{m(s,\eta(r)(\beta(r)))}
  \end{align*}
  is measurable.
  We have %
  \(p=\Mtfun{(\beta\Comp\Proj2)}{(m\Comp\Proj1)}\in\Mcms{\Simpls
    BC}_{Y\times X}\) and by our assumption about \(\eta\) we know
  that %
  \begin{align*}
    \Absm{(r,s,r')\in X\times Y\times X}{p(s,r,\eta(r'))}
    =\Absm{(r,s,r')\in X\times Y\times X}{m(s,\eta(r')(\beta(r)))}
  \end{align*}
  is measurable and hence so is \(\phi\) and we have shown that %
  \(\Ev\in\STAB(\With{(\Simpls BC)}{B},C)\).

  We prove that %
  \((\Simpls BC,\Ev)\) is the internal hom of \(B,C\) in the cartesian
  category \(\STAB\).
  So let \(f\in\STAB(\With DB,C)\).
  For each given \(z\in\Cuball{\Mcca D}\) we see easily that %
  \(g=\Absm{x\in\Cuball{\Mcca B}}{f(z,x)}\in\Cuball{\Mcca{\Simpls
      BC}}\), %
  it remains to check that \(g\in\STAB(D,\Simpls BC)\). %
  Let us first check that \(g\) is totally monotonic so let %
  \(n\in\Nat\) and \(z,\List w1n\in\Mcca D\) be such that %
  \(z+\sum_{i=1}^nw_i\in\Cuball{\Mcca D}\).
  We must prove that %
  \begin{align*}
    h^-=\sum_{I\in\Npset n}g(z+\sum_{i\in I}w_i)
    \leq
    \sum_{I\in\Ppset n}g(z+\sum_{i\in I}w_i)=h^+
  \end{align*}
  in \(\Mcca{\Simpls BC}\).
  We use the characterization of the cone order in that cone
  given by Lemma~\ref{lemma:stable-order-charact}.
  So let \(k\in\Nat\) and let %
  \(x,\List u1k\in\Mcca B\) be such that %
  \(x+\sum_{j=1}^ku_j\in\Cuball B\).
  We must prove that %
  \begin{multline*}
    y^-=\sum_{J\in\Npset k}h^+(x+\sum_{i\in J}u_j)+
    \sum_{J\in\Ppset k}h^-(x+\sum_{i\in J}u_j)\\
    \leq
    \sum_{J\in\Ppset k}h^+(x+\sum_{i\in J}u_j)+
    \sum_{J\in\Npset k}h^-(x+\sum_{i\in J}u_j)=y^+
  \end{multline*}
  in \(\Mcca C\).
  We have
  \begin{align*}
    y^-&=\sum_{\Biind{I\in\Ppset n}{J\in\Npset k}}
    g(z+\sum_{i\in I}w_i,x+\sum_{i\in J}u_j)
    +\sum_{\Biind{I\in\Npset n}{J\in\Ppset k}}
    g(z+\sum_{i\in I}w_i,x+\sum_{i\in J}u_j)\\
    y^+&=\sum_{\Biind{I\in\Ppset n}{J\in\Ppset k}}
    g(z+\sum_{i\in I}w_i,x+\sum_{i\in J}u_j)
    +\sum_{\Biind{I\in\Npset n}{J\in\Npset k}}
    g(z+\sum_{i\in I}w_i,x+\sum_{i\in J}u_j)
  \end{align*}
  Notice that %
  \((\Ppset n\times\Npset k)\cap(\Npset n\times\Ppset k)=\emptyset\)
  and that there is a bijection
  \begin{align*}
    (\Ppset n\times\Npset k)\cup(\Npset n\times\Ppset k)
    &\to\Npset{n+k}\\
    (I,J)&\mapsto I\cup(J+n)
  \end{align*}
  and similarly that %
  \((\Ppset n\times\Ppset k)\cap(\Npset n\times\Npset k)=\emptyset\)
  and that there is a bijection
  \begin{align*}
    (\Ppset n\times\Ppset k)\cup(\Npset n\times\Npset k)
    &\to\Ppset{n+k}\\
    (I,J)&\mapsto I\cup(J+n)\,.
  \end{align*}
  We define a sequence %
  \((w'_l,u'_l)_{l=1}^{n+k}\) of elements of %
  \(\Mcca{D}\times\Mcca B\) as follows:
  \begin{align*}
    (w'_l,u'_l)=
    \begin{cases}
      (w_l,0) & \text{if }l\in\Eset{1,\dots,n}\\
      (0,u_{l-n}) & \text{if }l\in\Eset{n+1,\dots,n+k}\,.
    \end{cases}
  \end{align*}
  so that %
  \((z,x)+\sum_{l=1}^{n+k}(w'_l,u'_l)\in\Cuball{\Mcca
    D}\times\Cuball{\Mcca B}\).
  With these notations, we have
  \begin{align*}
    y^-&=\sum_{K\in\Npset{n+k}}g((z,x)+\sum_{l\in K}(w'_l,u'_l))\\
    y^+&=\sum_{K\in\Ppset{n+k}}g((z,x)+\sum_{l\in K}(w'_l,u'_l))
  \end{align*}
  and hence \(y^-\leq y^+\) since \(g\) is totally monotonic.

  The \(\omega\)-continuity of \(g\) results from
  Lemma~\ref{lemma:sup-stable-fns}.
  We prove that \(g\) is measurable
  so let \(X\in\ARCAT\) and %
  \(\delta\in\Cuball{\Mcca{\Cpath XD}}\), we must prove that %
  \(g\Comp\delta\in\Mcca{\Cpath X{\Simpls BC}}\) so let %
  \(Y\in\ARCAT\) and \(p\in\Mcms{\Simpls BC}_Y\).
  Let \(\beta\in\Cuball{\Mcca{\Cpath YB}}\) and %
  \(m\in\Mcms C_Y\) be such that \(p=\Mtfun\beta m\), we have
  \begin{align*}
    \Absm{(s,r)\in Y\times X}{p(s,g(\delta(r)))}
    &=\Absm{(s,r)\in Y\times X}{m(s,g(\delta(r))(\beta(s)))}\\
    &=\Absm{(s,r)\in Y\times X}{m(s,f(\delta(r),\beta(s)))}
  \end{align*}
  and this map is measurable because %
  \(\Tuple{\delta\Comp\Proj2,\beta\Comp\Proj1}
  \in\Cuball{\Mcca{\Cpath{Y\times X}{\With DB}}}\) and by
  measurability of \(f\).
\end{proof}

\begin{remark}
  Note that contrarily to \(\ICONES\) it is very likely that the
  category \(\SCONES\) does not have all equalizers and therefore
  is not complete.
  For instance we have \(f,g:\SCONES(\Sone,\Sone)\) given by
  \(f(x)=x\) and \(g(x)=x^2\), and the set of all
  \(x\in\Cuball{\Mcca\Sone}=\Intercc01\) such that \(f(x)=g(x)\) is
  \(\{0,1\}\) which does not look like a cone.
  It would be interesting to understand if the set of solutions of
  such an equation could be considered as some kind of manifold, with
  a local structure of integrable cone, as in differential geometry.
  The same observation applies to the category \(\ACONES\) studied in
  Section~\ref{sec:analytic-functions-exp}.
\end{remark}

So we have a functor %
\(\Simpls\_\_:\Op\STAB\times\STAB\to\STAB\) mapping \((B,C)\) to
\(\Simpls BC\) and \(f\in\STAB(B',B),g\in\STAB(C,C'))\) to %
\(\Simpls fg\in\STAB(\Simpls BC,\Simpls{B'}{C'})\) which is given by %
\((\Simpls fg)(h)=g\Comp h\Comp f\).
Observe that if \(g\in\ICONES(C,C')\) we have %
\(\Simpls fg\in\ICONES(\Simpls BC,\Simpls{B'}{C'})\) so that in the
sequel we consider only \(\Simpls\_\_\) as a functor
\(\Op\ICONES\times\ICONES\to\ICONES\) (using implicitly a
pre-composition with \(\Derfuns\)).

\begin{theorem}
  \label{th:derfuns-preserves-limits}
  The functor \(\Derfuns:\ICONES\to\STAB\) preserves all limits.
\end{theorem}
\begin{proof}
  By Theorem~\ref{th:Icones-adjoint-functor} it suffices to prove that
  \(\Derfuns\) preserves all products and all equalizers.
  Since products are defined in the same way in both categories, the
  first property is obvious, let us check the second one.

  Let \(B,C\) be objects of \(\ICONES\) and \(f,g\in\ICONES(B,C)\), we
  have already defined the equalizer \((E,e\in\ICONES(E,B))\) of
  \(f,g\) in the proof of Theorem~\ref{th:mcones-complete}.
  We just have to check that \((E,e)\) is the equalizer of \(f,g\) in
  \(\STAB\) as well.
  So let \(H\) be an integrable cone and %
  \(h\in\STAB(H,B)\) be such that \(f\Comp h=g\Comp h\).
  This simply means that
  \(h(\Cuball{\Mcca H})\subseteq\Cuball{\Mcca E}\) from which it
  follows that \(h\in\STAB(H,E)\) because the cone order relation
  of \(E\) is the restriction of that of \(B\) to
  \(\Mcca E\subseteq\Mcca B\) (and similarly for the measurability
  structure).
  Let us call \(h'\) this version of \(h\) ranging in
  \(\Cuball{\Mcca E}\) instead of \(\Cuball{\Mcca B}\), so that
  \(h=e\Comp h'\).
  It is obvious that \(h'\) is the only morphism in \(\STAB\) having
  this property.
\end{proof}
\noindent 
The study of the exponential induced by this cartesian closed
structure on the category \(\ICONES\) is developed in
Section~\ref{sec:lin-nonlin-adj}.

\section{Analytic and integrable functions on cones}
\label{sec:analytic-functions-exp}
Our goal now is to associate with \(\ICONES\) another cartesian closed
category based on a notion of morphisms which are analytic in the sense
that they are limits of polynomial functions.
These analytic functions are actually stable and measurable, but their
definition is based on a notion of multilinear maps%
\footnote{As in~\cite{KerjeanTasson18} but without the support of
  complex analysis.} %
in \(\ICONES\) which preserve integrals so that analytic functions
have an implicit ``integral preservation'' property%
\footnote{A property that we don't really know yet how to express
  simply and directly in terms of the functions; of course it is not
  plain integral preservation which cannot be expected from non-linear
  maps.
  We know that it is a property of the analytic functions themselves
  because the symmetric multilinear functions of their Taylor
  expansion at \(0\) are associated with analytic functions in a
  unambiguous way by means of standard polarization formulas as we shall
  see.
}
that general stable and measurable functions don't have.

\subsection{The cone of multilinear and symmetric functions}
The basic ingredient for defining our analytic functions is the notion
of multilinear morphisms.
More precisely, they will allow first to define homogeneous polynomial
functions (obtained by applying an \(n\)-linear function to
``diagonal'' tuples \((x,\dots,x)\)), and then analytic functions as
converging sums thereof.

\begin{definition}
  Let \(\List B1n,C\) be integrable cones.
  A function %
  \(f:\prod_{i=1}^n\Mcca{B_i}\to\Mcca C\) is said to be \emph{multilinear and continuous}
  if it is linear and continuous, separately, with respect to each of
  its \(n\) arguments.

  Observe that when this holds, \(f\) is bounded (use for instance
  Lemma~\ref{lemma:line-cont-bounded} and monoidal closedness in a
  proof by induction on \(n\)).
  One says that \(f\) is \emph{symmetric} if %
  \(B_1=\cdots=B_n=B\) and, for all \(\sigma\in\Symgrp n\) (the group
  of permutations on \(\Eset{1,\dots,n}\)), one has %
  \begin{align*}
    \forall \List x1n\in\Mcca B\quad %
    f(x_1,\dots,x_n)=f(x_{\sigma(1)},\dots,x_{\sigma(n)})\,.
  \end{align*}
  One says that \(f\) is \emph{measurable} if for all %
  \(X\in\ARCAT\) and \((\beta_i\in\Mcca{\Cpath X{B_i}})_{i=1}^n\), %
  one has %
  \(f\Comp\Tuple{\beta_1,\dots,\beta_n}\in\Mcca{\Cpath XC}\).
  Lastly, one says that \(f\) is \emph{integrable} (or that it \emph{preserves integrals})
  if it is separately integrable with respect to each of its arguments.
\end{definition}

The multilinear, continuous, symmetric, measurable and integrable
functions %
\(\Mcca B^n\to\Mcca C\) are easily seen to form a cone \(\Mlsym nBC\)
with operations defined pointwise and norm defined by
\begin{align*}
  \Norm f=\sup\{\Norm{f(x_1,\dots,x_n)}\St \List x1n\in\Cuball{\Mcca B}\}\,.
\end{align*}
We equip this cone with a measurability structure
\((\Mcms{\Mlsym nBC}_X)_{X\in\ARCAT}\) where \(\Mcms{\Mlsym nBC}_X\)
is the set of all %
\(p=\Mtlfun{\Vect\beta}{m}\) where %
\(\Vect\beta=(\beta_i\in\Mcca{\Cpath XB})_{i=1}^n\) and %
\(m\in\Mcms C_X\), given by %
\(p(f)=\Absm{r\in X}{m(r,f(\beta_1(r),\dots,\beta_n(r)))}\).
The order relation in this cone is the pointwise order as easily
checked.

Notice last that this cone is integrable.
Let indeed \(X\in\ARCAT\), \(\mu\in\Mcca{\Cmeas(X)}\) %
and \(\eta\in\Mcca{\Cpath X{\Mlsym nBC}}\) and let us define a
function \(f:\Mcca B^n\to\Mcca C\) by
\begin{align*}
  f(\List x1n)=\int\eta(r)(\List x1n)\mu(dr)
\end{align*}
then it is easy to check as usual that \(f\) is well defined,
\(f\in\Mcca{\Mlsym nBC}\) and that \(f=\int\eta(r)\mu(dr)\).

\begin{remark}
  This integrable cone is a ``subcone'' of the integrable cone %
  \(\Limpl{B\ITens\cdots\ITens B}{C}\), but we have not developed the
  notion of subcone in the present paper.
\end{remark}

\begin{remark}
  We could define \(\Mlsym nBC\) more abstractly as follows: first,
  generalizing Definition~\ref{def:bilinear}, we define the integrable
  cone of \(n\)-linear morphisms
  \(\Limplp{B_1,\dots,B_n}C=\Limplp{B_1}{\cdots\Limpl{}{\Limpl{B_n}C}}\),
  then we observe that, when \(B_1=\cdots=B_n=B\), for each permutation
  \(\sigma\in\Symgrp n\) there is an automorphism on that cone which
  acts by permutation of the arguments of \(n\)-linear functions, and
  we define \(\Mlsym nBC\) as the equalizer of all these automorphisms
  using the completeness of \(\ICONES\).
  Of course we would have obtained the same object of \(\ICONES\)
  (possibly up to an isomorphism), but the explicit description above
  will be useful.
\end{remark}

\subsection{The cone of homogeneous polynomial functions}
As announced, this is the next basic concept in the definition of
analytic functions.
\begin{definition}
  An \(n\)-homogeneous polynomial function from \(B\) to \(C\) is a
  function %
  \(f:\Mcca B\to\Mcca C\) such that there exists %
  \(h\in\Mcca{\Mlsym nBC}\) satisfying
  \begin{align*}
    \forall x\in\Mcca B\quad f(x)=h(x,\dots,x)\,.
  \end{align*}
  Then \(h\) is called a \emph{linearization} of \(f\).
  We use \(\Monom nBC\) for the set of \(n\)-homogeneous polynomial
  functions from \(B\) to \(C\) and set \(\Mlmon n(h)=f\).
\end{definition}
Notice that we would define exactly the same class of functions
without requiring \(h\) to be symmetric.
We make this choice only to reduce the number of notions at hand.

\begin{lemma} %
  \label{lemma:monomial-tot-mono}
  If \(f\in\Monom nBC\) then the restriction of \(f\) to %
  \(\Cuball{\Mcca B}\) is totally monotonic.
\end{lemma}
\begin{proof}
  Let \(h\) be a linearization of \(f\).
  By Lemma~\ref{lemma:multilin-tot-mon} we know that the restriction
  of \(h\) to \(\Cuball{\Mcca B}^n\to\Mcca C\) is totally monotonic,
  and since the diagonal map \(d:\Mcca B\to\Mcca B^n\) is linear of
  norm \(\leq 1\), the map \(f=h\Comp d\) is totally monotonic.
\end{proof}

\begin{lemma} %
  \label{lemma:unique-linearization}
  An \(n\)-homogeneous polynomial function \(f\) has exactly one
  linearization \(\Linhp nf\).
  Moreover
  \begin{align*}
    \Norm{\Linhp nf}\leq\frac{n^n}{\Factor n}\Norm f
  \end{align*}
  where \(\Norm f=\sup_{x\in\Cuball{\Mcca B}}\Norm{f(x)}\).
\end{lemma}
Notice that we also have \(\Norm f\leq\Norm{\Linhp nf}\) so that we
could interpret this lemma as expressing that the norms of
\(\Mcca{\Mlsym nBC}\) and \(\Monom nBC\) are equivalent and
that these two cones are isomorphic in a weak sense (remember that in
\(\ICONES\), isomorphisms must have norm
\(\leq 1\)).
\begin{proof}
  Let \(f:\Mcca B\to\Mcca C\) be an \(n\)-homogeneous polynomial and let %
  \(h\in\Mcca{\Mlsym nBC}\) be a linearization of \(f\).
  By Lemma~\ref{lemma:monomial-tot-mono} we can define a function %
  \(h':\Mcca B^n\to\Mcca C\) by
  \begin{align*}
    h'(x_1,\dots,x_n)
    =\frac 1{\Factor n}\Big(\sum_{I\in\Ppset n}f(\sum_{i\in I}x_i)
    -\sum_{I\in\Npset n}f(\sum_{i\in I}x_i)\Big)
    =\frac 1{\Factor n}\Fdiff f{\List x1n}(0)
  \end{align*}
  and the usual proof of the polarization theorem shows that
  necessarily \(h'=h\).
  So \(\Linhp nf=h\) is completely determined by \(f\), proving our
  contention.

  Next, given \(\List x1n\in\Cuball{\Mcca B}\),  we have
  \begin{align*}
    \Norm{h(x_1,\dots,x_n)}
    &\leq\frac1{\Factor n}\Norm{f(\sum_{i=1}^nx_i)}
    \text{\quad by Lemma~\ref{lemma:fdiff-upper}}\\
    &=\frac{n^n}{\Factor n}
      \Norm{h(\frac 1n\sum_{i=1}^nx_i,\dots,\frac 1n\sum_{i=1}^nx_i)}\\
    &=\frac{n^n}{\Factor n}
      \Norm{f(\frac 1n\sum_{i=1}^nx_i)}\\
    &\leq\frac{n^n}{\Factor n}\Norm f
  \end{align*}
  since \(\Norm{\frac 1n\sum_{i=1}^nx_i}\leq 1\).
\end{proof}
\noindent 
The set \(\Monom nBC\) is canonically a precone.
Indeed if \(f,g\in\Monom nBC\) then \(f+g\) (defined pointwise)
belongs to \(\Monom nBC\) because clearly
\(f+g=\Mlmon n(\Linhp nf+\Linhp ng)\) and we know that
\(\Linhp nf+\Linhp ng\in\Mlsym nBC\).
Multiplication by a scalar in \(\Realp\) is dealt with similarly.
Notice that this reasoning also shows that the maps %
\(\Linhp n:\Monom nBC\to\Mcca{\Mlsym nBC}\) and %
\(\Mlmon n:\Mcca{\Mlsym nBC}\to\Monom nBC\) are linear.

We define \(\Norm\__{\Monom nBC}\) as usual by %
\(\Norm f_{\Monom nBC}=\sup_{x\in\Cuball{\Mcca B}}\Norm{f(x)}_C\) %
so that clearly %
\(\Norm f_{\Monom nBC}\leq\Norm{\Linhp nf}_{\Mlsym nBC}\).
Equipped with this norm, \(\Monom nBC\) is a cone: the only non obvious
property is completeness, so let %
\((f_k\in\Cuball{\Monom nBC})_{k=1}^\infty\) be an increasing sequence
and let \(f:\Mcca B\to\Mcca C\) be the pointwise lub of this sequence
which is well defined since for each \(x\in\Mcca B\) we have
\begin{align*}
  \Norm{f_k(x)}
  \leq\Norm{\Linhp n{f_k}}\Norm x^n
  \leq\frac{n^n}{\Factor n}\Norm{f_k}\Norm x^n
  \leq\frac{n^n}{\Factor n}\Norm x^n\,.
\end{align*}
The sequence \((\Linhp n{f_k})_{k=1}^\infty\) is increasing in %
\(\frac{n^n}{\Factor n}\Cuball{\Mcca{\Mlsym nBC}}\) and has therefore
a lub \(h\in\Cuball{\Mcca{\Mlsym nBC}}\) and remember that this lub is
defined pointwise on \(\Mcca B^n\).
It follows that \(f=\Mlmon n(h)\in\Monom nBC\) and that we have %
\begin{align*}
  \forall x\in\Mcca B\quad f(x)=\sup_{k}f_k(x)\,.
\end{align*}
Since \(\forall k\ f_k\leq f\) by monotonicity of \(\Mlmon n\) it
follows that \(f=\sup_kf_k\).
Last observe that \(\Norm f\leq 1\) which ends the proof of
\(\omega\)-completeness of \(\Cuball{\Monom nBC}\).

So we have shown that \(\Monom nBC\) is a cone, and also that the
linear maps \(\Linhp n\) and \(\Mlmon n\) are \(\omega\)-continuous. %

\begin{remark}
  It is important to notice that, contrarily to %
  \(\Mlsym nBC\), the cone order relation of \(\Monom nBC\) is
  not the pointwise order.
  As an example take \(n=2\), \(B=\With\Sone\Sone\), \(C=\Sone\), and
  consider \(f,g\in\Monom nBC\) given by \(f(x,y)=2xy\) and
  \(g(x,y)=x^2+y^2\).
  Then %
  \begin{align*}
    \Linhp 2f((x_1,y_1),(x_2,y_2))
    &=\frac12(f(x_1+x_2,y_1+y_2)-f(x_1,y_1)-f(x_2,y_2)))\\
    &=x_1y_2+x_2y_1
  \end{align*}
  and similarly \(\Linhp 2g((x_1,y_1),(x_2,y_2))=x_1x_2+y_1y_2\) and
  therefore we do not have \(\Linhp2f\leq\Linhp2g\) %
  (we have \(\Linhp 2f((1,0),(0,1))=1\) and
  \(\Linhp 2g((1,0),(0,1))=0\)) %
  whereas
  \(\forall(x,y)\in\Mcca{\With\Sone\Sone}=\Realp\times\Realp\
  f(x,y)\leq g(x,y)\).
\end{remark}

Given \(X\in\ARCAT\), \(\beta\in\Mcca{\Cpath XB}\) and %
\(m\in\Mcms C_X\) we define
\(\Mtfun\beta m:X\times\Monom nBC\to\Realp\) as usual by
\((\Mtfun\beta m)(r,f)={m(r,f(\beta(r)))}\) for all
\(f\in\Monom nBC\).
Notice that
\begin{align*}
  \Absm {r\in X}{m(r,f(\beta(r)))}
  =\Absm {r\in X}{m(r,\Linhp nf(\Rep{\beta(r)}n))}
\end{align*}
and this function is measurable because %
\(\Linhp nf\in\Mcca{\Mlsym nBC}\), which implies that %
\(\Linhp nf\Comp\Tuple{\beta,\dots,\beta}\in\Mcca{\Cpath XC}\). 
Then it is easily checked that setting %
\(\cM_X =\Eset{\Mtfun\beta m\St\beta\in\Mcca{\Cpath XB} \text{ and
  }m\in\Mcms C_X}\) we define a measurability structure on %
\(\Monom nBC\) so that \(E=(\Monom nBC,(\cM_X)_{X\in\ARCAT})\) is a
measurable cone that we denote as \(\Monomic nBC\).

\begin{lemma}
  \label{lemma:Linhp-norm}
  \(\Linhp n\in\MCONES(\Monom nBC,\frac{n^n}{\Factor n}\Mlsym nBC)\).
\end{lemma}
\begin{proof}
  In view of what we know about \(\Linhp n\), it suffices to prove
  that \(\Linhp n:\Mcca{\Monomic nBC}\to\Mcca{\Mlsym nBC}\) is
  measurable so let \(X\in\ARCAT\) and
  \(\eta\in\Mcca{\Cpath X{\Monomic nBC}}\), we must check that
  \(\Linhp n\Comp\eta\in\Mcca{\Cpath X{\Mlsym nBC}}\).
  Let \(Y\in\ARCAT\) and \(p\in\Mcms{\Mlsym nBC}_Y\).
  We have %
  \(p=\Mtlfun{\Vect\beta}{m}\) where %
  \(\Vect\beta\in\Mcca{\Cpath YB}^n\) and \(m\in\Mcms B_Y\) so that %
  \begin{align*}
    \Absm{(s,r)\in Y\times X}{p(s,\Linhp n(\eta(r)))}
    =\Absm{(s,r)\in Y\times X}{m(s,\Linhp n(\eta(r))(\Vect\beta(s)))}
  \end{align*}
  which, %
  coming back to the characterization of \(\Linhp n\) explicitly
  provided in the proof of Lemma~\ref{lemma:unique-linearization},
  is measurable by measurability of addition and subtraction in
  \(\Real\) and linearity of \(m(s,\_)\).
\end{proof}

\begin{theorem}
  \label{th:Hpol-icone}
  \(\Monomic nBC\) is an integrable cone and %
  \[\Linhp n\in\ICONES(\Monom nBC,\frac{n^n}{\Factor n}\Mlsym nBC)\,.\]
\end{theorem}
\begin{proof}
  We prove integrability of %
  \(E=\Monomic nBC\) so let %
  \(X\in\ARCAT\), \(\eta\in\Mcca{\Cpath XE}\) and
  \(\mu\in\Mcca{\Cmeas(X)}\), we define %
  \(f:\Mcca B\to\Mcca C\) by %
  \(f(x)=\int\eta(r)(x)\mu(dr)\) using the fact that \(C\) is an
  integrable cone.
  Since \(\Linhp n\) is measurable we can define
  \(h:\Mcca B^n\to\Mcca C\) by
  \begin{align*}
    h(\Vect x)=\int\Linhp n(\eta(r))(\Vect x)\mu(dr)\,.
  \end{align*}
  which is clearly symmetric.
  It is \(n\)-linear, \(\omega\)-continuous, measurable by
  Lemma~\ref{lemma:int-mesurable}, and integrable by the first
  statement of Theorem~\ref{th:meas-path-equiv}.
  So we have \(h\in\Mcca{\Mlsym nBC}\) and
  \(h(x,\dots,x)=\int\eta(r)(x)\mu(dx)=f(x)\) for all \(x\in\Mcca B\)
  which proves that \(f\in\Mcca{\Monomic nBC}\).
  Last let \(p\in\Mcms{\Monomic nBC}_0\) so that %
  \(p=\Mtfun xm\) for some \(x\in\Mcca B\) and \(m\in\Mcms C_0\).
  We have %
  \(p(f)=m(f(x))=m(\int\eta(r)(x)\mu(dx))=\int m(\eta(r)(x))\mu(dr)\)
  by definition of an integral in \(C\).
  So \(p(f)=\int p(\eta(r))\mu(dr)\) by definition of \(p\), which
  shows that \(f\) is the integral of \(\eta\) in \(\Monomic nBC\) and
  hence that this measurable cone is also integrable.

  The integrability of \(\Linhp n\) results from its definition and
  from the fact that integrals commute with finite sums and
  differences.
\end{proof}

\subsection{The cone of analytic functions} %
\label{sec:analytic-functions}

We can finally define and study our analytic functions.
\begin{definition}
  A function \(f:\Cuball{\Mcca B}\to\Mcca C\) is analytic if it is
  bounded, and there is a sequence \((f_n\in\Monom nBC)_{n\in\Nat}\)
  such that
  \begin{align}\label{eq:ana-homo-sum}
    \forall x\in\Cuball{\Mcca B}\quad f(x)=\sum_{n=0}^\infty f_n(x)\,.
  \end{align}
  Such a sequence \((f_n)_{n\in\Nat}\) is called a \emph{homogeneous
    polynomial decomposition} of \(f\).
\end{definition}

Notice that the precise meaning of~\Eqref{eq:ana-homo-sum} is that,
for all \(x\in\Cuball{\Mcca B}\), the increasing sequence
\((\sum_{n=0}^N f_n(x))_{N\in\Nat}\) is bounded in \(\Mcca C\) (in the
sense of the norm) and has \(f(x)\) as lub.

\begin{lemma} %
  \label{lemma:unique-monomial}
  If \(f:\Cuball{\Mcca B}\to\Mcca C\) is analytic, then \(f\) has
  exactly one homogeneous polynomial decomposition.
\end{lemma}
\begin{proof}
  Let \((f_n)_{n\in\Nat}\) be a homogeneous polynomial decomposition
  of an analytic %
  \(f\) that we can assume without loss of generality to range in
  \(\Cuball{\Mcca C}\) since \(f\) is bounded. %
  Let \(x'\in\Mcca{\Cdual B}\) and \(x\in\Cuball{\Mcca B}\).
  Let
  \begin{align*}
    \phi:\Intercc 01 &\to\Realp\\
    t&\mapsto x'(f(tx))\,.
  \end{align*}
  We have \(\phi(t)=\sum_{n=0}^\infty x'(f_n(x))t^n\) by linearity and
  continuity of \(x'\) and hence
  \begin{align*}
    \forall n\in\Nat\quad x'(f_n(x))
    =\frac 1{\Factor n}\phi^{(n)}(0)
    =\frac 1{\Factor n}\Evreal{\frac{d^n}{dt^n}x'(f(tx))}{t=0}
  \end{align*}
  so that if \((g_n)_{n\in\Nat}\) is another homogeneous polynomial
  decomposition of \(f\) we have \(x'(f_n(x))=x'(g_n(x))\) for all
  \(x\), \(n\) and \(x'\).
  Since this holds in particular for all \(x'=m\in\Mcms B_0\) our
  claim is proven by \Mssepr.
\end{proof}
\noindent 
If \(f:\Cuball{\Mcca B}\to C\) is analytic, we use \(\Hpan n(f)\) for
the \(n\)th component of its unique homogeneous polynomial
decomposition and we set \(\Linan n=\Factor n(\Linhp n\Comp\Hpan n)\)
so that \(\Linan nf\in\Mcca{\Mlsym nBC}\) and we have
\begin{align*}
  f(x)=\sum_{n=0}^\infty\frac 1{\Factor n}\Linan nf(\Rep xn)
\end{align*}
which can be understood as the Taylor expansion of \(f\), motivating
the notation: %
\(\Linan nf\) can be understood as the \(n\)th derivative of \(f\) at
\(0\), which is an \(n\)-linear symmetric function.
As usual we say that \(f\) is measurable if, for all %
\(X\in\ARCAT\) and \(\beta\in\Mcca{\Cpath XB}\), one has %
\(f\Comp\beta\in\Mcca{\Cpath XC}\).

We define now a cone of analytic and measurable functions %
\(\Cuball{\Mcca B}\to\Mcca C\) so let \(P\) be the set of these
functions.
We define the algebraic operations on \(P\) pointwise: if %
\(f,g\in P\) then \(f+g\in P\) since %
\((f+g)(x)=f(x)+g(x)=\sum_{n=0}^\infty(\Hpan nf(x)+\Hpan ng(x))\) by
continuity of addition.
Notice that if \(f,g\in P\) then
\begin{align*}
  f\leq g\Equiv\forall n\in\Nat\quad \Linan nf\leq\Linan ng\,.
\end{align*}
since \(f\leq g\) means that %
\(\forall x\in\Cuball{\Mcca B}\ f(x)\leq g(x)\)
and \(\Absm{x\in\Cuball{\Mcca B}}{(g(x)-f(x))}\in P\).
Each map \(\Linan n:P\to\Mcca{\Mlsym nBC}\) is linear by %
Lemmas~\ref{lemma:unique-linearization}
and~\ref{lemma:unique-monomial}.

We set as usual
\begin{align*}
  \Norm f=\sup\Eset{\Norm{f(x)}\St x\in\Cuball{\Mcca B}}
\end{align*}
and define in that way a cone.
Let indeed \((f^k)_{k\in\Nat}\) be an increasing sequence in
\(\Cuball P\).
For each \(k,n\in\Nat\) we have %
\(\Norm{\Linan n{f^k}}\leq{n^n}\) by Lemma~\ref{lemma:Linhp-norm} and
the sequence %
\((\Linan n{f^k})_{k\in\Nat}\) is increasing and hence has a lub %
\(h_n\in\Mcca{\Mlsym nBC}\) and we have
\begin{align*}
  \forall\List x1n\in\Mcca B\quad h_n(\List x1n)
  =\sup_{k\in\Nat}\Linan n{f^k}(\List x1n)\,.
\end{align*}
In particular we can define the homogeneous polynomial map %
\(f_n=\Mlmon n(\frac 1{\Factor n}h_n)\), which means
\begin{align*}
  f_n(x)=\frac 1{\Factor n}h_n(\Rep xn)
  =\sup_{k\in\Nat}\frac 1{\Factor n}\Linan n{f^k}(\Rep xn)\,.
\end{align*}
Let \(f:\Cuball{\Mcca B}\to\Mcca C\) be defined by
\(f(x)=\sup_{k\in\Nat}f^k(x)\), we have
\begin{align*}
  f(x)
  &=\sup_{k\in\Nat}f^k(x)\\
  &=\sup_{k\in\Nat}\sum_{n=0}\frac 1{\Factor n}\Linan n{f^k}(\Rep xn)\\
  &=\sum_{n=0}\sup_{k\in\Nat}\frac 1{\Factor n}\Linan n{f^k}(\Rep xn)\\
  &=\sum_{n=0}^\infty f_n(x)
\end{align*}
which shows that \(f\) is analytic and is the lub \((f^k)_{k\in\Nat}\)
in \(\Cuball P\) since \(f\) is clearly measurable (as usual by the
monotone convergence theorem).

Then we define a family %
\(\cM=(\cM_X)_{X\in\ARCAT}\) of sets of measurability tests by
stipulating that \(p\in\cM_X\) if \(p=\Mtfun\beta m\) where %
\(\beta\in\Cuball{\Mcca{\Cpath XB}}\) and \(m\in\Mcms C_X\), and, if %
\(f\in P\) and \(r\in X\) then \(p(r,f)=m(r,f(\beta(r)))\).
It is easily checked that \((P,\cM)\) is a measurable cone, that we
denote as \(\Simpla BC\).

We check that \(\Simpla BC\) is integrable so let %
\(\eta\in\Mcca{\Cpath X{\Simpla BC}}\) for some \(X\in\ARCAT\) and
let \(\mu\in\Mcca{\Cmeas(X)}\).
We define a function \(f:\Cuball{\Mcca B}\to\Mcca C\) by
\begin{align*}
  \forall x\in\Cuball{\Mcca B}
  \quad f(x)=\int^C_X \eta(r)(x)\mu(dr)\,.
\end{align*}
This function is well defined because for each given %
\(x\in\Cuball{\Mcca B}\) one has %
\(\Absm{r\in X}{\eta(r)(x)}\in\Mcca{\Cpath XC}\). %
For each \(r\in X\) we can write %
\begin{align*}
  \eta(r)(x)=\sum_{n=0}^\infty
  \frac 1{\Factor n}\Linan n(\eta(r))(\Rep xn)
\end{align*}
and hence
\begin{align*}
  f(x)=\sum_{n=0}^\infty\frac1{\Factor n}
        \int^C_X\Linan n(\eta(r))(\Rep xn)\mu(dr)
      =\sum_{n=0}^\infty\frac1{\Factor n}
        \Big(\int^{\Mlsym nBC}_X\Linan n(\eta(r))\mu(dr)\Big)(\Rep xn)
\end{align*}
by definition of integrals in the integrable cone \(\Mlsym nBC\), and
hence \(f\in\Mcca{\Simpla BC}\).
Let \(p=(\Mtlfun xm)\in\Mcms{\Simpla BC}_\Measterm\) where
\(x\in\Mcca B\) and \(m\in\Mcms C_0\), we have
\begin{align*}
  p(f)
  &=m(f(x))\\
  &=m\Big(
    \sum_{n=0}^\infty \frac1{\Factor n}
    \int\Linan n(\eta(r))(\Rep xn)\mu(dr)\Big)\\
  &=\sum_{n=0}^\infty \frac1{\Factor n}
    m\Big(\int\Linan n(\eta(r))(\Rep xn)\mu(dr)\Big)
    \text{\quad by lin.~and cont.~of }m\\
  &=\sum_{n=0}^\infty\frac1{\Factor n}
    \int m(\Linan n(\eta(r))(\Rep xn))\mu(dr)
    \text{\quad by def.~of integrals in }C\\
  &=\int
    \Big(\sum_{n=0}^\infty \frac1{\Factor n}
    m(\Linan n(\eta(r))(\Rep xn))\Big)\mu(dr)
    \text{\quad by the monotone convergence th.}\\
  &=\int
    m\Big(\sum_{n=0}^\infty \frac1{\Factor n}
    \Linan n(\eta(r))(\Rep xn)\Big)\mu(dr)\\
  &=\int
    m(\eta(r)(x))\mu(dr)=\int p(\eta(r))\mu(dr)
\end{align*}
which shows that \(f=\int\eta(r)\mu(dr)\), and hence the measurable cone %
\(\Simpla BC\) is integrable.

\begin{lemma}
  \label{lemma:Hpan-integrable}
  For each \(n\in\Nat\), the function %
  \(\Hpan n:\Mcca{\Simpla BC}\to\Mcca{\Monomic nBC}\) is linear,
  continuous, measurable, integrable and has norm \(\leq 1\).
\end{lemma}
\begin{proof}
  Linearity and continuity result straightforwardly from the fact that
  the homogeneous polynomial decomposition %
  \((f_n=\Hpan n(f))_{n\in\Nat}\) of
  \(f\) is uniquely determined by its defining property:
  \begin{align*}
    \forall x\in\Cuball{\Mcca B}
    \quad
    f(x)=\sum_{n\in\Nat}f_n(x)\,.
  \end{align*}
  Let \(\eta\in\Mcca{\Cpath X{\Simpla BC}}\), we must check next
  that %
  \(\Hpan n\Comp\eta\in\Mcca{\Cpath X{\Monomic nBC}}\) so let %
  \(Y\in\ARCAT\),
  \(\beta\in\Mcca{\Cpath YB}\) %
  and \(m\in\Mcms C_Y\), we must prove that %
  \begin{align*}
    \theta
    &=\Absm{(s,r)\in Y\times X}{(\Mtpath\beta m)(s,\Hpan n(\eta(r)))}\\
    &=\Absm{(s,r)\in Y\times X}{m(s,\Hpan n(\eta(r))(\beta(s)))}
  \end{align*}
  is measurable \(Y\times X\to\Realp\).
  This results from the fact that
  \begin{align*}
    \theta(s,r)
    =\frac 1{\Factor n}\Evreal{\frac{d^n}{dt^n}m(s,\eta(r,t\beta(s)))}{t=0}
  \end{align*}
  and from the measurability and smoothness wrt.~\(t\) of the
  map %
  \((s,r,t)\mapsto m(s,\eta(r,t\phi(s)))\).
  Indeed the following is standard: if \(Z\) is a measurable space
  then if a function \(Z\times\Interco 01\to\Realp\) is measurable,
  and smooth in its second argument, then so is its derivative
  wrt.~its second argument.

  Last we check integrability of \(\Hpan n\) so let moreover %
  \(\mu\in\Cmeas(X)\), and let \(p\in\Mcms{\Monomic nBC}_\Measterm\)
  so that %
  \(p=\Mtfun xm\) for some \(x\in\Mcca B\) and
  \(m\in\Mcms C_\Measterm\), we have
  \begin{align*}
    p\Big(\Hpan n\Big(\int^{\Monomic nBC}_Y\eta(s)\mu(ds)\Big)\Big)
    &=\frac 1{\Factor n}
      \Evreal{\frac{d^n}{dt^n}m\Big(\int^C_Y\eta(s,tx)\mu(ds)\Big)}{t=0}\\
    &=\frac 1{\Factor n}
      \Evreal{\Big(\frac{d^n}{dt^n}\int_Y m(\eta(s,tx))\mu(ds)\Big)}{t=0}\\
    &=\frac 1{\Factor n}
      \int_Y\Evreal{\frac{d^n}{dt^n}m(\eta(s,tx))}{t=0}\mu(ds)\\
    &=\int_Y p(\Hpan n(\eta(s)))\mu(ds)
  \end{align*}
  by standard properties of integration.
  The fact that \(\Norm{\Hpan n}\leq 1\) results from the obvious fact
  that \(\Hpan nf(x)\leq f(x)\) for all \(x\in\Cuball{\Mcca B}\).
\end{proof}

\begin{theorem} %
  \label{th:linan-lin-morph}
  For all \(n\in\Nat\) we have %
  \(\Linan n\in\ICONES(\Simpla BC,n^n\Mlsym nBC)\).
\end{theorem}
\begin{proof}
  Remember that \(\Linan n=\Factor n(\Linhp n\Comp\Hpan n)\) and apply %
  Theorem~\ref{th:Hpol-icone} and Lemma~\ref{lemma:Hpan-integrable}.
\end{proof}

\begin{theorem}
  Each analytic function is stable and measurable.
\end{theorem}
\begin{proof}
  Immediate consequence of the definition of analytic functions and of
  Lemma~\ref{lemma:monomial-tot-mono}.
\end{proof}
\begin{remark}
  The converse is not true, as shown by
  Remark~\ref{rem:discrete-meas-comp}.
  Indeed since the stable and measurable function %
  \(\ContinuousPart\) introduced in
  Remark~\ref{rk:continuous-part-measure} is actually linear, if
  \(\ContinuousPart\) were analytic we would have
  \(\Linan 1\ContinuousPart=\ContinuousPart\) and
  \(\Linan n\ContinuousPart=0\) if \(n\not=1\) which is not possible
  since \(\ContinuousPart\) does not preserve integrals.
\end{remark}


\subsection{The category of integrable cones and analytic functions}
Our goal in this section is to show that integrable cones, together
with analytic functions as morphisms, form a category which is
cartesian closed.

\subsubsection{Composing analytic functions} %
We start with a special case of composition that we can think of as
the restriction of an analytic function to a \emph{local cone} in the
sense of Section~\ref{sec:local-cone}.
\begin{theorem} %
  \label{th:loc-ana-fun}
  Let \(x\in\Cuball{\Mcca B}\) and let \(f\in\Mcca{\Simpla BC}\). %
  Then the function \(g:\Cuball{\Mcca{\Cloc Bx}}\to\Mcca C\) defined
  by \(g(u)=f(x+u)\) is analytic, that is %
  \(g\in\Mcca{\Simpla{\Cloc Bx}C}\).
\end{theorem}
\begin{proof}
  Given \(u\in\Cuball{\Mcca{\Cloc Bx}}\) we have
  \begin{align*}
    g(u)
    &=f(x+u)\\
    &=\sum_{n=0}^\infty\frac 1{\Factor n}\Linan nf(\Rep{x+u}n)\\
    &=\sum_{n=0}^\infty\frac 1{\Factor n}
      \sum_{k=0}^n\Binom nk\Linan nf(\Rep u{n-k},\Rep xk)\\
    &=\sum_{n=0}^\infty\sum_{k=0}^n
      \frac 1{\Factor k\Factor{(n-k)}}\Linan nf(\Rep u{n-k},\Rep xk)\\
    &=\sum_{k=0}^\infty\sum_{n=k}^\infty
      \frac 1{\Factor k\Factor{(n-k)}}\Linan nf(\Rep u{n-k},\Rep xk)\\
    &=\sum_{k=0}^\infty\sum_{l=0}^\infty
      \frac 1{\Factor k\Factor l}\Linan {l+k}f(\Rep ul,\Rep xk)\\
    &=\sum_{l=0}^\infty\frac 1{\Factor l}
      \sum_{k=0}^\infty\frac 1{\Factor k}\Linan{l+k}f(\Rep ul,\Rep xk)
  \end{align*}
  so it suffices to show that for each \(l\in\Nat\) the function %
  \(g_l:\Cuball{\Mcca{\Cloc Bx}}\to\Mcca C\) defined by
  \begin{align*}
    g_l(u)=\sum_{k=0}^\infty\frac 1{\Factor k}\Linan{l+k}f(\Rep ul,\Rep xk)
  \end{align*}
  satisfies \(g_l(u)=\phi_l(\Rep ul)\) for some
  \(\phi_l\in\Mlsym l{\Cloc Bx}{C}\). We show that we can set
  \begin{align*}
    \phi_l(\Vect u)=\sum_{k=0}^\infty
    \frac 1{\Factor k}\Linan{l+k}f(\Vect u,\Rep xk)
  \end{align*}
  for all \(\Vect u=(\List u1l)\in\Mcca{\Cloc Bx}^l\). %
  So let \(\Vect u=(\List u1l)\in\Mcca{\Cloc Bx}^l\) and let %
  \(\lambda\geq\max_{i=1}^l\Norm{u_i}_{\Mcca{\Cloc Bx}}\) be such that
  \(\lambda>0\) so that for each \(i\) we have
  \(\frac 1\lambda u_i\in\Cuball{\Mcca{\Cloc Bx}}\).

  For \(N\in\Nat\) let %
  \(\phi_l^N(\Vect u)=\sum_{k=0}^N \frac 1{\Factor
    k}\Linan{l+k}f(\Vect u,\Rep xk)\) so that %
  \(\phi_l^N\in\Mcca{\Mlsym l{\Cloc Bx}{C}}\); actually we even have %
  \(\phi_l^N\in\Mcca{\Mlsym l{B}{C}}\). Observe that, setting %
  \(u=\frac 1{l\lambda}\sum_{i=1}^lu_i\in\Cuball{\Mcca{\Cloc Bx}}\) we
  have %
  \(u_i\leq l\lambda u\) for each \(i=1,\dots,l\) so that
  \begin{align*}
    \phi_l^N(\Vect u)
    &\leq \phi_l^N(\Rep{l\lambda u}l)
      =(l\lambda)^l\phi_l^N(\Rep{u}l)\\
    &\leq (l\lambda)^lg(u)=(l\lambda)^lf(x+u)
  \end{align*}
  so that %
  \(\Norm{\phi_l^N(\Vect u)}_C\leq (l\lambda)^l\Norm f\) and
  since neither \(l\) nor \(\lambda\) depend on \(N\) %
  the sequence \((\phi_l^N(\Vect u))_{N\in\Nat}\) is increasing in %
  \((l\lambda)^l\Cuball{\Mcca C}\), it has a lub which is %
  \(\phi_l(\Vect u)\) which is therefore well-defined and belongs to %
  \((l\lambda)^l\Cuball{\Mcca C}\).
  The fact that the map %
  \(\phi_l:\Mcca{\Cloc Bx}\to\Mcca C\) defined in that way is %
  \(l\)-linear symmetric and \(\omega\)-continuous results from the
  \(\omega\)-continuity of addition, scalar multiplication and from the
  basic properties of lubs.
  The measurability and integrability of \(\phi_l\) result as usual
  from the monotone convergence theorem.
\end{proof}

\noindent 
Let \(f\in\Cuball{\Mcca{\Simpla BC}}\) and %
\(g\in\Mcca{\Simpla CD}\), since %
\(g(\Cuball{\Mcca B})\subseteq\Cuball{\Mcca C}\), the function %
\(g\Comp f:\Cuball{\Mcca B}\to\Mcca D\) is well defined and bounded.
We assume first that \(f(0)=0\) so that the first term of the Taylor
expansion of \(f\) vanishes and we have
\begin{align*}
  g(f(x))
  &=\sum_{n=0}^\infty\frac 1{\Factor n}\Linan n g\Big(
    \Rep{\sum_{k=1}^\infty\frac 1{\Factor k}\Linan kf(\Rep xk)}n\Big)\\
  &=\sum_{n=0}^\infty\frac 1{\Factor n}\sum_{\sigma:\Intset n\to\Natnz}
    \frac {\Factor n}{\Factor\sigma}\Linan ng(\Linan{\sigma(1)}f(\Rep x{\sigma(1)}),\dots,
    \Linan{\sigma(n)}f(\Rep x{\sigma(n)}))
\end{align*}
by multilinearity and continuity of the \(\Linan nf\)'s, with the
notation \(\Factor\sigma=\prod_{i=1}^n\Factor{\sigma(i)}\).
If \(n,l\in\Nat\) we define \(\Intfun nl\) as the set of all %
\(\sigma:\Intset n=\Eset{1,\dots,n}\to\Natnz\) such that
\(\sum_{i=1}^n\sigma(i)=l\).
This set is finite and empty as soon as \(n>l\) (it is for obtaining
this effect that we have assumed that \(f(0)=0\)).
We have %
\begin{align*}
  g(f(x))
  &=\sum_{l=0}^\infty\frac 1{\Factor l}\sum_{n=0}^l
    \sum_{\sigma\in\Intfun nl}
    \frac{\Factor l}{\Factor\sigma}
    \Linan nf(\Linan{\sigma(1)}g(\Rep x{\sigma(1)}),\dots,
    \Linan{\sigma(n)}g(\Rep x{\sigma(n)}))\,.
\end{align*}
For each \(l\in\Nat\), the function
\begin{align*}
  h_l:\Mcca B^l&\to\Mcca D\\
  (\List x1l)&\mapsto
               \sum_{n=0}^l\sum_{\sigma\in\Intfun nl}
               \frac{\Factor l}{\Factor\sigma}
               \Linan lf(\Linan{\sigma(1)}g(x_1,\dots,x_{\sigma(1)}),\dots,
    \Linan{\sigma(n)}g(x_{l-\sigma(n)+1},\dots,x_l))
\end{align*}
is \(l\)-linear, measurable and integrable as a finite sum of such
functions, however it is not necessarily symmetric (for instance, for
\(l=4\), this sum contains the expression\\ %
\(\frac{\Factor 4}{(\Factor 2)^2}
\Linan 2f(\Linan 2g(x_1,x_2),\Linan 2g(x_3,x_4))\)), %
but not
\(\frac{\Factor 4}{(\Factor 2)^2}
\Linan 2f(\Linan 2g(x_1,x_3),\Linan 2g(x_2,x_4))\)), %
so we set
\begin{align*}
  k_l(\Vect x)
  =\frac 1{\Factor l}\sum_{\theta\in\Symgrp l}
  h_l(x_{\theta(1)},\dots,x_{\theta(l)})
\end{align*}
and \(k_l\) is again a finite sum of \(l\)-linear, measurable and
integrable functions and hence obviously belongs to
\(\Mcca{\Mlsym lBD}\), and we have
\begin{align*}
  g(f(x))=\sum_{l=0}^\infty\frac 1{\Factor l}
  k_l(\Rep xl)
\end{align*}
for each \(x\in\Cuball{\Mcca B}\) which proves that %
\(g\Comp f\) is analytic since this function is obviously bounded.

Now we don't assume anymore that \(f(0)=0\), and we define an
obviously analytic function
\(f_0\in\Cuball{\Mcca B}\to\Cuball{\Mcca{\Cloc C{f(0)}}}\) by
\(f_0(x)=f(x)-f(0)\). By Lemma~\ref{th:loc-ana-fun} the function %
\(g_0:\Cuball{\Mcca{\Cloc C{f(0)}}}\to\Mcca D\) given by %
\(g_0(v)=g(f(0)+v)\) is analytic and hence %
\(g\Comp f=g_0\Comp f_0\) is analytic since \(f_0(0)=0\).
The measurability of \(g\Comp f\) is obvious so
\(g\Comp f\in\Mcca{\Simpla BD}\).

This shows that we have defined a category %
\(\ACONES\) whose objects are the integrable cones and where a
morphism from \(B\) to \(C\) is a \(f\in\Mcca{\Simpla BC}\) such
that \(\Norm f\leq 1\).
We aim now at proving that this category is cartesian closed.

\begin{lemma}
  For all measurable cones \(B,C\) we have
  \(\ICONES(B,C)\subseteq\ACONES(B,C)\).
\end{lemma}
\noindent 
This is obvious and shows that there is a forgetful faithful functor
\(\Derfuna:\ICONES\to\ACONES\) which acts as the identity on objects
and morphisms.

\begin{proposition}
  \label{prop:acones-cartesian}
  The category \(\ACONES\) has all (small) products.
\end{proposition}
\begin{proof}
  We already know that each family \((B_i)_{i\in I}\) of integrable
  cones has a product \(B=\Bwith_{i\in I}B_i\) in \(\ICONES\) with
  projections \((\Proj i\in\ICONES(B,B_i))_{i\in I}\).
  We show that \(B\) is also the product of the family
  \((B_i)_{i\in I}\) with projections \((\Derfuna(\Proj i))_{i\in I}\) 
  in \(\ACONES\). Remember that an element of \(\Mcca B\) is a family %
  \((x_i\in\Mcca{B_i})_{i\in I}\) such that the family %
  \((\Norm{x_i}_{B_i})_{i\in I}\) is bounded in \(\Realp\).

  So let \((f_i\in\ACONES(C,B_i))_{i\in I}\), it suffices to prove
  that the function %
  \(f:\Cuball{\Mcca C}\to\Mcca B\) given by %
  \(f(y)=(f_i(y))_{i\in I}\) belongs to %
  \(\ACONES(C,B)\).
  The fact that %
  \(\forall y\in\Cuball{\Mcca C}\ f(y)\in\Cuball{\Mcca B}\) results
  from the definition of the norm of \(B\) and from the fact that
  \(\forall i\in I\ \Norm{f_i}\leq 1\).
  The measurability of \(f\) results trivially from its definition and
  from the definition of \(\Mcms B\).
  We know that %
  \(f_i(y)=\sum_{n=0}^\infty\frac 1{\Factor n}\Linan n{f_i}(\Rep
  yn)\). For each \(n\in\Nat\) the map %
  \(\phi_n:\Mcca C^n\to\Mcca B\) defined by %
  \(\phi_n(\Vect y)=(\Linan n{f_i}(\Vect y))_{i\in I}\) belongs to %
  \(\Mcca{\Mlsym nCB}\) since we know that %
  \(\Norm{\Linan nf}\leq n^n\) by
  Theorem~\ref{th:linan-lin-morph}.
  It follows that \(f\) is analytic since %
  \(f(y)=\sum_{n=0}^\infty\frac 1{\Factor n}\phi_n(\Rep yn)\).
\end{proof}

\begin{theorem}
  The category \(\ACONES\) is cartesian closed.
\end{theorem}
\begin{proof}
  We already know that \(\Simpla BC\) is an integrable cone and we
  have an obvious function
\begin{align*}
  \Ev:\Mcca{\With{(\Simpla BC)}{B}}={(\Mcca{\Simpla BC})}\times{\Mcca B}
  &\to\Mcca C\\
  (f,x)&\mapsto f(x)
\end{align*}
which satisfies \(\Norm\Ev\leq 1\), we show that it is measurable.
Let
\(\theta\in\Mcca{\Cpath{X}{\With{({\Simpla BC})}{B}}}\) %
for some \(X\in\ARCAT\), so that \(\theta=\Tuple{\eta,\beta}\) %
where %
\(\eta\in\Mcca{\Cpath X{\Simpla BC}}\) and %
\(\beta\in\Mcca{\Cpath XB}\), we must prove that %
\(\Ev\Comp\Tuple{\eta,\beta}\in\Mcca{\Cpath XC}\) so let %
\(m\in\Mcms C_Y\) for some \(Y\in\ARCAT\), %
we must prove that the function
\(
  \phi=\Absm{(s,r)\in Y\times X}{m(s,\eta(r)(\beta(r)))}
  :Y\times X\to\Realp
\)
is measurable.
We build %
\(p=\Mtfun{(\beta\Comp{\Proj 1})}{(m\Comp\Proj 2)}
\in\Mcms{\Simpla BC}_{X\times Y}\)
and since \(\eta\in\Mcca{\Cpath d{\Simpla BC}}\) the map %
\begin{align*}
  \psi&=\Absm{(r_1,s,r_2)\in X\times Y\times X}{p(r_1,s,\eta(r_2))}\\
      &=\Absm{(r_1,s,r_2)\in X\times Y\times X}{m(s,\eta(r_2)(\beta(r_1)))}
\end{align*}
is measurable which shows that %
\(\phi=\Absm{(s,r)\in Y\times X}{\psi(r,s,r)}\)
is measurable.

We prove that \(\Ev\) is analytic.
We have
\begin{align*}
  \Ev(f,x)&=f(x)\\
          &=\sum_{n=0}^\infty\frac 1{\Factor n}\Linan nf(\Rep x n)\\
          &=\sum_{n=0}^\infty\frac 1{\Factor n}\phi_n(\Rep{(f,x)}{n+1})\\
          &=\sum_{n=0}^\infty\frac 1{\Factor{(n+1)}}(n+1)\phi_n(\Rep{(f,x)}{n+1})
\end{align*}
where %
\(\phi_n:(\Mcca{\With{(\Simpla BC)}{B}})^{n+1}\to\Mcca C\) is given by %
\begin{align*}
  \phi_n((f_1,x_1),\dots,(f_{n+1},x_{n+1}))
  =\frac 1{n+1}\sum_{i=1}^{n+1}
  \Linan n{f_i}(x_1,\dots,x_{i-1},x_{i+1},\dots,x_{n+1})
\end{align*}
and therefore belongs to %
\(\Mcca{\Mlsym{n+1}BC}\); the measurability of \(\phi_n\) follows from
that of \(\Linan nf\).
If follows that \(\Ev\) is analytic, with
\begin{align*}
  \Linan 0\Ev()
  &=0\\
  \Linan{n+1}\Ev((f_1,x_1),\dots,(f_{n+1},x_{n+1}))
  &=\sum_{i=1}^{n+1}\Linan n{f_i}(x_1,\dots,x_{i-1},x_{i+1},\dots,x_{n+1})\,.
\end{align*}
\noindent 
Now we deal with the Curry transpose of analytic functions.
So let \(D\) be an integrable cone and let
\(f\in\ACONES(\With DB,C)\).
Given \(z\in\Cuball{\Mcca D}\) let %
\(f_z:\Cuball{\Mcca B}\to\Mcca C\) be given by %
\(f_z(x)=f(z,x)\).
We know that \(f_z\in\Mcca{\Simpla BC}\) by %
Theorem~\ref{th:loc-ana-fun} applied at
\((z,0)\in\Cuball{\Mcca{\With DB}}\) and by precomposing the obtained
``local'' analytic function \(g:\Cloc{\Withp DB}{(z,0)}\to C\) defined
by \(g(w,y)=f(z+w,y)\) with the obviously analytic function
\(x\mapsto (0,x)\): this composition of functions coincides with
\(f_z\).

We are left with proving that the function
\(g:\Cuball{\Mcca D}\to\Mcca{\Simpla BC}\) defined by %
\(g(z)=f_z\) belongs to \(\ACONES(D,\Simpla BC)\).
It is obvious that \(\Norm g\leq 1\) so let us check that \(g\) is
measurable.
Let \(\delta\in\Mcca{\Cpath XD}\), we must prove that %
\(g\Comp\delta\in\Mcca{\Cpath X{\Simpla BC}}\) so let %
\(Y\in\ARCAT\) and \(p\in\Mcms{\Simpla BC}_Y\), we must prove that
the function
\(
  \phi=\Absm{(s,r)\in Y\times X}{p(s,g(\delta(r)))}:Y\times X\to\Realp
\)
is measurable.
We have \(p=\Mtfun\beta m\) where \(\beta\in\Mcca{\Cpath YB}\) and %
\(m\in\Mcms D_Y\) so that
\(
  \phi=\Absm{(s,r)\in Y\times X}{m(s,g(\delta(r))(\beta(s)))}
  =\Absm{(s,r)\in Y\times X}{m(s,f(\delta(r),\beta(s)))}
  \)
is measurable because \(f\) is measurable and %
\(\Absm{(r,s)\in Y\times X}{(\delta(r),\beta(s))}
\in\Mcca{\Cpath{Y\times X}{\With DB}}\).
We are left with proving that \(g\) is analytic.
For \(z\in\Cuball{\Mcca D}\) we have
\begin{align*}
  g(z)
  &=\Absm{x\in\Cuball{\Mcca B}}{f(z,x)}\\
  &=\Absm{x\in\Cuball{\Mcca B}}
    {\sum_{n=0}^\infty\frac 1{\Factor n}\Linan nf(\Rep{(z,x)}n)}\\
  &=\sum_{n=0}^\infty\frac 1{\Factor n}
    \Absm{x\in\Cuball{\Mcca B}}{\Linan nf(\Rep{(z,0)+(0,x)}n)}\\
  &=\sum_{n=0}^\infty\frac 1{\Factor n}
    \Absm{x\in\Cuball{\Mcca B}}
    {\sum_{k=0}^n\Binom{n}{k}\Linan nf(\Rep{(z,0)}k,\Rep{(0,x)}{n-k})}\\
  &=\sum_{n=0}^\infty\sum_{k=0}^n
    \frac 1{\Factor k\Factor{(n-k)}}
    {\Absm{x\in\Cuball{\Mcca B}}{\Linan nf(\Rep{(z,0)}k,\Rep{(0,x)}{n-k})}}\\
  &=\sum_{k=0}^\infty
    \frac 1{\Factor k}\sum_{l=0}^\infty\frac 1{\Factor l}
    {\Absm{x\in\Cuball{\Mcca B}}{\Linan{k+l}f(\Rep{(z,0)}k,\Rep{(0,x)}{l})}}
    =\sum_{k=0}^\infty\frac 1{\Factor k}h_k(\Rep zk)
\end{align*}
where
\begin{align*}
  h_k(\List z1k)=\sum_{l=0}^\infty\frac 1{\Factor l}
  {\Absm{x\in\Cuball{\Mcca B}}
  {\Linan{k+l}f((z_1,0),\dots,(z_k,0),\Rep{(0,x)}{l})}}
\end{align*}
is well defined for all \(\List z1k\in\Mcca D\).
Indeed, as usual it suffices to take some \(\lambda>0\) such that
\(\lambda\geq\Norm{z_i}_D\) for \(i=1,\dots,k\) and observe that for
all \(N\in\Nat\) one has, setting \(z=\sum_{i=1}^kz_i\) so that
\(\frac 1{k\lambda}z\in\Cuball{\Mcca D}\),
\begin{align*}
  \sum_{l=0}^N\frac 1{\Factor l}
  {\Absm{x\in\Cuball{\Mcca B}}
  {\Linan {k+l}f((z_1,0),\dots,(z_k,0),\Rep{(0,x)}{l})}}
  &\leq h_k(\Rep zk)\\
  &=(k\lambda)^kh_k(\frac1{k\lambda}\Rep zk)\\
  &\leq \Factor k(k\lambda)^kg(\frac 1{k\lambda}z)\,.
\end{align*}
The map \(h_k\) is multilinear by \(\omega\)-continuity of the algebraic
operations in each cone, it is obviously symmetric by the symmetry of
the \(\Linan nf\).
Its \(\omega\)-continuity follows from that of the \(\Linan nf\) and from
commutations of lubs. Last, measurability and integrability follow as
usual from the monotone convergence theorem.
So we have \(h_k\in\Mlsym nD{\Simpla BC}\) and this shows that \(g\)
is analytic.

To prove that \(\ACONES\) is cartesian closed it suffices to prove
that \(g\) is the unique morphism in \(\ACONES(D,\Simpla BC)\) such that
\begin{align*}
  \Ev\Comp\Withp{g}{\Id_B}=f
\end{align*}
which results straightforwardly from the fact that \(\Ev\) is defined
exactly as in \(\SET\).
\end{proof}

\begin{theorem}
  \label{th:derfuna-preserves-limits}
  The functor %
  \(\Derfuna:\ICONES\to\ACONES\) preserves all limits.
\end{theorem}
\begin{proof}
  Preservation of categorical products resulting easily from the
  construction of products in \(\ACONES\)
  (Proposition~\ref{prop:acones-cartesian}) and in \(\ICONES\)
  (Theorem~\ref{th:mcones-complete}), let us deal with equalizers.
  So let \(f,g\in\ICONES(B,C)\) and let \((E,e)\) be their equalizer
  in \(\ICONES\): remember that
  \(\Mcca E=\{x\in\Mcca B\St f(x)=g(x)\}\) and that
  \(e\in\ICONES(E,B)\) is the obvious injection.
  Let \(h\in\ACONES(D,B)\) be such that \(f\Comp h=g\Comp h\).
  This means that \(h(\Cuball{\Mcca D})\subseteq\Cuball{\Mcca E}\).
  We know that
  \begin{align*}
    h(z)=\sum_{n=0}^\infty\frac 1{\Factor n}h_n(z)
  \end{align*}
  where \(h_n\in\Monomic nDB\) is fully characterized by
  \begin{align}
    \label{eq:ana-hom-dec-charact}
    \forall x'\in\Mcca{\Cdual B}
    \quad x'(h_n(z))=\Evreal{\frac{d^n}{dt^n}x'(h(tz))}{t=0}\,,
  \end{align}
  see the proof of Lemma~\ref{lemma:unique-monomial}.
  We contend that \(f\Comp h_n=g\Comp h_n\) so let %
  \(z\in\Cuball{\Mcca D}\) and let \(p\in\Mcms C_\Measterm\), we
  have %
  \begin{align*}
    p(f(h_n(z)))=\Evreal{\frac{d^n}{dt^n}p(f(h(tz))))}{t=0}\,,
  \end{align*}
  by Equation~\Eqref{eq:ana-hom-dec-charact} applied with %
  \(x'=p\Compl f\in\Mcca{\Cdual B}\) and hence %
  \(p(f(h_n(z)))=p(g(h_n(z)))\) which proves our contention by
  \Mssepr.
  This shows that %
  \(h_n(\Cuball{\Mcca D})\subseteq\Cuball{\Mcca E}\).
  Therefore since the operator %
  \(\Linhp n\) is defined in terms of addition, subtraction and
  multiplication by non-negative real numbers we have %
  \(\Linhp n{h_n}\in\Mcca{\Mlsym nBE}\) --~measurability and
  integrability follow from the fact that measurability tests and
  integrals in \(E\) are defined as in \(B\).
  Finally this shows that %
  \(\Linan nh\in\Mcca{\Mlsym nBE}\) and hence %
  \(h\in\ACONES(D,E)\) which shows that \((E,e)\) is the equalizer of
  \(f,g\) in \(\ACONES\).
\end{proof}

\section{The linear-non-linear adjunction, in the stable and
  analytic cases} %
\label{sec:lin-nonlin-adj}

From now on we use \(\cC\) to denote one of the two cartesian closed
categories \(\SCONES\) and \(\ACONES\), which are both locally small.
In both cases we use \(\Derfun\) to denote the functor
\(\ICONES\to\cC\) (which was denoted by \(\Derfuns\) when
\(\cC=\SCONES\) and by \(\Derfuna\) when \(\cC=\ACONES\)).

Remember that \(\Derfun\) preserves all limits, see
Theorems~\ref{th:derfuns-preserves-limits}
and~\ref{th:derfuna-preserves-limits}.

\renewcommand\Derfuna{\Derfun}
\renewcommand\Derfuns{\Derfun}
\renewcommand\SCONES{\cC}
\renewcommand\STAB{\cC}
\renewcommand\ACONES{\cC}
\renewcommand\Eana{\mathsf{E}}
\renewcommand\Estab{\mathsf{E}}
\renewcommand\Expadjs{\Theta}
\renewcommand\Expadja{\Theta}
\renewcommand\Excls{\oc}
\renewcommand\Excla{\oc}
\renewcommand\Exclls{\oc\oc}
\renewcommand\Ders{\mathsf{der}}
\renewcommand\Diggs{\mathsf{dig}}
\renewcommand\Unistab{\mathsf{nl}}
\renewcommand\Proms[1]{{#1}^{\oc}}
\renewcommand\Promms[1]{{#1}^{\oc\oc}}
\renewcommand\Proma[1]{{#1}^{\oc}}
\renewcommand\Promma[1]{{#1}^{\oc\oc}}
\renewcommand\Simpls[2]{#1\Rightarrow#2}
\renewcommand\Coalga{\Coalgs}
\noindent 
For that reason the two categories \(\ICONES\) and \(\SCONES\) can be
related by a linear-non-linear adjunction in the sense
of~\cite{Mellies09}, and hence form a categorical model of
Intuitionistic \(\LL\).
We describe directly the associated Seely category.

Let \(\Estab:\STAB\to\ICONES\) be the left adjoint of \(\Derfuns\),
which exists by Theorem~\ref{th:Icones-adjoint-functor}, and let us
introduce the notation %
\(\Expadjs_{B,C}:\ICONES(\Estab B,C)\to\STAB(B,\Derfuns
C)=\STAB(B,C)\) for the associated natural bijection (remember that
\(\Derfuns C=C\)).

\begin{remark}
  Just as for the tensor product (see
  Remark~\ref{rk:tensor-concrete-pres}), we have no concrete
  description of the \(\Estab\) functor for the time being.
\end{remark}

We use %
\((\Excls,\Ders,\Diggs)\) for the induced comonad on \(\ICONES\) whose
Kleisli category is (equivalent to) \(\STAB\) since
\begin{align*}
  \Kleisli{\ICONES}{\Excls}(B,C)
  &=\ICONES(\Excls B,C)\\
  &=\ICONES(\Estab\Compl\Derfuns B,C)\\
  &\Isom\STAB(\Derfuns B,\Derfuns C)\\
  &=\STAB(B,C)\,.
\end{align*}
Notice that actually \(\Excls B=\Estab B\) and similarly for
morphisms.
The notation \(\Der{}\) for the counit of this comonad comes from the
dereliction rule of \(\LL\), and the notation \(\Digg{}\) for the
comultiplication comes from the \(\LL\) digging derived rule.

Let
\(\Unistab_B=\Expadjs_{B,\Estab B}(\Id_{\Estab B})\in\STAB(B,\Excls
B)\) be the unit of the adjunction,
which is the ``universal nonlinear map'' on \(B\)
in the sense that for each %
integrable cone \(C\) and each \(f\in\STAB(B,C)\) one has %
\(f=\phi\Comp\Unistab_B\) for a unique \(\phi\in\ICONES(\Excls B,C)\),
namely \(\phi=\Invp{\Expadjs_{B,C}}(f)\) (dropping the \(\Derfuns\)
symbol since this functor acts as the identity on objects and
morphisms considered as functions).
So that for \(h\in\ICONES(\Excls B,C)\), one has %
\[
  \Expadjs_{B,C}(h)=h\Comp\Unistab_B\,.
\]
For each \(x\in\Cuball{\Mcca B}\) we set %
\(\Proms x=\Unistab_B(x)\in\Cuball{\Mcca{\Estab B}}\) so that, for %
\(f\in\STAB(B,C)\) we have \(f(x)=\Invp{\Expadjs_{B,C}}(f)(\Proms x)\).

The next lemma is similar to Proposition~\ref{prop:fun-ttree-charact}.
\begin{lemma}
  \label{lemma:tens-excl-equal-charact}
  Let \(n\geq 1\), let \(\List B1n,C\) be objects of \(\ICONES\) and %
  \(f\) and \(g\) be elements of
  \(\ICONES(\Excls{B_1}\ITens\cdots\ITens\Excls{B_n},C)\) such that
  \(f(\Proms{x_1}\ITens\cdots\ITens\Proms{x_n})
  =g(\Proms{x_1}\ITens\cdots\ITens\Proms{x_n})\) for all
  \((x_i\in\Cuball{\Mcca{B_i}})_{i=1}^n\).
  Then \(f=g\).
\end{lemma}
\begin{proof}
  By induction on \(n\).
  For \(n=1\) this comes from the fact that
  \(f(x)=\Inv{(\Expadjs_{B_1,C})}f(\Proms{x})\) for all
  \(x\in\Cuball{\Mcca B}\) so that our assumption entails
  \(f=g\).

  For \(n>1\), we have
  \(\Curlin(f),\Curlin(g)
  \in\ICONES(\Excls{B_1},\Limpl{\Excls{B_2}\ITens\cdots\ITens\Excls{B_n}}C)\).
  For each \(x_1\in\Cuball{\Mcca{B_1}}\), the two functions %
  \(\Curlin(f)(\Proms{x_1}),\Curlin(g)(\Proms{x_1})
  \in\ICONES(\Excls{B_2}\ITens\cdots\ITens\Excls{B_n},C)\) satisfy
  \begin{align*}
    \forall x_2\in\Cuball{\Mcca{B_2}},\dots,x_n\in\Cuball{\Mcca{B_n}}
    \quad
    \Curlin(f)(\Proms{x_1})(\Proms{x_2}\ITens\cdots\ITens\Proms{x_n})
    =
    \Curlin(g)(\Proms{x_1})(\Proms{x_2}\ITens\cdots\ITens\Proms{x_n})
  \end{align*}
  and hence \(\Curlin(f)(\Proms{x_1})=\Curlin(g)(\Proms{x_1})\) by
  inductive hypothesis.
  As in the base case we get \(f=g\).
\end{proof}

The counit \(\Ders_B\in\ICONES(\Excls B,B)\) of the comonad
\(\Excls\_\) is also the counit of the adjunction.
It satisfies therefore %
\begin{align*}
  \forall x\in\Cuball{\Mcca B}\quad \Ders_B(\Proms x)=x\,.
\end{align*}
The comultiplication %
\(\Diggs_B\in\ICONES(\Excls B,\Excls{\Excls B})
=\ICONES(\Estab\Derfun B,\Estab\Derfun\Estab\Derfun B)\) is defined by
\(
\Diggs_B=\Estab(\Unistab_B)
\)
so that we have
\begin{align*}
  \forall x\in\Cuball{\Mcca B}\quad \Diggs_B(\Proms x)=\Promms x\,.
\end{align*}
since by naturality of \(\Unistab_B\) we have %
\(\Derfun(\Estab(\Unistab_B))\Comp\Unistab_B
=\Unistab_{\Derfun(\Estab B)}\Comp\Unistab_B\) in \(\STAB\).

\begin{lemma} %
  \label{lemma:excl-fun-prom}
  Let \(f\in\ICONES(B,C)\) and \(x\in\Cuball{\Mcca B}\).
  We have \((\Excls f)(\Proms x)=\Proms{f(x)}\).
\end{lemma}
\begin{proof}
  We have %
  \((\Excls f)(\Proms x)
  =(\Estab(\Derfun f))(\Proms x)
  =((\Derfun(\Estab(\Derfun f)))\Comp\Unistab_{\Derfun B})(x)\)
  where the composition is taken in \(\STAB\).
  By naturality we get %
  \((\Excls f)(\Proms x) =(\Unistab_{\Derfun C}\Comp\Derfun f)(x)
  =\Proms{f(x)}\).
\end{proof}

\noindent 
Consider the two functors %
\(L,R:\Op\ICONES\times\Op\ICONES\times\ICONES\to\ICONES\) defined on
objects by %
\(L(B,C,D)=(\Simpls B{(\Limpl CD))}\) and %
\(R(B,C,D)=(\Limpl C{(\Simpls BD))}\), and similarly on morphisms.

\begin{lemma} %
  \label{lemma:stab-lin-swap}
  Let \(B,C,D\) be integrable cones. There is an isomorphism in
  \(\ICONES\) from %
  \(L(B,C,D)=(\Simpls B{(\Limpl CD)})\) to %
  \(R(B,C,D)=(\Limpl C{(\Simpls BD))}\) %
  which is natural in \(B\), \(C\) and \(D\).
\end{lemma}
\begin{proof}[Proof sketch]
  This needs a separate proof in each case \(\cC=\STA\) and
  \(\cC=\ANA\), which follows a pattern that we have seen many
  times.
  The natural isomorphism maps %
  \(f\in\Mcca{\Simpls B{(\Limpl CD)}}\) to %
  \(\Absm{y\in\Mcca C}{\Absm{x\in\Cuball{\Mcca B}}{f(x,y)}}\).
\end{proof}
\noindent 
Then we have
\begin{align*}
  \ICONES(\Excls(\With{B_1}{B_2}),C)
  &\Isom\STAB(\With{B_1}{B_2},C)\\
  &\Isom\STAB(B_1,\Simpls{B_2}{C})\\
  &\Isom\ICONES(\Excls{B_1},\Simpls{B_2}{C})\\
  &\Isom\ICONES(\Sone,\Limpl{\Excls{B_1}}{(\Simpls{B_2}{C})})\\
  &\Isom\ICONES(\Sone,\Simpls{B_2}{\Limplp{\Excls{B_1}}{C}})
  \text{ by Lemma~\ref{lemma:stab-lin-swap}}\\
  &\Isom\STAB(\Stop,\Simpls{B_2}{\Limplp{\Excls{B_1}}{C}})\\
  &\Isom\STAB(B_2,\Limplp{\Excls{B_1}}{C})\\
  &\Isom\ICONES(\Excls{B_2},\Limplp{\Excls{B_1}}{C})\\
  &\Isom\ICONES(\Tens{\Excls{B_2}}{\Excls{B_1}},C)\\
  &\Isom\ICONES(\Tens{\Excls{B_1}}{\Excls{B_2}},C)
\end{align*}
by a sequence of natural bijections and hence by
Lemma~\ref{lemma:functor-yoneda-iso} we have a natural isomorphism
\(\Seelyt_{B_1,B_2}\) in
\[
  \ICONES(\Tens{\Excls{B_1}}{\Excls{B_2}},\Excls(\With{B_1}{B_2}))
\]
which satisfies
\begin{align*}
  \Seelyt_{B_1,B_2}(\Tens{\Proms{x_1}}{\Proms{x_2}})
=\Proms{\Tuple{x_1,x_2}}\,.
\end{align*}
Notice that this equation fully characterizes \(\Seelyt_{B_1,B_2}\) by
Lemma~\ref{lemma:tens-excl-equal-charact}.

Similarly we define an iso \(\Seelyz\in\ICONES(\Sone,\Excls\Stop)\)
which is such that \(\Seelyz(t)=t\,\Proms 0\) for all \(t\in\Realp\).
Then one can prove using %
Lemma~\ref{lemma:functor-yoneda-iso} again that \(\Excls\) is a strong
monoidal comonad.

\begin{theorem}
  Equipped with the strong monoidal comonad \(\Excls\), the category %
  \(\ICONES\) is a Seely category in the sense of~\cite{Mellies09}.
\end{theorem}
\begin{proof}
  Using
  Lemma~\ref{lemma:tens-excl-equal-charact} %
  it is easy to prove the remaining properties, which regard
  \(\Excls\_\) and its associated morphisms.
  As an example let us prove that the following diagram commutes.
  \begin{center}
    \begin{tikzcd}
      \Tens{\Excls{B_1}}{\Excls{B_2}}
      \ar[r,"\Seelyt_{B_1,B_2}"]
      \ar[dd,swap,"\Tens{\Diggs_{B_1}}{\Diggs_{B_2}}"]
      &[2em]
      \Excls{(\With{B_1}{B_2})}
      \ar[d,"\Diggs_{\With{B_1}{B_2}}"]
      \\
      &
      \Exclls{(\With{B_1}{B_2})}
      \ar[d,"\Excls{\Tuple{\Excls{\Proj1},\Excls{\Proj2}}}"]\\
      \Tens{\Exclls{B_1}}{\Exclls{B_2}}
      \ar[r,"\Seelyt_{\Excls{B_1},\Excls{B_2}}"]
      &
      \Excls{(\With{\Excls{B_1}}{\Excls{B_2}})}
    \end{tikzcd}
  \end{center}
  Given \((x_i\in\Cuball{\Mcca{B_i}})_{i=1,2}\), we have
  \begin{align*}
    \Seelyt_{\Excls{B_1},\Excls{B_2}}
    (\Tensp{\Diggs_{B_1}}{\Diggs_{B_2}}
    (\Tens{\Proms{x_1}}{\Proms{x_2}}))
    &=\Seelyt_{\Excls{B_1},\Excls{B_2}}
      (\Tens{\Promms{x_1}}{\Promms{x_2}})\\
    &=\Proms{\Tuple{\Proms{x_1},\Proms{x_2}}}
  \end{align*}
  and
  \begin{align*}
    \Excls{\Tuple{\Excls{\Proj1},\Excls{\Proj2}}}
    (\Diggs_{\With{B_1}{B_2}}
    (\Seelyt_{B_1,B_2}
    (\Tens{\Proms{x_1}}{\Proms{x_2}})))
    &=\Excls{\Tuple{\Excls{\Proj1},\Excls{\Proj2}}}
    (\Diggs_{\With{B_1}{B_2}}
    (\Proms{\Tuple{x_1,x_2}}))\\
    &=\Excls{\Tuple{\Excls{\Proj1},\Excls{\Proj2}}}
      (\Promms{\Tuple{x_1,x_2}})\\
    &=\Proms{(\Tuple{\Excls{\Proj1},\Excls{\Proj2}}
      (\Proms{\Tuple{x_1,x_2}}))}\\
    &=\Proms{\Tuple{
      (\Excls{\Proj1})(\Proms{\Tuple{x_1,x_2}}),
      (\Excls{\Proj2})(\Proms{\Tuple{x_1,x_2}})}}\\
    &=\Proms{\Tuple{\Proms{x_1},\Proms{x_2}}}
      \qedhere
  \end{align*}
\end{proof}

\begin{remark}
  \newcommand\SCONESV{\mathbf{SCones}}
  \label{rem:discrete-meas-comp}
  Assume that \(\cC=\SCONESV\) and that, as in Remark
  \ref{rmk:ar-main-example} and Section \ref{sec:icones-qbs},
  \(\ARCAT\) is the category whose only objects are \(\Real\) and
  \(\Measterm\) (the one element measurable space), and all measurable
  functions as morphisms.
  Then the underlying set of \(\Mcca{\Cmeas(\Real)}\) is the set of all
  finite measures on \(\Real\).
  
  Consider, as in
  Remark~\ref{rk:continuous-part-measure}, the map
  \(\ContinuousPart:\Mcca{\Cmeas(\Real)}\to\Mcca{\Cmeas(\Real)}\) which
  extracts the continuous part of each measure on \(\Real\).
  In addition to being linear and \(\omega\)-continuous, this map is
  also measurable (because the map
  \(\mu \mapsto \frac{d\mu}{d\lambda}\) is measurable~\cite[Theorem
  1.28]{kallenberg17}).
  However, as noted in %
  Remark~\ref{rk:continuous-part-measure}, this map does not commute
  with integrals and is therefore not a morphism in \(\ICONES\).
  
  Nevertheless, \(\ContinuousPart\) is a morphism in \(\SCONESV\), so
  there exists
  \(f\in\ICONES(\Excls{\Cmeas(\Real)},\allowbreak\Cmeas(\Real))\) such
  that %
  \(\ContinuousPart(\mu)=f(\Proms\mu)\).
  It would be interesting to understand how this function \(f\) works
  to get some insight on the internal structure of the stable
  exponential, which is defined in a rather implicit way (by the
  special adjoint functor theorem).
\end{remark}

\subsection{The coalgebra structure of \(\Cmeas(X)\)}
\label{sec:stable-exp-meas-coalg}
In this section we derive additional consequences of the integrability
condition on cones and linear morphisms.
First, we show that for each \(X\in\ARCAT\), the integrable cone
\(\Cmeas(X)\) has a structure of \(\oc\)-coalgebra.
The definition of this coalgebra structure strongly uses the fact that
\(\Excl{\Cmeas(X)}\) is an integrable cone, that is, all measurable
paths valued in that cone have an integral wrt.~each subprobability
measure on the measurable space (belonging to \(\ARCAT\)) where it is
defined.
This is typically a property which was not available
in~\cite{EhrhardPaganiTasson18}.

Let \(X\in\ARCAT\).
In Section~\ref{sec:cat-prop-integ} we defined
the Dirac path \(\Dirac X\in\Cuball{\Mcca{\Cpath X{\Cmeas(X)}}}\)
which maps \(r\in X\) to \(\Dirac X(r)\), the Dirac measure at
\(r\).
Since morphisms in \(\STAB\) are measurable we have
\begin{align*}
  \Unistab_{\Cmeas(X)}\Comp\Dirac X
  \in\Cuball{\Mcca{\Cpath X{\Excls{\Cmeas(X)}}}}
\end{align*}
and we define
\begin{align*}
  \Coalgs X=\Mcint{\Excls{\Cmeas(X)}}_X(\Unistab_{\Cmeas(X)}\Comp\Dirac X)
  \in\ICONES(\Cmeas(X),\Excls{\Cmeas(X)})
\end{align*}
using Theorem~\ref{th:meas-path-equiv}.
In other words \(\Coalgs X\) is defined by
\begin{align*}
  \Coalgs X(\mu)=\int_{r\in X}\Proms{\Dirac X(r)}\mu(dr)
\end{align*}
and satisfies \(\Coalgs X(\Dirac X(r))=\Proms{\Dirac X(r)}\).

\begin{theorem}
  \label{th:meas-cone-coalgebra-stab}
  Equipped with \(\Coalgs X\), the object \(\Cmeas(X)\) of \(\ICONES\)
  is a coalgebra of the comonad \(\Excls\_\).
  Moreover for each \(\phi\in\ARCAT(X,Y)\), we have
  \[
    \Cmeas(\phi)=\Pushf\phi\in\Em\ICONES(\Cmeas(X),\Cmeas(Y))
  \]
  so that \(\Cmeas\) is a functor \(\ARCAT\to\Em\ICONES\).
\end{theorem}
\begin{proof}
  We must first prove: %
  \(\Ders_{\Cmeas(X)}\Compl\Coalgs X
  =\Id_{\Cmeas(X)}\in\ICONES(\Cmeas(X),\Cmeas(X))\).
  By Theorem~\ref{th:dirac-dense} this results from the fact that for
  all \(r\in X\) one has %
  \((\Ders_{\Cmeas(X)}\Compl\Coalgs X)(\Dirac X(r))
  =\Ders_{\Cmeas(X)}(\Proms{\Dirac X(r)})=\Dirac X(r)\).

  Next we must prove that
  \(f_1=f_2\in\ICONES(\Cmeas(X),\Excls{\Excls(\Cmeas(X))})\) where %
  \begin{align*}
    f_1=\Diggs_{\Cmeas(X)}\Compl\Coalgs X
    \text{\quad and\quad}
    f_2=\Excls{\Coalgs X}\Compl\Coalgs X\,.
  \end{align*}
  Let \(r\in X\), we have %
  \begin{align*}
    f_1(\Dirac X(r))&=\Diggs_{\Excls{\Cmeas(X)}}(\Proms{\Dirac X(r)})
                       =\Promms{\Dirac X(r)}\\
    f_2(\Dirac X(r))&=\Excls{\Coalgs X}(\Proms{\Dirac X(r)})
  =\Proms{(\Coalgs X(\Dirac X(r)))}
  \end{align*}
  by Lemma~\ref{lemma:excl-fun-prom}.
  And hence \(f_2(\Dirac X(r))=\Promms{\Dirac X(r)}\) so that %
  \(f_1=f_2\) by Theorem~\ref{th:dirac-dense}.

  Let now \(\phi\in\ARCAT(X,Y)\), we must prove that %
  \(f_1=f_2\) where \(f_1=\Coalgs Y\Compl\Pushf\phi\) and
  \(f_2=\Exclp{\Pushf\phi}\Compl\Coalgs X\).
  Let \(r\in X\), we have %
  \(f_1(\Dirac X(r))=\Coalgs Y(\Dirac Y(\phi(r))=\Proms{\Dirac
    Y(\phi(r))}\) and %
  \(f_2(\Dirac X(r))=\Excls{(\Pushf\phi)}(\Proms{\Dirac X(r)})
  =\Proms{(\Pushf\phi(\Dirac X(r)))}=\Proms{\Dirac Y(\phi(r))}\) by
  Lemma~\ref{lemma:excl-fun-prom} and so \(f_1=f_2\) by
  Theorem~\ref{th:dirac-dense}.
  This proves the second part of the theorem.
\end{proof}

\begin{remark} %
  \label{rk:linear-stable-non-integrable}
  One of the main goals in introducing integrable cones was precisely to get
  this additional structure for each cone \(\Cmeas(X)\) (notably because
  this structure is required in order to interpret call-by-value languages).
  It means more specifically that for each %
  \(f\in\STAB(\Cmeas(X),B)=\ICONES(\Excls{\Cmeas(X)},B)\) we can
  define %
  \(g=f\Compl\Coalgs X\in\ICONES(\Cmeas(X),B)\) such that
  \begin{align*}
    \forall\mu\in\Mcca{\Cmeas(X)}\quad
    g(\mu)=\int_{r\in X} f(\Dirac X(r))\mu(dr)
  \end{align*}
  which is a ``linearization'' of \(f\) allowing to interpret the
  sampling operation of probabilistic programming languages: one
  samples a \(r\in X\) (a real number if \(X=\Real\)) according to the
  distribution \(\mu\) and feeds the program \(f\) with the value
  \(r\) represented as the Dirac measure \(\Dirac X(r)\).
  This Dirac measure represents, in our semantics, the real number
  \(r\) considered as a value.
  This observation strongly supports the idea of taking the objects
  or \(\ARCAT\) (such as the real line, or the set of natural numbers)
  as our basic data-types and the measurable functions
  \(\phi\in\ARCAT(X,Y)\) as the basic functions of our programming
  language, through the functor \(\Cmeas:\ARCAT\to\ICONES\): remember
  that
  \(\Cmeas(\phi)(\Dirac X(r))=\Pushf\phi(\Dirac X(r))=\Dirac
  X(\phi(r))\).

  From the viewpoint of \(\LL\) this means that each
  \(X\in\ARCAT\) can be seen as a positive type, that is, a type
  equipped with structural rules allowing to erase and duplicate its
  values, see for
  instance~\cite{Girard91a,LaurentRegnier03,EhrhardTasson19}.
  Another way to understand \(\Coalgs X\) is to see it as a
  \emph{storage operator} in the sense of~\cite{Krivine94b}, that is,
  \(\Cmeas(X)\) is a \emph{data-type}.
  This idea will be confirmed in Section~\ref{sec:arcat-full-subcat-EM}
  where we will see that \(\Cmeas\) is a full and faithful functor
  from \(\ARCAT\) to \(\Em\ICONES\).
\end{remark}

\begin{example}
  \label{ex:sampling-programming}
  A typical probabilistic programming language that we can interpret
  in the model \(\ICONES\) (we assume that \(\Real\in\ARCAT\)) is the
  probabilistic version of PCF presented
  in~\cite{EhrhardPaganiTasson18} (to which we refer for more details
  and examples), which features continuous data-types and can be
  extended in various ways.
  Such a language could feature a type \(\rho\) of real numbers,
  a constant \(\mathsf{unif}\) such that
  \(\Tseq\Gamma{\mathsf{unif}}\rho\) corresponding to the uniform
  probability distribution on the interval \(\Intercc 01\) \Etc{}
  All the types of this language, which are given by the following
  grammar
  \begin{align*}
    \sigma,\tau\dots\Bnfeq\rho\Bnfor\Timpl\sigma\tau\Bnfor\cdots
  \end{align*}
  are interpreted as objects of \(\ICONES\):
  \(\Tsem\rho=\Cmeas(\Real)\),
  \(\Tsem{\Timpl\sigma\tau}=\Simpls{\Tsem\sigma}{\Tsem\tau}\), \Etc{}
  A term \(M\) such that \(\Tseq{\Gamma}M\tau\) where
  \(\Gamma=(x_1:\sigma_1,\dots,x_k:\sigma_k)\) is a typing context,
  will then be interpreted as a stable and measurable morphism, or as
  an analytic morphism
  \(\Psem M\Gamma\in\Kl
  \ICONES(\Tsem{\sigma_1}\IWith\cdots\IWith\Tsem{\sigma_k},\Tsem\tau)\).
  So if \(\Tseq{\Gamma,x:\rho}M\sigma\) has a free variable of type
  \(\rho\) and %
  \(\Tseq\Gamma N\rho\), we should consider that \(N\) represents a
  (sub)probability distribution \(\Psem N\Gamma\) on \(\Real\) and we
  may want to sample a real number \(r\) along this distribution and
  feed \(M\) with the resulting real value that \(M\) will use
  \emph{as many times as it wants}: this value will be represented as
  the Dirac measure \(\Dirac\Real(r)\).
  In our language, the corresponding construct is a simple
  \(\mathsf{let}\) which allows to deal with the ground type \(\rho\)
  in a call-by-value way%
  \footnote{This idea was already central
    in~\cite{EhrhardPaganiTasson14,EhrhardPaganiTasson18b,EhrhardTasson19},
    in the discrete setting of probabilistic coherence spaces.}%
  :
  \begin{align*}
    \Tseq\Gamma{\mathsf{let}(x,N,M)}\tau
  \end{align*}
  and the semantics of this construct is
  \begin{align*}
    \Psem{\mathsf{let}(x,N,M)}\Gamma=\int_{r\in\Real}^{\Tsem\tau}\Psem{M}\Gamma
    (\Dirac\Real(r))\Psem N\Gamma(dr)
  \end{align*}
  which is well defined since \(\Psem M\Gamma\) is stable and
  measurable (or analytic), the function
  \(\Absm{r\in\Real}{\Dirac\Real(r)}\) belongs to
  \(\Cuball{\Mcca{\Cpath\Real{\Cmeas(\Real)}}}\) and the cone
  \(\Tsem\tau\) is integrable.
  The constant \(\mathsf{unif}\) is interpreted as the probability
  measure on \(\Real\) which maps a measurable set \(U\) to the
  Lebesgue measure of \(U\cap\Intercc01\).
  For each \(r\in\Real\), the language has a constant \(\Num r\) of
  type \(\rho\) and
  \(\Psem{\Num r}\Gamma=\Dirac\Real(r)\in\Cuball{\Mcca{\Cmeas(\Real)}}\) (in
  each context \(\Gamma\)).
  Our language will also have constructs \(\mathsf{log}(M)\),
  \(\mathsf{sqrt}(M)\) \Etc{} corresponding to the usual functions
  which are all measurable, and typed for instance by
  \begin{center}
    \begin{prooftree}
      \hypo{\Tseq\Gamma M\rho}
      \infer1{\Tseq\Gamma{\mathsf{log}(M)}\rho}
    \end{prooftree}
  \end{center}
  with semantics given by push-forward:
  \(\Psem{\mathsf{log}(M)}\Gamma=\Pushf\log(\Psem M\Gamma)\).
  So for instance we can define a closed term \(N\) such that
  \(\Tseq {}N\rho\) by
  \begin{align*}
    N=\mathsf{let}(x,\mathsf{unif},
    \mathsf{let}(y,\mathsf{unif},
    \mathsf{mult}(
    \mathsf{sqrt}(\mathsf{mult}(\Num{-2.},\mathsf{log}(x))),
    \mathsf{cos}(\mathsf{mult}(\Num{6.28\cdots},y))
    )
    ))
  \end{align*}
  and then \(\Psem N{}\) is the normal distribution \(\cN(0,1)\)
  defined by the Box Muller method, and we can define a term \(N'\)
  with two free variable \(x\) and \(s\) for \(\cN(x,s)\) as %
  \[
    N'=\mathsf{let}(y,N,\mathsf{plus}(\mathsf{mult}(s,y),x)))
  \]
  such that %
  \(\Tseq{x:\rho,s:\rho}{N'}\rho\).
  Then
  \(\Psem{\Substbis{N'}{\Num{4.2}/x,\Num{0.7}/s}}{}
  =\Psem{N'}{x:\rho,s:\rho}(\Dirac\Real(4.2),\Dirac\Real(0.7))
  \in\Mcca{\Cmeas(\Real)}\)
  is the measure \(\cN(4.2,0.7)\).
\end{example}

\subsubsection{The category of measurable functions as a full subcategory
  of the Eilenberg Moore category}
\label{sec:arcat-full-subcat-EM}
So we have extended the operation \(\Cmeas\) on the measurable spaces
of \(\ARCAT\) into a functor \(\ARCAT\to\Em\ICONES\) which acts on
morphisms by push-forward.
This functor is clearly faithful, we prove that, under very reasonable
assumptions about \(\ARCAT\), it is also full, which is
quite a remarkable fact: the Eilenberg-Moore category of \(\oc\)
contains \(\ARCAT\) as a full subcategory.
Again, integration is an essential ingredient in the proof of this
result.

A Polish space is a complete metric space which has a countable dense
subset.

We will need two lemmas which are folklore in measure theory.
\begin{lemma}
  \label{lemma:polish-01-measure}
  Let \(X\) be a Polish space, equipped with its standard Borel
  \(\sigma\)-algebra \(\Sigalg X\).
  Let \(\mu\) be a probability measure on \(X\) and assume that
  \(\forall U\in\Sigalg X\ \mu(U)\in\Eset{0,1}\).
  Then \(\mu\) is a Dirac measure.
\end{lemma}
\begin{proof}
  Given \(r\in X\) and \(\epsilon\geq 0\) we use
  \(B(r,\epsilon)\subseteq X\) for the closed ball of radius
  \(\epsilon\) centered at \(r\).
  Let \(D\) be a countable dense subset of \(X\).
  Let \(F\subseteq X\) be closed and such that \(\mu(F)=1\) and let
  \(\epsilon>0\), we have %
  \(F\subseteq\Union_{r\in D\cap F}B(r,\epsilon)\) and hence %
  \(1=\mu(F)\leq\sum_{r\in D\cap F}\mu(B(r,\epsilon))\) and hence %
  \(\exists r\in D\cap F\ \mu(B(r,\epsilon))=1\).
  We define a sequence \((r_n)_{n\in\Nat}\) of elements of \(D\) such
  that \(\forall n\in\Nat\ \mu(B(r_n,2^{-n}))=1\) as follows.
  We obtain \(r_0\) by applying the property above with \(F=X\) and
  \(\epsilon=1\).
  We get \(r_{n+1}\) by applying the property above with
  \(F=B(r_n,2^{-n})\) and \(\epsilon=2^{-(n+1)}\).
  Then the sequence \((r_n)_{n\in\Nat}\) is Cauchy and has therefore a
  limit \(r\) and we have \(\{r\}=\bigcap_{n\in\Nat}B(r_n,2^{-n})\) so
  that \(\mu(\Eset r)=\inf_{n\in\Nat}\mu(B(r_n,2^{-n}))=1\) since
  \(\mu\) is a measure.
  It follows that \(\mu(U)=0\) for each measurable \(U\) such that
  \(r\notin U\) since we must have \(\mu(\Eset r\cup U)=1\) and hence
  \(\mu=\Dirac X(r)\).
\end{proof}

\begin{lemma}
  \label{lemma:Dirac-kernel-pushf}
  If \(X\) and \(Y\) are measurable spaces such that the
  \(\sigma\)-algebra of \(Y\) contains all singletons (this is true in
  particular if \(Y\) is a Polish space), and if \(\kappa\) is a
  kernel from \(X\) to \(Y\) such that for all \(r\in X\) the
  measure \(\kappa(r)\) is a Dirac measure on \(Y\), then there is a
  uniquely defined measurable function \(\phi:X\to Y\) such that
  \(\kappa=\Dirac Y\Comp\phi\).
\end{lemma}
This is obvious.

\begin{theorem}
  Let \(X,Y\in\ARCAT\) be such that \(Y\) is a Polish space
  and let \(f\) be a morphism from \(\Cmeas(X)\) to \(\Cmeas(Y)\) in
  \(\ICONES\).
  Then \(f\) is a coalgebra morphism from \((\Cmeas(X),\Coalga X)\) to
  \((\Cmeas(Y),\Coalga Y)\) iff there is a \(\phi\in\ARCAT(X,Y)\)
  such that \(f=\Cmeas(\phi)=\Pushf\phi\).

  As a consequence, if we assume that \(X\) is a Polish space
  for all \(X\in\ARCAT\), then \(\ARCAT\) is a full subcategory of
  the Eilenberg Moore category of the comonad \(\Excla\_\) through the
  \(\Cmeas\) functor.
\end{theorem}
Most measurable spaces which appear in probability theory are Polish
spaces: discrete spaces, the real line, countable products and
measurable subspaces of Polish spaces (and hence the Cantor Space and
the Baire Space, the Hilbert Cube \Etc{}) are Polish spaces.
So the restriction to Polish spaces is not a serious one.
\begin{proof}
  Saying that \(f\) is a coalgebra morphism means that the following
  diagram commutes in \(\ICONES\):
  \begin{equation*}
    \begin{tikzcd}
      \Cmeas(X)\ar[r,"f"]\ar[d,swap,"\Coalga X"]
      & \Cmeas(Y)\ar[d,"\Coalga Y"]\\
      \Excla{\Cmeas(X)\ar[r,"\Excla f"]} & \Excla{\Cmeas(Y)}
    \end{tikzcd}
  \end{equation*}
  which, by Theorem~\ref{th:dirac-dense}, is equivalent to
  \begin{align}
    \label{eq:meas-coalg-morph-charact}
    \Proma{(f(\Dirac X(r)))}
    =\Excla f(\Coalga X(\Dirac X(r))
    =\Coalga Y(f(\Dirac X(r))
    =\int_{s\in Y}^{\Excla{\Cmeas(Y)}} \Proma{\Dirac Y(s)}f(\Dirac X(r))(ds)
  \end{align}
  for all \(r\in X\), and this equation trivially holds if
  \(f=\Pushf\phi\).
  Assume conversely that \(f\) satisfies
  \Eqref{eq:meas-coalg-morph-charact}.
  Let \(V\) be a measurable subset of \(Y\) and let
  \(g\in\ACONES(\Cmeas(Y),\Sone)\) be defined by \(g(\nu)=\nu(V)^2\)
  and let
  \(g_0
  =\Invp{\Expadja_{\Cmeas(Y),\Sone}}(g)\in\ICONES(\Excla{\Cmeas(Y)},\Sone)\)
  which is characterized by
  \(\forall\nu\in\Mcca{\Cmeas(Y)}\ g(\nu)=g_0(\Proma\nu)\).
  We have %
  \(g_0(\Proma{f(\Dirac X(r))})=g(f(\Dirac X(r)))=f(\Dirac
  X(r))(V)^2\) and, since \(g_0\) preserves integrals,
  \begin{align*}
    f(\Dirac X(r))(V)^2
    &=g_0\Big(\int^{\Excla{\Cmeas(Y)}}_{s\in Y}
      \Proma{\Dirac Y(s)}f(\Dirac X(r))(ds)\Big)
    \text{\quad by Equation~\Eqref{eq:meas-coalg-morph-charact}}\\
    &=\int_{s\in Y} g_0(\Proma{\Dirac Y(s)})f(\Dirac X(r))(ds)\\
    &=\int_{s\in Y} {\Dirac Y(s)}(V)^2f(\Dirac X(r))(ds)\\
    &=\int_{s\in Y} {\Dirac Y(s)}(V)f(\Dirac X(r))(ds)\\
    &=f(\Dirac X(r))(V)
  \end{align*}
  so we have \(f(\Dirac X(r))(V)\in\Eset{0,1}\) for all
  \(V\in\Sigalg Y\).
  Let \(g\in\ACONES(\Cmeas(Y),\Sone)\) be defined now by \(g(\nu)=1\)
  and let
  \(g_0
  =\Invp{\Expadja_{\Cmeas(Y),\Sone}}(g)\in\ICONES(\Excla{\Cmeas(Y)},\Sone)\),
  we have %
  \(g_0(\Proma{f(\Dirac X(r))})=g(f(\Dirac X(r)))=1\) and, since
  \(g_0\) preserves integrals,
  \begin{align*}
    1=g_0\Big(\int^{\Excla{\Cmeas(Y)}}_{s\in Y}
    \Proma{\Dirac Y(s)}f(\Dirac X(r))(ds)\Big)
    &=\int_{s\in Y} g_0(\Proma{\Dirac Y(s)})f(\Dirac X(r))(ds)\\
    &=\int_{s\in Y} f(\Dirac X(r))(ds)\\
    &=f(\Dirac X(r))(Y)
  \end{align*}
  and hence the measure \(f(\Dirac X(r))\) is a Dirac measure by
  Lemma~\ref{lemma:polish-01-measure} and it follows that
  \(f=\Cmeas(\phi)=\Pushf\phi\) for a uniquely determined
  \(\phi\in\ARCAT(X,Y)\) by Lemma~\ref{lemma:Dirac-kernel-pushf}.
\end{proof}

\subsection{Fixpoint operators in the cartesian closed category} %
\label{sec:ccc-fix}
Remember that \(\cC\) is a CCC having the following property:
\begin{quote}
  The objects of \(\cC\) are integrable cones.
  In particular, for each object \(B\) of \(\cC\), the set
  \(\Cuball{\Mcca B}\) has a structure of \(\omega\)-cpo by the condition
  \Cnormcr{}, with \(0\) as least element.
  And each \(f\in\cC(B,B)\) is in particular an increasing and
  \(\omega\)-continuous function
  \(\Cuball{\Mcca B}\to\Cuball{\Mcca B}\) and therefore has a least
  fixpoint which is \(\sup_{n=0}^\infty f^n(0)\in\Cuball{\Mcca B}\).
\end{quote}
\noindent 
It is completely standard to apply this property to the map %
\(\cZ\in\cC(\Impl{(\Impl BB)}{B},\Impl{(\Impl BB)}{B})\) given by
\begin{align*}
  \cZ(F)(f)=f(F(f))
\end{align*}
which is well-defined and belongs to
\(\cC(\Impl{(\Impl BB)}{B},\Impl{(\Impl BB)}{B})\) by cartesian
closedness of \(\cC\).
The least fixpoint \(\cY\) of \(\cZ\) is an element of
\(\cC(\Impl BB,B)\) which is easily seen to satisfy
\begin{align*}
  \cY(f)=\sup_{n=0}^\infty f^n(0)
\end{align*}
and is therefore a least fixpoint operator that we have proven here to
be a morphism in \(\cC\), that is, a stable and measurable or an
analytic map depending on the considered category \(\cC\). This
morphism \(\cY\) is the key ingredient to interpret recursively
defined functional programs in the CCC \(\cC\).

\begin{example}
  The function \(f:\Intercc01\to\Intercc01\) given by
  \(f(x)=\frac 12+\frac14 x^2\) belongs to \(\cC(\Sone,\Sone)\) so it
  has a least fixpoint \(x\in\Intercc01\) which must satisfy
  \(x^2-4x+2=0\) and is therefore \(2-{\sqrt 2}\).
  Of course, when \(a\in\Realp\) and \(a>0\), the function
  \(x\mapsto x+a\) from \(\Realp\) to \(\Realp\) has no fixpoint, but
  this is not a contradiction because it does not restrict to a
  function \(\Intercc01\to\Intercc01\).
\end{example}

\section{Probabilistic coherence spaces as integrable cones}
\label{sec:pcs-integrable}
So far we have seen several ways of building integrable cones: as
spaces of measures, or of paths, as products and tensor products, as
spaces of analytic maps \Etc{}
As announced in Example~\ref{ex:cones-analytic-functions} we describe
here another source of integrable cones: the probabilistic coherence
spaces.
Intuitively, they form a model of \(\LL\) based on
\emph{discrete} but not necessarily finite probabilities.
So their definition does not require measure theory.

We use \(\Realpc\) for the completed real half-line, that is
\(\Realpc=\Real\cup\Eset\infty\), considered as a semi-ring with
multiplication satisfying \(0\,\infty=0\), which is the only possible
choice since we want multiplication to be \(\omega\)-continuous.

Let \(I\) be a set. If \(i\in I\) we use \(\Base i\) for the element
of \(\Realpto I\) such that \(\Base i_j=\Kronecker ij\).

If \(\cP\subseteq\Realpcto I\) we define
\(\Orth\cP\subseteq\Realpcto I\) by
\begin{align*}
  \Orth\cP=\{x'\in\Realpto I\St\forall x\in\cP\ \sum_{i\in I}x_ix'_i\leq 1\}
\end{align*}
and we use the notation \(\Eval x{x'}=\sum_{i\in I}x_ix'_i\).
As usual we have \(\cP\subseteq\cQ\Implies\Orth\cQ\Implies\Orth\cP\)
and \(\cP\subseteq\Biorth\cP\), and as a consequence
\(\Orth\cP=\Triorth\cP\).
In other words, it is equivalent to say that \(\cP=\Biorth\cP\) or to
say that \(\cP=\Orth\cQ\) for some \(\cQ\).

\begin{theorem}
  Let \(\cP\subseteq\Realpcto I\), one has \(\cP=\Biorth\cP\) iff the
  following conditions hold
  \begin{itemize}
  \item \(\cP\) is convex (that is, if \(x,y\in\cP\) and
    \(\lambda\in\Intcc01\) then \(\lambda x+(1-\lambda)y\in\cP\))
  \item \(\cP\) is down-closed for the product order
  \item and, for each sequence \((x(n))_{n\in\Nat}\) of element of
    \(\cP\) which is increasing for the pointwise order, the pointwise
    lub \(\in\Realpcto I\) of this sequence belongs to \(\cP\).
  \end{itemize}
\end{theorem}
\noindent 
A proof is outlined in~\cite{Girard04a} and a complete proof can be
found in~\cite{Ehrhard22a}.

\begin{definition} %
  \label{def:pcs}
  A probabilistic coherence space (PCS) is a pair %
  \(\cX=(\Web\cX,\Pcoh\cX)\) where \(\Web\cX\) is a set which is at most
  countable%
  \footnote{This countability assumption is crucial in the present
    setting, again because of our use of the monotone convergence
    theorem.} %
  and \(\Pcoh\cX\subseteq\Realpto{\Web\cX}\) satisfies
  \begin{itemize}
  \item \(\Pcoh\cX=\Biorth{\Pcoh\cX}\)
  \item for all \(a\in\Web\cX\) there is \(x\in\Pcoh\cX\) such that
    \(x_a>0\)
  \item and for all \(a\in\Web\cX\) the set
    \(\{x_a\St x\in\Pcoh\cX\}\subseteq\Realp\) is bounded.
  \end{itemize}
\end{definition}
\noindent 
The 2nd and 3rd conditions are required to keep the coefficients
finite and are dual of each other.

Given two sets \(I,J\), a vector \(u\in\Realpcto I\) and a matrix
\(w\in\Realpcto{I\times J}\), we define \(\Matappa wu\in\Realpcto J\) by
\begin{align*}
  \Matappa wu=\Big(\sum_{i\in I}w_{i,j}u_i\Big)_{j\in J}
\end{align*}
and then, given \(v\in\Realpcto J\), observe that
\begin{align*}
  \Eval{\Matappa wu}{v}=\Eval w{\Tenspcs uv}
  =\sum_{i\in I,j\in J}w_{i,j}u_iv_j\in\Realpc
\end{align*}
where \(\Tenspcs uv\in\Realpcto{I\times J}\) is defined by
\((\Tenspcs uv)_{i,j}=u_iv_j\) (we use this notation here to avoid
confusions with the tensor operations we have introduced for integrable
cones).

Given \(w_1\in\Realpcto{I_1\times I_2}\) and
\(w_2\in\Realpcto{I_2\times I_3}\), one defines
\(w_2\Compl w_1\in\Realpcto{I_1\times I_3}\) (product of matrices
written in reversed order) by
\begin{align*}
  (w_2\Compl w_1)_{i_1,i_3}
  =\sum_{i_2\in I_2}(w_1)_{i_1,i_2}(w_2)_{i_2,i_3}\,.
\end{align*}

It is easily checked that, given PCSs \(\cX\) and \(\cY\), one defines a 
PCS \(\Limpl \cX\cY\) by %
\(\Web{\Limpl \cX\cY}=\Web \cX\times\Web \cY\) and
\begin{align*}
  \Pcohp{\Limpl \cX\cY}
  =\{t\in\Realpto{\Web{\Limpl \cX\cY}}
  \St\forall x\in\Pcoh \cX\ \Matappa tx\in\Pcoh \cY\}\,.
\end{align*}
Indeed one can check that
\begin{align*}
  \Pcohp{\Limpl \cX\cY}
  =\Orth{\{\Tenspcs x{y'}\St x\in\Pcoh \cX\text{ and }y'\in\Orth{\Pcoh \cX}\}}\,.
\end{align*}
Then given \(s\in\Pcohp{\Limpl \cX\cY}\) and \(t\in\Pcohp{\Limpl \cY\cZ}\) one has
\begin{align*}
  t\Compl s\in\Pcoh{\Limplp \cX\cZ}
\end{align*}
and the diagonal matrix
\(\Id=(\Kronecker{a}{a'})_{(a,a')\in\Web{\Limpl \cX\cX}}\) belongs to
\(\Pcohp{\Limpl \cX\cX}\).
This defines the category \(\PCOH\) of probabilistic coherence spaces.

Let \(t\in\PCOH(\cX,\cY)\). We use
\(\Lfun t:\Pcoh \cX\to\Pcoh \cY\) for the function defined by %
\(\Lfun t(x)=\Matappa tx\).

The orthogonal (or linear negation) \(\Orth \cX\) of a PCS \(\cX\) is
defined by \(\Web{\Orth \cX}=\Web \cX\) and
\(\Pcohp{\Orth \cX}=\Orth{(\Pcoh \cX)}\) so that \(\Biorth \cX=\cX\). We use
\(\Sbot\) for the PCS such that \(\Web{\Sbot}=\{\Sonelem\}\) and
\(\Pcoh\Sbot=\Intercc01\).
Setting \(\Sone=\Orth\Sbot\) we have obviously \(\Sone=\Sbot\) and
\(\Orth \cX\) is trivially isomorphic to \(\Limpl \cX\Sbot\).
Under this iso, the function \(\Lfun{x'}:\Pcoh \cX\to\Intercc01\)
associated with \(x'\in\Pcoh{\Orth \cX}\) is given by
\(\Lfun{x'}(x)=\Eval x{x'}\).

This linear negation is a functor \(\Op\PCOH\to\PCOH\), mapping
\(t\in\PCOH(\cX,\cY)\) to it transpose \(\Orth t\) defined by
\((\Orth t)_{b,a}=t_{a,b}\).

Each PCS \(\cX\) induces a measurable cone \(\Pcohc \cX\) defined by %
\(\Mcca{\Pcohc \cX} =\{\lambda x\St x\in\Pcoh \cX\text{ and
}\lambda\in\Realp\}\) with algebraic operations defined in the obvious
pointwise way.
Notice that if \(t\in\PCOH(\cX,\cY)\) we can extend \(\Lfun t\) to a
function \(\Mcca{\Pcohc \cX}\to\Mcca{\Pcohc \cY}\) by setting
\(\Lfun t(x)=\Inv\lambda\Lfun t(\lambda x)\) for each \(\lambda>1\)
such that \(\lambda x\in\Pcoh \cX\) (the function does not depend on the
choice of \(\lambda\)).

The norm of this cone is defined by
\begin{align*}
  \Norm x_{\Pcohc \cX}
  =\sup_{x'\in\Cuball{\Cdual{\Mcca{\Pcohc\cX}}}}\Eval x{x'}
  =\inf\{\lambda>0\St x\in\lambda\Pcoh \cX\}
\end{align*}
so that \(\Cuball{\Mcca{\Pcohc \cX}}=\Pcoh \cX\).
The measurability structure of \(\Pcohc \cX\) is given by
\(\Mcms{\Pcohc \cX}_\Measterm=\{\Lfun{x'}\St x'\in\Pcoh{\Orth \cX}\}\)
and %
\(\Mcms{\Pcohc \cX}_X\) is the set of all constant functions from
\(X\in\ARCAT\) to \(\Mcms{\Pcohc \cX}_\Measterm\).

\begin{lemma}
  \label{lemma:PCS-norm-predual}
  Let \(\cX\) be a PCS and let \(\cP\subseteq\Realpto{\Web \cX}\) be such
  that \(\Pcoh \cX=\Orth\cP\).
  Then
  \begin{align*}
    \Norm x_{\Pcohc \cX}=\sup_{x'\in\cP}\Eval x{x'}\,.
  \end{align*}
\end{lemma}
\noindent 
The proof is easy.
Notice that for each \(a\in\Web \cX\) one has
\(\Base a\in\Mcca{\Pcohc \cX}\) by the second condition in the
definition of a PCS, and that \(\Norm{\Base a}\) is not necessarily
equal to \(1\).

\begin{lemma} %
  \label{lemma:pcohc-path-charact}
  Let \(\cX\) be a PCS and \(X\in\ARCAT\).
  A function %
  \(\beta:X\to\Mcca{\Pcohc \cX}\) is a measurable path of the
  measurable cone \(\Pcohc \cX\) iff %
  \(\beta(X)\) is bounded and, for all \(a\in\Web \cX\), the
  function \(\Absm{r\in X}{\beta(r)_a}:X\to\Realp\) is
  measurable.
\end{lemma}
\begin{proof}
  The \(\Rightarrow\) direction results from the observation that, for
  each \(a\in\Web \cX\) there is a \(\lambda>0\) such that
  \(\lambda\Base a\in\Pcoh{\Orth \cX}\).
  For the \(\Leftarrow\) direction let \(\beta:X\to\Pcoh \cX\)
  be such that the function \(\Absm{r\in X}{\beta(r)_a}\) is
  measurable for all \(a\in\Web \cX\).
  Let \(x'\in\Pcoh{\Orth \cX}\), we must prove that
  \(\phi=\Absm{r\in X}{\Eval{\beta(r)}{x'}}\) is measurable.
  Since \(\Web \cX\) is countable, this results from the monotone
  convergence theorem and from the fact that
  \begin{align*}
    \phi(r)&=\sum_{a\in\Web \cX}\beta(r)_ax'_a\,.
    \qedhere
  \end{align*}
\end{proof}

\begin{theorem}
  For each PCS \(\cX\) the measurable cone \(\Pcohc \cX\) is integrable.
\end{theorem}
\begin{proof}
  Let \(\beta:X\to\Pcoh \cX\) be a measurable path (by
  Lemma~\ref{lemma:pcohc-path-charact} this is equivalent to saying
  that \(\beta_a=\Absm{r\in X}{\beta(r)_a}\) is measurable
  \(X\to\Realp\) for all \(a\in\Web \cX\) since
  \(\forall r\in X\ \Norm{\beta(r)}\leq 1\)).
  Let \(\mu\in\Mcca{\Cmeas(X)}\).
  We define \(x\in\Realpto{\Web \cX}\) by
  \begin{align*}
    x_a=\int\beta_a(r)\mu(dr)
  \end{align*}
  which is a well defined element of \(\Realp\) since the function
  \(\beta_a\) is bounded by definition of a PCS.
  Let \(x'\in\Pcoh{\Orth \cX}\), we have, applying the monotone
  convergence theorem,
  \begin{align*}
    \Eval x{x'}
    &=\sum_{a\in\Web \cX}\Big(\int\beta_a(r)\mu(dr)\Big)x'_a\\
    &=\sum_{a\in\Web \cX}\int\left(\beta_a(r))x'_a\right)\mu(dr)\\
    &=\int\Eval{\beta(r)}{x'}\mu(dr)\leq\Norm\mu
  \end{align*}
  so if \(\lambda>0\) is such that \(\lambda\Norm\mu\leq 1\) we get %
  \(\Eval{\lambda x}{x'}\leq 1\) for all \(x'\in\Pcoh{\Orth \cX}\) so
  that \(x\in\Mcca{\Pcohc \cX}\).
  The equation \(\Eval x{x'}=\int\Eval{\beta(r)}{x'}\mu(dr)\) which
  holds for all \(x'\in\Pcoh{\Orth \cX}\) shows that \(x\) is the
  integral of \(\beta\) over \(\mu\) by definition of
  \(\Mcms{\Pcohc \cX}\).
\end{proof}

\begin{theorem} %
  \label{th:lfun-pcoh-full-faithful}
  If \(t\in\Pcoh{\Limplp \cX\cY}\) then
  \(\Lfun t\in\ICONES(\Pcohc \cX,\Pcohc \cY)\) and extended to morphisms
  in that way, the operation \(\Pcohcn\) is a full and faithful
  functor \(\PCOH\to\ICONES\).
\end{theorem}
\begin{proof}
  The fact that \(\Lfun t\) is linear and continuous is easy (the
  proof can be found in~\cite{DanosEhrhard08} for
  instance).
  Measurability and integral preservation of \(\Lfun t\) boil down
  again to the monotone convergence theorem. Faithfulness results from
  the fact that \(t\) is completely determined by the action of
  \(\Lfun t\) on the elements \(\Base a\) of \(\Mcca{\Pcohc \cX}\) (for
  all \(a\in\Web \cX\); remember that indeed
  \(\forall a\in\Web \cX\ \Base a\in\Mcca{\Pcohc \cX}\)). Last let
  \(f\in\ICONES(\Pcohc \cX,\Pcohc \cY)\).
  We define \(t\in\Realpto{\Web{\Limpl \cX\cY}}\) by
  \(t_{a,b}=f(\Base a)_b\).
  Given \(x\in\Pcoh \cX\) and \(y'\in\Pcoh{\Orth \cY}\) we have
  \begin{align*}
    \Eval{\Matappa tx}{y'}
    &=\sum_{a\in\Web \cX,b\in\Web \cY}t_{a,b}x_ay'_b\\
    &=\sum_{a\in\Web \cX,b\in\Web \cY}f(\Base a)_bx_ay'_b\\
    &=\sum_{b\in\Web \cY}f(x)_by'_b
      \text{\quad by linearity and continuity of }f\\
    &=\Eval{f(x)}{y'}\leq 1
  \end{align*}
  since \(\Norm f\leq 1\), which shows that \(\Matappa tx\in\Pcoh \cY\)
  and hence \(t\in\PCOH(\cX,\cY)\).
  The equation \(\Eval{\Matappa tx}{y'}=\Eval{f(x)}{y'}\) for all
  \(y'\in\Pcoh{\Orth \cY}\) also shows that \(\Matappa tx=f(x)\) and
  hence the functor \(\Pcohcn\) is full.
\end{proof}
\noindent 
We use \(\Tenspcs \cX\cY\) for the tensor product operation in \(\PCOH\),
that is \(\Web{\Tenspcs \cX\cY}=\Web \cX\times\Web \cY\) and
\(\Pcohp{\Tenspcs \cX\cY}=\Biorth{\{\Tenspcs xy\St x\in\Pcoh \cX\text{ and
  }y\in\Pcoh \cY\}}=\Orth{\Limplp{\cX}{\Orth \cY}}\).

\begin{theorem} %
  \label{th:limpl-pcoh-icones-isom}
  If \(\cX,\cY\) are PCSs then \(\Lfunn\) is an iso from the integrable
  cone \(\Pcohc{\Limpl \cX\cY}\) to the integrable cone
  \(\Limpl{\Pcohc \cX}{\Pcohc \cY}\) in \(\ICONES\), and this iso is
  natural in \(\cX\) and \(\cY\).
\end{theorem}
\begin{proof}[Proof sketch]
  We know by Theorem~\ref{th:lfun-pcoh-full-faithful} that \(\Lfunn\)
  is an iso of cones.
  We need to prove that \(\Lfunn\) and \(\Inv\Lfunn\) are measurable
  and that \(\Lfunn\) preserve integrals (then \(\Inv\Lfunn\) also
  preserves integrals by injectivity of \(\Lfunn\)).

  Let \(X\in\ARCAT\) and
  \(\eta\in\Mcca{\Cpath X{\Pcohc{\Limpl \cX\cY}}}\), we show that
  \[
    \Lfunn\Comp\eta\in\Mcca{\Cpath X{\Limpl{\Pcohc \cX}{\Pcohc \cY}}}
  \]
  so let \(Y\in\ARCAT\), \(\beta\in\Mcca{\Cpath Y{\Pcohc \cX}}\) and
  \(m\in\Mcms{\Pcohc \cY}_Y\) meaning that \(m=\Lfun{y'}\) for some
  \(y'\in\Pcoh{\Orth \cY}\), we have, for all \((s,r)\in Y\times X\),
  \begin{align*}
    (\Mtfun\beta m)(s,\Lfun{\eta(r)})
    &=\Eval{\Lfun{\eta(r)}(\beta(s))}{y'}\\
    &=\Eval{\Matappa{\eta(r)}{\beta(s)}}{y'}\\
    &=\sum_{(a,b)\in\Web \cX\times\Web \cY}\eta(r)_{a,b}\beta(s)_ay'_b
  \end{align*}
  and the function
  \(\Absm{(s,r)\in Y\times X}{(\Mtfun\beta m)(s,\Lfun{\eta(r)})}\) is
  measurable as a countable sum of measurable functions.
  Conversely let now
  \(\eta\in\Mcca{\Cpath X{\Limpl{\Pcohc \cX}{\Pcohc \cY}}}\), we must
  prove that %
  \(\Inv\Lfunn\Comp\eta\in\Mcca{\Cpath X{\Pcohc{\Limpl \cX\cY}}}\) so
  let %
  \(Y\in\ARCAT\) and \(p\in\Mcms{\Pcohc{\Limpl \cX\cY}}_Y\), that is %
  \(p=\Lfun z\) for some
  \(z\in\Pcoh{\Orth{\Limplp \cX\cY}}=\Pcohp{\Tenspcs \cX{\Orth \cY}}\), we
  have
  \begin{align*}
    \Absm{(s,r)\in Y\times X}{p(s,\Inv{\Lfunn}(\eta(r)))}
    &=\Absm{(s,r)\in Y\times X}{\Eval{z}{\Inv{\Lfunn}(\eta(r))}}\\
    &=\Absm{(s,r)\in Y\times X}
      {\sum_{(a,b)\in\Web \cX\times\Web \cY}z_{a,b}\eta(r)(\Base a)_b}
  \end{align*}
  which is measurable as a countable sum of measurable functions since
  we know that for all \(a,b\) the function %
  \(\Absm{r\in X}{\eta(\Base a)_b}\) is measurable by our
  assumption that \(\eta\) is a measurable path.

  The fact that \(\Lfunn\) preserves integrals results from the
  pointwise definition of integration in
  \(\Limpl{\Pcohc \cX}{\Pcohc \cY}\).
\end{proof}

\begin{theorem}
  \label{th:iso-pcohc-tensor}
  There is a natural isomorphism %
  \(\phi_{\cX,\cY}
  \in\ICONES(\Tens{\Pcohc \cX}{\Pcohc \cY},\Pcohc{\Tenspcs \cX\cY})\).
\end{theorem}
\begin{proof}[Proof sketch]
  The map %
  \(\Absm{(x,y)\in\Mcca{\Pcohc \cX}\times\Mcca{\Pcohc \cY}}{\Tenspcs xy}\)
  is easily seen to be bilinear, \(\omega\)-continuous, measurable and
  separately integrable so that we have an associated %
  \(\phi_{\cX,\cY} %
  \in\ICONES(\Tens{\Pcohc \cX}{\Pcohc \cY},\Pcohc{\Tenspcs \cX\cY})\)
  characterized by \(\phi_{\cX,\cY}(\Tens xy)=\Tenspcs xy\).
  We define now
  \(\psi_{\cX,\cY}:\Pcohc{\Tenspcs \cX\cY}\to\Tens{\Pcohc \cX}{\Pcohc
    \cY}\).
  First, given \((a,b)\in\Web{\Tenspcs \cX\cY}\) we set
  \(\psi_{\cX,\cY}(\Base{a,b})=\Tens{\Base a}{\Base b}\).
  Next given \(z\in\Mcca{\Pcohc{\Tenspcs \cX\cY}}\) such that
  \(\Supp z=\Eset{(a,b)\in\Web{\Tenspcs \cX\cY}\St z_{(a,b)}\not=0}\) is
  finite we set
  \(\psi_{\cX,\cY}(z) =\sum_{(a,b)\in\Web{\Tenspcs
      \cX\cY}}z_{(a,b)}\Tens{\Base a}{\Base b}\) which is a well defined
  finite sum in the cone \(\Mcca{\Tens{\Pcohc \cX}{\Pcohc \cY}}\).
  We contend that
  \begin{align*}
    \Norm{\psi_{\cX,\cY}(z)}_{\Tens{\Pcohc \cX}{\Pcohc \cY}}
    \leq\Norm{z}_{\Pcohc{\Tenspcs \cX\cY}}
  \end{align*}
  so let \(\epsilon>0\) and assume without loss of generality that
  \(\Norm z\leq 1\).
  By Proposition~\ref{th:norm-dual} there is
  \(g\in\Cuball{\Mcca{\Limpl{\Tens{\Pcohc \cX}{\Pcohc \cY}}{\Sbot}}}\)
  such that
  \(\Norm{\psi_{\cX,\cY}(z)}_{\Tens{\Pcohc \cX}{\Pcohc \cY}}\leq
  g(\psi_{\cX,\cY}(z))+\epsilon\).
  Let %
  \(h\in\Cuball{\Mcca{\Limplp{\Pcohc \cX}{\Limplp{\Pcohc
          \cY}{\Sbot}}}}\) %
  be the the bilinear morphism associated to \(g\) by the iso of
  Theorem~\ref{th:icones-tens-limpl-isom}.
  We have
  \begin{align*}
    g(\psi_{\cX,\cY}(z))
    &=\sum_{(a,b)\in\Web{\Tens \cX\cY}}z_{(a,b)}g(\Tens{\Base a}{\Base b})\\
    &=\sum_{(a,b)\in\Web{\Tens \cX\cY}}z_{(a,b)}h(\Base a,\Base b)
    \leq 1
  \end{align*}
  because
  \((h(\Base a,\Base b))_{(a,b)\in\Web{\Tens \cX\cY}}\in\Pcoh{\Orth{(\Tenspcs
      \cX\cY)}}\) by Theorem~\ref{th:limpl-pcoh-icones-isom} and by our
  assumption that \(\Norm z\leq 1\).
  So we have
  \(\Norm{\psi_{\cX,\cY}(z)}_{\Tens{\Pcohc \cX}{\Pcohc \cY}}\leq 1+\epsilon\)
  and since this holds for all \(\epsilon>0\) our contention is
  proven. Now let \(z\) be any element of
  \(\Mcca{\Pcohc{\Tenspcs \cX\cY}}\) and assume again that
  \(\Norm z\leq 1\). Let \((I_n)_{n\in\Nat}\) be an increasing sequence
  of finite sets such that \(\Union I_n=\Web \cX\times\Web \cY\) and let
  \(z(n)\in\Mcca{\Pcohc{\Tenspcs \cX\cY}}\) be defined by
  \begin{align*}
    z(n)_{(a,b)}=
    \begin{cases}
      z_{(a,b)}&\text{if }(a,b)\in I_n\\
      0&\text{otherwise.}
    \end{cases}
  \end{align*}
  so that the sequence \((z(n))_{n\in\Nat}\) is increasing and has \(z\)
  as lub in \(\Mcca{\Pcohc{\Tenspcs \cX\cY}}\).
  The sequence \((\psi_{\cX,\cY}(z(n)))_{n\in\Nat}\) is increasing and all
  its elements have norm \(\leq 1\) since each \(z(n)\) has finite
  support and norm \(\leq 1\) and hence it has a lub in
  \(\Mcca{\Tens{\Pcohc \cX}{\Pcohc \cY}}\).
  It is easy to check that this lub does not depend on the choice of
  the \(I_n\)'s, so we can set
  \(\psi_{\cX,\cY}(z)=\sup_{n\in\Nat}\psi_{\cX,\cY}(z(n))\) so that actually
  \begin{align*}
    \psi_{\cX,\cY}(z)
    =\sum_{(a,b)\in\Web{\cX}\times\Web \cY}z_{a,b}\Tens{\Base a}{\Base b}\,.
  \end{align*}
  The proof that
  \(\psi_{\cX,\cY}\in\ICONES(\Pcohc{(\Tenspcs \cX\cY)},\Tens{\Pcohc \cX}{\Pcohc
    \cY})\) follows the standard pattern and it is obvious that it is
  the inverse of \(\phi_{\cX,\cY}\).
\end{proof}

\subsection{More constructions}
We outline very briefly the additive and exponential constructions on
PCSs.
The categorical product \(\cX=\Bwith_{i\in I}\cX_i\) of a family
\((\cX_i)_{i\in I}\) of PCSs can be described by %
\(\Web\cX=\Union_{i\in I}\{i\}\times\Web{\cX_i}\) and
\(x\in\Realpto{\Web\cX}\) belongs to \(\Pcoh\cX\) if
\(\forall i\in I\ \Matappa{\Proj i}x\in\Pcoh{\cX_i}\) where %
\(\Proj i\in\Realpto{\Web{\cX}\times\Web{\cX_i}}\) is given by %
\((\Proj i)_{(j,a),b}=\Kronecker ij\Kronecker ab\), so that
\(\Proj i\in\Pcoh(\cX,\cX_i)\) for each \(i\in I\).
Then it is easy to check that \((\cX,(\Proj i)_{i\in I})\) is the
categorical product of the family \((\cX_i)_{i\in I}\) in \(\PCOH\) and
that there is a natural isomorphism from
\(\Bwith_{i\in I}{\Pcohc{\cX_i}}\) to \(\Pcohc{\cX}\).

For the coproduct \(\cY=\Bplus_{i\in I}\cX_i\) we can take %
\(\cY=\Orth{(\Bwith_{i\in I}\Orth{\cX_i})}\) so that %
\(\Web\cY=\Web\cX\), and %
\(\Pcoh\cY=\{x\in\Pcoh\cX\St\sum_{i\in I}\Norm{\Matappa{\Proj i}x}\leq
1\}\), equipped with injections
\(\Inj i=\Orth{\Proj i}\in\PCOH(\cX_i,\cY)\).
So for instance the coproduct \(\Plus\Sone\Sone\) has \(\{0,1\}\) as
web, and \(\Pcohp{\Plus\Sone\Sone}=\{u\in\Realp^2\St u_0+u_1\leq 1\}\).

For the exponential, we introduce the following notations.
Given \(u\in\Realpto I\) we define \(\Prom u\in\Realpto{\Mfin I}\) by
\(\Prom u_m=u^m=\prod_{i\in I}u_i^{m(i)}\).
Given a PCS \(\cX\) we define a PCS \(\Excl\cX\) by %
\(\Web{\Excl\cX}=\Mfin{\Web{\cX}}\) and %
\(\Pcohp{\Excl\cX}=\Biorth{\{\Prom x\St x\in\Pcoh\cX \}}\) so that %
\(t\in\PCOH(\Excl\cX,\cY)\) means exactly that %
\(t\in\Realpto{\Mfin{\Web\cX}\times\Web\cY}\) and %
\begin{align*}
  \Fun t(x)=\Big(\sum_{m\in\Mfin{\Web\cX}}
  t_{m,b}x^m\Big)_{b\in\Web\cY}\in\Pcoh\cY
\end{align*}
from which it is not hard to derive that the integrable cones %
\(\Pcohc{\Limpl{\Excl\cX}{\cY}}\) and %
\(\Simpla{\Pcohc\cX}{\Pcohc\cY}\) are isomorphic.

\begin{remark}
  \label{rk:pushf-EM}
  There is a morphism %
  \(\phi_\cX\in\ICONES(\Exclana{\Pcohc \cX},\Pcohc{\Excl \cX})\) for
  all PCSs \(\cX\) such that \(\phi_\cX(\Promana x)=\Prom x\) for all
  \(x\in\Pcoh \cX\). This morphism is similar to the one of
  Theorem~\ref{th:iso-pcohc-tensor}, we conjecture that it is an iso.
\end{remark}

\subsection{Example: the Cantor Space as an equalizer of \(\PCOH\) morphisms}
Since \(\ICONES\) is a complete category, the equalizer of two
parallel \(\PCOH\) morphisms is an integrable cone.
As we shall see now, this cone needs not be a PCS which means that,
contrarily to the larger category \(\ICONES\), the category \(\PCOH\)
is not complete.
This example also shows that interesting ``non discrete'' cones arise
as limits of diagrams in \(\PCOH\).

Consider the PCS \(\cS\) such that \(\Web\cS=\Eset{0,1}^{<\omega}\)
is the set of finite sequences of \(0\)'s and \(1\)'s and where
\(x\in\Realpto{\Web\cS}\) belongs to \(\Pcoh\cS\) if, for each
\(u\subseteq\Web \cS\) which is an antichain (meaning that if
\(s,s'\in u\) then \(s\leq s'\Implies s=s'\) where \(\leq\) is the
prefix order), one has \(\sum_{s\in u}x_s\leq 1\).
Since \(\Pcoh\cS=\Orth A\) where \(A\) is the set of all characteristic
functions of antichains, \(\cS\) is a PCS (the second and third
conditions of Definition~\ref{def:pcs} result from the observation
that each singleton is an antichain).
Notice that
\begin{align}
  \label{eq:norm-eq-cantor}
  \Norm x_\cS=\sup_{x'\in A}\Eval{x}{x'}
\end{align}
by Lemma~\ref{lemma:PCS-norm-predual}.

The PCS \(\cS\) is the ``least solution'' (in the sense explained
in~\cite{DanosEhrhard08,EhrhardTasson19}) of the equation
\(\cS=\With\Sone{(\Plus \cS\cS)}\).

There is a morphism \(\theta\in\PCOH(\cS,\cS)\) which is given by
\begin{align*}
  \theta_{s,t}=
  \begin{cases}
    1 & \text{if }s=ta\text{ for some }a\in\Eset{0,1}\\
    0 & \text{otherwise}
  \end{cases}
\end{align*}
where we use simple juxtaposition for concatenation.
Indeed given an antichain \(u\) and \(x\in\Pcoh\cS\) we have
\begin{align*}
  \sum_{t\in u}(\Matappa\theta x)_t=\sum_{s\in v}x_s\leq 1
\end{align*}
where \(v=\Eset{sa\St s\in u\text{ and }a\in\Eset{0,1}}\) is an
antichain since \(u\) is an antichain. Let \(C\) be the integrable
cone which is the equalizer of \(\theta\) and \(\Id_S\), considered as
morphisms of \(\ICONES\) through the full and faithful functor
\(\Pcohcn\).

\begin{theorem}
  The integrable cone \(C\) is isomorphic to \(\Cmeas(\Cantor)\) where
  \(\Cantor\) is the Cantor Space equipped with the Borel sets of its
  usual topology (the product topology of \(\Eset{0,1}^\omega\) where
  \(\Eset{0,1}\) has the discrete topology).
\end{theorem}
\begin{proof}
  We have
  \begin{align*}
    \Mcca{C}=\{x\in\Pcohc\cS\St \Matappa\theta x=x\}\,,
  \end{align*}
  that is, an element of \(\Mcca{C}\) is an \(x\in\Realpto{\Web\cS}\) %
  such that \(x\in\Pcoh\cS\) and
  \begin{align*}
    \forall s\in\Web\cS\quad x_s=x_{s0}+x_{s1}\,.
  \end{align*}
  Given \(s\in\Web\cS\) we set
  \(\Cantbase s=\{\alpha\in\Cantor\St s<\alpha\}\subseteq\Cantor\),
  which is a clopen of \(\Cantor\).
  Let \(U\) be an open subset of \(\Cantor\), the set
  \(\Cantopenbase U\) of all \(s\in\Web\cS\) which are minimal (for
  the prefix order) such that \(\Cantbase s\subseteq U\) is an
  antichain, and we have
  \begin{align} %
    \label{eq:cantor-open-clopens}
    U=\Union\{\Cantbase s\St s\in\Cantopenbase U\}
  \end{align}
  by definition of the topology of \(\Cantor\).
  Given \(x\in\Mcca{\Pcohc\cS}\) we define a function
  \(\Cantormeas(x):\Opens\Cantor\to\Realp\) on the open sets of
  \(\Cantor\) by
  \begin{align*}
    \Cantormeas(x)(U)=\sum_{s\in\Cantopenbase U}x_s
  \end{align*}
  and we have \(\Cantormeas(x)(U)\leq\Norm x\) by our assumption that
  \(x\in\Mcca{\Pcohc \cS}\).
  This function \(\Cantormeas(x)\) is additive (that is
  \(\Cantormeas(x)(\Union_{i\in I}U_i)=\sum_{i\in
    I}\Cantormeas(x)(U_i)\) for each countable family
  \((U_i)_{i\in I}\) of pairwise disjoint open subsets of
  \(\Cantor\)).
  And so \(\Cantormeas(x)\) extends to a uniquely defined finite
  measure on the Borel sets of the Cantor Space, that is to an element
  of \(\Mcca{\Cmeas(\Cantor)}\).
  Notice that \(\Cantormeas:\Mcca{\Pcohc\cS}\to\Cmeas(\Cantor)\) is
  linear and satisfies
  \(\Norm{\Cantormeas(x)}=x_{\Seqempty}\leq\Norm x\) where
  \(\Seqempty\in\Web{\cS}\) is the empty sequence.

  Let \(\mu\in\Mcca{\Cmeas(\Cantor)}\), we define
  \(\Cantorvect(\mu)\in\Realpto{\Web\cS}\) by
  \(\Cantorvect(\mu)_s=\mu(\Cantbase s)\). Given an antichain
  \(u\subseteq\Web\cS\) notice that the clopens
  \((\Cantbase s)_{s\in u}\) are pairwise disjoint and that
  \(U=\Union_{s\in u}\Cantbase s\) is open and hence measurable, so,
  since \(\mu\) is a measure, we have
  \begin{align*}
    \sum_{s\in u}\Cantorvect(\mu)_s
    =\sum_{s\in u}\mu(\Cantbase s)
    =\mu\big(\Union_{s\in u}\Cantbase s\big)
    =\mu(U)\leq\mu(\Cantor)\,.
  \end{align*}
  Since this holds for each antichain \(u\) we have shown that
  \(\Cantorvect(\mu)\in\Pcoh\cS\).
  Notice that for each \(s\in\Web S\) we have
  \(\Cantbase s=\Cantbase{s0}\cup\Cantbase{s1}\) and that this union
  is disjoint, so that
  \(\mu(\Cantbase s)=\mu(\Cantbase s0)+\mu(\Cantbase s1)\) since
  \(\mu\) is a measure, that is
  \(\Cantorvect(\mu)\in\Mcca{C}\).
  Notice also that the function \(\Cantorvect\) is linear and
  satisfies \(\Cantorvect(\mu)\leq\Norm\mu\) since for each antichain
  \(u\subseteq\Web\cS\) one has
  \(\sum_{s\in u}\Cantorvect(\mu)_s=\sum_{s\in u}\mu(\Cantbase
  s)=\mu(\Union_{s\in u}\Cantbase s)\leq\mu(\Cantor)=\Norm\mu\).

  We prove that the functions \(\Cantormeas\) and \(\Cantorvect\) are
  inverse of each other. Let first \(x\in\Mcca{C}\), we have, for all
  \(s\in\Web\cS\),
  \begin{align*}
    \Cantorvect(\Cantormeas(x))_s
    =\Cantormeas(x)(\Cantbase s)
    =x_s
  \end{align*}
  since \(\Cantopenbase{\Cantbase s}=\Eset s\). Let now
  \(\mu\in\Cmeas(\Cantor)\) and let \(U\in\Opens\Cantor\) we have
  \begin{align*}
    \Cantormeas(\Cantorvect(\mu))(U)
    =\sum_{s\in\Cantopenbase U}\Cantorvect(\mu)_s
    =\sum_{s\in\Cantopenbase U}\mu(\Cantbase s)
    =\mu(U)
  \end{align*}
  by Formula~\Eqref{eq:cantor-open-clopens}. It follows that
  \(\Cantormeas(\Cantorvect(\mu))=\mu\).

  It follows that \(\Cantormeas\) and \(\Cantorvect\) define an order
  isomorphism between \(\Mcca{\Cmeas(\Cantor)}\) and \(\Mcca C\) and
  therefore are \(\omega\)-continuous (this uses also the fact that
  \(\Norm{\Cantormeas(x)}=\Norm x\) since
  \(\Norm x=\Norm{\Cantorvect(\Cantormeas(x))}\leq\Norm{\Cantormeas(x)}\)
  and similarly \(\Norm{\Cantorvect(\mu)}=\Norm\mu\)).
  
  The fact that the map
  \(\Cantormeas:\Mcca{\Pcohc\cS}\to\Mcca{\Cmeas(\Cantor)}\) is linear
  is measurable and integrable results as usual from the monotone
  convergence theorem.
  So we have \(\Cantormeas\in\ICONES(\Pcohc\cS,\Cmeas(\Cantor))\) and
  hence by restriction
  \(\Cantormeas\in\ICONES(\Mcca C,\Cmeas(\Cantor))\) since
  \(\Norm{\Cantormeas(x)}=\Cantormeas(x)(\Cantor)\leq\Norm x\) for all
  \(x\in\Mcca C\).
  Again, checking that \(\Cantorvect\) is
  measurable and integrable is routine; as an example let us prove the
  last property so let \(X\in\ARCAT\) and let
  \(\kappa\in\Mcca{\Cpath X{\Cmeas(\Cantor)}}\).
  Let \(m\in\Mcms C_\Measterm\), that is \(m=\Lfun{x'}\) for some
  \(x'\in\Pcoh{\Orth\cS}\).
  We have
  \begin{align*}
    m\Big(\int^C\Cantorvect(\kappa(r))\mu(dr)\Big)
    &=\sum_{s\in\Web\cS}x'_s
      \Big(\int^{\Pcohc S}\Cantorvect(\kappa(r))\mu(dr)\Big)_s\\
    &=\sum_{s\in\Web\cS}x'_s\int\Cantorvect(\kappa(r))_s\mu(dr)\\
    &=\int\big(\sum_{s\in\Web\cS}x'_s\kappa(r)(\Cantbase s)\big)\mu(dr)\\
    &=\int m(\Cantorvect(\kappa(r)))\mu(dr)\,.
  \end{align*}
  By Formula~\Eqref{eq:norm-eq-cantor} we have
  \(\Norm{\Cantorvect}\leq 1\) and hence
  \(\Cantorvect\in\ICONES(\Cmeas(\Cantor),C)\).
\end{proof}

\section*{Conclusion}
Elaborating on earlier work by the first author (together with
Michele~Pagani and Christine~Tasson) on a denotational semantics based
on measurable cones and by the second author on a notion of
convex QBS where integration is the fundamental algebraic
operation~\cite{Geoffroy22},
we have developed a theory of integration for measurable cones,
introducing the category of \emph{integrable cones} and of linear
morphisms preserving integrals.
We have shown that this category is a model of Intuitionistic \(\LL\)
featuring two exponential comonads; for defining the tensor
product and the exponentials we have used the special adjoint functor
theorem which avoids providing explicit combinatorial constructions of
these objects.

The construction is parameterized by a small full subcategory
\(\ARCAT\) of the category of measurable spaces and measurable
functions.
The model obtained in that way has many pleasant properties.
\begin{itemize}
\item It contains the category of probabilistic coherence
  spaces as a full subcategory.
\item It contains the category whose objects are those of \(\ARCAT\)
  and whose morphisms are the substochastic kernels as a full
  subcategory.
\item For both exponentials, the associated Eilenberg Moore category
  contains \(\ARCAT\) as a full subcategory, if we assume that all the
  objects of \(\ARCAT\) are standard Borel spaces which btw.~is a very
  natural and harmless requirement.
\end{itemize}
The two latter properties strongly rely on the fact that the morphisms
of the underlying linear category preserve integrals.
The last one means that \(\ARCAT\) can be considered as a category of
basic data-types (the objects of \(\ARCAT\)) and basic operations on
them (the morphisms of \(\ARCAT\)).

In future work we will explain how this model can be used for
interpreting call-by-value or even call-by-push-value probabilistic
functional programming languages with continuous data-types
(interpreted as the aforementioned coalgebras) as well as recursive
types.

\section*{Acknowledgments}
This work has been partly funded by the ANR PRC project Probabilistic
Programming Semantics (PPS), ANR-19-CE48-0014.

We would like to thank warmly the reviewers for their extremely
careful reading of the paper and many useful comments and
suggestions. The paper owes in particular very much to one of the two
reviewers who suggested many mathematical improvements and
simplifications, as well as the counter-example of
Remark~\ref{rk:cone-non-separe}. We could implement most of them and
it is clear that if the article has reached a reasonable level of
readability, this is mainly thanks to their crucial contributions.


\bibliographystyle{alphaurl}
\bibliography{newbiblio.bib}

\end{document}